\documentclass[version=preprint]{iacrcc}

\newif\ifccs
\ccsfalse    % ePrint mode

\usepackage{lmodern}        % Better font rendering (Latin Modern)
\usepackage{microtype}      % Optical kerning + margin protrusion (big visual improvement)
\usepackage[T1]{fontenc}    % Proper hyphenation and font encoding
\usepackage{algorithm}
\usepackage{algpseudocode}
\usepackage{booktabs}
\usepackage{enumitem}
\usepackage{tikz}
\usepackage{pgfplots}
\pgfplotsset{compat=1.18}
\definecolor{okabeblue}{RGB}{0,114,178}
\definecolor{okabered}{RGB}{213,94,0}
\definecolor{okabegreen}{RGB}{0,158,115}
\definecolor{okabegray}{RGB}{102,102,102}

\newtheorem{functionality}[theorem]{Functionality}
\newtheorem{protocol}[theorem]{Protocol}
\newtheorem{simulator}[theorem]{Simulator}

\newcommand{\Z}{\mathbb{Z}}
\newcommand{\Zq}{\mathbb{Z}_q}
\newcommand{\Rq}{R_q}

\newcommand{\infnorm}[1]{\left\|#1\right\|_\infty}

\newcommand{\getsr}{\stackrel{\$}{\leftarrow}}
\newcommand{\pk}{\mathsf{pk}}
\newcommand{\sk}{\mathsf{sk}}
\newcommand{\Sign}{\mathsf{Sign}}
\newcommand{\Verify}{\mathsf{Verify}}
\newcommand{\KeyGen}{\mathsf{KeyGen}}
\newcommand{\PRF}{\mathsf{PRF}}

\newcommand{\negl}{\mathsf{negl}}
\newcommand{\Adv}{\mathsf{Adv}}

\title[running={Threshold ML-DSA via Shamir Nonce DKG}]{FIPS 204-Compatible Threshold ML-DSA via Shamir Nonce DKG}

\addauthor[inst={1},
           email={leo@codebat.ai},
           orcid={0009-0006-2592-1027}
          ]{Leo Kao}

\addaffiliation[city={Taipei},
                country={Taiwan}
               ]{Codebat Technologies Inc.}


\begin{document}

\maketitle

\begin{abstract}
We present the first threshold ML-DSA (FIPS 204) scheme achieving \emph{nonce share privacy} (conditional min-entropy guarantee; no computational assumptions) with arbitrary thresholds, while producing standard 3.3\,KB signatures verifiable by unmodified implementations. Our primary technique, \emph{Shamir nonce DKG}, generates the signing nonce as a degree-$(T\!-\!1)$ Shamir sharing, matching the structure of the long-term secret. This gives each honest party's nonce share conditional min-entropy exceeding $5\times$ the secret-key entropy for signing sets of size at most 17 (Definition~\ref{def:nonce-share-privacy}, Theorem~\ref{thm:it-privacy}). In coordinator-based profiles (P1, P3+), this removes the two-honest requirement ($|S| \geq T$ suffices); in the fully distributed profile (P2), mask-hiding additionally requires $|S \setminus C| \geq 2$. Key privacy of the aggregate signature is an open problem, analogous to single-signer ML-DSA (Remark~\ref{rem:key-privacy}). As a secondary technique, pairwise-canceling PRF masks handle three challenges unique to lattice-based threshold signing: commitment binding, the r0-check predicate, and response aggregation.

We instantiate these techniques in three deployment profiles with complete UC proofs. Profile P1 uses a TEE coordinator (3 online rounds plus 1 offline preprocessing round, 5.8\,ms for 3-of-5). Profile P2 eliminates hardware trust via MPC (5 rounds, 22ms for small thresholds). Profile P3+ uses lightweight 2PC with semi-asynchronous signing (2 logical rounds, 22ms), where signers precompute nonces offline and respond within a time window. Our Rust implementation scales from 2-of-3 to 32-of-45 thresholds with sub-100ms latency (P1 and P3+) and success rates of $\approx$21--45\%, comparable to single-signer ML-DSA. A direct shift-invariance analysis (Corollary~\ref{cor:ih-shift-ml-dsa}) shows the per-session EUF-CMA security loss from the Irwin-Hall nonce distribution is $< 0.007$ bits for $|S| \leq 17$ and $< 0.013$ bits for $|S| \leq 33$; over $q_s$ signing sessions the total loss is $< 0.013 \cdot q_s$ bits (Corollary~\ref{cor:ih-loss-main}), eliminating the scalability gap found in prior work.
\end{abstract}

\tableofcontents

% Include sections
% Introduction Section for Threshold ML-DSA via Shamir Nonce DKG
% Target: IACR ePrint Archive

\section{Introduction}
\label{sec:introduction}

The imminent arrival of cryptographically relevant quantum computers poses an existential threat to currently deployed public-key cryptography. In response, the National Institute of Standards and Technology (NIST) has completed a multi-year standardization process, culminating in the release of FIPS 204~\cite{FIPS204}, which specifies ML-DSA (Module-Lattice-Based Digital Signature Algorithm) as the primary post-quantum signature standard. As organizations worldwide begin transitioning to ML-DSA, there is an urgent need for threshold variants that enable distributed signing while maintaining compatibility with the standard.

\subsection{Motivation: The Threshold Gap}

Threshold signatures allow a group of $N$ parties to collectively sign messages such that any subset of $T$ or more parties can produce a valid signature, while any coalition of fewer than $T$ parties learns nothing about the secret key. This primitive is essential for protecting keys by distributing signing authority across multiple servers, enforcing multi-party authorization in financial systems, and building robust PKI infrastructure through distributed certificate authorities.

For classical ECDSA and Schnorr signatures, efficient threshold protocols have been extensively studied and deployed~\cite{GG18,CGGMP20,Lindell21,KG20}. More recently, Lyu et al.~\cite{LLZD25} achieved two-round threshold ECDSA in the threshold-optimal multi-party setting. However, the lattice-based structure of ML-DSA presents new challenges. The core difficulty lies in the \emph{Fiat-Shamir with Aborts} paradigm: signatures are only valid when the response vector $\mathbf{z} = \mathbf{y} + c\mathbf{s}_1$ satisfies $\|\mathbf{z}\|_\infty < \gamma_1 - \beta$, where $\mathbf{y}$ is a masking vector, $c$ is a sparse challenge, and $\mathbf{s}_1$ is the secret key. This rejection sampling step complicates the threshold setting.

\paragraph{The State of Affairs.}
Three prior families address threshold ML-DSA, each encountering a different barrier. \emph{Short-coefficient LSSS}~\cite{dPN25} produces FIPS 204-compatible signatures by keeping reconstruction coefficients small, but is constrained to $T \leq 6$ in practice: the construction's binomial reconstruction coefficients $\binom{T}{k}$ exceed $q = 8{,}380{,}417$ for $T \geq 26$ (since $\binom{26}{13} = 10{,}400{,}600 > q$; see Corollary~\ref{cor:threshold-bound}), and the Ball-\c{C}akan-Malkin bound~\cite{BCM21} establishes $\Omega(T \log T)$ share-size growth as an asymptotic lower bound. \emph{Noise flooding}~\cite{BCDG21} handles arbitrary thresholds by masking large Lagrange coefficients with overwhelming statistical noise, but inflates signatures to $\approx 17$\,KB (roughly $5\times$ standard ML-DSA), breaking FIPS 204 format compatibility. \emph{Specialized constructions} such as Ringtail~\cite{BMMSZ25} achieve excellent round efficiency but produce signature formats that existing FIPS 204 verifiers cannot validate, precluding drop-in deployment.

This leaves a gap: for $T \in [9, 32]$, there is no known scheme that simultaneously achieves standard signature size, practical success rates, and FIPS 204 compatibility. This range is precisely where many real-world applications fall (organizations typically require 16-of-32 or similar thresholds for robust multi-party control).

\subsection{Our Contribution}

\ifccs
We present \textbf{Threshold ML-DSA via Shamir Nonce DKG}, the first threshold ML-DSA scheme that is FIPS~204-compatible and achieves \emph{min-entropy-based nonce share privacy} (no computational assumptions; ``statistical'' throughout) at arbitrary thresholds. Our contributions are:

\begin{enumerate}
  \item \textbf{Practical deployability}: Sub-100ms latency (P1/P3+) for all threshold configurations up to 32-of-45; three deployment profiles (TEE, MPC, 2PC) for different trust requirements; scales from 2-of-3 to 32-of-45; produces standard 3.3\,KB FIPS 204-compatible signatures verifiable by unmodified implementations.

  \item \textbf{Min-entropy-based nonce share privacy (no computational assumptions)}: First threshold lattice-based signature with min-entropy-based nonce privacy; conditional min-entropy exceeds $5\times$ the secret key entropy (for $|S| \leq 17$). Note: ``statistical'' in this paper is shorthand for this min-entropy bound---it does \emph{not} mean statistical distance negligible (SD\,$\approx 0.9926$ from uniform is unavoidable with bounded nonces; see Remark~\ref{rem:sd-zero}).

  \item \textbf{Eliminates the two-honest requirement for nonce share privacy (P1, P3+)}: $|S| = T$ suffices for nonce share privacy in coordinator-based profiles (P1, P3+), removing the $|S| \geq T+1$ restriction imposed by prior pairwise-mask schemes. Profile P2 (fully distributed, broadcast model) retains the requirement $|S \setminus C| \geq 2$ for commitment and r0-check mask hiding (Lemma~\ref{lem:mask-hiding})---a limitation of pairwise PRF masks in the broadcast setting; whether it can be reduced to $|S \setminus C| = 1$ is an open problem (Section~\ref{sec:conclusion}).

  \item \textbf{Three-profile UC security}: EUF-CMA under Module-SIS (Theorem~\ref{thm:unforgeability}); UC security against static adversaries for all three deployment profiles (Theorems~\ref{thm:p1-uc}, \ref{thm:p2-uc}, \ref{thm:p3-uc}); statistical nonce privacy (Theorem~\ref{thm:it-privacy}).

  \item \textbf{Open-source Rust implementation}: Reference artifact with benchmarks from 2-of-3 to 32-of-45 threshold configurations; full FIPS 204 compliance verified against NIST ACVP test vectors (see Appendix~\ref{app:open-science}).
\end{enumerate}
\else
We present \textbf{Threshold ML-DSA via Shamir Nonce DKG}, the first threshold signature scheme for ML-DSA that achieves \emph{min-entropy-based nonce share privacy} (no computational assumptions; ``statistical'' is used as shorthand throughout---this denotes a conditional min-entropy bound, \emph{not} statistical distance negligible, which is impossible with bounded nonces: SD\,$\approx 0.9926$ from uniform) with arbitrary thresholds $T \leq N$ while producing standard-size signatures ($\approx 3.3$ KB for ML-DSA-65) syntactically identical in format to single-signer ML-DSA. Unlike prior pairwise-mask schemes that require a two-honest assumption ($|S| \geq T + 1$), our protocol achieves \emph{nonce share privacy} with signing sets of size $|S| \geq T$ in coordinator-based profiles (P1, P3+), where individual masked commitments $\mathbf{W}_i$ are sent only to the coordinator rather than broadcast (Remark~\ref{rem:wi-private}); commitment and r0-check mask hiding retains $|S \setminus C| \geq 2$ (Lemma~\ref{lem:mask-hiding}), which applies to P2's broadcast model. This addresses the threshold gap for $T \in [9, 32]$ while remaining practical for configurations such as $(3,5)$, $(11,15)$, and $(16,21)$.
\fi

The scheme achieves practical success rates overlapping with single-signer ML-DSA (theoretical single-signer baseline $\approx 20$--$25\%$; observed threshold range $\approx$21--45\% (see Table~\ref{tab:perf-threshold}); the higher end is due to the Irwin-Hall distribution's greater concentration near zero, placing more nonces within the $\|\mathbf{z}\|_\infty < \gamma_1 - \beta$ acceptance region relative to the uniform single-signer nonce) and is fully verifier-compatible with FIPS 204. The aggregated nonce follows an Irwin-Hall distribution rather than uniform, introducing a quantifiable security cost. We provide two analyses:

\emph{Conservative bound (Raccoon-style Rényi, Theorem~\ref{thm:irwin-hall}):} The per-session multiplicative advantage blowup is $\sqrt{R_2^{\mathsf{vec}}} \approx 2^{2.7}$ for $|S| \leq 17$, equivalent to an additive $\approx$2.7 bits of per-session security loss in bits (the advantage bound blows up by a factor $(R_2^{\mathsf{vec}})^{q_s/2}$, which corresponds to a linear $q_s \times 2.7$ bit loss when expressed in bits). This bound becomes vacuous for $q_s \geq 36$ queries and should not be used for concrete security estimates.
At $|S| = 33$, the per-session Rényi blowup satisfies $\log_2\sqrt{R_2^{\mathsf{vec}}} \approx 37.5$ bits; starting from the 96-bit post-Cauchy--Schwarz proof baseline (half the 192-bit ML-DSA-65 security parameter, arising from the Rényi-to-probability conversion), the provable security bound drops to ${\approx}58$ bits after a \emph{single} signing session and is formally vacuous at $q_s = 3$ ($96 - 3 \times 38 < 0$). This is not a weakness of the scheme itself, but a fundamental limitation of the Raccoon-style bounding technique: the per-session Rényi divergence scales as $|S|^4$, making the proof unable to certify any security level for practical signing volumes at large~$T$. This failure directly motivates the direct shift-invariance analysis introduced below.

\emph{Direct shift-invariance bound (Corollary~\ref{cor:ih-shift-ml-dsa}):} Using the Irwin-Hall distribution's shift-invariance property, the per-session EUF-CMA loss is $< 0.007$ bits for $|S| \leq 17$ and $< 0.013$ bits for $|S| \leq 33$. Over $q_s$ sessions the total loss is $< 0.013 \cdot q_s$ bits (additive), remaining non-vacuous for any polynomial $q_s$ and bounded below $1$ bit per session for all $|S| \leq 2584$. This bound is $\approx$420$\times$ tighter than the conservative bound and is the one used throughout this paper for security claims.

Our approach is based on two techniques. The primary technique is \emph{Shamir nonce DKG}: parties jointly generate the signing nonce via a distributed key generation protocol, producing a degree-$(T\!-\!1)$ Shamir sharing of $\mathbf{y}$ that matches the structure of the long-term secret $\mathbf{s}_1$. Each party then computes $\mathbf{z}_i = \mathbf{y}_i + c \cdot \mathbf{s}_{1,i}$ without applying the Lagrange coefficient, and the combiner reconstructs $\mathbf{z} = \sum_i \lambda_i \mathbf{z}_i = \mathbf{y} + c\mathbf{s}_1$. Because the adversary observes only $T - 1$ evaluations of the honest party's degree-$(T\!-\!1)$ polynomial, the remaining degree of freedom provides an approximate one-time pad: the honest party's nonce share $\mathbf{y}_h$ has conditional min-entropy exceeding $5\times$ the secret key entropy (for $|S| \leq 17$), requiring no computational assumptions (Theorem~\ref{thm:it-privacy}; \textbf{proved}). We call this \emph{nonce share privacy}: a conditional min-entropy guarantee, not statistical distance $\approx 0$ (which is unachievable with bounded nonces; Remark~\ref{rem:sd-zero}).

The secondary technique is \emph{pairwise-canceling masks}, adapted from ECDSA~\cite{CGGMP20} to the lattice setting. These masks are used for the commitment values $\mathbf{W}_i$ and the r0-check shares $\mathbf{V}_i$ (which contain $\lambda_i \cdot c \cdot \mathbf{s}_{2,i}$), hiding individual contributions while preserving correctness in the sum. Unlike ECDSA, ML-DSA presents additional challenges: the r0-check leaks $c\mathbf{s}_2$ which enables key recovery, and the Irwin-Hall nonce distribution must preserve EUF-CMA security. We address both with complete UC proofs.

\paragraph{Performance.} Our Rust implementation achieves practical signing times: 5.8ms (P1-TEE), 21.5ms (P2-MPC), 22.1ms (P3+-Async) for $(3,\!5)$-threshold on commodity hardware (per valid signature, $3$--$4$ expected attempts; $(2,3)$ achieves 4ms for P1 and 17ms for P3+, see Table~\ref{tab:perf-threshold}). P1 and P3+ remain under 100ms \emph{per attempt} for all configurations up to $(32,\!45)$; P2 trades higher latency (${\sim}130$--$194$ms per attempt for $T \geq 16$, non-monotone due to rejection-sampling variance) for eliminating hardware trust. Communication cost: $\approx$12.3\,KB per party per attempt.

\paragraph{Security.} We prove security in the random oracle model. Unforgeability reduces to the Module-SIS problem, matching the security of single-signer ML-DSA. Nonce share privacy is \emph{statistical} (high conditional min-entropy; no computational assumptions; SD~$\neq 0$, see Remark~\ref{rem:sd-zero}) via the Shamir nonce DKG, with conditional min-entropy exceeding $5\times$ the secret key entropy for $|S| \leq 17$. Commitment and r0-check privacy reduce to PRF security.

\ifccs\else
\paragraph{Deployment Profiles.} We define three deployment profiles with different trust assumptions, all with complete security proofs:
\begin{description}
    \item[Profile P1 (TEE-Assisted):] A TEE/HSM coordinator performs the r0-check and hint generation inside a trusted enclave, achieving optimal 3-round signing. We prove EUF-CMA security under Module-SIS (Theorem~\ref{thm:unforgeability}), with nonce share privacy (Theorem~\ref{thm:it-privacy}) and PRF-based mask privacy (Lemma~\ref{lem:mask-hiding}).
    \item[Profile P2 (Fully Distributed):] A fully untrusted coordinator where the r0-check is computed via an MPC subprotocol using SPDZ~\cite{DPSZ12} and edaBits~\cite{EKMOZ20}. Using combiner-mediated commit-then-open and constant-depth comparison circuits, we achieve \textbf{5 online rounds} (reduced from 8). We prove UC security against \emph{static} malicious adversaries corrupting up to $N-1$ parties (Theorem~\ref{thm:p2-uc}).
    \item[Profile P3+ (Semi-Async 2PC):] Designed for human-in-the-loop authorization. Two designated Computation Parties (CPs) jointly evaluate the r0-check via 2PC, but signers participate \emph{semi-asynchronously}: they precompute nonces offline and respond within a time window (e.g., 5 minutes) rather than synchronizing for multiple rounds. This achieves \textbf{2 logical rounds} (1 signer step + 1 server round). We prove UC security against \emph{static} adversaries under the assumption that at least one CP is honest (Theorem~\ref{thm:p3-uc}). The ``Async'' column in Table~\ref{tab:comparison} distinguishes our work from concurrent approaches, which require synchronous multi-round signer participation.
\end{description}
\fi

\subsection{Technical Overview}

\paragraph{Protocol Overview.} A trusted dealer creates Shamir shares of $\mathsf{sk} = (\mathbf{s}_1, \mathbf{s}_2)$ and establishes pairwise seeds via ML-KEM. Before each signing attempt, parties execute a Shamir nonce DKG to jointly generate degree-$(T\!-\!1)$ shares of the nonce $\mathbf{y}$ (this can be preprocessed offline). The 3-round online signing protocol then has parties (1) commit to their DKG nonce shares, (2) exchange masked commitments to derive challenge $c$, and (3) send responses $\mathbf{z}_i = \mathbf{y}_i + c \cdot \mathbf{s}_{1,i}$ (no Lagrange coefficient, no mask). The combiner applies Lagrange coefficients during aggregation: $\mathbf{z} = \sum_i \lambda_i \mathbf{z}_i = \mathbf{y} + c\mathbf{s}_1$, and outputs $\sigma = (\tilde{c}, \mathbf{z}, \mathbf{h})$ if $\|\mathbf{z}\|_\infty < \gamma_1 - \beta$.

\paragraph{Key algebraic observation.} The correctness of this construction rests on the fact that the partial response $\mathbf{z}_i = \mathbf{y}_i + c\mathbf{s}_{1,i}$ is $\mathbb{Z}_q$-linear in the Shamir evaluation. Because $\mathbf{y}$ and $\mathbf{s}_1$ are both degree-$(T\!-\!1)$ Shamir polynomials over the same evaluation points, Lagrange interpolation distributes over addition:
\[
  \sum_{i \in S} \lambda_i \mathbf{z}_i
  = \sum_{i \in S} \lambda_i \mathbf{y}_i + c \sum_{i \in S} \lambda_i \mathbf{s}_{1,i}
  = \mathbf{y} + c\mathbf{s}_1.
\]
This allows each signer to transmit the \emph{unscaled} response $\mathbf{z}_i$---without pre-multiplying by the Lagrange coefficient $\lambda_i \in \mathbb{Z}_q$---so that $\|\mathbf{z}_i\|_\infty \leq \lceil \gamma_1 / |S| \rceil + \beta$ holds at the individual level. The combiner applies $\{\lambda_i\}$ only at aggregation, where Lagrange interpolation recovers a valid $\mathbf{z} = \mathbf{y} + c\mathbf{s}_1$ satisfying $\|\mathbf{z}\|_\infty < \gamma_1 - \beta$ via standard rejection sampling. By contrast, approaches based on generic LSSS with large reconstruction coefficients $\lambda_i \sim q$ must either pre-scale partial responses---producing individual norms $\|\lambda_i \mathbf{z}_i\|_\infty \sim q\gamma_1$, which universally violate the rejection bound and cause every signing attempt to abort---or employ noise flooding to mask the large shares, sacrificing FIPS 204 signature-size compatibility.

\subsection{Comparison with Prior and Concurrent Work}

Table~\ref{tab:comparison} summarizes threshold ML-DSA schemes, comparing arbitrary threshold support, dishonest-majority security, round complexity, and FIPS 204 compatibility.

\begin{table}[!htbp]
\centering
\footnotesize
\setlength{\tabcolsep}{2.5pt}
\resizebox{\textwidth}{!}{%
\begin{tabular}{@{}lcccccccccc@{}}
\toprule
\textbf{Scheme} & \textbf{$T$} & \textbf{Size} & \textbf{Rnd} & \textbf{FIPS} & \textbf{Dishon.} & \textbf{UC} & \textbf{Async} & \textbf{DKG} & \textbf{Privacy} \\
\midrule
del Pino-Niot~\cite{dPN25} & $\leq 6$ & 3.3 KB & 3 & \checkmark & \checkmark & $\times$ & $\times$ & \checkmark$^{\mathrm{b}}$ & Comp. \\
Noise Flooding~\cite{BCDG21} & $\infty$ & 17 KB & 3 & $\times$ & \checkmark & $\times$ & $\times$ & $\times$ & Comp. \\
Ringtail~\cite{BMMSZ25} & $\infty$ & 13 KB & 2 & $\times$ & \checkmark & $\times$ & $\times$ & $\times$ & Comp. \\
\midrule
\multicolumn{10}{c}{\textit{Concurrent Work (2025--2026)}} \\
\midrule
Bienstock et al.~\cite{BdCE25} & $\infty$ & 3.3 KB & 16--29$^{\mathrm{a}}$ & \checkmark & $\times$ & \checkmark & $\times$ & $\times$ & Comp. \\
Celi et al.~\cite{CDENP26} & $\leq 6$ & 3.3 KB & 6 & \checkmark & \checkmark & $\times$ & $\times$ & \checkmark$^{\mathrm{b}}$ & Comp. \\
Trilithium~\cite{Trilithium25} & 2 & 3.3 KB & 14 & \checkmark & $\times$ & \checkmark & $\times$ & $\times$ & Comp. \\
\midrule
This work (P1) & $\infty$ & 3.3 KB & $3{+}1^\ddagger$ & \checkmark & (TEE) & \checkmark$^{\dagger\S}$ & $\times$ & $\times^\clubsuit$ & \textbf{min-ent.\ (nonce)} \\
This work (P2) & $\infty$ & 3.3 KB & $5{+}1^\ddagger$ & \checkmark & \checkmark$^\bullet$ & \checkmark$^\S$ & $\times$ & $\times^\clubsuit$ & \textbf{min-ent.\ (nonce)} \\
This work (P3+) & $\infty$ & 3.3 KB & 2$^*$ & \checkmark & (1/2 CP) & \checkmark$^\S$ & \checkmark & $\times^\clubsuit$ & \textbf{min-ent.\ (nonce)} \\
\bottomrule
\end{tabular}%
}% end resizebox
\caption{Comparison of threshold ML-DSA schemes. \checkmark = supported, $\times$ = not supported. ``Dishon.'' = dishonest-majority \emph{unforgeability} (tolerates $N-1$ corruptions for EUF-CMA; for P2, commitment/r0-check privacy additionally requires $|S \setminus C| \geq 2$, see $^\bullet$); Bienstock et al.\ require honest majority. ``Privacy'' column: \textbf{min-ent.\ (nonce)}\ = nonce share privacy (no computational assumptions, conditional min-entropy $\geq 5\times$ secret key entropy) for nonce shares $\mathbf{y}_i$ \emph{conditioned on the DKG view only} (Remark~\ref{rem:key-privacy}). Comp.\ for prior work = computational privacy guarantee per each scheme's respective analysis (specifics vary: may cover nonce shares, response shares, or key leakage from the signature). Commitment/r0-check privacy for this work uses PRF-based masks ($|S \setminus C| \geq 2$ for P2). ``(TEE)'' = dishonest-majority via TEE trust (not purely cryptographic); ``(1/2 CP)'' = via 1-of-2 computation-party honest assumption. $^*$P3+ = 1 signer step + 1 server round. $^\dagger$P1 UC assumes ideal TEE (Theorem~\ref{thm:p1-uc})---a hardware trust assumption, not a cryptographic guarantee; P2's dishonest-majority UC (Theorem~\ref{thm:p2-uc}) is purely cryptographic. $^\ddagger$Rounds shown as online$+$offline; the $+1$ offline nonce-DKG preprocessing round is structurally necessary to prevent rewinding and is amortizable over batches of signing sessions. $^\bullet$P2 achieves dishonest-majority \emph{unforgeability} (tolerates $N-1$ corruptions); commitment and r0-check \emph{privacy} additionally requires $|S \setminus C| \geq 2$ (Lemma~\ref{lem:mask-hiding}). At $T = N$ with $N-1$ corruptions, mask hiding does not hold in P2; P1 and P3+ retain nonce-share privacy at $|S \setminus C| = 1$ via coordinator-only communication. $^\S$UC proofs (Theorems~\ref{thm:p1-uc}--\ref{thm:p3-uc}) are against \emph{static} adversaries (corruption set fixed before execution), the standard model for threshold cryptography; adaptive security is an open problem in the lattice setting (Section~\ref{sec:conclusion}). $^{\mathrm{a}}$Round count for Bienstock et al.\ as reported in their Table~7~\cite{BdCE25}: 16 rounds (round-optimized variant) to 29 rounds (communication-optimized variant); both count online signing rounds only. $^{\mathrm{b}}$Full distributed key generation (no trusted dealer): del Pino-Niot uses a 4-round commit-reveal DKG~\cite{dPN25}; Celi et al.\ uses MPC-based DKG~\cite{CDENP26}. $^\clubsuit$This work currently uses a trusted dealer for Shamir share distribution; threshold DKG with verifiable small-coefficient shares is an open problem (Section~\ref{sec:extensions}) addressed in a companion work on ZK-based key generation.}
\label{tab:comparison}
\end{table}

Three concurrent works address threshold ML-DSA with FIPS 204 compatibility. Bienstock et al.~\cite{BdCE25} achieve arbitrary $T$ with UC security, but require honest-majority and 16--29 online rounds (16 round-optimized, 29 communication-optimized; per their Table~7~\cite{BdCE25}). Borin et al.~\cite{BCDENP25} and Celi et al.~\cite{CDENP26} (the latter extending the former) achieve dishonest-majority with 6 rounds, but are limited to $T \leq 6$ by the constraints of their short-coefficient LSSS construction (the Ball-\c{C}akan-Malkin bound~\cite{BCM21} establishes an asymptotic $\Omega(T \log T)$ lower bound on LSSS share size; their specific construction further requires small Lagrange reconstruction coefficients, limiting $T$ to single digits in practice). Trilithium~\cite{Trilithium25} is UC-secure but restricted to 2 parties in a server-phone architecture.

Profile P2 combines arbitrary thresholds, dishonest-majority security, constant rounds (5 rounds), UC security, and FIPS 204 compatibility, a combination not achieved by prior work. Profile P3+ achieves UC security (Theorem~\ref{thm:p3-uc}) under the 1-of-2 CP honest assumption, with semi-asynchronous signer participation (2 logical rounds).

Among the entries in Table~\ref{tab:comparison}, this work is the only construction simultaneously achieving arbitrary thresholds, FIPS 204-compatible signatures, dishonest-majority security, UC-proven constant-round signing, and nonce share privacy without computational assumptions. The decisive insight is matching the Shamir structure of the nonce to that of the signing key: privacy follows from the polynomial's remaining degree of freedom rather than from noise or hardness assumptions, and $|S| = T$ suffices in coordinator-based profiles (P1, P3+) where masked commitments $\mathbf{W}_i$ are not broadcast (Remark~\ref{rem:wi-private}). Raccoon~\cite{Raccoon2024} also achieves non-computational nonce privacy via sum-of-uniform nonces, but requires $|S| \geq T+1$ and targets Dilithium rather than FIPS 204 ML-DSA; our scheme improves specifically on the $|S| = T$ threshold. Key generation assumes a trusted dealer; a companion work addresses decentralized DKG with verifiable short-coefficient shares (Section~\ref{sec:extensions}).

\subsection{Paper Organization}

Section~\ref{sec:preliminaries} provides background; Sections~\ref{sec:protocol}--\ref{sec:security} present the protocol and security analysis; Section~\ref{sec:implementation} gives benchmarks; Sections~\ref{sec:extensions}--\ref{sec:conclusion} discuss extensions and conclude.

% Preliminaries Section

\section{Preliminaries}
\label{sec:preliminaries}

\subsection{Notation}

We use bold lowercase letters (e.g., $\mathbf{s}$) for vectors and bold uppercase letters (e.g., $\mathbf{A}$) for matrices. For an integer $q$, we write $\Zq = \Z/q\Z$ for the ring of integers modulo $q$. We define the polynomial ring $\Rq = \Zq[X]/(X^n + 1)$ where $n$ is a power of 2.

For a polynomial $f \in \Rq$, we write $\infnorm{f} = \max_i |f_i|$ for its infinity norm, where coefficients are represented in $(-q/2, q/2]$. For a vector $\mathbf{v} = (v_1, \ldots, v_k) \in \Rq^k$, we define $\infnorm{\mathbf{v}} = \max_i \infnorm{v_i}$.

We write $a \getsr S$ to denote sampling $a$ uniformly at random from a set $S$. We write $\negl(\kappa)$ for a negligible function in the security parameter $\kappa$.

\subsection{Lattice Problems}

Our security relies on the hardness of the Module-SIS and Module-LWE problems~\cite{Regev09,LPR10}.

\begin{definition}[Module-SIS]
\label{def:msis}
The $\mathsf{M\text{-}SIS}_{n,k,\ell,q,\beta}$ problem is: given a uniformly random matrix $\mathbf{A} \getsr \Rq^{k \times \ell}$, find a non-zero vector $\mathbf{z} \in \Rq^\ell$ such that $\mathbf{A}\mathbf{z} = \mathbf{0} \mod q$ and $\infnorm{\mathbf{z}} \leq \beta$, where $\infnorm{\cdot}$ denotes the $\ell^\infty$ norm computed with polynomial coefficients centered in $(-q/2, q/2]$.
\end{definition}

\begin{remark}[Notation: $\beta$ disambiguation]
The parameter $\beta$ in the M-SIS definition above is a generic norm bound following standard lattice notation~\cite{LPR10}. The ML-DSA signing parameter $\beta = \tau\eta = 196$ (see parameters below) is a distinct, concrete value: the SIS norm bound instantiated for our unforgeability proof is $\gamma_1 - \beta_{\mathsf{sign}} = 2^{19} - 196 \approx 524{,}092$ (Remark~\ref{rem:selfmsis}). No clash arises because the SIS definition is parameterized generically and the ML-DSA $\beta = 196$ is always used in the signing context.
\end{remark}

\begin{remark}[SelfTargetMSIS]
\label{rem:selfmsis}
The security reduction in Theorem~\ref{thm:unforgeability} extracts a solution to the \emph{inhomogeneous} \textsf{SelfTargetMSIS} variant~\cite{KLS18}: given $[\mathbf{A}|\!-\!\mathbf{t}_1 \cdot 2^d]$, find a short $(\mathbf{v}, \Delta c)$ with $\|\mathbf{v}\|_\infty < \gamma_1 - \beta$ such that $[\mathbf{A}|\!-\!\mathbf{t}_1\cdot 2^d]\bigl[\begin{smallmatrix}\mathbf{v}\\\Delta c\end{smallmatrix}\bigr] = 2\gamma_2\mathbf{w}_1 + \mathbf{r}_0$ where $\mathbf{w}_1$ is a known target and $\|\mathbf{r}_0\|_\infty \leq \gamma_2$ (the $\mathsf{LowBits}$ residual from $\mathsf{HighBits}$ rounding). By~\cite{KLS18}, \textsf{SelfTargetMSIS} hardness follows from M-SIS hardness under a standard parameter embedding; we therefore state Theorem~\ref{thm:unforgeability} in terms of M-SIS.
\end{remark}

\begin{definition}[Module-LWE]
\label{def:mlwe}
The $\mathsf{M\text{-}LWE}_{n,k,\ell,q,\chi}$ problem is to distinguish between:
\begin{itemize}
    \item $(\mathbf{A}, \mathbf{A}\mathbf{s} + \mathbf{e})$ where $\mathbf{A} \getsr \Rq^{k \times \ell}$, $\mathbf{s} \getsr \chi^\ell$, $\mathbf{e} \getsr \chi^k$
    \item $(\mathbf{A}, \mathbf{u})$ where $\mathbf{u} \getsr \Rq^k$
\end{itemize}
where $\chi$ is a distribution over $\Rq$ with small coefficients.
\end{definition}

For ML-DSA-65 parameters ($n = 256$, $q = 8380417$, $k = 6$, $\ell = 5$), both problems are believed to be hard at the 192-bit security level.

\subsection{ML-DSA (FIPS 204)}

ML-DSA~\cite{FIPS204}, based on the CRYSTALS-Dilithium design~\cite{DKL18}, is a lattice-based digital signature scheme using the Fiat-Shamir with Aborts paradigm~\cite{Lyu12,PS00}. Security in the quantum random oracle model follows from~\cite{KLS18}. We summarize the key algorithms.

\paragraph{Parameters.} For ML-DSA-65 (NIST Security Level 3):
\begin{itemize}[nosep]
    \item $n = 256$: polynomial degree
    \item $q = 8380417 = 2^{23} - 2^{13} + 1$: modulus
    \item $k = 6$, $\ell = 5$: matrix dimensions
    \item $\eta = 4$: secret key coefficient bound
    \item $\gamma_1 = 2^{19}$: masking range
    \item $\gamma_2 = (q-1)/32$: rounding parameter
    \item $\tau = 49$: challenge weight (number of $\pm 1$ entries)
    \item $\beta = \tau \cdot \eta = 196$: signature bound
    \item $\omega = 55$: hint weight bound (max number of $1$-entries in hint $\mathbf{h}$)
    \item $d = 13$: dropping bits in $\mathsf{Power2Round}$ ($\mathbf{t} = \mathbf{t}_1 \cdot 2^d + \mathbf{t}_0$)
\end{itemize}

\paragraph{Key Generation.}
\begin{enumerate}
    \item Sample seed $\rho \getsr \{0,1\}^{256}$ and expand to matrix $\mathbf{A} \in \Rq^{k \times \ell}$
    \item Sample short secrets $\mathbf{s}_1 \getsr \chi_\eta^\ell$ and $\mathbf{s}_2 \getsr \chi_\eta^k$ with coefficients in $[-\eta, \eta]$
    \item Compute $\mathbf{t} = \mathbf{A}\mathbf{s}_1 + \mathbf{s}_2$ and decompose $\mathbf{t} = \mathbf{t}_1 \cdot 2^d + \mathbf{t}_0$
    \item Output $\pk = (\rho, \mathbf{t}_1)$ and $\sk = (\rho, \mathbf{s}_1, \mathbf{s}_2, \mathbf{t}_0)$
\end{enumerate}

\paragraph{Signing.}
Sample nonce $\mathbf{y} \getsr \{-\gamma_1+1, \ldots, \gamma_1\}^{n\ell}$, compute $\mathbf{w}_1 = \mathsf{HighBits}(\mathbf{A}\mathbf{y}, 2\gamma_2)$, derive challenge $c = H(\mu \| \mathbf{w}_1)$, and set $\mathbf{z} = \mathbf{y} + c\mathbf{s}_1$. If $\infnorm{\mathbf{z}} \geq \gamma_1 - \beta$, restart (rejection sampling). Compute hint $\mathbf{h}$ and output $\sigma = (\tilde{c}, \mathbf{z}, \mathbf{h})$.

\paragraph{Verification.}
Expand $c$ from $\tilde{c}$, compute $\mathbf{w}_1' = \mathsf{UseHint}(\mathbf{h}, \mathbf{A}\mathbf{z} - c\mathbf{t}_1 \cdot 2^d)$, and verify $\infnorm{\mathbf{z}} < \gamma_1 - \beta$ and $\tilde{c} = H(\mu \| \mathbf{w}_1')$.

The z-bound check and r0-check together yield $\approx 20$--$25\%$ per-attempt success (4--5 expected attempts). This is inherent to Fiat-Shamir with Aborts~\cite{Lyu12}, not threshold-specific. \emph{Note:} In our threshold protocol, the aggregated nonce follows an Irwin-Hall distribution (sum of $|S|$ bounded uniform variables) rather than a single uniform draw. The Irwin-Hall distribution has greater density near zero, shifting more nonces into the acceptance region $\|\mathbf{z}\|_\infty < \gamma_1 - \beta$. As a result, observed threshold signing success rates of 21--45\% are consistent with (and can slightly exceed) the single-signer baseline; see Table~\ref{tab:perf-threshold}.

\subsection{Shamir Secret Sharing}

Shamir's $(T, N)$-threshold secret sharing~\cite{Shamir79} allows a secret $s \in \Zq$ to be distributed among $N$ parties such that any $T$ parties can reconstruct $s$, but any $T-1$ parties learn nothing. For vector (module) secrets $\mathbf{s} \in \Rq^\ell$, sharing is applied coordinate-by-coordinate over $\Zq$.

\paragraph{Sharing.} To share secret $s$:
\begin{enumerate}
    \item Sample random polynomial $f(X) = s + r_1 X + \cdots + r_{T-1} X^{T-1} \in \Zq[X]$
    \item Give share $s_i = f(i)$ to party $i$ for $i \in [N]$
\end{enumerate}

\paragraph{Reconstruction.} Given shares from set $S$ with $|S| \geq T$:
\[
s = f(0) = \sum_{i \in S} \lambda_i \cdot s_i \mod q
\]
where the \textbf{Lagrange coefficients} are:
\[
\lambda_i = \prod_{j \in S, j \neq i} j \cdot (j - i)^{-1} \bmod q
\]

\begin{remark}[Lagrange Coefficient Magnitude]
\label{rem:lagrange}
For the standard evaluation points $\{1, 2, \ldots, T\}$, Lagrange coefficients satisfy $\lambda_i = (-1)^{i-1} \binom{T}{i}$, growing as $\approx 2^T / \sqrt{\pi T / 2}$ (by Stirling's formula). For $T \leq 8$, coefficients remain small ($\binom{8}{4} = 70$), enabling short-coefficient schemes~\cite{dPN25}. For $T \geq 26$, coefficients exceed $q = 8{,}380{,}417$ (since $\binom{25}{12} = 5{,}200{,}300 < q$ but $\binom{26}{13} = 10{,}400{,}600 > q$), aligning with the Ball-\c{C}akan-Malkin bound~\cite{BCM21}. For \emph{short-coefficient LSSS} designs~\cite{dPN25,BCM21} that require Lagrange products $\lambda_i \cdot \mathbf{z}_i$ to remain short \emph{before} modular reduction, this is a binding obstacle: the product has norm $O(q \cdot \gamma_1) \gg \gamma_1 - \beta$. When $|S| > T$, the maximum Lagrange coefficient magnitude grows as $\binom{|S|}{\lfloor |S|/2 \rfloor}$ (larger than for $|S| = T$), so the constraint on $T$ is the binding one in practice. In our construction, shares $\mathbf{z}_i \in \mathbb{Z}_q^\ell$ are large, but the aggregate satisfies $\sum_i \lambda_i \mathbf{z}_i = \mathbf{y} + c\mathbf{s}_1$ by the Shamir reconstruction identity; since $\|\mathbf{y} + c\mathbf{s}_1\|_\infty \leq \gamma_1 - \beta < q/2$, the computation in $\mathbb{Z}_q$ yields the same result as integer arithmetic, and the z-bound check operates on the correct short value.
\end{remark}

\subsection{Security Definitions for Threshold Signatures}

\begin{definition}[Threshold Signature Scheme]
\label{def:threshold-sig}
A $(T, N)$-threshold signature scheme consists of algorithms:
\begin{itemize}
    \item $\KeyGen(1^\kappa, N, T) \to (\pk, \{\sk_i\}_{i \in [N]})$: Generates public key and secret key shares
    \item $\Sign(\{\sk_i\}_{i \in S}, \mu) \to \sigma$: Interactive protocol among parties in $S$ (where $|S| \geq T$) to sign message $\mu$
    \item $\Verify(\pk, \mu, \sigma) \to \{0, 1\}$: Deterministic verification
\end{itemize}
\end{definition}

\begin{remark}[Signing Set Size]
\label{rem:t-notation}
With Shamir nonce DKG and coordinator-based communication (Remark~\ref{rem:wi-private}), Profiles P1 and P3+ require only $|S| \geq T$ signers---the minimum required by Shamir's $(T,N)$ threshold property for reconstruction. This improves upon standard pairwise-mask schemes~\cite{CGGMP20,DKOSS24} that require $|S| \geq T + 1$. The nonce DKG provides nonce share privacy (no computational assumptions) even with a single honest party (see Theorem~\ref{thm:it-privacy} and Remark~\ref{rem:two-honest}). In P2's broadcast model, $|S \setminus C| \geq 2$ is required for commitment privacy (Lemma~\ref{lem:mask-hiding}).
\end{remark}

\begin{definition}[Unforgeability (EUF-CMA)]
A threshold signature scheme is EUF-CMA secure if for any PPT adversary $\mathcal{A}$ corrupting up to $T-1$ parties:
\[
\Pr[\mathsf{EUF\text{-}CMA}_{\Pi}^{\mathcal{A}}(\kappa) = 1] \leq \negl(\kappa)
\]
where in the security game:
\begin{enumerate}
    \item $(\pk, \{\sk_i\}) \gets \KeyGen(1^\kappa, N, T)$
    \item $\mathcal{A}$ receives $\pk$ and $\{\sk_i\}_{i \in C}$ for corrupted set $C$ with $|C| < T$
    \item $\mathcal{A}$ can request signatures on messages of its choice (interacting with honest parties)
    \item $\mathcal{A}$ outputs $(m^*, \sigma^*)$
    \item $\mathcal{A}$ wins if $\Verify(\pk, m^*, \sigma^*) = 1$ and $m^*$ was never queried
\end{enumerate}
\end{definition}

\begin{definition}[Privacy (Informal)]
\label{def:privacy}
A threshold signature scheme provides privacy if the view of any coalition $C$ with $|C| < T$ can be efficiently simulated given only $\pk$ and the signatures produced.
\end{definition}

\begin{remark}[Definition~\ref{def:privacy} and UC Privacy]
The full UC-level privacy guarantees supersede Definition~\ref{def:privacy}, with explicit simulators and hybrid arguments established for each deployment profile: Theorem~\ref{thm:p1-uc} (P1), Theorem~\ref{thm:p2-uc} (P2), and Theorem~\ref{thm:p3-uc} (P3+).
\end{remark}

\subsection{Pseudorandom Functions}

\begin{definition}[PRF Security]
A function family $F: \mathcal{K} \times \mathcal{X} \to \mathcal{Y}$ is a secure PRF if for any PPT distinguisher $\mathcal{D}$:
\[
\left| \Pr[\mathcal{D}^{F_k(\cdot)} = 1] - \Pr[\mathcal{D}^{R(\cdot)} = 1] \right| \leq \negl(\kappa)
\]
where $k \getsr \mathcal{K}$ and $R: \mathcal{X} \to \mathcal{Y}$ is a truly random function.
\end{definition}

In our construction, we instantiate the PRF with SHAKE-256, modeled as a random oracle in the security analysis.

% Protocol Section

\section{Protocol}
\label{sec:protocol}

We now present our threshold ML-DSA protocol in full detail. The protocol consists of a setup phase (key generation) and a signing phase (three-round interactive protocol).

\subsection{Overview}

The core challenge in threshold ML-DSA is handling the large Lagrange coefficients that arise in Shamir secret sharing over $\Zq$. As noted in Remark~\ref{rem:lagrange}, the Lagrange coefficients for standard evaluation points $\{1,\ldots,N\}$ can be as large as $\binom{T}{\lfloor T/2\rfloor} \approx 2^T/\sqrt{T}$ over $\mathbb{Z}$ before modular reduction---for example, $\binom{16}{8} = 12870$ at $T = 16$ (see Lemma~\ref{lem:lagrange-magnitude}). When the naive reconstruction $\mathbf{z} = \sum_{i \in S} \lambda_i \mathbf{z}_i$ is computed, each response share $\mathbf{z}_i = \mathbf{y}_i + c \cdot \mathbf{s}_{1,i}$ is scaled by $\lambda_i$, inflating the $\ell_\infty$-norm by up to $12870\times$ (at $T=16$)---far beyond the rejection-sampling bound $\gamma_1 - \beta = 524092$. For large $T$, the coefficient magnitude eventually exceeds $q$, causing aliasing in the modular reduction.

Our solution is based on two key techniques:

\begin{enumerate}
    \item \textbf{Shamir nonce DKG}: Parties jointly generate the nonce $\mathbf{y}$ via a Shamir-style distributed key generation, so that both $\mathbf{y}$ and the long-term secret $\mathbf{s}_1$ are degree-$(T\!-\!1)$ Shamir sharings. Each party computes $\mathbf{z}_i = \mathbf{y}_i + c \cdot \mathbf{s}_{1,i}$ (without Lagrange coefficients), and the combiner reconstructs $\mathbf{z} = \sum_i \lambda_i \mathbf{z}_i = \mathbf{y} + c\mathbf{s}_1$. This achieves \emph{nonce share privacy} (conditional min-entropy $\geq 5\times$ secret key entropy for $|S| \leq 17$, requiring no computational assumptions), in coordinator-based profiles (P1, P3+), eliminates the two-honest requirement ($|S| \geq T$ suffices; see Remark~\ref{rem:wi-private}), and requires no pairwise masks on the response $\mathbf{z}_i$.

    \item \textbf{Pairwise-canceling masks}: For the commitment values $\mathbf{W}_i$ and the r0-check shares $\mathbf{V}_i$, we add pairwise-canceling masks that hide individual contributions while preserving correctness in the sum.
\end{enumerate}

\subsection{Threat Model and Deployment Profiles}
\label{sec:threat-model}

We consider a $(T, N)$-threshold setting where an adversary may corrupt up to $T-1$ parties. Table~\ref{tab:profiles} summarizes our three deployment profiles.

\begin{table}[ht]
\centering
\caption{Deployment Profiles: Trust Assumptions and Security Guarantees}
\label{tab:profiles}
\ifccs\small\fi
\ifccs\resizebox{\columnwidth}{!}{\fi
\begin{tabular}{lccc}
\toprule
\textbf{Property} & \textbf{Profile P1} & \textbf{Profile P2} & \textbf{Profile P3+} \\
\midrule
Coordinator & TEE/HSM & Untrusted & 2PC pair \\
Trust assumption & TEE integrity & None (MPC) & 1-of-2 CP honest \\
r0-check location & Inside enclave & MPC subprotocol & 2PC subprotocol \\
$c\mathbf{s}_2$ exposure & Never leaves TEE & Never reconstructed & Split between CPs \\
Online rounds & 3 & 5 & 2$^*$ \\
Signer sync & Multi-round & Multi-round & Semi-async \\
Latency & Low & Medium & Lowest \\
Complexity & Low & High & Medium \\
\midrule
\multicolumn{4}{l}{\textbf{Common to all profiles:}} \\
\midrule
Adversary model & \multicolumn{3}{c}{Static, malicious (up to $T-1$ parties)} \\
Signing set size & \multicolumn{3}{c}{$|S| \geq T$} \\
Security goal & \multicolumn{3}{c}{EUF-CMA under Module-SIS} \\
Privacy goal & \multicolumn{3}{c}{Statistical share privacy (Theorem~\ref{thm:it-privacy})} \\
Output format & \multicolumn{3}{c}{FIPS 204-compatible signature} \\
\bottomrule
\end{tabular}
\ifccs}\fi

\vspace{0.5em}
\footnotesize{$^*$P3+ = 1 signer step (within time window) + 1 server round (2PC). Signers precompute nonces offline.}
\end{table}

\ifccs\else
%% fig-protocol-flows.tex — v2: professional styling
%% Three panels side-by-side in minipage; labels (a)/(b)/(c) above each panel.

\tikzset{
  %% Actor box: clean rectangle, white bg, colored border per role
  actor/.style={
    draw=gray!60, fill=white, line width=0.55pt,
    minimum width=1.55cm, minimum height=0.62cm,
    align=center, font=\scriptsize\bfseries, inner sep=3pt,
    rounded corners=1.5pt,
  },
  actortee/.style={actor, draw=teal!65!black,   fill=teal!8,   text=teal!80!black},
  actormpc/.style={actor, draw=brown!60!black,  fill=brown!8,  text=brown!80!black},
  actortpc/.style={actor, draw=violet!55!black, fill=violet!8, text=violet!80!black},
  %% Lifelines: thin dashed
  tline/.style={densely dashed, gray!38, line width=0.3pt},
  %% Arrows
  arr/.style={->, >=stealth, line width=0.55pt},
  biarr/.style={<->, >=stealth, line width=0.55pt},
  %% Round label (small, italic, muted)
  rlbl/.style={font=\tiny\itshape, text=gray!55},
  %% Message labels: upgraded from \tiny to \scriptsize
  msgl/.style={font=\scriptsize, above=1pt, inner sep=1pt},
  msglb/.style={font=\scriptsize, below=1pt, inner sep=1pt},
}

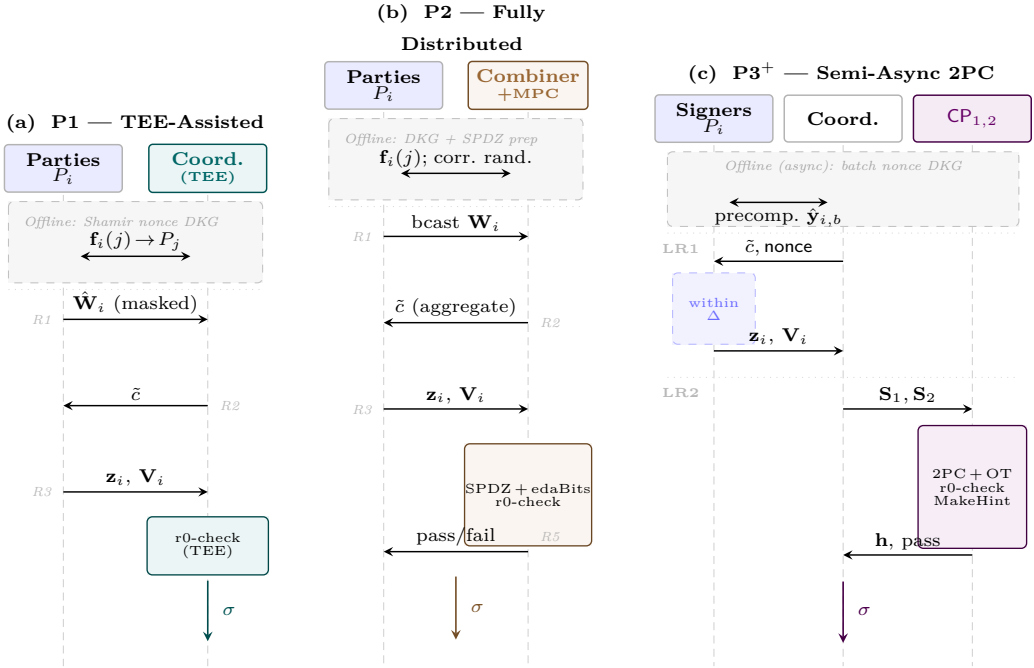
\begin{figure*}[!t]
\centering
%
%──────────────────────────────────────────────────────────────────────────────
%  Panel (a): P1 — TEE-Assisted
%──────────────────────────────────────────────────────────────────────────────
\begin{minipage}[b]{3.65cm}
\centering
{\scriptsize\bfseries (a)\; P1 — TEE-Assisted}\\[3pt]
\begin{tikzpicture}[x=1.9cm, y=-1.1cm, >=stealth]
  %% Actors
  \node[actor, fill=blue!8]  at (0,0) {Parties\\[-2pt]$P_i$};
  \node[actortee]             at (1,0) {Coord.\\[-2pt]{\tiny (TEE)}};
  %% Timelines
  \draw[tline] (0, 0.32) -- (0, 6.0);
  \draw[tline] (1, 0.32) -- (1, 6.0);
  %% Offline preprocessing
  \fill[gray!7] (-0.38,0.40) rectangle (1.38,1.38);
  \draw[dashed,gray!48,rounded corners=2pt,line width=0.4pt]
    (-0.38,0.40) rectangle (1.38,1.38);
  \node[font=\tiny\itshape,gray!70,anchor=north west] at (-0.34,0.47)
    {Offline: Shamir nonce DKG};
  \draw[biarr] (0.12,1.06) -- (0.88,1.06);
  \node[msgl] at (0.5,1.06) {$\mathbf{f}_i(j)\!\to\!P_j$};
  %% Online separator
  \draw[dotted,gray!28,line width=0.4pt] (-0.38,1.46) -- (1.38,1.46);
  %% R1
  \node[rlbl,left=1pt] at (0,1.82) {R1};
  \draw[arr] (0,1.82) -- (1,1.82);
  \node[msgl] at (0.5,1.82) {$\hat{\mathbf{W}}_i$ (masked)};
  %% R2
  \node[rlbl,right=1pt] at (1,2.86) {R2};
  \draw[arr] (1,2.86) -- (0,2.86);
  \node[msgl] at (0.5,2.86) {$\tilde{c}$};
  %% R3
  \node[rlbl,left=1pt] at (0,3.90) {R3};
  \draw[arr] (0,3.90) -- (1,3.90);
  \node[msgl] at (0.5,3.90) {$\mathbf{z}_i,\,\mathbf{V}_i$};
  %% r0-check box
  \fill[teal!8] (0.58,4.20) rectangle (1.42,4.90);
  \draw[teal!55!black,rounded corners=2pt,line width=0.5pt]
    (0.58,4.20) rectangle (1.42,4.90);
  \node[font=\tiny,align=center] at (1,4.55)
    {r0-check\\[-1pt]{\tiny (TEE)}};
  %% Output
  \draw[->,line width=0.6pt,teal!60!black] (1,4.97) -- (1,5.70);
  \node[font=\scriptsize,right=2pt,text=teal!65!black] at (1,5.33) {$\sigma$};
\end{tikzpicture}
\end{minipage}
\hspace{0.4cm}
%
%──────────────────────────────────────────────────────────────────────────────
%  Panel (b): P2 — Fully Distributed, MPC
%──────────────────────────────────────────────────────────────────────────────
\begin{minipage}[b]{3.65cm}
\centering
{\scriptsize\bfseries (b)\; P2 — Fully Distributed}\\[3pt]
\begin{tikzpicture}[x=1.9cm, y=-1.1cm, >=stealth]
  %% Actors
  \node[actor,fill=blue!8] at (0,0) {Parties\\[-2pt]$P_i$};
  \node[actormpc]          at (1,0) {Combiner\\[-2pt]{\tiny +MPC}};
  %% Timelines
  \draw[tline] (0, 0.32) -- (0, 7.0);
  \draw[tline] (1, 0.32) -- (1, 7.0);
  %% Offline
  \fill[gray!7] (-0.38,0.40) rectangle (1.38,1.38);
  \draw[dashed,gray!48,rounded corners=2pt,line width=0.4pt]
    (-0.38,0.40) rectangle (1.38,1.38);
  \node[font=\tiny\itshape,gray!70,anchor=north west] at (-0.34,0.47)
    {Offline: DKG + SPDZ prep};
  \draw[biarr] (0.12,1.06) -- (0.88,1.06);
  \node[msgl] at (0.5,1.06) {$\mathbf{f}_i(j)$;\;corr.~rand.};
  %% Online separator
  \draw[dotted,gray!28,line width=0.4pt] (-0.38,1.46) -- (1.38,1.46);
  %% R1
  \node[rlbl,left=1pt] at (0,1.82) {R1};
  \draw[arr] (0,1.82) -- (1,1.82);
  \node[msgl] at (0.5,1.82) {bcast $\mathbf{W}_i$};
  %% R2
  \node[rlbl,right=1pt] at (1,2.86) {R2};
  \draw[arr] (1,2.86) -- (0,2.86);
  \node[msgl] at (0.5,2.86) {$\tilde{c}$ (aggregate)};
  %% R3
  \node[rlbl,left=1pt] at (0,3.90) {R3};
  \draw[arr] (0,3.90) -- (1,3.90);
  \node[msgl] at (0.5,3.90) {$\mathbf{z}_i,\,\mathbf{V}_i$};
  %% R4-5 SPDZ box
  \node[rlbl,right=1pt] at (1,4.50) {R4};
  \fill[brown!8] (0.56,4.32) rectangle (1.44,5.55);
  \draw[brown!55!black,rounded corners=2pt,line width=0.5pt]
    (0.56,4.32) rectangle (1.44,5.55);
  \node[font=\tiny,align=center] at (1,4.935)
    {SPDZ\,+\,edaBits\\[-1pt]r0-check};
  \node[rlbl,right=1pt] at (1,5.44) {R5};
  %% pass/fail
  \draw[arr] (1,5.62) -- (0,5.62);
  \node[msgl] at (0.5,5.62) {pass/fail};
  %% Output
  \draw[->,line width=0.6pt,brown!60!black] (0.5,5.92) -- (0.5,6.68);
  \node[font=\scriptsize,right=2pt,text=brown!65!black] at (0.5,6.30) {$\sigma$};
\end{tikzpicture}
\end{minipage}
\hspace{0.4cm}
%
%──────────────────────────────────────────────────────────────────────────────
%  Panel (c): P3+ — Semi-Async 2PC
%──────────────────────────────────────────────────────────────────────────────
\begin{minipage}[b]{5.1cm}
\centering
{\scriptsize\bfseries (c)\; P3$^+$ — Semi-Async 2PC}\\[3pt]
\begin{tikzpicture}[x=1.7cm, y=-1.1cm, >=stealth]
  %% Actors
  \node[actor,fill=blue!8] at (0,0) {Signers\\[-2pt]$P_i$};
  \node[actor]             at (1,0) {Coord.};
  \node[actortpc]          at (2,0) {$\mathsf{CP}_{1,2}$};
  %% Timelines
  \draw[tline] (0, 0.32) -- (0, 6.6);
  \draw[tline] (1, 0.32) -- (1, 6.6);
  \draw[tline] (2, 0.32) -- (2, 6.6);
  %% Offline (async)
  \fill[gray!7] (-0.36,0.40) rectangle (2.36,1.30);
  \draw[dashed,gray!48,rounded corners=2pt,line width=0.4pt]
    (-0.36,0.40) rectangle (2.36,1.30);
  \node[font=\tiny\itshape,gray!70] at (1.0,0.58)
    {Offline (async): batch nonce DKG};
  \draw[biarr] (0.12,1.01) -- (0.88,1.01);
  \node[msglb] at (0.5,1.01) {precomp.~$\hat{\mathbf{y}}_{i,b}$};
  %% LR1
  \draw[dotted,gray!28,line width=0.4pt] (-0.36,1.38) -- (2.36,1.38);
  \node[font=\tiny\bfseries,gray!55,left=2pt] at (0,1.57) {LR1};
  \draw[arr] (1,1.72) -- (0,1.72);
  \node[msgl] at (0.5,1.72) {$\tilde{c},\mathsf{nonce}$};
  %% Response window
  \fill[blue!5] (-0.32,1.86) rectangle (0.32,2.72);
  \draw[blue!28,dashed,rounded corners=2pt,line width=0.4pt]
    (-0.32,1.86) rectangle (0.32,2.72);
  \node[font=\tiny,blue!58,align=center] at (0,2.29) {within\\[-1pt]$\Delta$};
  \draw[arr] (0,2.80) -- (1,2.80);
  \node[msgl] at (0.5,2.80) {$\mathbf{z}_i,\,\mathbf{V}_i$};
  %% LR2
  \draw[dotted,gray!28,line width=0.4pt] (-0.36,3.12) -- (2.36,3.12);
  \node[font=\tiny\bfseries,gray!55,left=2pt] at (0,3.30) {LR2};
  \draw[arr] (1,3.50) -- (2,3.50);
  \node[msgl] at (1.5,3.50) {$\mathbf{S}_1,\mathbf{S}_2$};
  %% 2PC box
  \fill[violet!7] (1.58,3.70) rectangle (2.42,5.18);
  \draw[violet!50!black,rounded corners=2pt,line width=0.5pt]
    (1.58,3.70) rectangle (2.42,5.18);
  \node[font=\tiny,align=center] at (2,4.44)
    {2PC\,+\,OT\\[-1pt]r0-check\\[-1pt]MakeHint};
  %% h output
  \draw[arr] (2,5.26) -- (1,5.26);
  \node[msgl] at (1.5,5.26) {$\mathbf{h}$, pass};
  %% Output
  \draw[->,line width=0.6pt,violet!55!black] (1,5.58) -- (1,6.32);
  \node[font=\scriptsize,right=2pt,text=violet!60!black] at (1,5.93) {$\sigma$};
\end{tikzpicture}
\end{minipage}
\vspace{4pt}
\caption{Protocol message flows for the three deployment profiles.
Shaded top regions are offline preprocessing; dashed lines are actor lifelines.
\textbf{(a)~P1}: masked commitments $\hat{\mathbf{W}}_i$ sent only to trusted TEE coordinator,
enabling $|S\setminus C|=1$.
\textbf{(b)~P2}: all values broadcast; dishonest-majority security via SPDZ+edaBits;
r0-check computed inside MPC.
\textbf{(c)~P3$^+$}: semi-async signers respond within window~$\Delta$;
two CPs run garbled-circuit 2PC for r0-check.
LR = Logical Round; R = Protocol Round.}
\label{fig:protocol-flows}
\end{figure*}

\fi

\paragraph{Profile P1 (Primary).}\label{sec:profile-p1} A TEE/HSM coordinator aggregates masked contributions and performs the r0-check inside a trusted enclave. The value $c\mathbf{s}_2$ is reconstructed inside the enclave and never exposed. This is our recommended deployment model, combining threshold distribution among parties with hardware-isolated aggregation.

\paragraph{Profile P2 (Fully Distributed).} A fully distributed protocol where no single party learns $c\mathbf{s}_2$. The r0-check is computed via an MPC subprotocol using SPDZ~\cite{DPSZ12} and edaBits~\cite{EKMOZ20} (see Section~\ref{sec:extensions}). Using combiner-mediated commit-then-open and constant-depth comparison circuits, we achieve \textbf{5 online rounds} (optimized from 8). We prove UC security against malicious adversaries corrupting up to $N-1$ parties in Theorem~\ref{thm:p2-uc}.

\paragraph{Profile P3+ (Semi-Async 2PC).} Designed for human-in-the-loop authorization scenarios. Two designated \emph{Computation Parties} (CPs) jointly evaluate the r0-check via lightweight 2PC. Each CP receives additive shares of $c\mathbf{s}_2$; neither learns the complete value unless they collude.

The key innovation is \emph{semi-asynchronous signer participation}: signers precompute nonces $(\mathbf{y}_i, \mathbf{w}_i, \mathsf{Com}_i)$ offline at any time, then respond within a bounded time window (e.g., 5 minutes) after receiving the challenge. A client daemon handles the cryptographic operations transparently, completing the signing process in under 1 second.

This achieves \textbf{2 logical rounds} (1 signer step + 1 server round). The 2PC uses pre-garbled circuits for reduced online latency ($\sim$15ms). Profile P3+ is suited for deployments requiring human authorization where signers cannot coordinate for multiple synchronous rounds.

\subsection{Setup Phase}

\paragraph{Pairwise Seed Establishment.}
Before key generation, each pair of parties $(i, j)$ with $i < j$ establishes a shared secret seed $\mathsf{seed}_{i,j} \in \{0,1\}^{256}$. This may be done via ML-KEM key exchange (each party generates a keypair; pairs perform encapsulation and decapsulation to derive the shared seed) or via pre-shared keys from a trusted setup.

Each party $i$ stores seeds $\{\mathsf{seed}_{\min(i,j), \max(i,j)}\}_{j \neq i}$ for all other parties.

\paragraph{Key Generation.}
The key generation can be performed by a trusted dealer or via a distributed key generation protocol. We present the trusted dealer version; see Section~\ref{sec:extensions} for DKG.

\begin{remark}[Key Generation Scope]
\label{rem:keygen-scope}
Threshold key generation is an \emph{application-layer} component outside the scope of FIPS~204. The resulting public key $\pk = (\rho, \mathbf{t}_1)$ is byte-identical to a standard FIPS~204 ML-DSA-65 public key and is usable with any FIPS~204-compliant verification implementation without modification. The threshold secret key shares $\{\sk_i\}$ are specific to this construction and replace the single-party FIPS~204 signing key; they do not affect the signature format or verifiability.
\end{remark}

\begin{algorithm}[htbp]
\caption{$\KeyGen(1^\kappa, N, T)$: Threshold Key Generation}
\label{alg:keygen}
\begin{algorithmic}[1]
    \State Sample seed $\rho \getsr \{0,1\}^{256}$
    \State Expand $\mathbf{A} = \mathsf{ExpandA}(\rho) \in \Rq^{k \times \ell}$
    \State Sample $\mathbf{s}_1 \getsr \chi_\eta^\ell$, $\mathbf{s}_2 \getsr \chi_\eta^k$ \Comment{Short secret keys}
    \State Compute $\mathbf{t} = \mathbf{A}\mathbf{s}_1 + \mathbf{s}_2$
    \State Decompose $(\mathbf{t}_1, \mathbf{t}_0) \gets \mathsf{Power2Round}(\mathbf{t}, d)$ \Comment{FIPS 204 Alg.~29; $\mathbf{t} = \mathbf{t}_1 \cdot 2^d + \mathbf{t}_0$}
    \State Set $\pk = (\rho, \mathbf{t}_1)$
    \For{each secret polynomial $s$ in $\mathbf{s}_1$ and $\mathbf{s}_2$}
        \State Create Shamir sharing: $p_s(X) = s + a_1 X + \cdots + a_{T-1} X^{T-1}$, \quad $a_j \getsr \Rq$
        \State Compute shares: $s_i = p_s(i)$ for $i \in [N]$
    \EndFor
    \State \Return $\pk$, $\{\sk_i = (\mathbf{s}_{1,i}, \mathbf{s}_{2,i}, \{\mathsf{seed}_{*,*}\})\}_{i \in [N]}$
\end{algorithmic}
\end{algorithm}

\subsection{Shamir Nonce DKG}
\label{sec:nonce-dkg}

The primary technique enabling nonce share privacy is the \emph{Shamir nonce DKG}: parties jointly generate the signing nonce $\mathbf{y}$ as a degree-$(T\!-\!1)$ Shamir sharing, matching the structure of the long-term secret $\mathbf{s}_1$. This is executed once per signing attempt and can be preprocessed offline.

\begin{algorithm}[htbp]
\caption{$\mathsf{NonceDKG}(S)$: Distributed Nonce Generation}
\label{alg:nonce-dkg}
\begin{algorithmic}[1]
    \Statex \textbf{Input:} Signing set $S$ with $|S| \geq T$
    \For{each party $i \in S$ in parallel}
        \State Sample secret contribution $\hat{\mathbf{y}}_i \getsr \{-\lfloor\gamma_1/|S|\rfloor, \ldots, \lfloor\gamma_1/|S|\rfloor\}^{n\ell}$
        \State Construct degree-$(T\!-\!1)$ polynomial over $\Rq^\ell$:
        \Statex \qquad $\mathbf{f}_i(x) = \hat{\mathbf{y}}_i + \mathbf{a}_{i,1} x + \cdots + \mathbf{a}_{i,T-1} x^{T-1}$
        \Statex \qquad where $\mathbf{a}_{i,k} \getsr \Rq^\ell$ for $k = 1, \ldots, T\!-\!1$
        \State Send $\mathbf{f}_i(j)$ to each party $j \in S$ over a private channel \Comment{Encrypted via pairwise ML-KEM seeds}
    \EndFor
    \For{each party $j \in S$}
        \State Compute nonce share: $\mathbf{y}_j = \sum_{i \in S} \mathbf{f}_i(j)$
    \EndFor
\end{algorithmic}
\end{algorithm}

The combined polynomial is $\mathbf{F}(x) = \sum_{i \in S} \mathbf{f}_i(x)$, with $\mathbf{F}(0) = \sum_i \hat{\mathbf{y}}_i = \mathbf{y}$ (the nonce) and $\mathbf{F}(j) = \mathbf{y}_j$ (the shares). This is a valid degree-$(T\!-\!1)$ Shamir sharing of $\mathbf{y}$ over $\Rq^\ell$.

\begin{remark}[Why Uniform Higher-Degree Coefficients]
\label{rem:uniform-coefficients}
The coefficients $\mathbf{a}_{i,k}$ for $k \geq 1$ are sampled uniformly over $\Rq^\ell$ (not from a short distribution). This is essential: the adversary observes $T - 1$ evaluations of $\mathbf{f}_h$ (the honest party's polynomial), leaving exactly one degree of freedom. The uniform \emph{prior} distribution of $\mathbf{a}_{h,T-1}$ ensures high entropy in this free parameter; combined with the bounded-nonce constraint, the honest party's share $\mathbf{f}_h(x_h)$ retains conditional min-entropy $\geq n\ell \cdot \log_2\bigl(2\lfloor \gamma_1/|S| \rfloor + 1\bigr)$ bits per session (the support of each per-party share coordinate is $\{-\lfloor\gamma_1/|S|\rfloor,\ldots,\lfloor\gamma_1/|S|\rfloor\}$, a range of $2\lfloor\gamma_1/|S|\rfloor+1$ values counting both endpoints; the approximation $\approx n\ell\cdot(20 - \log_2|S|)$ bits holds to within $1$ bit for $|S| \leq 17$). See Theorem~\ref{thm:it-privacy} and Remark~\ref{rem:sd-zero}.
\end{remark}

\begin{remark}[Communication Cost]
\label{rem:dkg-cost}
The nonce DKG requires $|S| \cdot (|S| - 1)$ point-to-point messages, each of size $n\ell \cdot \lceil\log_2 q\rceil$ bits ($\approx 3.6$ KB for ML-DSA-65; $256 \times 5 \times 23 = 29{,}440$ bits). For a $(3,5)$-threshold with $|S| = 3$, this is 6 messages totaling $\approx 22$ KB; for $|S| = 32$, the DKG requires $32 \times 31 = 992$ messages totaling $\approx 3.6$ MB. This $O(|S|^2)$ offline preprocessing cost does not appear in the online round complexity comparison; it is amortizable over batches of signing sessions (multiple nonce DKG rounds can be preprocessed during idle time) and is absent from the online communication cost ($\approx 12.3$ KB per party per online attempt). The DKG round is independent of the message to be signed and can be executed offline as a preprocessing step, so online signing latency is unchanged.
\end{remark}

\ifccs\else
\begin{remark}[Why No VSS for the Nonce DKG]
\label{rem:no-vss-nonce}
The \emph{key} DKG uses Feldman commitments~\cite{Feldman87} because the long-term secret $\mathbf{s}_1$ has short coefficients and parties must verifiably bind their contributions during a one-time setup. The \emph{nonce} DKG has a different structure: the higher-degree polynomial coefficients are sampled $\Rq^\ell$-uniformly (by design, for entropy), and the constant term $\hat{\mathbf{y}}_i$ must remain hidden to preserve nonce privacy (Theorem~\ref{thm:it-privacy}). Feldman-style commitments publish $g^{a_k}$ for each coefficient, which would not hide $\hat{\mathbf{y}}_i$ against a discrete-log adversary and would require a hiding commitment scheme instead---substantially complicating the offline round. We instead adopt an \emph{optimistic} approach: shares are sent over private authenticated channels, and any cheating party is detectable \emph{a posteriori} by the blame protocol (Section~\ref{sec:blame}), which asks parties to reveal their full polynomials to the TEE. The TEE recomputes each share and verifies consistency with the submitted $\mathbf{z}_i$ values. The optimistic threshold of $K = 33$ consecutive aborts ensures blame is triggered before a malicious party can cause meaningful disruption (honest probability $(0.75)^{33} < 10^{-4}$).
\end{remark}
\fi

\subsection{Pairwise-Canceling Masks}

Pairwise-canceling masks are used to hide the commitment values $\mathbf{W}_i$ (for challenge derivation) and the r0-check shares $\mathbf{V}_i$ (which contain $\lambda_i \cdot c \cdot \mathbf{s}_{2,i}$). For any signing set $S$, each party $i \in S$ computes a mask $\mathbf{m}_i$ such that $\sum_{i \in S} \mathbf{m}_i = \mathbf{0}$.

\begin{definition}[Pairwise-Canceling Mask]
\label{def:mask}
For signing set $S$ and PRF context $\mathsf{ctx} \in \{0,1\}^*$ (composed as $\mathsf{ctx} = \mathsf{nonce}^* \| \mathsf{label}$, where $\mathsf{nonce}^*$ is the appropriate per-session identifier defined below, and $\mathsf{label}$ is a fixed domain tag), party $i$ computes:
\[
\mathbf{m}_i^{(\mathsf{ctx})} = \sum_{j \in S: j > i} \PRF(\mathsf{seed}_{i,j}, \mathsf{ctx}) - \sum_{j \in S: j < i} \PRF(\mathsf{seed}_{j,i}, \mathsf{ctx})
\]
where $\PRF: \{0,1\}^{256} \times \{0,1\}^* \to \Rq^k$ is instantiated with SHAKE-256 (here $k=6$ is the module rank, not the dropping-bits parameter $d=13$). The per-session nonce ensures masks are fresh in each signing session (see Mask Hiding, Lemma~\ref{lem:mask-hiding}).

We use $\mathsf{ctx} = \mathsf{nonce}_0 \| \texttt{"comm"}$ for commitment masks $\mathbf{m}_i^{(w)} \in \Rq^k$, where $\mathsf{nonce}_0 = H(\texttt{"nonce0"} \| \mathsf{Com}_{i_1} \| \cdots)$ is the \emph{pre-challenge} nonce derived at the start of Round~2 (before the challenge $c$ is known). We use $\mathsf{ctx} = \mathsf{nonce} \| \texttt{"s2"}$ for r0-check masks $\mathbf{m}_i^{(s2)} \in \Rq^k$, where $\mathsf{nonce}$ is the post-challenge per-session identifier.
\end{definition}

\begin{lemma}[Mask Cancellation]
\label{lem:mask-cancel}
For any signing set $S$ and nonce, $\sum_{i \in S} \mathbf{m}_i = \mathbf{0}$.
\end{lemma}

\begin{proof}
Each pair $(i, j)$ with $i < j$ contributes $+\PRF(\mathsf{seed}_{i,j}, \mathsf{nonce})$ from party $i$ (since $j > i$) and $-\PRF(\mathsf{seed}_{i,j}, \mathsf{nonce})$ from party $j$ (since $i < j$); these cancel exactly in the sum, so $\sum_{i \in S} \mathbf{m}_i = \mathbf{0}$.
\end{proof}

\subsection{Signing Protocol}

The signing protocol is a three-round interactive protocol among parties in the signing set $S$ (where $|S| \geq T$), preceded by an offline nonce DKG preprocessing step (Algorithm~\ref{alg:nonce-dkg}). We assume a broadcast channel or a designated combiner (TEE/HSM in Profile P1).

\ifccs
\begin{algorithm}[htbp]
\caption{$\Sign$: Three-Round Threshold Signing Protocol}
\label{alg:sign}
\small
\begin{algorithmic}[1]
    \Statex \textbf{Public:} Message $\mu$, $\pk = (\rho, \mathbf{t}_1)$, signing set $S$
    \Statex \textbf{Private:} Party $i$ holds $\sk_i = (\mathbf{s}_{1,i}, \mathbf{s}_{2,i}, \{\mathsf{seed}_{*,*}\})$
    \Statex
    \Statex \underline{\textbf{Preprocessing: Nonce DKG (offline)}}
    \State Execute $\mathsf{NonceDKG}(S)$ (Alg.~\ref{alg:nonce-dkg}) $\to$ nonce shares $\{\mathbf{y}_i\}_{i \in S}$
    \Statex

    \Statex \underline{\textbf{Round 1: Nonce Commitment}}
    \For{each party $i \in S$ in parallel}
        \State Compute $\mathbf{w}_i = \mathbf{A} \mathbf{y}_i$
        \State Sample $r_i \getsr \{0,1\}^{256}$
        \State Compute $\mathsf{Com}_i = H(\texttt{"com"} \| \mathbf{y}_i \| \mathbf{w}_i \| r_i)$; broadcast $\mathsf{Com}_i$
        \State Store $(\mathbf{y}_i, \mathbf{w}_i, r_i)$ as local state
    \EndFor
    \Statex

    \Statex \underline{\textbf{Round 2: Masked Reveal and Challenge}}
    \State \textbf{Pre-round:} $\mathsf{nonce}_0 = H(\texttt{"nonce0"} \| \mathsf{Com}_{i_1} \| \cdots \| \mathsf{Com}_{i_{|S|}})$
    \For{each party $i \in S$ in parallel}
        \State Compute $\lambda_i = \prod_{j \in S, j \neq i} \frac{j}{j-i} \bmod q$
        \State Compute commitment mask $\mathbf{m}_i^{(w)} \in \Rq^k$ using $\mathsf{nonce}_0 \| \texttt{"comm"}$
        \State Send $(\mathbf{W}_i = \lambda_i \cdot \mathbf{w}_i + \mathbf{m}_i^{(w)},\, r_i)$ to coordinator
    \EndFor
    \State \textbf{Coordinator:}
    \State \quad Aggregate $\mathbf{W} = \sum_{j \in S} \mathbf{W}_j = \mathbf{A}\mathbf{y}$ \Comment{Masks cancel; Lagrange on $\{\mathbf{y}_j\}$}
    \State \quad Compute $\mathbf{w}_1 = \mathsf{HighBits}(\mathbf{W}, 2\gamma_2)$
    \State \quad Compute $\tilde{c} = H_{\mathsf{chal}}(\mu \| \mathbf{w}_1)$, expand $c \in \Rq$
    \State \quad Derive $\mathsf{nonce} = H(\texttt{"nonce"} \| \tilde{c} \| \mu \| \mathsf{sort}(S))$
    \State \quad Broadcast $(\mathbf{w}_1, \tilde{c}, \mathsf{nonce})$ \Comment{Individual $\mathbf{W}_i$ stay private}
    \Statex

    \Statex \underline{\textbf{Round 3: Response}}
    \For{each party $i \in S$ in parallel}
        \State Compute $\mathbf{z}_i = \mathbf{y}_i + c \cdot \mathbf{s}_{1,i}$ \Comment{No $\lambda_i$, no mask}
        \State Compute $\mathbf{V}_i = \lambda_i \cdot c \cdot \mathbf{s}_{2,i} + \mathbf{m}_i^{(s2)}$
        \State Send $(\mathbf{z}_i, \mathbf{V}_i)$ to combiner
    \EndFor
    \Statex

    \Statex \underline{\textbf{Aggregation and r0-Check (Section~\ref{sec:r0check})}}
    \State Combiner: $\mathbf{z} = \sum_{i \in S} \lambda_i \cdot \mathbf{z}_i \bmod q$
    \If{$\infnorm{\mathbf{z}} \geq \gamma_1 - \beta$} \Return $\bot$ \EndIf
    \State Combiner: $c\mathbf{s}_2 \gets \sum_{i \in S} \mathbf{V}_i$ \Comment{Lagrange + mask cancellation}
    \State \textbf{r0-check:} Verify $\infnorm{\mathsf{LowBits}(\mathbf{w} - c\mathbf{s}_2, 2\gamma_2)} < \gamma_2 - \beta$
    \If{r0-check fails} \Return $\bot$ \EndIf
    \State Compute $\mathbf{r} = \mathbf{A}\mathbf{z} - c\mathbf{t}_1 \cdot 2^d$
    \State Set $-c\mathbf{t}_0 \gets \mathbf{w} - \mathbf{r} - c\mathbf{s}_2$
    \State Compute $\mathbf{h} = \mathsf{MakeHint}(-c\mathbf{t}_0, \mathbf{r}, 2\gamma_2)$
    \If{$\mathsf{weight}(\mathbf{h}) > \omega$} \Return $\bot$ \EndIf
    \State \Return $\sigma = (\tilde{c}, \mathbf{z}, \mathbf{h})$
\end{algorithmic}
\end{algorithm}
\else
\begin{algorithm}[htbp]
\caption{$\Sign$: Three-Round Threshold Signing Protocol}
\label{alg:sign}
\begin{algorithmic}[1]
    \Statex \textbf{Public Input:} Message $\mu$, public key $\pk = (\rho, \mathbf{t}_1)$, signing set $S$
    \Statex \textbf{Private Input:} Party $i$ holds $\sk_i = (\mathbf{s}_{1,i}, \mathbf{s}_{2,i}, \{\mathsf{seed}_{*,*}\})$
    \Statex
    \Statex \underline{\textbf{Preprocessing: Nonce DKG (offline)}}
    \State Execute $\mathsf{NonceDKG}(S)$ (Algorithm~\ref{alg:nonce-dkg}) to obtain nonce shares $\{\mathbf{y}_i\}_{i \in S}$
    \Statex

    \Statex \underline{\textbf{Round 1: Nonce Commitment}}
    \For{each party $i \in S$ in parallel}
        \State Compute $\mathbf{w}_i = \mathbf{A} \mathbf{y}_i$ \Comment{Using DKG nonce share}
        \State Sample $r_i \getsr \{0,1\}^{256}$
        \State Compute $\mathsf{Com}_i = H(\texttt{"com"} \| \mathbf{y}_i \| \mathbf{w}_i \| r_i)$
        \State Broadcast $\mathsf{Com}_i$
        \State Store $(\mathbf{y}_i, \mathbf{w}_i, r_i)$ as local state
    \EndFor
    \Statex

    \Statex \underline{\textbf{Round 2: Masked Reveal and Challenge}}
    \State \textbf{Pre-round:} Derive preliminary nonce $\mathsf{nonce}_0 = H(\texttt{"nonce0"} \| \mathsf{Com}_{i_1} \| \cdots \| \mathsf{Com}_{i_{|S|}})$ where $(i_1, \ldots, i_{|S|}) = \mathsf{sort}(S)$
    \For{each party $i \in S$ in parallel}
        \State Compute Lagrange coefficient for reconstruction at point 0: $\lambda_i = \prod_{j \in S, j \neq i} \frac{j}{j-i} \bmod q$ \Comment{Shamir interpolation at $x=0$}
        \State Compute commitment mask $\mathbf{m}_i^{(w)} \in \Rq^k$ using $\mathsf{nonce}_0 \| \texttt{"comm"}$
        \State Send $(\mathbf{W}_i = \lambda_i \cdot \mathbf{w}_i + \mathbf{m}_i^{(w)}, r_i)$ to coordinator \Comment{NOT broadcast; see Remark~\ref{rem:two-honest}}
    \EndFor
    \State \textbf{Coordinator:}
    \State \quad Aggregate $\mathbf{W} = \sum_{j \in S} \mathbf{W}_j = \sum_j (\lambda_j \mathbf{w}_j + \mathbf{m}_j^{(w)})$ \Comment{Substitute definition of $\mathbf{W}_j$}
    \State \qquad $= \sum_j \lambda_j \mathbf{w}_j + \sum_j \mathbf{m}_j^{(w)} = \sum_j \lambda_j \mathbf{A}\mathbf{y}_j + \mathbf{0}$ \Comment{Masks cancel (Lemma~\ref{lem:mask-cancel}); $\mathbf{w}_j = \mathbf{A}\mathbf{y}_j$}
    \State \qquad $= \mathbf{A}\sum_j \lambda_j \mathbf{y}_j = \mathbf{A}\mathbf{y}$ \Comment{Linearity of $\mathbf{A}$; Lagrange reconstruction of $\mathbf{y}$ from $\{\mathbf{y}_j\}_{j \in S}$}
    \State \quad Set $\mathbf{w} \gets \mathbf{W}$ \Comment{$\mathbf{w} = \mathbf{A}\mathbf{y}$; used in r0-check and hint computation}
    \State \quad Compute $\mathbf{w}_1 = \mathsf{HighBits}(\mathbf{w}, 2\gamma_2)$
    \State \quad Compute $\tilde{c} = H_{\mathsf{chal}}(\mu \| \mathbf{w}_1)$ and expand to challenge $c \in \Rq$ \Comment{FIPS 204 spec; no invertibility check needed}
    \State \quad Derive $\mathsf{nonce} = H(\texttt{"nonce"} \| \tilde{c} \| \mu \| \mathsf{sort}(S))$
    \State \quad Broadcast $(\mathbf{w}_1, \tilde{c}, \mathsf{nonce})$ to all parties \Comment{Only $\mathsf{HighBits}$; individual $\mathbf{W}_i$ stay private}
    \Statex

    \Statex \underline{\textbf{Round 3: Response}}
    \For{each party $i \in S$ in parallel}
        \State Compute response: $\mathbf{z}_i = \mathbf{y}_i + c \cdot \mathbf{s}_{1,i}$ \Comment{No $\lambda_i$, no mask}
        \State Compute $\mathbf{V}_i = \lambda_i \cdot c \cdot \mathbf{s}_{2,i} + \mathbf{m}_i^{(s2)}$ \Comment{Masked, for r0-check}
        \State Send $(\mathbf{z}_i, \mathbf{V}_i)$ to combiner
    \EndFor
    \Statex

    \Statex \underline{\textbf{Aggregation and r0-Check (see Subsection~\ref{sec:r0check})}}
    \State Combiner computes $\mathbf{z} = \sum_{i \in S} \lambda_i \cdot \mathbf{z}_i \mod q$ \Comment{Lagrange at combiner}
    \If{$\infnorm{\mathbf{z}} \geq \gamma_1 - \beta$}
        \State \Return $\bot$ (z-bound abort, retry from Round 1)
    \EndIf
    \State Combiner computes $c\mathbf{s}_2 \gets \sum_{i \in S} \mathbf{V}_i = \sum_{i \in S} (\lambda_i c\mathbf{s}_{2,i} + \mathbf{m}_i^{(s2)})$ \Comment{Substitute definition of $\mathbf{V}_i$}
    \State \qquad $= c\sum_{i \in S} \lambda_i \mathbf{s}_{2,i} + \sum_{i \in S} \mathbf{m}_i^{(s2)} = c\mathbf{s}_2 + \mathbf{0}$ \Comment{Lagrange reconstruction of $\mathbf{s}_2$ from shares $\{\mathbf{s}_{2,i}\}_{i \in S}$ using same $\lambda_i$ as Line~206; masks cancel (Lemma~\ref{lem:mask-cancel})}
    \State \textbf{r0-check:} Verify $\infnorm{\mathsf{LowBits}(\mathbf{w} - c\mathbf{s}_2, 2\gamma_2)} < \gamma_2 - \beta$
    \If{r0-check fails}
        \State \Return $\bot$ (r0 abort, retry from Round 1)
    \EndIf
    \State Compute $\mathbf{r} = \mathbf{A}\mathbf{z} - c\mathbf{t}_1 \cdot 2^d$ \Comment{$= \mathbf{w} - c\mathbf{s}_2 + c\mathbf{t}_0$}
    \State Set $-c\mathbf{t}_0 \gets \mathbf{w} - \mathbf{r} - c\mathbf{s}_2$ \Comment{$= \mathbf{w} - (\mathbf{A}\mathbf{z} - c\mathbf{t}_1 \cdot 2^d) - c\mathbf{s}_2$}
    \State Compute hint $\mathbf{h} = \mathsf{MakeHint}(-c\mathbf{t}_0, \mathbf{r}, 2\gamma_2)$
    \If{$\mathsf{weight}(\mathbf{h}) > \omega$}
        \State \Return $\bot$ (hint weight abort)
    \EndIf
    \State \Return signature $\sigma = (\tilde{c}, \mathbf{z}, \mathbf{h})$
\end{algorithmic}
\end{algorithm}
\fi

\ifccs\else
\begin{remark}[Commitment Model]
\label{rem:protocol-commitment}
The Round~1 commitments $\mathsf{Com}_i$ bind parties to their nonces before seeing others' values, preventing adaptive attacks. However, \emph{we do not publicly verify these commitments during normal signing}---the masked values $\mathbf{W}_i$ cannot be verified without revealing private mask seeds, and the responses $\mathbf{z}_i$ are privacy-protected by the nonce DKG (Theorem~\ref{thm:it-privacy}). Instead, we follow an \emph{optimistic} model: if signing fails (excessive aborts, invalid signatures), the TEE coordinator triggers a blame protocol where parties reveal their commitments for verification. This ``commit always, open only on blame'' approach minimizes latency on the normal path while enabling accountability. See Section~\ref{sec:blame} for details.
\end{remark}
\fi

\begin{remark}[Why $\mathbf{W}_i$ Must Not Be Broadcast]
\label{rem:wi-private}
Individual masked commitments $\mathbf{W}_i = \lambda_i \mathbf{A}\mathbf{y}_i + \mathbf{m}_i^{(w)}$ are sent \emph{only to the coordinator}, not broadcast to all parties. This is critical for privacy when $|S \setminus C| = 1$: if $\mathbf{W}_h$ were broadcast and the adversary controls all other parties in $S$, the pairwise mask $\mathbf{m}_h^{(w)}$ would be fully computable (all seeds are shared with corrupted parties), revealing $\lambda_h \mathbf{A}\mathbf{y}_h$.

\textbf{Recovery of $\mathbf{y}_h$ from $\mathbf{A}\mathbf{y}_h$ (when $k > \ell$).} For ML-DSA-65, the public matrix $\mathbf{A} \in \Rq^{k \times \ell}$ with $k=6$, $\ell=5$ (so $k > \ell$) is generated via the random oracle ExpandA. The Module-LWE assumption~\cite{LS15,ADPS16} requires that $\mathbf{A}$ has full column rank $\ell$ with overwhelming probability over the random oracle sampling (see~\cite[Section~3.4]{FIPS204} for ExpandA). When $\mathbf{A}$ has full column rank and $k \geq \ell$, the left inverse $((\mathbf{A}^\top\mathbf{A})^{-1}\mathbf{A}^\top)$ exists over $R_q$ (modulo invertibility of $\mathbf{A}^\top\mathbf{A}$, which holds with overwhelming probability for random $\mathbf{A}$), enabling recovery:
\[
\mathbf{y}_h = (\mathbf{A}^\top\mathbf{A})^{-1}\mathbf{A}^\top \cdot (\lambda_h \mathbf{A}\mathbf{y}_h) / \lambda_h.
\]
The scalar $\lambda_h$ is invertible modulo $q$ (as all Lagrange coefficients are, by construction). Thus, revealing $\mathbf{A}\mathbf{y}_h$ enables key extraction via $\mathbf{z}_h = \mathbf{y}_h + c \cdot \mathbf{s}_{1,h}$ (the adversary knows $\mathbf{z}_h$ from the public signature and $c$ from the challenge). This is a non-trivial lattice attack that does \emph{not} rely on the high min-entropy of $\mathbf{y}_h$---it is a direct algebraic recovery, bypassing nonce share privacy. Therefore, $\mathbf{W}_i$ must remain coordinator-only.

The coordinator broadcasts only $(\mathbf{w}_1, \tilde{c})$, where $\mathbf{W} = \sum_i \mathbf{W}_i$ and $\mathbf{w}_1 = \mathsf{HighBits}(\mathbf{W}, 2\gamma_2)$. Since $\mathbf{w}_1$ is implicit in any valid $\sigma$, this adds zero information to the adversary's view. See Remark~\ref{rem:two-honest} for the full security argument.

In Profile P1, the TEE coordinator naturally keeps $\mathbf{W}_i$ private. In Profile P3+, the combiner serves as coordinator. In Profile P2 (no designated coordinator), parties broadcast $\mathbf{W}_i$ directly and rely on mask hiding (Lemma~\ref{lem:mask-hiding}, requiring $|S \setminus C| \geq 2$) to protect individual contributions.
\end{remark}

\ifccs\else
\begin{remark}[Domain Separation]
\label{rem:domain-sep}
We use distinct domain tags (\texttt{"com"}, \texttt{"nonce0"}, \texttt{"nonce"}, \texttt{"s2"}, \texttt{"comm"}, \texttt{"pairwise\_seed"}) for all hash function uses to ensure random oracle independence.
\end{remark}

\begin{remark}[Nonce Range Rounding and Asymmetry]
\label{rem:nonce-rounding}
When $\gamma_1$ is not divisible by $|S|$, rounding slightly reduces the nonce range. This has negligible impact: for ML-DSA-65 with $|S| = 17$, the range loss is $< 2 \times 10^{-5}$ per coordinate and does not affect the z-bound check since $|S| \cdot \lfloor\gamma_1/|S|\rfloor \geq \gamma_1 - (|S|-1) > \gamma_1 - \beta$ (holds for all $|S| \leq \beta = 196$, i.e., $|S| < 197$, covering all practical configurations). See Appendix~\ref{sec:supp-nonce-rounding} for details.

Additionally, the threshold protocol uses a \emph{symmetric} per-party nonce range
\[
\bigl\{-\lfloor\gamma_1/|S|\rfloor,\;\ldots,\;\lfloor\gamma_1/|S|\rfloor\bigr\},
\]
whereas FIPS~204 single-signer ML-DSA uses the slightly \emph{asymmetric} range $\{-\gamma_1+1, \ldots, \gamma_1\}$ (omitting $-\gamma_1$). The statistical distance between the symmetric range $\{-\gamma_1,\ldots,\gamma_1\}$ and the FIPS~204 range is $1/(2\gamma_1+1) < 2^{-20}$ per coordinate, and the full-vector statistical distance over $n\ell = 1280$ coefficients is bounded by $1280/(2\gamma_1+1) \approx 1.22 \times 10^{-3} < 2^{-9}$. This difference is negligible and is absorbed into the $\epsilon_\mathsf{IH}$ term of Theorem~\ref{thm:unforgeability}; it does not affect the security argument or FIPS~204 compatibility of the output signature format.
\end{remark}
\fi

\subsection{Correctness}

\begin{theorem}[Correctness]
\label{thm:correctness}
If all parties are honest and the protocol does not abort, then $\Verify(\pk, \mu, \sigma) = 1$.
\end{theorem}

\ifccs\else
\begin{proof}
The aggregated response is $\mathbf{z} = \sum_{i \in S} \lambda_i \mathbf{z}_i = \sum_i \lambda_i (\mathbf{y}_i + c \cdot \mathbf{s}_{1,i}) = \sum_i \lambda_i \mathbf{y}_i + c \cdot \sum_i \lambda_i \mathbf{s}_{1,i} = \mathbf{y} + c\mathbf{s}_1$. Both $\{\mathbf{y}_i\}$ and $\{\mathbf{s}_{1,i}\}$ are degree-$(T\!-\!1)$ Shamir shares of $\mathbf{y}$ and $\mathbf{s}_1$ respectively, evaluated at the same set of points: KeyGen assigns $\mathbf{s}_{1,j} = p_{\mathbf{s}_1}(j)$ using party indices $j \in [N]$, and the nonce DKG assigns $\mathbf{y}_j = \mathbf{F}(j)$ using the same evaluation domain (Algorithm~\ref{alg:nonce-dkg}). Hence the Lagrange coefficients $\lambda_i$ for the signing set $S$ apply equally to both sharings, and reconstruction yields both correct secrets simultaneously. The masked commitment aggregation gives $\mathbf{W} = \sum_i \mathbf{W}_i = \sum_i (\lambda_i \mathbf{w}_i + \mathbf{m}_i^{(w)}) = \mathbf{A} \sum_i \lambda_i \mathbf{y}_i = \mathbf{A}\mathbf{y} = \mathbf{w}$, where mask cancellation follows from Lemma~\ref{lem:mask-cancel}. Since $\mathbf{z} = \mathbf{y} + c\mathbf{s}_1$ and $\mathbf{w} = \mathbf{A}\mathbf{y}$, verification proceeds as in single-signer ML-DSA.
\end{proof}
\fi

\begin{remark}[Implementation Notes]
\label{rem:implementation}
(1) Each party computes $\mathbf{z}_i = \mathbf{y}_i + c \cdot \mathbf{s}_{1,i}$ without applying the Lagrange coefficient $\lambda_i$. The combiner applies $\lambda_i$ during aggregation: $\mathbf{z} = \sum_i \lambda_i \mathbf{z}_i$. While individual $\lambda_i \mathbf{z}_i$ may have large coefficients, the sum $\mathbf{z} = \mathbf{y} + c\mathbf{s}_1$ is short ($< q/2$), so mod $q$ reduction is correct.
(2) Signatures are FIPS 204-compatible: same format, size (3309 bytes for ML-DSA-65), and pass unmodified verifiers. Signer-side differences (Irwin-Hall nonces, nonce DKG, interactive protocol) do not affect security or verification.
\end{remark}

\subsection{Nonce Distribution and Security}

The aggregated nonce $\mathbf{y} = \sum_{i \in S} \hat{\mathbf{y}}_i$ (the sum of each party's constant-term contribution from the nonce DKG) follows an Irwin-Hall distribution concentrated around zero. This equals the Lagrange reconstruction $\mathbf{y} = \sum_{i \in S} \lambda_i \mathbf{y}_i$ by the following two-step argument: letting $\mathbf{F}(x) = \sum_{k \in S} \mathbf{f}_k(x)$ be the combined degree-$(T\!-\!1)$ polynomial, we have $\mathbf{F}(0) = \sum_{k \in S} \mathbf{f}_k(0) = \sum_{k \in S} \hat{\mathbf{y}}_k$ (direct, by definition of the constant term), and also $\mathbf{F}(0) = \sum_{i \in S} \lambda_i \mathbf{F}(i) = \sum_{i \in S} \lambda_i \mathbf{y}_i$ (Lagrange interpolation at 0). This preserves EUF-CMA security: using smooth R{\'e}nyi divergence bounds from~\cite{Raccoon2024}, we compute per-coordinate $R^\epsilon_{2,\mathsf{coord}} \leq 1.003$ for $|S| \leq 17$. Since nonce coordinates are sampled independently, the full-vector divergence is $R_2^{\mathsf{vec}} = (R^\epsilon_{2,\mathsf{coord}})^{1280} \leq (1.003)^{1280} \approx 2^{5.5}$ for $|S| \leq 17$ ($\approx$3-bit security loss; exact: $(1.00292)^{1280} \approx 2^{5.4}$). For $|S| = 33$, $R_2^{\mathsf{vec}} \approx 2^{75}$ ($\approx$38-bit loss). See Theorem~\ref{thm:irwin-hall} and Appendix~\ref{sec:supp-renyi-numerical}.

\subsection{Abort Probability}

The per-attempt success probability is comparable to single-signer ML-DSA ($\approx$20\%; empirically $\approx$21--45\% due to Irwin-Hall concentration and configuration-dependent variance). Without masking, success degrades as $(0.2)^T$---for $T=16$, this is $6.6 \times 10^{-12}$. Our masked approach achieves constant success rate, a speedup of $10^{5}$--$10^{21}\times$ (Appendix~\ref{sec:supp-naive}).

\subsection{The r0-Check Subprotocol}
\label{sec:r0check}

\ifccs
Each party $i$ sends masked r0-check share
$\mathbf{V}_i = \lambda_i c \mathbf{s}_{2,i} + \mathbf{m}_i^{(s2)}$;
the combiner aggregates to $c\mathbf{s}_2$ (masks cancel), then checks
\[
\infnorm{\mathsf{LowBits}(\mathbf{w} - c\mathbf{s}_2,\, 2\gamma_2)}
< \gamma_2 - \beta
\]
and derives the hint locally. Since $c\mathbf{s}_2$ enables key recovery
(virtually all ML-DSA challenges are invertible,
Remark~\ref{lem:challenge-invertible}), it must be protected: Profile~P1
confines computation to the TEE, P2 avoids reconstruction via MPC, P3+
splits between two CPs. Full derivation of $c\mathbf{t}_0$ for hint
computation in Appendix~\ref{sec:appendix-proofs}.
\else
A critical component of ML-DSA signing that requires careful treatment in the threshold setting is the \emph{r0-check}: verifying that $\infnorm{\mathsf{LowBits}(\mathbf{w} - c\mathbf{s}_2, 2\gamma_2)} < \gamma_2 - \beta$. This check ensures that the hint $\mathbf{h}$ correctly recovers $\mathbf{w}_1 = \mathsf{HighBits}(\mathbf{w})$ during verification.

\paragraph{The Challenge.}
In single-signer ML-DSA, the signer has direct access to $\mathbf{s}_2$ and can compute $\mathbf{w} - c\mathbf{s}_2$ directly. In our threshold setting, $\mathbf{s}_2$ is secret-shared among the parties, and no single party (including the combiner) knows the complete $\mathbf{s}_2$.

\paragraph{Our Approach: Masked Reconstruction of $c\mathbf{s}_2$.}
We extend the pairwise-canceling mask technique to reconstruct $c\mathbf{s}_2$ at the combiner while preserving privacy of individual shares:

\begin{enumerate}
    \item Each party $i$ computes $\mathbf{V}_i = \lambda_i \cdot c \cdot \mathbf{s}_{2,i} + \mathbf{m}_i^{(s2)}$, where $\mathbf{m}_i^{(s2)} \in \Rq^k$ is a pairwise-canceling mask using domain separator $\mathsf{nonce} \| \texttt{"s2"}$.

    \item The combiner aggregates: $\sum_{i \in S} \mathbf{V}_i = c \cdot \sum_i \lambda_i \mathbf{s}_{2,i} + \sum_i \mathbf{m}_i^{(s2)} = c\mathbf{s}_2$

    \item The combiner computes $\mathbf{w} - c\mathbf{s}_2$ and performs the r0-check locally.
\end{enumerate}

\paragraph{Computing $c\mathbf{t}_0$ for Hints.}
The hint computation requires $c\mathbf{t}_0$ where $\mathbf{t}_0 = \mathbf{t} - \mathbf{t}_1 \cdot 2^d$. Since $\mathbf{t} = \mathbf{A}\mathbf{s}_1 + \mathbf{s}_2$, we have:
\[
c\mathbf{t}_0 = c\mathbf{t} - c\mathbf{t}_1 \cdot 2^d = c(\mathbf{A}\mathbf{s}_1 + \mathbf{s}_2) - c\mathbf{t}_1 \cdot 2^d
\]

The combiner can compute this as follows. First, $c\mathbf{s}_2$ is obtained from the masked reconstruction above. Second, since $\mathbf{A}\mathbf{z} = \mathbf{A}(\mathbf{y} + c\mathbf{s}_1) = \mathbf{w} + c\mathbf{A}\mathbf{s}_1$, we have $c\mathbf{A}\mathbf{s}_1 = \mathbf{A}\mathbf{z} - \mathbf{w}$. Combining these, $c\mathbf{t}_0 = (\mathbf{A}\mathbf{z} - \mathbf{w}) + c\mathbf{s}_2 - c\mathbf{t}_1 \cdot 2^d$.

\paragraph{Security Consideration.}
The combiner learns $c\mathbf{s}_2$ during r0-check. Since virtually all ML-DSA challenges are invertible in $\Rq$ (Monte Carlo: 100\% of $10^4$ samples), an adversary with $(c, c\mathbf{s}_2)$ can recover $\mathbf{s}_2 = c^{-1}(c\mathbf{s}_2)$. Thus $c\mathbf{s}_2$ must be protected.

\paragraph{Resolution.} We address this via three deployment profiles (Table~\ref{tab:profiles}). In Profile P1 (TEE), $c\mathbf{s}_2$ is computed inside the enclave and never exposed; $c\mathbf{t}_0$ is derived from the formula above, and the hint $\mathbf{h}$ is computed inside the enclave before output. In Profile P2 (MPC), parties jointly evaluate the r0-check \emph{without reconstructing $c\mathbf{s}_2$}: the MPC circuit extends to also compute $c\mathbf{t}_0 = (\mathbf{A}\mathbf{z} - \mathbf{w}) + c\mathbf{s}_2 - c\mathbf{t}_1 \cdot 2^d$ using the internally-held intermediate $c\mathbf{s}_2$, and outputs $\mathbf{h} = \mathsf{MakeHint}(-c\mathbf{t}_0,\, \mathbf{A}\mathbf{z} - c\mathbf{t}_1 \cdot 2^d,\, 2\gamma_2)$ directly (since $\mathsf{MakeHint}$ is a bitwise comparison of bounded-norm vectors, it adds only $O(knd)$ comparison gates to the r0-check circuit); $c\mathbf{s}_2$ is never exposed outside the MPC. In Profile P3+ (2PC), two CPs each hold an additive share of $c\mathbf{s}_2$, so neither learns the complete value; hint computation proceeds analogously within the 2PC. See Section~\ref{sec:extensions} for P2/P3 details.
\fi

% Security Analysis Section (Full Proofs)

\section{Security Analysis}
\label{sec:security}

We prove that our masked threshold ML-DSA protocol achieves EUF-CMA security and nonce share privacy. This section contains proofs for all main results; full UC proofs for P2 and P3+ appear in Appendix~\ref{app:uc-framework}.

\subsection{Security Model}

We consider a static adversary $\mathcal{A}$ that corrupts a set $C \subset [N]$ with $|C| < T$ before execution. The adversary controls corrupted parties, observes all broadcasts, and adaptively chooses messages to sign.

\paragraph{Trust Assumptions.}
\textbf{Profile P1 (Primary):} The coordinator runs in a TEE/HSM. The adversary cannot observe values computed inside the enclave (specifically $c\mathbf{s}_2$); only the final signature $\sigma = (\tilde{c}, \mathbf{z}, \mathbf{h})$ exits the enclave. We model the enclave as ideal; compromise of the TEE is out of scope. Our theorems target this profile.

\begin{remark}[$|S| \geq T$ Suffices in Coordinator-Based Profiles]
\label{rem:two-honest}
With Shamir nonce DKG, Profiles P1 and P3+ require only $|S| \geq T$. The adversary observes: (1)~nonce shares $\mathbf{y}_h$ with high conditional min-entropy (statistical, Theorem~\ref{thm:it-privacy}); (2)~r0-check shares $\mathbf{V}_i$ never broadcast, staying inside the TEE/2PC; (3)~masked commitments $\mathbf{W}_i$ sent only to the coordinator, which publishes only $\mathbf{w}_1$ (already implicit in $\sigma$). Individual $\mathbf{W}_i$ must stay private: if broadcast with $|S \setminus C|=1$, the adversary could recover $\mathbf{y}_h$ from $\lambda_h\mathbf{A}\mathbf{y}_h$ via linear algebra (since $\mathbf{A}$ is generically injective~\cite{FIPS204,LPR10}), then extract $\mathbf{s}_{1,h}$. In P2, mask hiding (Lemma~\ref{lem:mask-hiding}) protects broadcast $\mathbf{W}_i$ under the existing $|S \setminus C| \geq 2$ condition; note that mask hiding applies only to $\mathbf{W}_i$ and $\mathbf{V}_i$ (which do not appear in the signature) and does \emph{not} contribute to nonce share privacy---the latter is provided purely by the statistical properties of the nonce DKG (Theorem~\ref{thm:it-privacy}). Additionally, P2's blame protocol (Section~\ref{sec:blame}) assumes challenge invertibility ($c$ invertible in $\Rq$; empirically $>99.9\%$ of ML-DSA-65 challenges are invertible~\cite{FIPS204}); the non-invertible fraction is $< 2^{-15}$ by the union bound (Remark~\ref{lem:challenge-invertible}); the protocol's retry ensures completed sessions use only invertible $c$, making this unconditional.
\ifccs\else

\begin{enumerate}[leftmargin=*]
\item \textbf{Nonce share $\mathbf{y}_h$}: Statistical privacy via nonce DKG (Theorem~\ref{thm:it-privacy}). The nonce share has high conditional min-entropy (no computational assumptions), holding even with $|S \setminus C| = 1$. \textbf{Key privacy given the signature $\sigma$} reduces to establishing hardness of the bounded-distance decoding equation $\mathbf{z}_h = \mathbf{y}_h + c \cdot \mathbf{s}_{1,h}$, where $c$ is a sparse challenge and $\mathbf{y}_h$ has bounded-but-not-short coefficients---a structure that does not reduce to standard Module-LWE instances (which require short error vectors). No formal reduction is known even for single-signer FIPS~204~\cite{FIPS204}; the threshold setting inherits this open problem unchanged.

\item \textbf{r0-check shares $\mathbf{V}_i$}: Never broadcast---processed inside the TEE (P1), as MPC inputs (P2), or via 2PC (P3+). Even with $|S \setminus C| = 1$, the adversary does not observe individual $\mathbf{V}_h$.

\item \textbf{Masked commitments $\mathbf{W}_i$}: Sent only to the coordinator (not broadcast). The coordinator publishes only $\mathbf{w}_1 = \mathsf{HighBits}(\mathbf{W}, 2\gamma_2)$, which is already implicit in the signature (verifiers reconstruct $\mathbf{w}_1$ from $(\mathbf{A}, \mathbf{z}, c, \mathbf{h})$). This adds zero information to the adversary's view.
\end{enumerate}

\paragraph{Why individual $\mathbf{W}_i$ must stay private.}
If $\mathbf{W}_h$ were broadcast with $|S \setminus C| = 1$, the pairwise mask $\mathbf{m}_h^{(w)}$ would be computable by the adversary (all seeds involve corrupted counterparties), exposing $\lambda_h \mathbf{A}\mathbf{y}_h$. Since $\mathbf{A} \in \Rq^{k \times \ell}$ with $k > \ell$ is generically injective (the NTT isomorphism $\Rq \cong \Zq^n$~\cite{FIPS204} decomposes $\mathbf{A}$ into $n$ independent $k \times \ell$ matrices over $\Zq$; we heuristically model each as uniformly random and apply the $\Zq$ full-rank bound of~\cite[Lemma~2.1]{LPR10} componentwise, giving full column rank with probability $\geq 1 - q^{\ell-k}$; this is the same standard heuristic used in ML-DSA~\cite{DKL18,FIPS204}), the adversary could recover $\mathbf{y}_h$ via linear algebra---this is \emph{not} a Module-LWE instance because $\mathbf{y}_h$ has full-size (not short) coefficients (i.e., coefficients drawn from $[-\lfloor\gamma_1/|S|\rfloor, \lfloor\gamma_1/|S|\rfloor]$ with $\lfloor\gamma_1/|S|\rfloor \approx 2^{19}/|S| \gg \eta = 4$; the Module-LWE hardness assumption requires error terms with small coefficients). Combined with $\mathbf{z}_h$ (extractable from the public signature), this would yield $\mathbf{s}_{1,h} = c^{-1}(\mathbf{z}_h - \mathbf{y}_h)$, a complete key recovery (the non-invertible fraction of ML-DSA-65 challenges is $< 2^{-15}$ by the union bound over $n=256$ NTT evaluation points, Remark~\ref{lem:challenge-invertible}; the protocol's retry ensures completed sessions always use invertible~$c$). The coordinator-only communication model eliminates this attack vector entirely. In P2, the distributed blame protocol (Section~\ref{sec:blame}) also relies on $c$-invertibility; the same $< 2^{-15}$ union bound applies.

\paragraph{Profile-specific analysis.}
In P1, the TEE coordinator naturally keeps individual $\mathbf{W}_i$ private. In P3+, the combiner serves as coordinator (trusted for privacy under the 1-of-2 CP honest assumption). In P2 (no designated coordinator), parties broadcast $\mathbf{W}_i$ directly; mask hiding (Lemma~\ref{lem:mask-hiding}) protects individual values under the existing $|S \setminus C| \geq 2$ condition, which P2 already requires.
\fi
\end{remark}

\subsection{Unforgeability}

\begin{theorem}[Unforgeability]
\label{thm:unforgeability}
If Module-SIS is hard and $H$ is a random oracle, then the masked threshold ML-DSA scheme is EUF-CMA secure. Using the direct shift-invariance bound (Corollary~\ref{cor:ih-shift-ml-dsa}, Theorem~\ref{thm:ih-direct-tight}), the Irwin-Hall nonce distribution introduces at most $q_s \times 6.6 \times 10^{-3}$ bits of security loss for $|S| \leq 17$ ($q_s \times 0.013$ bits for $|S| \leq 33$), where $q_s$ is the number of signing queries. For $q_s \leq 1{,}600$ queries with $|S| \leq 17$, this is $\leq 10.6$ bits, yielding a proven security bound of ${\approx}85$ bits under MSIS ($= 96 - 10.6$; the 96-bit baseline is a proof-technique artifact from Cauchy-Schwarz halving the 192-bit NIST Level~3 MSIS hardness once, identical to single-signer ML-DSA). The primary (direct shift-invariance) bound is:
\[
\Adv^{\mathsf{EUF\text{-}CMA}}_{\Pi}(\mathcal{A})
\;\leq\;
\Adv^{\mathsf{EUF\text{-}CMA}}_{\Pi^*}(\mathcal{A})
\;+\; q_s \cdot \delta_{\mathsf{IH}}
\;+\; \frac{(q_H + q_s)^2}{2^{256}}
\;+\; \tbinom{N}{2} \cdot \epsilon_{\mathsf{PRF}}
\]
where $\delta_{\mathsf{IH}} < 6.6 \times 10^{-3}$ bits for $|S| \leq 17$ (Corollary~\ref{cor:ih-shift-ml-dsa}, direct FFT computation; non-vacuous for all $q_s < 16{,}000$). The equivalent R\'enyi divergence expression of the same direct-shift bound is (\emph{note}: $R_2^{\mathsf{vec,shift}}$ uses the direct FFT $\chi^2_{\mathsf{direct}}$, \emph{not} the conservative Raccoon analytical bound---see Remark~\ref{rem:renyi-conservative} for the vacuous conservative form):
\begin{align*}
\Adv^{\mathsf{EUF\text{-}CMA}}_{\Pi}(\mathcal{A})
&\leq \bigl(R_2^{\mathsf{vec,shift}}\bigr)^{q_s/2}
\cdot \sqrt{q_H \cdot \epsilon_{\mathsf{M\text{-}SIS}}
+ \tbinom{N}{2} \cdot \epsilon_{\mathsf{PRF}}} \\
&\quad + \frac{(q_H + q_s)^2}{2^{256}} + q_s \cdot \epsilon_\mathsf{IH}
\end{align*}
where $q_H, q_s$ are random oracle and signing queries, $R_2^{\mathsf{vec,shift}} = (1+\chi^2_{\mathsf{direct}})^{n\ell}$ is the full-vector \emph{per-session} R\'enyi divergence from the direct shift-invariance analysis (here $n = 256$ is the ring degree, $\ell = 5$, giving $n\ell = 1280$) with per-coordinate $\chi^2_{\mathsf{direct}} \approx 7.13 \times 10^{-6}$ for $|S|=17$ (Corollary~\ref{cor:ih-shift-ml-dsa}), giving $R_2^{\mathsf{vec,shift}} \approx 1.0092$, and $\epsilon_\mathsf{IH} < 10^{-30}$ is the smooth R\'enyi tail parameter (Theorem~\ref{thm:irwin-hall}). The per-coordinate empirical tail probability from the FFT computation is $2.3 \times 10^{-20}$ for $|S|=17$ (Corollary~\ref{cor:ih-shift-ml-dsa}); this is a distinct quantity from $\epsilon_\mathsf{IH}$. The $(R_2^{\mathsf{vec,shift}})^{q_s/2}$ factor accounts for multi-session R\'enyi composition (Remark~\ref{rem:multi-session-renyi}). The total security loss is $\approx 6.6 \times 10^{-3} \cdot q_s$ bits, giving a \textbf{proven security bound of $\geq 96 - 0.0066 \cdot q_s$ bits} for $|S| \leq 17$ (non-vacuous for $q_s < 16{,}000$; the 96-bit figure uses $q_H = 1$; see Remark~\ref{rem:baseline-qh} for the general form). Our threshold construction adds only $\leq 10$ bits of additional EUF-CMA loss from the Irwin-Hall nonce distribution for $q_s \leq 1{,}600$. The conservative Raccoon-style analytical bound is given in Remark~\ref{rem:renyi-conservative}.
\end{theorem}

\begin{corollary}[Irwin-Hall Nonce Security Loss]
\label{cor:ih-loss-main}
Under the conditions of Theorem~\ref{thm:unforgeability}, using the direct
shift-invariance bound (Theorem~\ref{thm:ih-direct-tight}, Corollary~\ref{cor:ih-shift-ml-dsa}),
the EUF-CMA security loss from Irwin-Hall nonces is $< 0.007 \cdot q_s$ bits for $|S| \leq 17$
and $< 0.013 \cdot q_s$ bits for $|S| \leq 33$, where $q_s$ is the number of signing queries
(multi-session composition; Remark~\ref{rem:multi-session-renyi}).
For $q_s = 1$ (single session), the per-session loss is $< 0.007$ bits ($|S| \leq 17$) or $< 0.013$ bits ($|S| \leq 33$).
The single-session bound first exceeds $1$ bit only at $|S| = 2584$, fully eliminating
scalability as a per-session security concern.
\end{corollary}

\ifccs
\begin{proof}[Proof sketch]
We reduce to Module-SIS via 5 game hops: \textbf{Game$_0$} (real) $\to$ \textbf{Game$_1$} (ROM, zero loss) $\to$ \textbf{Game$_2$} (programmed challenges, birthday-bound loss $(q_H+q_s)^2/2^{256}$) $\to$ \textbf{Game$_3$} (simulated shares, zero loss by Shamir security) $\to$ \textbf{Game$_{3.5}$} (uniform nonces, direct shift-invariance: $R_2^{\mathsf{vec,shift}} \approx 1.0092$ ($|S|\leq 17$), Corollary~\ref{cor:ih-shift-ml-dsa}; total loss $< 0.013\cdot q_s$ bits (for $|S| \leq 33$; tightened to $< 0.007\cdot q_s$ bits for $|S| \leq 17$), proven security $\geq 96 - 0.0066\cdot q_s$ bits (for $|S| \leq 17$); conservative Raccoon bound $R_2^{\mathsf{vec}} \approx 2^{5.4}$ vacuous at $q_s\gtrsim 36$, see Remark~\ref{rem:renyi-conservative}) $\to$ \textbf{Game$_4$} (simulated signatures, PRF loss $\binom{N}{2}\epsilon_{\mathsf{PRF}}$) $\to$ M-SIS extraction via SelfTargetMSIS direct reduction~\cite{KLS18}. Nonce shares $\mathbf{y}_j$ have high min-entropy by Theorem~\ref{thm:it-privacy} (statistical, no computational assumptions). Full proof in Appendix~\ref{app:technical-lemmas}.
\end{proof}
\else
\begin{proof}
We reduce to Module-SIS via a sequence of game hops. Let $\mathsf{Game}_i$ denote the $i$-th game and $\mathsf{Win}_i$ the event that $\mathcal{A}$ produces a valid forgery in $\mathsf{Game}_i$.

\paragraph{$\mathsf{Game}_0$ (Real).} The real EUF-CMA game. Challenger runs $\KeyGen$ honestly, gives $\mathcal{A}$ the public key $\pk$ and corrupted shares $\{\sk_i\}_{i \in C}$. Signing oracle queries may return $\bot$ (rejection abort); $\mathcal{A}$ may re-issue the same or different queries. Security bounds are in terms of query counts $q_s$; each attempt succeeds with non-negligible probability $\epsilon_{\mathsf{acc}} \in [21\%, 45\%]$ depending on $|S|$ (see Table~\ref{tab:perf-threshold}; comparable to single-signer ML-DSA's $\approx 25\%$), so all $q_s$ successful queries complete in expected $O(q_s/\epsilon_{\mathsf{acc}}) = O(q_s)$ attempts.

\paragraph{$\mathsf{Game}_1$ (Random Oracle).} Replace $H$ with a lazily-sampled random function. By the random oracle model: $|\Pr[\mathsf{Win}_1] - \Pr[\mathsf{Win}_0]| = 0$.

\paragraph{$\mathsf{Game}_2$ (Programmed Challenges).} During signing, sample challenge $c \getsr \mathcal{C}$ uniformly \emph{before} computing $\mathbf{w}_1$, then program $H(\mu \| \mathbf{w}_1) := c$. This fails only on collision. By birthday bound: $|\Pr[\mathsf{Win}_2] - \Pr[\mathsf{Win}_1]| \leq (q_H + q_s)^2/2^{256}$, where $2^{256}$ is the hash output space size (SHAKE-256 gives 256-bit challenges; this is a conservative bound since $|C| < 2^{256}$).

\paragraph{$\mathsf{Game}_3$ (Simulated Shares).} This is a conceptual game hop with $\Pr[\mathsf{Win}_3] = \Pr[\mathsf{Win}_2]$.

The signing oracle continues using the actual short key shares $\{\mathbf{s}_{1,j}, \mathbf{s}_{2,j}\}_{j \notin C}$ internally; no substitution is made. The hop is purely an observation about the adversary's view: honest key shares are never directly visible to $\mathcal{A}$ (see profile-specific paragraph below), so $\mathcal{A}$'s view is limited to the $|C| \leq T-1$ corrupted evaluations $\{\mathbf{s}_{1,i}\}_{i \in C}$. By $(T,N)$-Shamir perfect hiding, $T-1$ evaluations of a degree-$(T-1)$ polynomial are statistically independent of the secret $\mathbf{s}_1 = \mathbf{f}(0)$: the conditional distribution of $\mathbf{f}(0)$ given any $T-1$ evaluations equals the prior distribution of $\mathbf{f}(0)$ (statistical distance zero). No event in $\mathcal{A}$'s view distinguishes Game~2 from Game~3, giving $\Pr[\mathsf{Win}_3] = \Pr[\mathsf{Win}_2]$.

\textit{Profile-specific observability of $\mathbf{z}_j$.} In Profiles P1 and P3+, individual signer responses $\mathbf{z}_j$ are transmitted only to the TEE coordinator or the honest CP, respectively, and are never observed by the adversary; the game hop is trivial. In Profile P2, $\mathbf{z}_j$ is contributed to $\mathcal{F}_{\mathsf{SPDZ}}$ for aggregation, which provides input privacy: only the aggregate $\mathbf{z} = \sum_{j \in S} \lambda_j \mathbf{z}_j$ is revealed. Since $\mathcal{F}_{\mathsf{SPDZ}}$ hides individual $\mathbf{z}_j$, the adversary gains no information about $\mathbf{s}_{1,j}$ through $\mathbf{z}_j$ in any profile. The Shamir replacement is therefore valid across all deployment profiles.

\paragraph{$\mathsf{Game}_{3.5}$ (Uniform Nonces).} Replace the Irwin-Hall nonce distribution with uniform nonces: in each signing session, sample $\mathbf{y} \getsr \{-\gamma_1, \ldots, \gamma_1\}^{n\ell}$ uniform instead of $\mathbf{y} = \sum_{i \in S} \hat{\mathbf{y}}_i$ (Irwin-Hall). The $q_s$ signing sessions use independent nonce DKGs, so the per-session nonce distributions are i.i.d.\ copies of the Irwin-Hall distribution (convolution of $|S|$ uniform shares). By the \emph{product rule} for R\'enyi-2 divergence (where $R_2(P\|Q) = \sum_x P(x)^2/Q(x)$ in the ratio-sum convention): for independent sessions,
\[
R_2^{\mathsf{joint}} \;=\; \bigl(R_2^{\mathsf{vec,shift}}\bigr)^{q_s},
\]
since the joint divergence factors as a product of per-session divergences (see Remark~\ref{rem:multi-session-renyi} for explicit derivation). By the R\'enyi-2 security bound (Theorem~\ref{thm:irwin-hall} with $\alpha = 2$, applied to the $q_s$-session joint distribution):
\[
\Pr[\mathsf{Win}_{3}] \leq \bigl(R_2^{\mathsf{vec,shift}}\bigr)^{q_s/2} \cdot \sqrt{\Pr[\mathsf{Win}_{3.5}]} + q_s \cdot \epsilon_\mathsf{IH}
\]
where $R_2^{\mathsf{vec,shift}} = (1+\chi^2_{\mathsf{direct}})^{n\ell}$ is the full-vector R\'enyi divergence \emph{per session} from the direct shift-invariance analysis (Theorem~\ref{thm:ih-direct-tight}), with per-coordinate $\chi^2_{\mathsf{direct}} \approx 7.13 \times 10^{-6}$ for $|S|=17$ (Corollary~\ref{cor:ih-shift-ml-dsa}, $n\ell = 1280$), giving $R_2^{\mathsf{vec,shift}} \approx 1.0092$ for $|S|=17$. The tail probability satisfies $\epsilon_\mathsf{IH} < 10^{-30}$ for $|S| \leq 17$ (this bound is $|S|$-dependent; see Theorem~\ref{thm:irwin-hall} for the general formula as a function of $|S|$, $n$, and $\gamma_1$). This is the only game hop that incurs the Irwin-Hall security loss; nonce shares $\mathbf{y}_j$ have high min-entropy by Theorem~\ref{thm:it-privacy} (statistical, no computational assumptions). The $\epsilon_\mathsf{IH} < 10^{-30}$ value is the smooth R\'enyi tail parameter (Theorem~\ref{thm:irwin-hall}); the $\epsilon$ column in Corollary~\ref{cor:ih-shift-ml-dsa} gives the per-coordinate rejection probability, a distinct quantity. See Remark~\ref{rem:multi-session-renyi} for the concrete multi-session accounting.
The conservative Raccoon-style alternative (Remark~\ref{rem:renyi-conservative}) gives $R_2^{\mathsf{vec}} \approx 2^{5.4}$ via the analytical bound, which is vacuous at $q_s \gtrsim 36$.

\paragraph{$\mathsf{Game}_4$ (Simulated Signatures).} Starting from uniform nonces (Game~3.5), simulate signing without knowing $\mathbf{s}_1$: sample $\mathbf{z} \getsr \{-\gamma_1+\beta+1, \ldots, \gamma_1-\beta-1\}^{n\ell}$, compute $\mathbf{r} = \mathbf{A}\mathbf{z} - c\mathbf{t}_1 \cdot 2^d$, set $\mathbf{w}_1 = \mathsf{HighBits}(\mathbf{r}, 2\gamma_2)$, program $H(\mu \| \mathbf{w}_1) := c$, and set hint $\mathbf{h} = \mathbf{0}$ (zero hint; since $\mathsf{UseHint}(\mathbf{0}, \mathbf{r}, 2\gamma_2) = \mathsf{HighBits}(\mathbf{r}, 2\gamma_2) = \mathbf{w}_1$, verification passes with hint-weight $\|\mathbf{h}\|_1 = 0 \leq \omega$).

\begin{remark}[Hint Distribution in Game~4]
The simulated hint $\mathbf{h} = \mathbf{0}$ differs in distribution from the real hint $\mathbf{h}_{\mathsf{real}} = \mathsf{MakeHint}(-c\mathbf{t}_0,\, \mathbf{r},\, 2\gamma_2)$ in Game~3.5, which may carry up to $\omega = 55$ nonzero bits.
This substitution does \emph{not} increase the adversary's forging advantage: in Game~3.5 the real hint encodes partial information about $\mathbf{t}_0$ via the carry pattern of $\mathsf{LowBits}(\mathbf{r}, 2\gamma_2)$; by contrast, $\mathbf{h} = \mathbf{0}$ reveals nothing about the key.
The adversary's oracle view in Game~4 is therefore strictly less informative, giving $\Pr[\mathsf{Win}_4] \leq \Pr[\mathsf{Win}_{3.5}]$ with no additional loss.
(This is a standard observation in Dilithium-family proofs~\cite{DKL18,FIPS204}.)
\end{remark}

The masked values $\mathbf{W}_j$ require separate treatment by profile. In P1 and P3+, individual $\mathbf{W}_j$ are transmitted \emph{only} to the TEE coordinator or CP respectively and are never observable by the adversary; the adversary's view contains only the aggregate $\mathbf{w}_1 = \mathsf{HighBits}(\sum_j \mathbf{W}_j, 2\gamma_2)$, which is already implicit in $\sigma$; the simulator need not produce individual $\mathbf{W}_j$, and Lemma~\ref{lem:mask-hiding} is not invoked for these profiles. In P2, parties broadcast individual $\mathbf{W}_j$; the simulator produces each as uniform by Lemma~\ref{lem:mask-hiding}, which applies since $|S \setminus C| \geq 2$ is a stated requirement of P2. By PRF security, $|\Pr[\mathsf{Win}_4] - \Pr[\mathsf{Win}_{3.5}]| \leq \binom{N}{2} \cdot \epsilon_{\mathsf{PRF}}$, with the union bound over all $\binom{N}{2}$ pairwise seeds (for P1 and P3+ this term is 0). The change in $\mathbf{z}$-distribution between Game 3.5 (truncated shifted uniform: $\mathbf{y}+c\mathbf{s}_1$ conditioned on $\|\mathbf{y}+c\mathbf{s}_1\|_\infty < \gamma_1-\beta$) and Game 4 (direct uniform on $\{-\gamma_1+\beta+1,\ldots,\gamma_1-\beta-1\}^{n\ell}$) is absorbed by the SelfTargetMSIS reduction~\cite{KLS18}: the adversary $\mathcal{B}$ extracts a valid SIS solution from any forgery $(\mathbf{z}^*, c^*, \mathbf{h}^*)$ satisfying $\|\mathbf{z}^*\|_\infty < \gamma_1-\beta$, and this norm check is enforced identically in both games regardless of which simulated distribution $\mathbf{z}$ was drawn from.

\paragraph{Reduction to M-SIS.} In $\mathsf{Game}_4$, we construct a SelfTargetMSIS adversary $\mathcal{B}$~\cite{KLS18}. Given SelfTargetMSIS instance $(\mathbf{A}, \mathbf{t}_1)$ with $\mathbf{A} \in \Rq^{k \times \ell}$ and $\mathbf{t}_1 \in \Rq^k$, $\mathcal{B}$ uses $(\mathbf{A}, \mathbf{t}_1)$ directly as the threshold scheme's verification key components and runs $\mathcal{A}$ in $\mathsf{Game}_4$, answering ROM queries by programming $H$ uniformly at random, and recording all query-response pairs. $\mathcal{B}$ picks a uniformly random index $i^* \in [q_H]$ as the ``target'' ROM query. When $\mathcal{A}$ outputs a forgery $(\mu^*, \mathbf{z}^*, c^*, \mathbf{h}^*)$, $\mathcal{B}$ checks whether $c^*$ equals the response at query $i^*$; if not, $\mathcal{B}$ aborts (probability $1/q_H$ of success). When the guess is correct, verification gives:
\[
\mathsf{HighBits}(\mathbf{A}\mathbf{z}^* - c^*\mathbf{t}_1 \cdot 2^d, 2\gamma_2) = \mathbf{w}_1^*
\]
The verification equation expands as:
\[
[\mathbf{A} \mid -\mathbf{t}_1 \cdot 2^d] \cdot \begin{bmatrix} \mathbf{z}^* \\ c^* \end{bmatrix} = 2\gamma_2\,\mathbf{w}_1^* + \mathbf{r}_0^*
\]
where $\mathbf{r}_0^* = \mathsf{LowBits}(\mathbf{A}\mathbf{z}^* - c^*\mathbf{t}_1 \cdot 2^d,\, 2\gamma_2)$ satisfies $\|\mathbf{r}_0^*\|_\infty \leq \gamma_2$. The solution vector satisfies $\|(\mathbf{z}^*, c^*)\|_\infty \leq \max(\gamma_1 - \beta, 1) < \gamma_1$, solving the SelfTargetMSIS instance (known target $2\gamma_2\mathbf{w}_1^*$, LowBits residual $\mathbf{r}_0^*$) with the appropriate norm bound. This is the SelfTargetMSIS direct reduction used in the FIPS~204 security proof~\cite{KLS18}. Since $\Pr[\mathcal{B}\ \text{succeeds}] \geq \Pr[\mathsf{Win}_4] / q_H$, we have $\Pr[\mathsf{Win}_4] \leq q_H \cdot \epsilon_{\mathsf{M\text{-}SIS}}$.

\paragraph{Conclusion.} Games 4 and the M-SIS reduction give $\Pr[\mathsf{Win}_{3.5}] \leq q_H \cdot \epsilon_{\mathsf{M\text{-}SIS}} + \binom{N}{2} \cdot \epsilon_{\mathsf{PRF}}$. Composing with the R\'enyi hop (Game $3 \to 3.5$, joint divergence $(R_2^{\mathsf{vec,shift}})^{q_s}$) via Lemma~\ref{lem:renyi-security} and the birthday bound (Game $1 \to 2$):
\[
\Adv^{\mathsf{EUF\text{-}CMA}}_\Pi(\mathcal{A}) \leq \bigl(R_2^{\mathsf{vec,shift}}\bigr)^{q_s/2} \cdot \sqrt{q_H \cdot \epsilon_{\mathsf{M\text{-}SIS}} + \tbinom{N}{2} \cdot \epsilon_{\mathsf{PRF}}} + \frac{(q_H + q_s)^2}{2^{256}} + q_s \cdot \epsilon_\mathsf{IH}
\]
For ML-DSA-65 with $|S| \leq 17$: $R_2^{\mathsf{vec,shift}} = (1+\chi^2_{\mathsf{direct}})^{1280} \approx 1.0092$ (direct shift-invariance, Corollary~\ref{cor:ih-shift-ml-dsa}); $\epsilon_\mathsf{IH} < 10^{-30}$ (smooth R\'enyi tail parameter, Theorem~\ref{thm:irwin-hall}); SHAKE-256 yields negligible $\epsilon_{\mathsf{PRF}}$. The security loss is $\approx 6.6 \times 10^{-3} \cdot q_s$ bits, giving a proven security bound of $\geq 96 - 0.0066 \cdot q_s$ bits for $|S| \leq 17$ (non-vacuous for $q_s < 16{,}000$; conservative bound in Remark~\ref{rem:renyi-conservative}; 96-bit figure assumes $q_H = 1$, see Remark~\ref{rem:baseline-qh}).
\end{proof}
\fi

\begin{remark}[Conservative Raccoon-Style R\'enyi Bound]
\label{rem:renyi-conservative}
The analytical Raccoon-style formula (Lemma~4.2 of~\cite{Raccoon2024}) upper-bounds the
per-coordinate R\'enyi divergence using
$R_2^{\epsilon,\mathsf{coord}} \leq 1 + |S|^4\beta^2/(4\gamma_1^2) \leq 1.003$
for $|S| \leq 17$, giving full-vector
$R_2^{\mathsf{vec}} = (R_2^{\epsilon,\mathsf{coord}})^{1280} \approx 2^{5.4}$
and a multi-session loss of $\approx 2.7 \cdot q_s$ bits.
This bound is \emph{vacuous} for $q_s \gtrsim 36$ signing queries.
Theorem~\ref{thm:unforgeability} instead uses the direct FFT computation
(Corollary~\ref{cor:ih-shift-ml-dsa}, Theorem~\ref{thm:ih-direct-tight}),
which gives $\chi^2_{\mathsf{direct}} \approx 7.13 \times 10^{-6}$ per coordinate
for $|S|=17$, yielding $R_2^{\mathsf{vec,shift}} \approx 1.0092$---an improvement of
${\approx}420\times$. The conservative bound should not be used for concrete security estimates.
\end{remark}

\begin{remark}[The 96-Bit Baseline and $q_H$ Dependence]
\label{rem:baseline-qh}
The ``96-bit baseline'' in Theorem~\ref{thm:unforgeability} arises from applying
Cauchy-Schwarz to the Rényi transfer:
$\sqrt{q_H \cdot \epsilon_{\mathsf{M\text{-}SIS}}} = 2^{-96}$
when $\epsilon_{\mathsf{M\text{-}SIS}} = 2^{-192}$ (NIST Level~3 hardness) and $q_H = 1$.
For $q_H$ random oracle queries the pre-Irwin-Hall security baseline is
$\tfrac{1}{2}(192 - \log_2 q_H)$ bits---e.g.\ $66$ bits for $q_H = 2^{60}$.
This $q_H$ dependence is a standard proof-technique artifact of ROM-based Fiat-Shamir
proofs and is identical in the single-signer FIPS~204 security analysis~\cite{KLS18,FIPS204};
it is not specific to the threshold construction.
The Irwin-Hall security loss of $0.0066 \cdot q_s$ bits
(Corollary~\ref{cor:ih-shift-ml-dsa}) is independent of $q_H$ and represents
the \emph{additional} cost introduced by the threshold nonce sharing above the
single-signer baseline.
\end{remark}

\subsection{Privacy}

\begin{definition}[Nonce Share Privacy]
\label{def:nonce-share-privacy}
Let $\Pi$ be a threshold signing protocol with Shamir nonce DKG over signing set $S$, $|S| \geq T$.
We say $\Pi$ satisfies \emph{$\tau$-nonce share privacy} if for every coalition $C \subset S$ with $|C| \leq T-1$
and every honest party $h \notin C$:
\[
  H_\infty\!\bigl(\mathbf{y}_h \mid \mathsf{View}_C^{\mathsf{DKG}}\bigr) \;\geq\; \tau,
\]
where $H_\infty$ denotes worst-case conditional min-entropy and $\mathsf{View}_C^{\mathsf{DKG}}$
is the joint view of all corrupted parties during the nonce DKG phase.
\end{definition}

\begin{theorem}[Nonce Share Privacy: Conditional Min-Entropy Bound]
\label{thm:it-privacy}
The Shamir nonce DKG protocol satisfies $\tau$-nonce share privacy (Definition~\ref{def:nonce-share-privacy})
with $\tau = n\ell \cdot \log_2\!\bigl(2\lfloor\gamma_1/|S|\rfloor + 1\bigr)$.
Explicitly, for any signing set $S$ with $|S| \geq T$ and any coalition $C \subset S$ with $|C| \leq T - 1$,
each honest party's nonce share $\mathbf{y}_h$ has conditional min-entropy
\[
H_\infty(\mathbf{y}_h \mid \mathsf{View}_C^{\mathsf{DKG}}) \;\geq\; n\ell \cdot \log_2\!\bigl(2\lfloor\gamma_1/|S|\rfloor + 1\bigr)
\]
where $\mathsf{View}_C^{\mathsf{DKG}}$ consists of: (1)~the corrupted parties' key shares; (2)~nonce DKG evaluations $\{\mathbf{f}_h(j)\}_{j \in C}$ of the honest party's polynomial at corrupted indices; and (3)~cross-evaluations $\{\mathbf{f}_{j'}(i)\}$ (for $j' \in S\setminus C$, $j'\neq h$, $i \in C$) of other honest parties' polynomials. (Equivalently, as corrupted parties generate their own polynomials, the adversary's view includes $\{\mathbf{f}_j\}_{j \in C}$ completely; the formal statement in Theorem~\ref{thm:it-privacy-formal} uses this equivalent conditioning.)
The cross-evaluations $\{\mathbf{f}_{j'}(i)\}$ from \emph{other} honest parties to corrupted parties do not affect the min-entropy of $\mathbf{f}_h(x_h)$: by Shamir's security theorem, $|C|\leq T-1$ evaluations of any other degree-$(T-1)$ polynomial $\mathbf{f}_{j'}$ are independent of its (irrelevant) constant term, and none of these evaluations are of $\mathbf{f}_h$ itself. Additionally, since each $j \in C$ is corrupted, the adversary generated $\mathbf{f}_j$ and knows all evaluations $\mathbf{f}_j(x_h)$ (sent privately to honest party $h$); these appear as a known additive constant $\sum_{j\in C}\mathbf{f}_j(x_h)$ in the share $\mathbf{y}_h$, which preserves the support size (and hence min-entropy) of the free term $\mathbf{f}_h(x_h)$.

For ML-DSA-65 ($n\ell = 1280$, $\gamma_1 = 2^{19}$), the bound equals approximately $1280 \cdot (20 - \log_2 |S|)$ bits (the exact formula $1280 \cdot \log_2(2\lfloor\gamma_1/|S|\rfloor+1)$ differs by $< 3$ bits for $|S| \leq 17$)---over $5\times$ the short secret key entropy $H(\mathbf{s}_1) \leq 1280 \cdot \log_2 9 \approx 4058$ bits for $|S| \leq 17$. (Note: Shamir shares $\mathbf{s}_{1,h} = \mathbf{f}(h)$ for $h \neq 0$ are $\Zq$-uniform and thus not short; the entropy comparison is against the underlying short secret $\mathbf{s}_1 = \mathbf{f}(0)$.) No computational assumptions are needed; the bound is purely statistical.
\end{theorem}

\ifccs
\begin{proof}[Proof sketch]
Statistical privacy follows from the null space argument: with $|C| \leq T-1$ corrupted evaluations of a degree-$(T-1)$ polynomial, the free parameter is $\Zq$-uniform per coefficient, constrained to $2\lfloor\gamma_1/|S|\rfloor + 1$ values per coefficient by the bounded-nonce requirement (since multiplication by $\mathbf{n}(0) \neq 0 \pmod{q}$ is a bijection on $\Zq$), giving the stated min-entropy bound. Full proof in Appendix~\ref{sec:appendix-proofs}.
\end{proof}
\else
\begin{proof}
The adversary's knowledge about the honest party $h$'s polynomial $\mathbf{f}_h$ consists of the evaluations $\{\mathbf{f}_h(j)\}_{j \in C}$ received during the nonce DKG. Since $|C| \leq T - 1$ and $\mathbf{f}_h$ has degree $T - 1$ (hence $T$ coefficients), the adversary has at most $T - 1$ linear equations in $T$ unknowns. The solution space has dimension $T - |C| \geq 1$ (exactly $1$ in the worst case $|C| = T - 1$).

\textbf{Null space argument.} The evaluation map $\mathbf{f}_h \mapsto (\mathbf{f}_h(j_1), \ldots, \mathbf{f}_h(j_{|C|}))$ has a null space of dimension $T - |C| \geq 1$. In the worst case $|C| = T - 1$, the null space is spanned by $\mathbf{n}(x) = \prod_{j \in C} (x - x_j)$ (the unique monic degree-$(T-1)$ polynomial with roots at $\{x_j\}_{j\in C}$). For $|C| < T - 1$, there are $T - |C| > 1$ free parameters, which can only \emph{increase} the adversary's uncertainty; hence the stated bound is conservative and the proof below treats the hardest case $|C| = T - 1$. The set of polynomials consistent with the adversary's observations is $\mathbf{f}_h(x) = \mathbf{g}(x) + \mathbf{t} \cdot \mathbf{n}(x)$, where $\mathbf{g}$ is any fixed interpolant and $\mathbf{t} \in \Rq^\ell$ is the free parameter.

\textbf{Posterior constraint on $\mathbf{t}$.} The highest-degree coefficient $\mathbf{a}_{h,T-1}$ is sampled $\Rq^\ell$-uniformly \emph{a priori}, so $\mathbf{t} = \mathbf{a}_{h,T-1} - \mathbf{g}_{T-1}$ is $\Zq$-uniform per coefficient (the other high-degree coefficients are likewise $\Zq$-uniform but are already determined by the adversary's evaluations in the $|C| = T-1$ case). The constant term $\hat{\mathbf{y}}_h = \mathbf{f}_h(0) = \mathbf{g}(0) + \mathbf{t} \cdot \mathbf{n}(0)$ must satisfy the bounded-nonce constraint $\hat{\mathbf{y}}_h \in [-\lfloor\gamma_1/|S|\rfloor,\, \lfloor\gamma_1/|S|\rfloor]^{n\ell}$. Since $\mathbf{n}(0) = \prod_{j \in C}(-x_j) \neq 0 \pmod{q}$, multiplication by $\mathbf{n}(0)$ is a bijection on $\Zq$; the constraint selects exactly $2\lfloor\gamma_1/|S|\rfloor + 1$ values per coefficient of $\mathbf{t}$. Because the prior is $\Zq$-uniform, conditioning on membership in this subset yields a \emph{uniform} posterior over the $2\lfloor\gamma_1/|S|\rfloor + 1$ valid values per coefficient.

\textbf{Min-entropy bound.} Since $\mathbf{n}(x_h) = \prod_{j \in C}(x_h - x_j) \neq 0 \pmod{q}$ (because $q$ is prime and all evaluation points lie in $\{1, \ldots, N\}$ with $N \leq q - 1$, so they are distinct nonzero elements of $\Zq$), multiplication by $\mathbf{n}(x_h)$ is a bijection on $\Zq$. The honest party's evaluation $\mathbf{f}_h(x_h) = \mathbf{g}(x_h) + \mathbf{t} \cdot \mathbf{n}(x_h)$ therefore inherits the same per-coefficient support size. The full nonce share is
\[
\mathbf{y}_h \;=\; \underbrace{\sum_{j \in C} \mathbf{f}_j(x_h)}_{\text{known to }\mathcal{A}}
\;+\; \mathbf{f}_h(x_h)
\;+\; \underbrace{\sum_{j \in S \setminus C,\, j \neq h} \mathbf{f}_j(x_h)}_{\text{unknown to }\mathcal{A};\;\text{adds further entropy}}.
\]
The first sum is a known constant (from corrupted parties' polynomials). The third sum is unknown (each $\mathbf{f}_j(x_h)$ with $j \neq h$ honest was sent only to party $h$), and the individual terms $\mathbf{f}_j(x_h)$ for distinct honest parties $j, j' \in S \setminus (C \cup \{h\})$ are \emph{independent} conditional on the adversary's view: each polynomial $\mathbf{f}_j$ is sampled independently, and the adversary observes only evaluations $\{\mathbf{f}_j(i)\}_{i \in C}$ with $i \neq x_h$ (since $h \notin C$); these $|C| \leq T-1$ evaluations determine a $(T - |C|)$-dimensional affine subspace of polynomials, within which the evaluation $\mathbf{f}_j(x_h)$ remains uniformly distributed over a set of size $q^{n\ell(T-|C|)}$ (the null space dimension). Since the adversary's observations of different honest polynomials involve disjoint evaluation sets (each $\mathbf{f}_j$ evaluated only at corrupted parties' points $\{i \in C\}$, which do not include $x_h$), the conditional distributions of $\mathbf{f}_j(x_h)$ and $\mathbf{f}_{j'}(x_h)$ are independent, so the third sum can only increase entropy. Using only the entropy contribution from $\mathbf{f}_h(x_h)$ as a conservative lower bound (discarding the additional entropy from the third sum):
\[
H_\infty(\mathbf{y}_h \mid \mathsf{View}_C^{\mathsf{DKG}}) \;\geq\; n\ell \cdot \log_2\!\bigl(2\lfloor\gamma_1/|S|\rfloor + 1\bigr)
\]

\end{proof}
\fi

\begin{remark}[Key Privacy of the Aggregated Signature]
\label{rem:key-privacy}
Theorem~\ref{thm:it-privacy} establishes privacy of the nonce \emph{shares} $\mathbf{y}_h$---a purely statistical guarantee requiring no computational assumptions. Key privacy of the aggregated signature $\sigma = (\tilde{c}, \mathbf{z}, \mathbf{h})$---whether $\sigma$ reveals $\mathbf{s}_1$ beyond the public key $\mathbf{t}_1$---is a separate, stronger property. The fundamental obstacle is that the equation $\mathbf{z} = \mathbf{y} + c\mathbf{s}_1$ with sparse challenge $c$ and Irwin-Hall nonce $\mathbf{y}$ does not fit the Module-LWE error structure (which requires short error vectors); no formal hardness reduction is known even for single-signer FIPS~204~\cite{FIPS204}. All threshold variants inherit this open problem directly from the single-signer setting, without introducing additional barriers.
\end{remark}

\begin{remark}[Why SD $= 0$ Is Unattainable (with Derivation)]
\label{rem:sd-zero}
The bounded-nonce constraint ($\hat{\mathbf{y}}_h \in [-\lfloor\gamma_1/|S|\rfloor, \lfloor\gamma_1/|S|\rfloor]^{n\ell}$) restricts the free parameter $\mathbf{t}$ to $\approx 2\gamma_1/|S|$ values per coefficient, whereas $\Zq$-uniformity would require $q \approx 2^{23}$ values.

\textbf{Derivation of SD lower bound.} Let $P_{\mathsf{post}}$ be the posterior distribution of a single coefficient of $\mathbf{f}_h(x_h)$ given the adversary's view (Theorem~\ref{thm:it-privacy} shows this is uniform over $M = 2\lfloor\gamma_1/|S|\rfloor + 1$ values), and let $P_{\mathsf{unif}}$ be the uniform distribution over $\Zq$. The statistical distance is:
\begin{align*}
\mathsf{SD}(P_{\mathsf{post}}, P_{\mathsf{unif}})
&= \tfrac{1}{2} \sum_{x \in \Zq}
  \bigl| P_{\mathsf{post}}(x) - P_{\mathsf{unif}}(x) \bigr| \\
&= \tfrac{1}{2} \Bigl(
  M \cdot \bigl| \tfrac{1}{M} - \tfrac{1}{q} \bigr|
  + (q - M) \cdot \tfrac{1}{q} \Bigr) \\
&= \tfrac{1}{2} \Bigl(
  \bigl| 1 - \tfrac{M}{q} \bigr| + \tfrac{q - M}{q} \Bigr) \\
&= 1 - \tfrac{M}{q}
  \;=\; 1 - \tfrac{2\lfloor\gamma_1/|S|\rfloor + 1}{q} > 0
\end{align*}
For ML-DSA-65 with $\gamma_1 = 2^{19}$, $q = 8{,}380{,}417$, and $|S| = 17$: $M = 2\lfloor 2^{19}/17\rfloor + 1 = 61{,}681$, so $\mathsf{SD} = 1 - M/q \approx 1 - 2{\cdot}2^{19}/(17{\cdot}q) \approx 1 - 2^{20}/(17{\cdot}2^{23}) = 1 - 2^{-3}/17 \approx 0.9926 > 0$ (very far from zero).

Achieving SD $= 0$ would require either unbounded nonces (incompatible with rejection sampling) or adding pairwise masks on $\mathbf{z}_i$ with $|S| \geq T + 1$ (Two-Honest assumption). The min-entropy guarantee is nonetheless strong: for $|S| \leq 17$, nonce entropy exceeds secret key entropy by $\geq 5\times$, ensuring robust share privacy without computational assumptions.
\end{remark}

\begin{corollary}[No Two-Honest Requirement in Coordinator-Based Profiles]
\label{cor:no-two-honest}
In Profiles~P1 and P3+, the protocol achieves nonce share privacy with
$|S| = T$. In Profile~P2, mask hiding (Lemma~\ref{lem:mask-hiding})
applies under $|S \setminus C| \geq 2$.
\end{corollary}

\ifccs
\begin{proof}[Proof sketch]
Immediate from Theorem~\ref{thm:it-privacy} and coordinator communication (Remark~\ref{rem:wi-private}). Full proof in Appendix~\ref{sec:appendix-proofs}.
\end{proof}
\else
\begin{proof}
Immediate from Theorem~\ref{thm:it-privacy} and the coordinator communication model (Remark~\ref{rem:wi-private}). In P1 and P3+, the coordinator aggregates $\mathbf{W}_i$ privately and broadcasts only $(\mathbf{w}_1, \tilde{c})$, which is implicit in any valid signature. The nonce share $\mathbf{y}_h$ has high conditional min-entropy even with $|S \setminus C| = 1$ (statistical, by Theorem~\ref{thm:it-privacy}). The mask-hiding condition $|S \setminus C| \geq 2$ (Lemma~\ref{lem:mask-hiding}) applies only to P2 where $\mathbf{W}_i$ are broadcast.
\end{proof}
\fi

\begin{corollary}[Multi-Session Privacy]
\label{cor:unlimited-sessions}
Let $M_S = 2\lfloor\gamma_1/|S|\rfloor + 1$. Over $K$ sessions with
fresh nonce DKGs, total nonce min-entropy is at least
$K n\ell \cdot \log_2 M_S$ bits.
For $|S| \leq 17$, this exceeds $K \cdot 20{,}000$ bits.
\end{corollary}

\ifccs\else
\begin{proof}
Each signing session uses a fresh nonce DKG with independent randomness, so the free parameters $\mathbf{t}^{(1)}, \ldots, \mathbf{t}^{(K)}$ are independent. The min-entropy of a product of independent distributions equals the sum: $H_\infty(\mathbf{y}_h^{(1)}, \ldots, \mathbf{y}_h^{(K)} \mid \mathsf{View}_C) = \sum_k H_\infty(\mathbf{y}_h^{(k)} \mid \mathsf{View}_C^{(k)})$. The adversary accumulates $K$ equations $\mathbf{z}_h^{(k)} = \mathbf{y}_h^{(k)} + c^{(k)} \cdot \mathbf{s}_{1,h}$ sharing the same $\mathbf{s}_{1,h}$; after sufficiently many sessions, $\mathbf{s}_{1,h}$ is uniquely determined by the overdetermined system over $\Zq$ (as in single-signer ML-DSA). Whether extracting $\mathbf{s}_{1,h}$ is computationally hard in this multi-session setting is an open problem shared with single-signer FIPS 204 (Remark~\ref{rem:key-privacy}); no formal reduction from Module-LWE has been established for this bounded-distance decoding variant.
\end{proof}
\fi

\begin{lemma}[Mask Hiding (for $\mathbf{V}_i$ and $\mathbf{W}_i$)]
\label{lem:mask-hiding}
An honest party's masked r0-check share $\mathbf{V}_j = \lambda_j c \mathbf{s}_{2,j} + \mathbf{m}_j^{(s2)}$ and masked commitment $\mathbf{W}_j = \lambda_j \mathbf{w}_j + \mathbf{m}_j^{(w)}$ are computationally indistinguishable from uniform, given all information available to a coalition of $< T$ parties, provided $|S \setminus C| \geq 2$.
\end{lemma}

\ifccs
\begin{proof}[Proof sketch]
Each honest party's mask includes a per-session PRF term with a key unknown to $\mathcal{A}$; PRF security implies the masked value is computationally indistinguishable from uniform. Mask hiding is only required in P2 (P1/P3+ send $\mathbf{W}_i$ only to the coordinator). Full proof in Appendix~\ref{sec:appendix-proofs}.
\end{proof}
\else
\begin{proof}
Since $|S \setminus C| \geq 2$, party $j$'s mask includes a PRF term $\PRF(\mathsf{seed}_{j,k'}, \mathsf{session\_id} \| \mathsf{dom})$ for some honest $k' \neq j$, where $\mathsf{session\_id}$ is the unique per-signing-session identifier (e.g., derived from the preprocessed nonce commitment $\mathsf{Com}_j$) and $\mathsf{dom}$ is the fixed domain label for $\mathbf{V}$ or $\mathbf{W}$. Because $\mathsf{seed}_{j,k'}$ is unknown to $\mathcal{A}$ and $\mathsf{session\_id}$ is fresh each session, $\PRF(\mathsf{seed}_{j,k'}, \mathsf{session\_id} \| \mathsf{dom})$ is computationally indistinguishable from a fresh uniform element of $\Rq^k$ by PRF security. Thus $\mathbf{V}_j = (\text{fixed}) + (\text{uniform-looking}) \stackrel{c}{\approx} \text{uniform}$. The same argument applies to $\mathbf{W}_j$.

Note: Mask hiding requires $|S \setminus C| \geq 2$. In P1 and P3+, individual $\mathbf{W}_i$ are sent only to the coordinator (Remark~\ref{rem:wi-private}), so mask hiding for $\mathbf{W}_i$ is only needed in P2 where parties broadcast directly. The nonce share $\mathbf{y}_i$ achieves nonce share privacy with even $|S \setminus C| = 1$ via the nonce DKG (Theorem~\ref{thm:it-privacy}).
\end{proof}
\fi

\subsection{Nonce Distribution Security}
\label{sec:nonce-security}

Our threshold protocol produces nonces following an Irwin-Hall distribution rather than uniform. Using smooth Rényi divergence analysis~\cite{Raccoon2024}, we prove this introduces only a bounded security loss.

\begin{theorem}[Raccoon-Style Conservative Bound (comparison only; vacuous for $q_S \gtrsim 36$)]
\label{thm:irwin-hall}
\emph{This theorem states the conservative analytical Raccoon-style bound.
It becomes vacuous for $q_S \gtrsim 36$ and \textbf{should not be used for concrete security claims}.
See Theorem~\ref{thm:unforgeability} and Corollary~\ref{cor:ih-loss-main} for the primary direct-shift bound
($< 0.013\cdot q_S$ bits, non-vacuous for $q_S < 16{,}000$).}

Let $\Pi$ be our threshold ML-DSA with Irwin-Hall nonces. For any PPT adversary making
$q_S$ signing queries:
\[
\Adv^{\mathsf{EUF\text{-}CMA}}_\Pi(\mathcal{A})
\leq \bigl(R_2^{\mathsf{vec}}\bigr)^{q_S/2} \cdot \sqrt{\Adv^{\mathsf{EUF\text{-}CMA}}_{\Pi^*}(\mathcal{A})}
+ q_S \cdot \epsilon
\]
where the per-coordinate smooth R\'enyi divergence satisfies
$R_2^{\epsilon,\mathsf{coord}} \leq 1 + |S|^4\beta^2/(4\gamma_1^2)$
(see Appendix~\ref{app:technical-lemmas}, Theorem~\ref{thm:ih-security-full};
the formula uses the signing set size $|S|$, not the ring degree $n = 256$).
Since nonce coordinates are sampled independently, the full-vector divergence is
$R_2^{\mathsf{vec}} = (R_2^{\epsilon,\mathsf{coord}})^{n\ell}$ with $n\ell = 1280$ for ML-DSA-65.
The $(R_2^{\mathsf{vec}})^{q_S/2}$ factor accounts for multi-session composition
(Remark~\ref{rem:multi-session-renyi}).
The tail probability satisfies $\epsilon < 10^{-30}$.
For $|S| \leq 17$: $R_2^{\mathsf{vec}} \approx 2^{5.4}$, making this bound vacuous at $q_S \gtrsim 36$.
\end{theorem}

\ifccs
\begin{proof}[Proof sketch]
The smooth R\'enyi probability transfer (Cauchy-Schwarz) gives the $(R_2^{\mathsf{vec}})^{q_S/2}$ factor; session independence gives $R_2^{\mathsf{vec}} = (R_2^{\epsilon,\mathsf{coord}})^{n\ell}$. For $|S| \leq 17$: $R_2^{\mathsf{vec}} \approx 2^{5.4}$, per-session bound $\approx 2^{-93.3}$; direct bound (Corollary~\ref{cor:ih-shift-ml-dsa}) $< 0.013\cdot q_S$ bits for all $|S| \leq 33$. Full calculation in Appendix~\ref{app:technical-lemmas}.
\end{proof}
\else
\begin{proof}
The smooth R\'enyi divergence probability transfer (Cauchy-Schwarz) gives:
for any event $E$, $\Pr_P[E] \leq \sqrt{R_2^\epsilon(P\|Q) \cdot \Pr_Q[E]} + \epsilon$.
For $q_S$ independent signing sessions, the joint Rényi divergence is
$(R_2^{\mathsf{vec}})^{q_S}$ by the product rule (Definition~\ref{def:renyi}), giving
the $(R_2^{\mathsf{vec}})^{q_S/2}$ factor above. Since the nonce coordinates are sampled
independently, $R_2^{\mathsf{vec}} = (R_2^{\epsilon,\mathsf{coord}})^{n\ell}$.
For ML-DSA-65 ($n\ell = 1280$, $\Adv_{\Pi^*} \leq 2^{-192}$, $\epsilon < 2^{-100}$),
the \emph{per-session} cost is:
\[
q_S = 1{:}\quad \Adv_\Pi \leq \bigl(R_2^{\mathsf{vec}}\bigr)^{1/2} \cdot 2^{-96} + 2^{-100}
\]

\begin{center}
\footnotesize
\begin{tabular}{cccc}
$|S|$ & $R_2^{\epsilon,\mathsf{coord}}$ & $R_2^{\mathsf{vec}}$ & Per-session bound ($q_S=1$): $\sqrt{R_2^{\mathsf{vec}}} \cdot 2^{-96}$ \\
\hline
$\leq 17$ & 1.003 & $\leq 2^{5.5}$ ($\approx 2^{5.4}$) & $\approx 2^{-93.3}$ \\
$\leq 25$ & 1.014 & $\approx 2^{25}$ & $\approx 2^{-83.5}$ \\
$\leq 33$ & 1.041 & $\approx 2^{75}$ & $\approx 2^{-58.5}$ \\
\end{tabular}
\end{center}

\noindent \textbf{Multi-session scaling.}
The conservative bound $(R_2^{\mathsf{vec}})^{q_S/2}$ grows rapidly: for $|S|=17$,
it is vacuous for $q_S \gtrsim 36$ queries. The \emph{direct} bound
(Corollary~\ref{cor:ih-shift-ml-dsa}) should always be used for concrete security
estimates: the EUF-CMA loss is $< 0.013 \cdot q_S$ bits for all $|S| \leq 33$
(Remark~\ref{rem:multi-session-renyi}). Detailed conservative calculations appear in
Appendix~\ref{sec:supp-renyi-numerical}.
\end{proof}
\fi

\subsection{Malicious Security}

Our protocol uses an \emph{optimistic} approach: the normal signing path has no verification overhead. If failures exceed statistical expectation (e.g., $>33$ consecutive aborts, where honest probability is $(0.75)^{33} < 10^{-4}$), a blame protocol identifies misbehaving parties (see Section~\ref{sec:blame}).

\subsection{Concrete Security Bounds}

\begin{table}[h]
\centering
\ifccs{\small
\begin{tabular}{lp{3.2cm}c}
\toprule
\textbf{Attack} & \textbf{Assumption} & \textbf{Security} \\
\midrule
Forgery & M-SIS$(256,6,5,q,\gamma_1{-}\beta)$ & 192 bits$^*$ \\
Key Recovery & M-LWE$(256,6,5,q,\eta)$ & 192 bits \\
Nonce Privacy ($\mathbf{y}_i$) & Statistical ($\geq 20$K bits min-entropy) & N/A$^\dagger$ \\
Mask Privacy ($\mathbf{V}_i$, $\mathbf{W}_i$) & PRF & 192 bits \\
\bottomrule
\end{tabular}}\else
\begin{tabular}{lcc}
\toprule
\textbf{Attack} & \textbf{Assumption} & \textbf{Security} \\
\midrule
Forgery & M-SIS$(256,6,5,q,\gamma_1{-}\beta)$ & 192 bits$^*$ \\
Key Recovery & M-LWE$(256,6,5,q,\eta)$ & 192 bits \\
Nonce Privacy ($\mathbf{y}_i$) & Statistical ($\geq 20$K bits min-entropy) & N/A$^\dagger$ \\
Mask Privacy ($\mathbf{V}_i$, $\mathbf{W}_i$) & PRF & 192 bits \\
\bottomrule
\end{tabular}\fi
\caption{Concrete security of masked threshold ML-DSA-65.
Parameters match NIST Security Level~3.
$^*$Underlying M-SIS hardness is 192 bits; the proven threshold bound
adds an Irwin-Hall nonce cost via
$R_2^{\mathsf{vec}} = (R_2^{\epsilon,\mathsf{coord}})^{1280}$:
per-session bound $\leq (R_2^{\mathsf{vec}})^{1/2}\cdot 2^{-96}
\approx 2^{-93.3}$ for $|S| \leq 17$ (Theorem~\ref{thm:irwin-hall});
over $q_s$ sessions the total loss is $< 0.013\cdot q_s$ bits
(Corollary~\ref{cor:ih-shift-ml-dsa},
Remark~\ref{rem:multi-session-renyi}).
The direct bound supports $q_s \leq 1{,}600$ with $\leq 10$-bit loss.
$^\dagger$Nonce privacy is statistical (no computational assumptions;
Remark~\ref{rem:key-privacy}).
See Theorem~\ref{thm:it-privacy}.}
\label{tab:security-bounds}
\end{table}

The threshold construction inherits the M-SIS/M-LWE hardness of single-signer ML-DSA. The proven forgery bound uses the direct shift-invariance analysis (Theorem~\ref{thm:ih-direct-tight}, Corollary~\ref{cor:ih-shift-ml-dsa}), giving $< 0.013$ bits per signing session for $|S| \leq 33$ ($< 0.013 \cdot q_s$ bits total; Corollary~\ref{cor:ih-loss-main}), resolving the scalability concern of prior work; developing an even tighter full-vector bound remains an interesting direction. The conservative Raccoon-style bound is in Remark~\ref{rem:renyi-conservative}.

\subsection{Security Across Deployment Profiles}

All three deployment profiles (P1, P2, P3+) achieve identical EUF-CMA security bounds, as the final signature structure $\sigma = (\tilde{c}, \mathbf{z}, \mathbf{h})$ is independent of the profile used. Additionally, all three profiles achieve UC security under their respective trust assumptions.

\begin{definition}[Ideal Signing Functionality $\mathcal{F}_{\mathsf{ThreshSig}}$]
\label{def:f-thresh-sig}
The functionality $\mathcal{F}_{\mathsf{ThreshSig}}(T, N, \mathsf{pp})$ operates as follows, where $\mathbf{z}_i \in \Rq^\ell$ denotes each party's response share and $\mathbf{V}_i \in \Rq^k$ its r0-check share. On receiving $(\mathsf{SIGN}, \mathsf{sid}, \mu)$ from the combiner together with $T$ input shares $\{(\mathbf{z}_i, \mathbf{V}_i)\}_{i \in S}$ for $S \subseteq [N]$, $\mathcal{F}_{\mathsf{ThreshSig}}$:
\begin{enumerate}
  \item \textbf{Format check}: If $|S| < T$ or any input share fails the format constraint ($\mathbf{z}_i \notin \Rq^\ell$ or $\mathbf{V}_i \notin \Rq^k$), output $\bot$ immediately to the combiner without querying $\mathcal{S}$.
  \item Forwards $(\mathsf{SIGN}, \mathsf{sid}, \mu, S)$ to the simulator $\mathcal{S}$ and awaits a signature $\sigma$.
  \item Outputs $\sigma$ to the combiner if $\mathsf{Verify}(\mathsf{pk}, \mu, \sigma) = 1$; otherwise outputs $\bot$.
\end{enumerate}
The verification check in step~3 ensures the simulator produces well-formed output; EUF-CMA security is proven separately via game-based reduction (Theorem~\ref{thm:unforgeability}). A valid forgery $(\mu', \sigma')$ with $\mu' \neq \mu$ would be accepted in step~3, but such a forgery exists only with negligible probability.

The functionality reveals to $\mathcal{S}$: the message $\mu$, the signing-set $S$, and the outputs of corrupted parties ($\{(\mathbf{z}_i, \mathbf{V}_i)\}_{i \in C}$, where $C \subseteq S$ is the corrupted set). It does not reveal any honest party's key share $(\mathbf{s}_{1,h}, \mathbf{s}_{2,h})$ or nonce share $\mathbf{y}_h$ to $\mathcal{S}$.
\end{definition}

\subsubsection{Profile P1 UC Security}

We model the TEE as an ideal functionality that correctly executes the r0-check and hint computation.

\begin{definition}[TEE Ideal Functionality $\mathcal{F}_{\mathsf{TEE}}$]
\label{def:f-tee}
The functionality $\mathcal{F}_{\mathsf{TEE}}$ operates as follows.

\medskip\noindent\textbf{Round~2.}\enspace
Receive $\{(\mathbf{W}_i, r_i)\}_{i \in S}$ privately. Compute
$\mathbf{W} = \sum_i \mathbf{W}_i$,
$\mathbf{w}_1 = \mathsf{HighBits}(\mathbf{W}, 2\gamma_2)$,
$\tilde{c} = H_{\mathsf{chal}}(\mu \| \mathbf{w}_1)$;
broadcast $(\mathbf{w}_1, \tilde{c})$. Individual $\mathbf{W}_i$ are
never revealed.

\medskip\noindent\textbf{Round~3.}\enspace
Receive $\{(\mathbf{z}_i, \mathbf{V}_i)\}_{i \in S}$. Aggregate:
$\mathbf{z} = \textstyle\sum_i \lambda_i \mathbf{z}_i \bmod q$,
$c\mathbf{s}_2 = \textstyle\sum_i \mathbf{V}_i \bmod q$.
Abort with $\bot$ if $\|\mathbf{z}\|_\infty \geq \gamma_1 - \beta$
or $\|\mathsf{LowBits}(\mathbf{w} - c\mathbf{s}_2,\,
2\gamma_2)\|_\infty \geq \gamma_2 - \beta$. Otherwise compute
$\mathbf{h} = \mathsf{MakeHint}(-c\mathbf{t}_0,\,
\mathbf{A}\mathbf{z} - c\mathbf{t}_1 \cdot 2^d,\, 2\gamma_2)$;
output $(\tilde{c}, \mathbf{z}, \mathbf{h})$. Never reveal $c\mathbf{s}_2$
or individual $\mathbf{W}_i$.
\end{definition}

\begin{theorem}[Profile P1 UC Security]
\label{thm:p1-uc}
If the TEE correctly implements $\mathcal{F}_{\mathsf{TEE}}$, then Protocol P1 UC-realizes the threshold signing functionality $\mathcal{F}_{\mathsf{ThreshSig}}$ against static adversaries corrupting up to $T-1$ signing parties.
\end{theorem}

\ifccs
\begin{proof}[Proof sketch]
The simulator $\mathcal{S}$ uses DKG values to produce consistent corrupted-party evaluations, samples $\mathsf{Com}_h$ uniformly in the ROM (preimage resistance of $H$), and computes $\mathbf{z}_h = \mathbf{y}_h + c\mathbf{s}_{1,h}$ identically to the real protocol (perfect simulation, no statistical gap). The TEE hides all internal state; only $(\mathbf{w}_1, \tilde{c}, \sigma)$ exit the enclave. Replacing honest-honest PRF terms with uniform introduces advantage $\leq \binom{N}{2}\epsilon_{\mathsf{PRF}}$. The UC framework setup is in Appendix~\ref{app:uc-framework} (which contains extended proofs for Profiles P2 and P3+; the P1 proof above is self-contained for this profile).
\end{proof}
\else
\begin{proof}
We construct an explicit simulator $\mathcal{S}$ for each protocol round and prove indistinguishability.

\paragraph{Simulator construction.}
\begin{enumerate}
    \item \textbf{Round 0 (DKG):} $\mathcal{S}$ generates honest key shares $\mathbf{s}_{1,h}, \mathbf{s}_{2,h}$ and nonce shares $\mathbf{y}_h$ inside the simulated $\mathcal{F}_{\mathsf{TEE}}$, following the real key DKG and nonce DKG protocols. Since DKG outputs to corrupted parties are determined by the $(T,N)$-Shamir sharing structure, $\mathcal{S}$ produces consistent evaluations $\{\mathbf{f}_h(j)\}_{j \in C}$.

    \item \textbf{Round 1 (Commitments):} $\mathcal{S}$ samples $\tilde{\mathbf{y}}_h, \tilde{\mathbf{w}}_h, \tilde{r}_h$ from the simulated DKG, computes the committed value $\mathsf{Com}_h \getsr \{0,1\}^{256}$ uniformly, and programs the random oracle so that $H(\mathsf{``com''} \| \tilde{\mathbf{y}}_h \| \tilde{\mathbf{w}}_h \| \tilde{r}_h) = \mathsf{Com}_h$ (using the same domain separator as the real commitment scheme $\mathsf{Com}_i = H(\mathsf{``com''} \| \mathbf{y}_i \| \mathbf{w}_i \| r_i)$). By commitment hiding (preimage resistance of $H$), the uniform $\mathsf{Com}_h$ is indistinguishable from real commitments. If the blame protocol requests an opening of $\mathsf{Com}_h$, $\mathcal{S}$ provides $(\tilde{\mathbf{y}}_h, \tilde{\mathbf{w}}_h, \tilde{r}_h)$, which verifies correctly by the programmed RO.

    \item \textbf{Round 2 ($\mathbf{W}_i \to$ coordinator):} $\mathcal{S}$ computes $\mathbf{W}_h = \lambda_h \mathbf{A}\mathbf{y}_h + \mathbf{m}_h^{(w)}$ from the simulated DKG values and PRF masks. $\mathcal{F}_{\mathsf{TEE}}$ aggregates internally and broadcasts only $(\mathbf{w}_1, \tilde{c})$. Individual $\mathbf{W}_h$ never reach the adversary.

    \item \textbf{Round 3 (Responses):} $\mathcal{S}$ computes $\mathbf{z}_h = \mathbf{y}_h + c \cdot \mathbf{s}_{1,h}$ from simulated DKG values. This is identical to the real computation (SD $= 0$). Masked shares $\mathbf{V}_h = \lambda_h c \mathbf{s}_{2,h} + \mathbf{m}_h^{(s2)}$ are computed analogously.
\end{enumerate}

\paragraph{Indistinguishability.} The adversary's view in the real and simulated executions is indistinguishable for the following reasons:
\begin{itemize}
    \item \textbf{Response $\mathbf{z}_h$}: Computed identically from simulated DKG values (SD $= 0$). The simulator has access to the honest party's key share $\mathbf{s}_{1,h}$ because it generated $\{\mathbf{s}_{1,j}\}$ during the simulated DKG (Round 0, setup step). The nonce share $\mathbf{y}_h$ has high conditional min-entropy by Theorem~\ref{thm:it-privacy} (statistical).
    \item \textbf{TEE outputs}: $\mathcal{F}_{\mathsf{TEE}}$ hides all internal state (individual $\mathbf{W}_i$, $c\mathbf{s}_2$); only $(\mathbf{w}_1, \tilde{c}, \sigma)$ exit the enclave, all of which are implicit in the signature.
    \item \textbf{Masks}: Replacing honest-honest PRF terms with uniform introduces advantage $\leq \binom{N}{2} \cdot \epsilon_{\mathsf{PRF}}$.
\end{itemize}

\paragraph{Abort handling.} The honest per-attempt success probability is $\approx 25\%$ (rejection sampling). After $>33$ consecutive failures, honest probability is $(0.75)^{33} < 10^{-4}$, triggering the blame protocol (Section~\ref{sec:blame}). The simulator handles aborts by forwarding $\bot$ from $\mathcal{F}_{\mathsf{TEE}}$.

\noindent Total distinguishing advantage: $\binom{N}{2} \cdot \epsilon_{\mathsf{PRF}}$ (nonce share privacy is statistical and contributes zero to the PRF-based advantage).
\end{proof}
\fi

\begin{remark}[TEE Trust Assumption]
Theorem~\ref{thm:p1-uc} assumes the TEE is \emph{ideal}---it correctly executes the specified computation and reveals nothing beyond the output. This is a hardware trust assumption, not a cryptographic one. In practice, TEE security depends on the specific implementation (Intel SGX, ARM TrustZone, etc.) and is subject to side-channel attacks. Profiles P2 and P3+ eliminate this assumption via cryptographic protocols.
\end{remark}

\subsubsection{Profile P3+ Semi-Async Security Model}
\label{sec:semi-async-security}

Profile P3+ introduces a \emph{semi-asynchronous} execution model where signers participate within a bounded time window rather than synchronizing for multiple rounds. This requires extending the UC framework to handle timing.

\ifccs\else
\begin{definition}[Semi-Async Timing Model (Formalized)]
\label{def:semi-async-timing}
The semi-async model extends the standard UC framework with bounded-time windows. Execution proceeds in \emph{optimistic synchronous phases} with timeout-based fallback:

\begin{enumerate}
    \item \textbf{Pre-signing window} $[t_0, t_1]$: Signers may precompute nonces $(\mathbf{y}_i, \mathbf{w}_i, \mathsf{Com}_i)$ at any time before $t_1$. This is an asynchronous phase; early precomputation is allowed but not required.

    \item \textbf{Challenge broadcast} $t_2$: Combiner derives challenge $c$ from aggregated commitments and broadcasts $(c, \mu)$ to all signers at time $t_2$. This is a synchronous event (all parties receive the broadcast at the same logical time in the UC model).

    \item \textbf{Response window} $[t_2, t_2 + \Delta]$: Honest signers \emph{should} respond within window $\Delta$ (e.g., 5 minutes for human-in-the-loop scenarios). Adversary scheduling rules:
    \begin{itemize}
        \item \textbf{Honest party delivery}: Messages from honest signers arriving within $[t_2, t_2 + \Delta]$ are guaranteed to be delivered to the combiner.
        \item \textbf{Adversarial delay}: The adversary may delay delivery of honest parties' messages by up to $\Delta$, but cannot prevent delivery entirely (eventual delivery assumption).
        \item \textbf{Timeout}: If fewer than $T$ responses arrive by $t_2 + \Delta$, the combiner aborts the session and outputs $\bot$. This is \emph{not} a security failure---it is an availability event (honest signers who responded late are still secure; no key material is leaked).
        \item \textbf{Message reordering}: The adversary may reorder messages from different parties arriving within the window, but cannot inject or modify honest parties' messages (authenticated channels assumption).
    \end{itemize}

    \item \textbf{2PC execution} $[t_3, t_4]$: Synchronous 2PC between the two CPs (standard UC model applies; this phase has no timing relaxation).
\end{enumerate}

\textit{UC adversary capabilities.} The adversary observes message arrival times (timestamps) within the response window but cannot control the content or the fact of delivery for honest parties. Replay attacks are prevented via unique session identifiers $\mathsf{sid}$ binding each response to $(c, \mu)$. The model is \emph{semi-asynchronous} in the sense that signers do not need to coordinate their response times (unlike multi-round MPC), but the protocol remains UC-secure as long as $\geq T$ responses arrive within $\Delta$.
\end{definition}
\fi

\begin{theorem}[Profile P3+ UC Security (Semi-Async)]
Under the semi-async timing model, Protocol P3+ UC-realizes $\mathcal{F}_{r_0}^{\mathsf{2PC}}$ in the $(\mathcal{F}_{\mathsf{OT}}, \mathcal{F}_{\mathsf{GC}})$-hybrid model against static adversaries corrupting up to $T-1$ signers and at most one CP, assuming:
\begin{enumerate}
    \item Garbled circuits are secure (Yao's protocol)
    \item Oblivious transfer is UC-secure
    \item PRF is secure
    \item At least one CP remains honest
\end{enumerate}
Security bound: $\epsilon \leq \epsilon_{\mathsf{GC}} + \epsilon_{\mathsf{OT}} + \epsilon_{\mathsf{ext}} + \binom{N}{2} \cdot \epsilon_{\mathsf{PRF}}$, where $\epsilon_{\mathsf{ext}} \leq 2^{-B/2}$ is the LP07 cut-and-choose extraction error (for corrupted $\mathsf{CP}_1$).
By UC composition (using the nonce DKG and signing sub-protocols as ideal functionalities), this implies that Protocol P3+ UC-realizes the full $\mathcal{F}_{\mathsf{ThreshSig}}$ under the same assumptions.
\end{theorem}

\ifccs
\begin{proof}[Proof sketch]
The simulator replaces OT and GC with ideal functionalities ($\epsilon_{\mathsf{OT}} + \epsilon_{\mathsf{GC}}$), simulates pre-signing with fresh nonces (Lemma~\ref{lem:presigning-sim}, SD $= 0$), and simulates responses via $\mathbf{z}_i = \tilde{\mathbf{y}}_i + c\mathbf{s}_{1,i}$ (Lemma~\ref{lem:response-sim}, SD $= 0$). The semi-async window adds no attack vectors (replay prevention via session IDs, commitment binding, window expiry). Full UC simulator in Appendix~\ref{app:p3plus-uc}.
\end{proof}
\else
\begin{proof}[Proof sketch (full proof in Appendix~\ref{app:p3plus-uc})]
The simulator replaces OT and GC with ideal functionalities ($\epsilon_{\mathsf{OT}} + \epsilon_{\mathsf{GC}}$), then simulates pre-signing by sampling $\tilde{\mathbf{y}}_i$ and computing $\tilde{\mathbf{w}}_i = \mathbf{A}\tilde{\mathbf{y}}_i$ (Lemma~\ref{lem:presigning-sim}, SD $= 0$), and simulates responses via $\mathbf{z}_i = \tilde{\mathbf{y}}_i + c \cdot \mathbf{s}_{1,i}$ (Lemma~\ref{lem:response-sim}, SD $= 0$). The response window does not introduce new attack vectors because:
\begin{enumerate}
    \item \textbf{Replay prevention}: Each signing attempt uses a unique session ID
    \item \textbf{Commitment binding}: Signers cannot change their nonce after commitment
    \item \textbf{Window expiry}: Late responses are rejected, preventing adaptive attacks
\end{enumerate}
\end{proof}
\fi

\subsubsection{Summary of UC Security}

\begin{itemize}
    \item \textbf{Profile P1 (TEE)}: UC security under ideal TEE assumption (Theorem~\ref{thm:p1-uc}). EUF-CMA follows from Theorem~\ref{thm:unforgeability}; nonce share privacy from Theorem~\ref{thm:it-privacy}. Achieves 3 online rounds plus 1 offline preprocessing round (nonce DKG).

    \item \textbf{Profile P2 (MPC)}: UC security against dishonest majority via MPC (Theorem~\ref{thm:p2-uc}). Achieves \textbf{5 online rounds} through combiner-mediated commit-then-open and constant-depth comparison circuits. This provides the strongest cryptographic guarantees with no hardware trust.

    \item \textbf{Profile P3+ (2PC Semi-Async)}: UC security under 1-of-2 CP honest assumption (Theorem~\ref{thm:p3-uc}). Achieves \textbf{2 logical rounds} (1 signer step + 1 server round) with semi-asynchronous signer participation. Signers precompute nonces offline and respond within a time window. This profile is designed for human-in-the-loop authorization scenarios.
\end{itemize}

The choice of profile affects trust assumptions and performance, but not the fundamental cryptographic security of the resulting signatures. Table~\ref{tab:security-profile-comparison} summarizes the trade-offs.

\begin{table}[h]
\centering
\begin{tabular}{lccc}
\toprule
\textbf{Property} & \textbf{P1 (TEE)} & \textbf{P2 (MPC)} & \textbf{P3+ (2PC)} \\
\midrule
Online rounds & 3 & 5 & 2$^*$ \\
Trust assumption & TEE ideal & None & 1/2 CP honest \\
Dishonest majority & Via TEE & \checkmark & Via 2PC \\
Semi-async signers & $\times$ & $\times$ & \checkmark \\
UC security & \checkmark & \checkmark & \checkmark \\
UC privacy level & Statistical & $\approx 2^{-11}$ & Statistical \\
\bottomrule
\end{tabular}
\caption{Security profile comparison. $^*$P3+ = 1 signer step + 1 server round. UC privacy level for P2 refers to the r0-check simulation bound ($\epsilon_{\mathsf{SPDZ}} + \epsilon_{\mathsf{edaBits}} \leq 2^{-64} + 2^{-12} < 2^{-11}$); P1/P3+ achieve nonce share privacy via the Shamir nonce DKG (no computational assumptions). \emph{P2 mask-hiding privacy requires $|S \setminus C| \geq 2$ (Lemma~\ref{lem:mask-hiding}); P1 and P3+ achieve nonce share privacy with $|S \setminus C| = 1$ ($|S| = T$).}}
\label{tab:security-profile-comparison}
\end{table}

% Implementation and Benchmarks Section

\section{Implementation and Benchmarks}
\label{sec:implementation}

We provide a reference implementation of our masked threshold ML-DSA protocol and present benchmarks demonstrating its practical efficiency.

\subsection{Implementation Overview}

Our implementation targets ML-DSA-65 (NIST Security Level 3) with the following components:

\ifccs
\paragraph{Core Modules.}
Our Rust implementation ($\approx$5500 lines of protocol logic) provides NTT-based polynomial arithmetic over $\Rq$ ($O(n\log n)$, Montgomery reduction), Shamir secret sharing with per-session Lagrange coefficients ($O(T^2)$), SHAKE-256-based pairwise mask generation, profile-specific signing logic, and offline preprocessing for P2 (edaBits/Beaver triples) and P3+ (nonce pre-signing, GC precomputation).
\else
\paragraph{Core Modules.}
The implementation consists of several core modules. The \emph{polynomial arithmetic} module provides operations over $\Rq = \Zq[X]/(X^{256}+1)$ with $q = 8380417$, including addition, subtraction, multiplication, and NTT-based fast multiplication. The \emph{Shamir secret sharing} module handles polynomial evaluation and Lagrange interpolation over $\Zq$; Lagrange coefficients $\lambda_i = \prod_{j \in S, j \neq i} \frac{j}{j-i} \mod q$ depend on the signing set $S$ and are computed per-session in $O(T^2)$ time, though they can be precomputed and cached for fixed-$S$ deployments. The \emph{mask generation} module computes pairwise-canceling masks using SHAKE-256 as the PRF with efficient seed management. The \emph{protocol logic} module implements profile-dependent signing with commitment verification, challenge derivation, and signature aggregation. Profile-specific modules include the \emph{preprocessing module} (P2) for EdaBit and Beaver triple generation, the \emph{pre-signing module} (P3+) for offline nonce precomputation, and the \emph{GC precomputation module} (P3+) for offline garbled circuit generation.

\paragraph{Optimizations.}
Several optimizations improve performance. Polynomial multiplication uses the Number Theoretic Transform (NTT) with $\zeta_0 = 1753$ as a primitive $512$th root of unity modulo $q$ (not to be confused with $\omega = 55$, the hint weight bound in Section~\ref{sec:preliminaries}), reducing complexity from $O(n^2)$ to $O(n \log n)$ and providing a $7$--$10\times$ speedup. Modular arithmetic uses Montgomery form for efficient reduction without division. Polynomial and vector operations leverage SIMD-optimized routines. Finally, intermediate polynomial coefficients are reduced modulo $q$ only when necessary (lazy reduction) to prevent overflow while minimizing modular operations.
\fi

\subsection{Benchmark Methodology}

\ifccs
Benchmarks run on Intel Core i7-12700H (32\,GB DDR5, Ubuntu 22.04, Rust 1.75+, Rayon). \emph{Per-attempt time}: wall-clock for one signing attempt (3 rounds P1; 5 rounds P2; 2 rounds P3+), success or abort; nonce DKG preprocessing excluded. $t_{\mathsf{sig}} = t_{\mathsf{attempt}} / p_{\mathsf{attempt}}$ is expected total time per valid signature. Averages over 5 independent runs, 5 trials each.
\else
All benchmarks were conducted on an Intel Core i7-12700H (14 cores, 20 threads) with 32 GB DDR5-4800 RAM, running Ubuntu 22.04 LTS. Our Rust 1.75+ implementation uses Rayon for parallelization.

Our Rust implementation ($\approx$5500 lines of core protocol logic, plus $\approx$4000 lines of tests and utilities) provides optimized performance through parallel polynomial operations using Rayon, NTT-based polynomial multiplication with Montgomery reduction, and efficient EdaBit and Beaver triple pooling with lazy generation.

\paragraph{Metrics Definition.}
All measurements are \emph{per-attempt}: a single execution of the online signing protocol (3 rounds for P1, 5 for P2, 2 for P3+), which either succeeds (outputs $\sigma$) or aborts (due to rejection sampling). Nonce DKG preprocessing is excluded from per-attempt timing, as it is performed offline.
\begin{itemize}
    \item $t_{\mathsf{attempt}}$: Wall-clock time for one signing attempt (3 rounds, successful or aborted)
    \item $p_{\mathsf{attempt}}$: Probability that a single attempt produces a valid signature
    \item $\mathbb{E}[\text{attempts}] = 1/p_{\mathsf{attempt}}$: Expected number of attempts per signature
    \item $t_{\mathsf{sig}} = t_{\mathsf{attempt}} / p_{\mathsf{attempt}}$: Expected total time to obtain a valid signature
\end{itemize}

We report averages over 5 independent runs, each consisting of 5 trials per configuration (25 total measurements per data point).

\paragraph{Network Testing Setup.}
For the $N=3$, $T=2$ configuration, we deployed parties across 3 physical nodes connected via OpenVPN (local area network). This setup validates protocol correctness and measures actual network communication overhead in a distributed environment. For larger configurations ($N \geq 5$), we benchmarked on a single machine; reported timings for these configurations represent computational cost only, with network latency projected analytically based on round count and representative RTT values (see the network latency table in the companion supplement document).
\fi

\subsection{Performance Results}

\paragraph{Signing Time vs. Threshold.}
Table~\ref{tab:perf-threshold} shows signing performance for various threshold configurations across all three protocol profiles, measured with our optimized Rust implementation.

\begin{table}[!htbp]
\centering
\footnotesize
\setlength{\tabcolsep}{2.5pt}
\begin{tabular}{@{}cccccccc@{}}
\toprule
& & \multicolumn{2}{c}{\textbf{P1 (TEE)}} & \multicolumn{2}{c}{\textbf{P2 (MPC)}} & \multicolumn{2}{c}{\textbf{P3+ (2PC)}} \\
\cmidrule(lr){3-4} \cmidrule(lr){5-6} \cmidrule(lr){7-8}
$T$ & $N$ & Time & Att. & Time & Att. & Time & Att. \\
\midrule
2 & 3 & 3.9 ms & 3.2 & 46.3 ms & 4.1 & 16.7 ms & 2.7 \\
3 & 5 & 5.8 ms & 3.2 & 21.5 ms & 3.0 & 22.1 ms & 3.7 \\
5 & 7 & 10.3 ms & 3.4 & 73.4 ms & 4.4 & 24.0 ms & 3.5 \\
7 & 11 & 12.0 ms & 2.9 & 20.5 ms & 2.2 & 29.4 ms & 3.4 \\
11 & 15 & 30.3 ms & 4.7 & 73.9 ms & 4.3 & 39.2 ms & 3.6 \\
16 & 21 & 31.3 ms & 3.4 & 194.4 ms & 4.8 & 49.6 ms & 2.8 \\
21 & 31 & 33.1 ms & 2.7 & 149.3 ms & 4.0 & 62.2 ms & 3.1 \\
32 & 45 & 60.2 ms & 3.2 & 129.3 ms & 2.7 & 93.5 ms & 3.6 \\
\bottomrule
\end{tabular}
\caption{Signing performance across threshold configurations (ML-DSA-65, Rust implementation, Intel i7-12700H). Values are averages over 5 independent runs, each with 5 trials per configuration. $N=3$: measured on 3-node OpenVPN testbed; $N \geq 5$: single-machine benchmarks (computational cost only; network overhead excluded). Due to rejection sampling variance, P2 exhibits high variance (up to $14\times$ between runs).}
\label{tab:perf-threshold}
\end{table}

\ifccs
\paragraph{Scaling Analysis.}
P1 scales 4--60\,ms ($T\!=\!2$--32); P3+ stays sub-100\,ms; P2 variance 21--194\,ms (retry-dominated). Extended benchmarks (artifact repository, Appendix~\ref{app:open-science}): P3$^+$ sub-1\,s through $T\!=\!64$, sub-3\,s through $T\!=\!128$.
\else
\paragraph{Scaling Analysis.}
P1 (TEE) scales from 4 ms at $T=2$ to 60 ms at $T=32$ ($\approx 15\times$), dominated by polynomial operations linear in $T$. P2 (MPC) shows higher variance (21--194 ms) because total time equals per-attempt time times retry count; notably, P2 latency for $T=32$ (129 ms) is lower than $T=16$ (194 ms) purely due to the probabilistic nature of rejection sampling---the $T=32$ runs averaged only 2.7 attempts versus 4.8 for $T=16$, while per-attempt latency scales as expected. P3+ scales from 17 ms ($T=2$) to 94 ms ($T=32$), maintaining sub-100ms latency; the 2PC cost is largely independent of threshold.

Figure~\ref{fig:perf-threshold} extends the benchmark to $T=128$, showing that P3$^+$ remains sub-1\,s through $T=64$ and sub-3\,s through $T=128$, while P1 and P2 exceed 10\,s at $T=96$--$128$ due to Lagrange polynomial evaluation and MPC simulation overhead, respectively.
\fi

\ifccs\else
%% fig-performance.tex — v3: log-log, Okabe-Ito colors, clean grid
%% \input{fig-performance} in implementation section

\begin{figure}[!t]
\centering
\begin{tikzpicture}
\begin{loglogaxis}[
  width=0.97\columnwidth,
  height=5.6cm,
  xlabel={Threshold $T$},
  ylabel={Latency (ms)},
  xmin=1.8, xmax=175,
  ymin=1.5, ymax=50000,
  xtick={2,4,8,16,32,64,128},
  xticklabels={2,4,8,16,32,64,128},
  ytick={1,10,100,1000,10000},
  grid=both,
  minor grid style={gray!10, very thin},
  major grid style={gray!25, thin},
  legend pos=north west,
  legend style={
    font=\scriptsize, fill=white, draw=gray!40,
    fill opacity=0.93, text opacity=1,
    inner sep=4pt, row sep=1pt,
    rounded corners=1pt,
  },
  tick label style={font=\scriptsize},
  label style={font=\small},
  mark size=2pt,
  line width=0.8pt,
]

%% P1 — TEE (blue, solid, filled squares)
\addplot[color=okabeblue, solid,
         mark=square*, mark options={fill=okabeblue, draw=okabeblue, scale=1.1}]
  coordinates {
    (2,2.04)(3,4.42)(5,3.91)(7,10.31)(11,15.82)
    (16,126.78)(21,183.19)(32,352.44)(45,1090.82)
    (64,3742.17)(96,12157.66)(128,17576.10)
  };
\addlegendentry{P1 · TEE \;(3\,rnd)}

%% P2 — MPC (vermilion, dashed, filled triangles)
\addplot[color=okabered, dashed,
         mark=triangle*, mark options={fill=okabered, draw=okabered, scale=1.2}]
  coordinates {
    (2,11.55)(3,32.39)(5,14.50)(7,23.96)(11,22.28)
    (16,224.04)(21,339.54)(32,369.85)(45,2119.53)
    (64,4035.67)(96,13389.69)(128,15083.10)
  };
\addlegendentry{P2 · MPC\;(5\,rnd)}

%% P3+ — 2PC (green, dash-dot, filled circles)
\addplot[color=okabegreen, densely dashdotted,
         mark=*, mark options={fill=okabegreen, draw=okabegreen, scale=1.1}]
  coordinates {
    (2,16.79)(3,15.65)(5,21.77)(7,24.45)(11,25.27)
    (16,42.91)(21,57.51)(32,151.63)(45,145.83)
    (64,716.87)(96,2297.30)(128,2853.80)
  };
\addlegendentry{P3$^+$ · 2PC\;(2\,rnd)}

%% 100 ms reference
\addplot[okabegray!70, thin, densely dashed, mark=none, forget plot]
  coordinates {(1.8,100)(175,100)};
\node[font=\tiny, text=okabegray, anchor=south west]
  at (axis cs:2,107) {100\,ms};

\end{loglogaxis}
\end{tikzpicture}
\caption{Signing latency vs.\ threshold $T$ (log--log scale, ML-DSA-65,
Intel i7-12700H, single-machine Rust benchmark; extended range $T\!>\!32$ uses $N\!=\!T\!+\!1$).
P3$^+$ scales most gracefully: sub-100\,ms through $T\!=\!32$, under 3\,s through $T\!=\!128$.
P1 and P2 exceed 10\,s at $T\!\geq\!96$.
High P2 variance at $T\!\approx\!16$--$32$ reflects rejection-sampling retry count
(e.g., 4.6 vs.\ 2.2 attempts at $T\!=\!16$ vs.\ $T\!=\!32$).
Values shown here are per-attempt times from the extended benchmark run (covering $T$=2--128, all using $N\!=\!T\!+\!1$); these may differ from Table~\ref{tab:perf-threshold} which reports per-valid-signature times (including retries) from a separate focused benchmark run at specific $(T,N)$ configurations. Run-to-run variance up to $\approx$2$\times$ is expected due to rejection sampling.}
\label{fig:perf-threshold}
\end{figure}
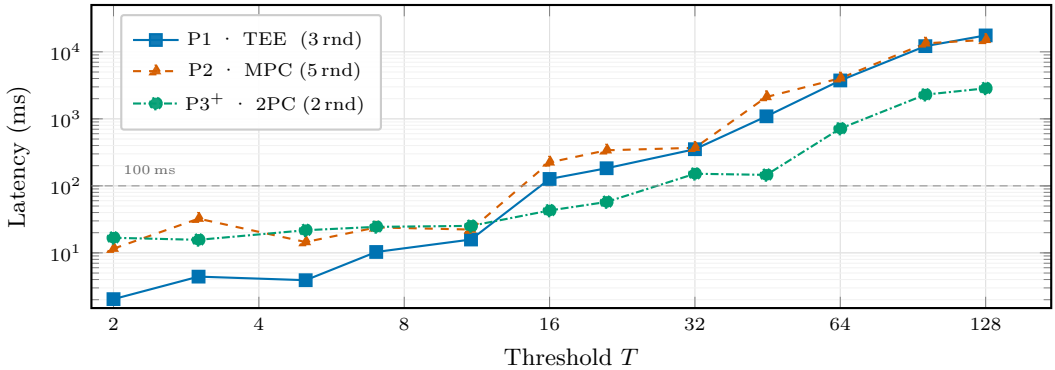

\fi

\subsection{Protocol Profile Comparison}

Table~\ref{tab:protocol-comparison} compares the three deployment profiles with experimental measurements for ML-DSA-65, $(N\!=\!5, T\!=\!3)$ using our optimized Rust implementation.

\begin{table}[!htbp]
\centering
\footnotesize
\setlength{\tabcolsep}{3pt}
\begin{tabular}{@{}lccccc@{}}
\toprule
\textbf{Profile} & \textbf{Rounds} & \textbf{Trust} & \textbf{Async} & \textbf{Online Time} & \textbf{Attempts} \\
\midrule
P1-TEE & 3 & TEE & $\times$ & 5.8 ms & 3.2 \\
P2-MPC & 5 & $T$-of-$N$ & $\times$ & 21.5 ms & 3.0 \\
P3+ & 2$^*$ & 1-of-2 CP & \checkmark & 22.1 ms & 3.7 \\
\bottomrule
\end{tabular}
\caption{Protocol profile comparison (ML-DSA-65, $N\!=\!5$, $T\!=\!3$, Rust implementation). $^*$P3+ = 1 signer step + 1 server round; signers precompute nonces offline.}
\label{tab:protocol-comparison}
\end{table}

\ifccs
\paragraph{Analysis.}
P1 (TEE): 5.8\,ms, optimal for HSM deployments. P2 (MPC): 21.5\,ms (3-of-5) via parallel edaBit A2B, 5-round structure; P2 latency exhibits high rejection-sampling variance (up to $14\times$ across runs; see Table~\ref{tab:perf-threshold} caption), so individual configurations may appear non-monotone in $T$. P3+: 22.1\,ms; signers precompute nonces offline (2.3\,KB/set), respond within a window; GC precomputation reduces 2PC to $\approx$12\,ms.
\else
\paragraph{Analysis.}
Profile P1 (TEE) is fastest at 5.8 ms but requires a trusted execution environment, making it suitable for HSM-based deployments. Profile P2 (5-round MPC) achieves 21.5 ms online time for 3-of-5 through parallel EdaBit-based A2B conversion and constant-depth AND tree for comparison; the 5-round structure merges commit-then-open (Rounds 3--4 $\to$ Round 2) and parallelizes comparison (Rounds 6--7 $\to$ Round 4). Profile P3+ provides 22.1 ms latency with FROST-like workflow: signers precompute nonces offline (2.3 KB per set) and respond within a time window (default: 5 min), requiring only a single user interaction. With GC precomputation, online 2PC evaluation drops to $\approx$12 ms.
\fi

\subsection{Comparison with Prior Work}

Table~\ref{tab:impl-comparison} compares our scheme with existing lattice-based threshold signature approaches.

\begin{table}[!htbp]
\centering
\footnotesize
\setlength{\tabcolsep}{2.5pt}
\begin{tabular}{@{}lcccccc@{}}
\toprule
\textbf{Scheme} & \textbf{$T$} & \textbf{Size} & \textbf{Rnd} & \textbf{FIPS} & \textbf{Async} & \textbf{Time} \\
\midrule
dPN25~\cite{dPN25} & 8 & 3.3 KB & 3 & Yes & $\times$ & -- \\
Noise flood~\cite{BCDG21} & $\infty$ & 17 KB & 3 & No & $\times$ & 300 ms \\
Ringtail~\cite{BMMSZ25} & $\infty$ & 13 KB & 2 & No & $\times$ & 200 ms \\
Ours (P1) & $\infty$ & 3.3 KB & 3 & Yes & $\times$ & 5.8 ms \\
Ours (P2) & $\infty$ & 3.3 KB & 5 & Yes & $\times$ & 21.5 ms \\
Ours (P3+) & $\infty$ & 3.3 KB & 2$^*$ & Yes & \checkmark & 22.1 ms \\
\bottomrule
\end{tabular}
\caption{Comparison with prior work. Times for prior schemes are taken from their original papers on potentially different hardware configurations; direct cross-scheme comparison is indicative only. Our times are from our Rust implementation on an Intel i7-12700H (ML-DSA-65, $N\!=\!5$, $T\!=\!3$). Our $\approx$21--45\% success rate is comparable to ML-DSA's expected $\approx$20--25\%. $^*$P3+ = 1 signer step + 1 server round.}
\label{tab:impl-comparison}
\end{table}

\ifccs
\paragraph{Analysis.}
dPN25 is FIPS-compatible but limited to $T \leq 8$ in practice (binomial reconstruction coefficients $\binom{T}{\lfloor T/2 \rfloor}$ exceed $q$ at $T \geq 26$; the Ball-\c{C}akan-Malkin bound~\cite{BCM21} independently gives an asymptotic share-size lower bound $\Omega(T \log T)$). Noise flooding produces $5\times$ larger signatures (not FIPS-compatible). Ringtail uses non-standard Dilithium. Our scheme achieves arbitrary $T$, standard 3.3\,KB signatures, and $\approx$21--45\% success rates overlapping with single-signer ML-DSA's $\approx$20--25\% (variation due to rejection-sampling variance and Irwin-Hall concentration effects).
\else
\paragraph{Analysis.}
The dPN25 compact scheme~\cite{dPN25} uses short-coefficient LSSS with binary reconstruction coefficients, achieving FIPS 204 compatibility but is limited to small thresholds ($T \leq 8$) in practice (binomial reconstruction coefficients exceed $q$ at $T \geq 26$; asymptotically, the Ball-\c{C}akan-Malkin bound~\cite{BCM21} gives $\Omega(T \log T)$ share size). Noise flooding~\cite{BCDG21} masks large Lagrange coefficients with overwhelming statistical noise, enabling arbitrary thresholds but producing $5\times$ larger signatures that break FIPS 204 compatibility; its near-100\% success rate comes from relaxed norm bounds. Ringtail~\cite{BMMSZ25} is based on a modified Dilithium variant (not standard ML-DSA) with custom parameters, achieving faster per-attempt time due to its 2-round structure.

Our scheme combines arbitrary thresholds with standard signature size and FIPS 204 verification compatibility. The observed $\approx$21--45\% success rate overlaps with single-signer ML-DSA's expected $\approx$20--25\%: the lower end (21\%, $T=16$) is in line with single-signer, while the higher end (45\%, $T=7$) reflects the Irwin-Hall distribution's greater concentration near zero relative to a uniform nonce, placing more nonces within the acceptance region $\|\mathbf{z}\|_\infty < \gamma_1 - \beta$ (see the Irwin-Hall success rate analysis in the companion supplement). The wide range (2× variation) is dominated by rejection-sampling variance, not systematic performance degradation.
\fi
\ifccs
Threshold overhead vs.\ single-party ML-DSA ($\approx$0.5\,ms): ${\approx}12\times$ (P1, 3-of-5) to $187\times$ (P3+, 32-of-45, from Table~\ref{tab:perf-threshold}); verification overhead $\approx 1\times$ (standard ML-DSA output).
\else
\subsection{Overhead vs.\ Standard ML-DSA}

Table~\ref{tab:overhead} quantifies the overhead of threshold signing compared to standard (single-party) ML-DSA for ML-DSA-65.

\begin{table}[!htbp]
\centering
\footnotesize
\setlength{\tabcolsep}{3pt}
\begin{tabular}{@{}lccccc@{}}
\toprule
\textbf{Config} & \textbf{Profile} & \textbf{Std Sign} & \textbf{Thresh Sign} & \textbf{Overhead} & \textbf{Attempts} \\
\midrule
3-of-5 & P1-TEE & 0.5 ms & 5.8 ms & 11.6$\times$ & 3.2 \\
3-of-5 & P2-MPC & 0.5 ms & 21.5 ms & 43.0$\times$ & 3.0 \\
3-of-5 & P3+ & 0.5 ms & 22.1 ms & 44.2$\times$ & 3.7 \\
\midrule
7-of-11 & P1-TEE & 0.5 ms & 12.0 ms & 24.0$\times$ & 2.9 \\
7-of-11 & P3+ & 0.5 ms & 29.4 ms & 58.8$\times$ & 3.4 \\
\midrule
32-of-45 & P1-TEE & 0.5 ms & 60.2 ms & 120$\times$ & 3.2 \\
32-of-45 & P3+ & 0.5 ms & 93.5 ms & 187$\times$ & 3.6 \\
\bottomrule
\end{tabular}
\caption{Threshold signing overhead vs.\ standard ML-DSA-65 (Rust implementation). Standard ML-DSA signing takes $\approx$0.5 ms in optimized Rust. Overhead scales with $T$ but remains practical even at $T=32$.}
\label{tab:overhead}
\end{table}

The threshold overhead (12--187$\times$) is inherent to any threshold scheme requiring coordination among $T$ parties. Verification has negligible overhead ($\approx$1.0$\times$) since threshold signatures are standard ML-DSA signatures.
\fi

\ifccs\else
\subsection{Scalability Analysis}

Our benchmarks demonstrate practical scalability up to $T=32$ signers. Profile P1-TEE scales from 4 ms ($T=2$) to 60 ms ($T=32$), approximately linear in $T$. Profile P3+ scales from 17 ms ($T=2$) to 94 ms ($T=32$), maintaining sub-100ms \emph{per-valid-signature} latency throughout (Table~\ref{tab:perf-threshold}; Figure~\ref{fig:perf-threshold} shows per-attempt times from a separate extended benchmark, which may be up to $2\times$ higher due to rejection-sampling retries). Profile P2-MPC shows higher variance (21--194 ms) due to MPC overhead sensitivity to retry count, exceeding 100ms for $T \geq 16$.

For larger thresholds ($T > 32$), the dominant cost becomes polynomial operations ($O(T \cdot n \cdot \log n)$ per signature). Extended benchmarks (Figure~\ref{fig:perf-threshold}) show P3$^+$ achieves sub-1\,s through $T=64$ and sub-3\,s through $T=128$; P1 and P2 exceed 10\,s at $T \geq 96$ due to Lagrange evaluation and MPC simulation overhead, respectively.
\fi

\ifccs
\paragraph{Communication and complexity.} Per-party online: 12.3\,KB per attempt. Computation is dominated by $\mathbf{A}\mathbf{y}_i$ ($O(\ell k n \log n)$, parallelised via Rayon). P2 preprocessing: 1536 edaBits $+$ ${\sim}$50k Beaver triples (225\,ms offline, ${<}$1\,ms online). P3+ pre-signing: 2.3\,KB per nonce set (21\,ms). Full analyses are in the appendix.
\else
\subsection{Preprocessing Analysis}

All profiles benefit from preprocessing. The Shamir nonce DKG (Section~\ref{sec:nonce-dkg}) adds one offline round of $O(|S|^2)$ point-to-point messages ($\approx$3.6\,KB each for ML-DSA-65) before each signing session. In practice, multiple nonce sharings can be precomputed during idle time and consumed on demand, amortizing the DKG cost across signing sessions.

Profile P2 (MPC) benefits from additional preprocessing. Table~\ref{tab:preprocessing} shows the amortization effect.

\begin{table}[!htbp]
\centering
\footnotesize
\setlength{\tabcolsep}{3pt}
\begin{tabular}{@{}lccc@{}}
\toprule
\textbf{Component} & \textbf{Per-Attempt} & \textbf{Offline Gen.} & \textbf{Storage} \\
\midrule
EdaBits & 1536 & 45 ms & 6.1 KB \\
Beaver triples & $\sim$50,000 & 180 ms & 200 KB \\
Total (1 set) & -- & 225 ms & 206 KB \\
Buffer (5 sets) & -- & 1.13 s & 1.03 MB \\
\bottomrule
\end{tabular}
\caption{P2 preprocessing requirements (ML-DSA-65, 5 parties, Rust implementation). Online consumption is $<$1 ms per set after preprocessing.}
\label{tab:preprocessing}
\end{table}

\paragraph{P2 Implementation Optimizations.}
Our Rust implementation achieves 21.5 ms online time (3-of-5) through several optimizations. Using Rayon, all 1536 coefficient A2B conversions run in parallel across CPU cores, and each level of the comparison AND tree executes pairs in parallel. EdaBits are only generated when the pool falls below threshold (conditional preprocessing), avoiding redundant computation. Beaver triples are pre-fetched per AND-tree level for cache efficiency.

The main bottleneck remains MPC communication overhead. For the $N=3$ OpenVPN testbed, measured round-trip latency averaged $\approx$2.1 ms (LAN setting). For larger configurations ($N \geq 5$) benchmarked on a single machine, network communication is not measured; in actual distributed deployments, network latency would dominate, making preprocessing even more critical.

\paragraph{Network Latency Analysis.}
For single-machine benchmarks ($N \geq 5$), we project distributed deployment latency analytically as $t_{\mathsf{total}} \approx t_{\mathsf{compute}} + \text{rounds} \times \text{RTT}$, where $t_{\mathsf{compute}}$ is measured and RTT varies by deployment scenario. On a LAN (1\,ms RTT), P1 takes $\approx$10\,ms and P3+ $\approx$19\,ms for $(3,5)$-threshold. Over a WAN (20\,ms RTT), P1 rises to $\approx$67\,ms and P2 to $\approx$120\,ms. Over intercontinental links (100\,ms RTT), P2 reaches $\approx$520\,ms, while P3+ remains under 220\,ms due to its 2-round structure. Detailed estimates appear in the supplementary material. P3+'s reduced round count makes it the most network-friendly profile for geographically distributed deployments.

\subsection{Pre-signing Analysis (P3+)}

Profile P3+ enables semi-asynchronous signing through offline nonce precomputation.

\begin{table}[!htbp]
\centering
\footnotesize
\setlength{\tabcolsep}{3pt}
\begin{tabular}{@{}lcc@{}}
\toprule
\textbf{Component} & \textbf{Size} & \textbf{Generation Time} \\
\midrule
Nonce $\mathbf{y}_i$ & 1.5 KB & 8 ms \\
Commitment $\mathbf{w}_i$ & 768 B & 12 ms \\
Commitment hash & 32 B & $<$1 ms \\
\midrule
Total (1 set) & 2.3 KB & 21 ms \\
Buffer (10 sets) & 23 KB & 210 ms \\
\bottomrule
\end{tabular}
\caption{P3+ pre-signing storage and generation (ML-DSA-65, per party). Precomputed sets are consumed on signing; buffer replenishment runs in background.}
\label{tab:presigning}
\end{table}

\paragraph{GC Precomputation.} The r0-check garbled circuit has $\approx$1.2M gates. Pre-garbling takes $\approx$45 ms; online evaluation takes $\approx$12 ms. With precomputation, the 2PC phase drops from 80 ms to 15 ms.

\subsection{Retry Efficiency Analysis}

Due to rejection sampling, signing requires multiple attempts. Table~\ref{tab:retry-efficiency} analyzes retry behavior across profiles (3-of-5 configuration).

\begin{table}[!htbp]
\centering
\footnotesize
\setlength{\tabcolsep}{3pt}
\begin{tabular}{@{}lccccc@{}}
\toprule
\textbf{Profile} & \textbf{$p_{\text{succ}}$} & \textbf{$\mathbb{E}[\text{att}]$} & \textbf{Per-Att Time} & \textbf{Retry Cost} & \textbf{$P$(5 att)} \\
\midrule
P1-TEE & 31\% & 3.2 & 1.8 ms & Full & 85\% \\
P2-MPC & 33\% & 3.0 & 7.2 ms & Full & 87\% \\
P3+ & 27\% & 3.7 & 6.0 ms & Partial$^*$ & 79\% \\
\bottomrule
\end{tabular}
\caption{Retry efficiency (3-of-5, ML-DSA-65). $p_{\text{succ}}$ = per-attempt success rate (inherent to FIPS 204 rejection sampling); $\mathbb{E}[\text{att}]$ = expected attempts; $P(5)$ = success probability within 5 attempts. $^*$P3+ retries reuse precomputed nonces, reducing per-retry latency.}
\label{tab:retry-efficiency}
\end{table}

Profile P3+ has a unique advantage: failed attempts only consume the 2PC evaluation cost ($\approx$10 ms with precomputation), while nonce regeneration happens asynchronously in the background. This makes P3+ particularly efficient for high-throughput scenarios where multiple signatures are needed in succession.

\subsection{Complexity Analysis}

\paragraph{Communication.} Per-party: 12.3 KB per attempt (ML-DSA-65). Total: $O(T \cdot n \cdot k)$ per signature, linear in $T$.

\paragraph{Computation.} Per-party: $O(\ell \cdot k \cdot n \log n)$ dominated by $\mathbf{A}\mathbf{y}_i$. With 8-thread parallelization: $3$--$4\times$ additional speedup.

\paragraph{Storage.} Per-party: $O(n \cdot (\ell + k))$ for shares, $O(N)$ for seeds. P3+ adds $O(B \cdot (\ell + k))$ for pre-signing buffer of size $B$.

\paragraph{Production Notes.} Our Rust implementation achieves 10--50$\times$ speedup over Python. Further optimization with AVX-512 SIMD could provide additional 2--4$\times$ improvement. HSM integration is recommended for secret share storage. The $\approx$21--45\% per-attempt success rate is inherent to FIPS 204's rejection sampling, not our construction.
\fi

\subsection{FIPS 204 Compliance Verification}

We verified our implementation against FIPS 204 specifications. Parameter verification confirmed all FIPS 204 constants ($q = 8380417$, $n = 256$, $d = 13$, $\omega = 55$, etc.). Algorithm correctness testing covered KeyGen, Sign, Verify, and all auxiliary functions (HighBits, LowBits, MakeHint, UseHint, NTT, SampleInBall, Power2Round, w1Encode) with FIPS 204 Known-Answer Tests (KAT) verifying each building block against specification-derived test vectors. Bound verification confirmed all signatures satisfy $\|\mathbf{z}\|_\infty < \gamma_1 - \beta$ and hint weight $\leq \omega$.

Our ``FIPS 204-Compatible'' claim refers to \emph{signature format compatibility}: threshold-generated signatures are byte-identical to single-signer ML-DSA signatures and pass any compliant FIPS 204 verifier without modification. Full NIST ACVP Known-Answer-Test (KAT) certification for production deployment would additionally require the deterministic encoding functions specified in FIPS 204 Section~7.2.

\ifccs
\paragraph{Artifact availability.}
The Rust implementation ($\approx$5500 lines core + $\approx$4000 lines tests) covers all three profiles, preprocessing, pre-signing, FIPS 204 compliance tests, and the benchmark harness. See Appendix~\ref{app:open-science} for the anonymous access URL.
\else
\subsection{Code Availability}

Our primary implementation is in Rust ($\approx$5500 lines core + $\approx$4000 lines tests), with production-ready performance including parallel processing. An earlier Python prototype was used during development but is no longer maintained; the Rust implementation is the reference.

Both implementations include:
\begin{itemize}
    \item All three protocol profiles (P1-TEE, P2-MPC 5-round, P3+-SemiAsync)
    \item Preprocessing modules for P2 (EdaBits, Beaver triples with parallel generation)
    \item Pre-signing and GC precomputation modules for P3+
    \item Comprehensive test suite: FIPS 204 KAT tests for all building blocks plus protocol integration tests across all profiles
    \item Benchmark harness for reproducible measurements across threshold configurations
\end{itemize}

Code will be released upon publication.
\fi

% Extensions Section (Condensed)

\section{Extensions}
\label{sec:extensions}

We briefly discuss extensions to our basic protocol. Full details are in the supplementary material.

\subsection{Distributed Key Generation}
\label{sec:dkg}

\ifccs
The \emph{key DKG} eliminates the trusted dealer via Feldman commitments~\cite{Feldman87}. The \emph{nonce DKG} (our primary contribution) generates a fresh degree-$(T\!-\!1)$ Shamir sharing per signing session; nonce share privacy follows from Theorem~\ref{thm:it-privacy}. Nonce DKG batches are preprocessed offline and amortized across sessions; see Supplementary Material.
\else
Our protocol employs two distinct DKG protocols. The \emph{key DKG} (executed once during setup) replaces a fully trusted dealer with a distributed protocol: each party contributes a random sharing of local secrets, parties aggregate their shares, and Feldman commitments~\cite{Feldman87} allow every participant to verify that all shares lie on a consistent polynomial without revealing the secret.

\emph{What Feldman VSS does and does not provide.}
Feldman VSS guarantees \emph{polynomial consistency}: each party can verify that their share $\mathbf{s}_{1,h} = f(h)$ is consistent with the committed polynomial $\{g^{a_0},\ldots,g^{a_{T-1}}\}$. Crucially, this does \emph{not} constrain the constant term $f(0) = \mathbf{s}_1$ to have short coefficients. The Shamir shares $f(h)$ for $h \geq 1$ lie in $\mathbb{Z}_q^\ell$ (large, as expected), which is fine for our protocol since the signing bound depends only on the reconstructed $\mathbf{s}_1$, not on individual share norms. The \emph{only} constraint requiring smallness is that $\mathbf{s}_1 = f(0)$ satisfies $\|\mathbf{s}_1\|_\infty \leq \eta = 4$, which is needed for the rejection-sampling bound in the signature.

\emph{Why enforcing short $\mathbf{s}_1$ without trust is hard.}
To verifiably guarantee $\|\mathbf{s}_1\|_\infty \leq \eta$ in a trustless DKG, parties would need to produce a zero-knowledge proof that the committed polynomial has a small constant term. Current lattice-based ZK range proofs for $\ell_\infty$ balls in $R_q$ carry significant communication overhead; a practical treatment is an open problem (see Section~\ref{sec:conclusion}). We note that this limitation is shared by all lattice-based threshold signature schemes in the literature---including Raccoon~\cite{Raccoon2024} and Dilithium-based constructions---that require small-coefficient secrets.

\emph{Practical alternative.}
Profile P1 (Section~\ref{sec:profile-p1}) uses a TEE coordinator that runs ML-DSA key generation in a hardware-isolated environment and distributes Shamir shares directly. This provides the short-$\mathbf{s}_1$ guarantee at the cost of a hardware trust assumption (Intel SGX or equivalent), which is already required by P1's security model. For deployments where hardware trust is acceptable---which covers a broad class of enterprise PKI scenarios---this is the recommended setup mechanism. For lattice-specific VSS/DKG, see~\cite{FS24,GJKR07}.

The \emph{nonce DKG} (executed before each signing session) is our primary technical contribution. As described in Section~\ref{sec:nonce-dkg}, each party generates a degree-$(T\!-\!1)$ polynomial whose higher-degree coefficients (degree $\geq 1$) are sampled uniformly from $\Rq^\ell$ and whose constant term is sampled from the bounded nonce range $[-\lfloor\gamma_1/|S|\rfloor, \lfloor\gamma_1/|S|\rfloor]^{n\ell}$, producing a fresh Shamir sharing of the signing nonce. This structural choice enables nonce share privacy (Theorem~\ref{thm:it-privacy}). The nonce DKG can be preprocessed offline and amortized across multiple signing sessions; see the Supplementary Material for the full protocol details.
\fi

\phantomsection\label{sec:blame}% label visible in both modes for cross-refs
\ifccs\else
\subsection{Optimistic Blame Protocol}

Our protocol uses an \emph{optimistic} approach: the normal signing path has no verification overhead. Only when failures exceed statistical expectation do we trigger blame.

\paragraph{Triggering Conditions.}
\begin{itemize}
    \item \textbf{Excessive aborts}: More than $K$ consecutive failures, where $K$ is set based on the deployment's per-attempt success rate $p$: $K = \lceil \log(10^{-4}) / \log(1-p) \rceil$. For $p \geq 25\%$ (all configurations with $T \leq 21$), $K = 33$ suffices (honest probability $(0.75)^{33} < 10^{-4}$); for worst-case $T=32$ ($p \approx 21\%$), set $K = 48$ to maintain the $10^{-4}$ guarantee.
    \item \textbf{Commitment mismatch}: Verifiable inconsistency in hash values (immediate blame)
\end{itemize}

\paragraph{Blame Protocol (Profile P1).}
Under TEE coordinator:
\begin{enumerate}
    \item TEE requests each signer to reveal their nonce DKG polynomials and mask inputs (sent only to TEE)
    \item TEE recomputes nonce shares $\mathbf{y}_i$, commitments $\mathbf{w}_i$, responses $\mathbf{z}_i = \mathbf{y}_i + c \cdot \mathbf{s}_{1,i}$, and masks, verifying consistency with submitted values
    \item Inconsistent parties are identified; TEE outputs only the cheater's identity
    \item Malicious party is excluded from future sessions
\end{enumerate}

\paragraph{Security Analysis.}
The blame protocol achieves the following properties:
\begin{itemize}
    \item \textbf{Privacy}: All sensitive values ($\mathbf{f}_h(x)$, $\mathbf{s}_{1,i}$, mask seeds) are transmitted only to the TEE coordinator and never revealed to other parties. The TEE outputs only the identity of misbehaving parties, preserving nonce share privacy of honest parties' shares (Theorem~\ref{thm:it-privacy} continues to hold even after blame execution).
    \item \textbf{Correctness}: The TEE recomputes all protocol values deterministically from revealed inputs. A party is blamed if and only if its revealed inputs are inconsistent with its submitted protocol messages (e.g., $\mathbf{W}_i \neq \lambda_i \mathbf{A}\mathbf{y}_i + \mathbf{m}_i^{(w)}$), ensuring honest parties are never falsely accused.
    \item \textbf{Trigger threshold}: The $K = 33$ consecutive-failure threshold is chosen such that the probability of triggering blame due to honest rejection sampling alone is $(1-p)^{33} < 10^{-4}$ for $p \geq 25\%$ (ML-DSA-65 with $T \leq 21$). For larger $T$, set $K = \lceil \log(10^{-4}) / \log(1-p) \rceil$ to maintain the false-positive rate below $10^{-4}$.
    \item \textbf{UC composition}: Under the TEE trust model ($\mathcal{F}_{\mathsf{TEE}}$ honest), the blame protocol UC-realizes an ideal functionality $\mathcal{F}_{\mathsf{Blame}}$ that takes inputs from all parties, identifies inconsistent parties via deterministic recomputation, and outputs cheater identities to all parties while keeping honest inputs private. Formal UC proof is deferred to the full version; the key observation is that the simulator can run the blame protocol honestly (since it has access to all corrupted parties' inputs and can program the TEE's output) and the view of honest parties consists only of the output (cheater identities), which is identically distributed in the real and ideal worlds.
\end{itemize}

\paragraph{Blame Protocol (Profile P2).}
Profile P2 does not require a separate blame protocol: SPDZ MAC verification already provides \emph{identifiable abort}. When a MAC check fails during the r0-check subprotocol, the SPDZ functionality $\mathcal{F}_{\mathsf{SPDZ}}$ outputs the identity of the party whose MAC is invalid (with probability $\geq 1 - 1/q \approx 1$ for $q = 8380417$). The cheating party is immediately identified and excluded from future sessions. Note that the SPDZ MAC check catches input manipulation (wrong shares), but not wrong inputs submitted to $\mathcal{F}_{\mathsf{SPDZ}}$; a party that consistently submits out-of-range inputs to the edaBit protocol is caught by the cut-and-choose parameter $B = 64$.

\paragraph{Blame Protocol (Profile P3+).}
Profile P3+ uses 2PC, and blame attribution proceeds as follows when failures exceed the threshold $K$. Since CP1 does not hold key shares, blame is routed to CP2 (or a designated auditor):
\begin{enumerate}
    \item CP2 requests each signer $i \in S$ to send: (a) their own nonce DKG polynomial $\mathbf{f}_i(x)$, and (b) all evaluations $\{\mathbf{f}_j(i)\}_{j \in S}$ they received during the DKG round.
    \item CP2 reconstructs each party's nonce share: $\mathbf{y}_i = \sum_{j \in S} \mathbf{f}_j(i)$.
    \item Signer $i$ also sends to CP2 (over a confidential channel): their key share $\mathbf{s}_{1,i}$ and commitment randomness $r_i$. CP2 verifies: (a) $\mathsf{Com}_i = H(\texttt{"com"} \| \mathbf{y}_i \| \mathbf{w}_i \| r_i)$ (commitment binding); (b) $\mathbf{w}_i = \mathbf{A}\mathbf{y}_i$ (nonce consistency); (c) $\mathbf{z}_i = \mathbf{y}_i + c \cdot \mathbf{s}_{1,i}$ (response correctness).
    \item Inconsistent signers are identified and excluded.
\end{enumerate}
\textbf{Note on nonce privacy after blame.} Revealing $\mathbf{f}_i(x)$ to CP2 discloses the constant term $\hat{\mathbf{y}}_i$ and all higher coefficients for the session under investigation, so per-session nonce privacy is intentionally forfeited during blame attribution. Theorem~\ref{thm:it-privacy} (nonce share privacy) holds for \emph{all sessions not under blame investigation}; for the session under review, nonce privacy is deliberately sacrificed for accountability. Key privacy is maintained because $\mathbf{s}_{1,i}$ is sent only to CP2 (trusted auditor), not broadcast. After any blame invocation, a key-refresh step is recommended to limit exposure.

\begin{remark}[Blame leakage across sessions]
Each blame invocation leaks at most $n\ell \cdot \log_2(2\lfloor\gamma_1/|S|\rfloor + 1) \approx 25{,}600$ bits of information (the full session nonce polynomial $\mathbf{f}_i(x)$, with coefficients bounded in $[-\lfloor\gamma_1/|S|\rfloor, \lfloor\gamma_1/|S|\rfloor]$). Because each session uses freshly sampled higher polynomial coefficients, successive blame invocations on \emph{different sessions} are statistically independent: the nonce polynomial for session $\tau$ reveals no information about session $\tau'$. There is therefore no cumulative cross-session leakage from the nonce DKG component alone. However, once $\mathbf{s}_{1,i}$ is disclosed to CP2 (for blame verification), CP2 is effectively a key-holder. A mandatory key-refresh should follow the \emph{first} blame invocation that requires revealing key material; subsequent blame invocations may reuse the refreshed key under the same analysis.
\end{remark}
\fi

\subsection{Profile P2: Fully Distributed Signing}
\label{sec:profile-p2}

\ifccs
Profile P2 eliminates TEE trust via SPDZ-based MPC~\cite{DPSZ12} with edaBits~\cite{EKMOZ20} for the r0-check predicate, avoiding reconstruction of $c\mathbf{s}_2$ at any party.
\else
For deployments where TEE trust is unacceptable, we provide a fully distributed protocol using secure multi-party computation. The core challenge is evaluating the r0-check predicate without reconstructing $c\mathbf{s}_2$ at any single point---since virtually all ML-DSA challenges $c$ are invertible in $\Rq$ (the non-invertible fraction is $< 2^{-15}$; see Remark~\ref{lem:challenge-invertible}), leaking $c\mathbf{s}_2$ would compromise the secret key.
\fi

\ifccs
The r0-check predicate $\infnorm{\mathsf{LowBits}(\mathbf{w} - c\mathbf{s}_2, 2\gamma_2)} < \gamma_2 - \beta$ is evaluated via SPDZ-based MPC using \emph{edaBits}~\cite{EKMOZ20}: parties mask $c\mathbf{s}_2$ shares and open the masked value; the correlated binary representation (edaBit) enables domain conversion without reconstruction, allowing constant-depth comparison. Combiner-mediated commit-then-open and parallel Beaver-triple comparisons reduce the base 8-round protocol to \textbf{5 online rounds}; by UC composition~\cite{Canetti01} the 5-round protocol UC-realizes the same functionality. P2 achieves dishonest-majority with arbitrary $T$ in 5 rounds vs. 23--79 rounds~\cite{BdCE25} and $T \leq 6$~\cite{CDENP26}. Full protocol and UC proof in Appendix~\ref{app:uc-framework}.
\else
\paragraph{The r0-Check Problem.}
The predicate $\infnorm{\mathsf{LowBits}(\mathbf{w} - c\mathbf{s}_2, 2\gamma_2)} < \gamma_2 - \beta$ requires:
\begin{enumerate}
    \item Computing $\mathbf{w}' = \mathbf{w} - c\mathbf{s}_2$ (parties hold shares of $c\mathbf{s}_2$)
    \item Extracting $\mathsf{LowBits}(\mathbf{w}', 2\gamma_2)$ (non-linear)
    \item Comparing the infinity norm against $\gamma_2 - \beta$
\end{enumerate}
Naive reconstruction of $c\mathbf{s}_2$ would leak the secret key (since virtually all challenges $c$ are invertible in $\Rq$; see Section~\ref{sec:r0check}).

\paragraph{Approach: edaBits-Based Arithmetic-to-Binary Conversion.}
We leverage \emph{edaBits}~\cite{EKMOZ20}, which provide correlated randomness $(\langle r \rangle_q, \langle r \rangle_2)$: the same random value shared in both arithmetic (mod $q$) and binary form. This enables efficient domain conversion:

\begin{enumerate}
    \item \textbf{Masked Opening}: Parties compute $\langle \mathbf{w}' + \mathbf{r} \rangle_q$ and open this masked value. The mask $\mathbf{r}$ hides $\mathbf{w}'$.

    \item \textbf{Binary Subtraction}: Using $\langle \mathbf{r} \rangle_2$, parties compute $\langle \mathbf{w}' \rangle_2 = (\mathbf{w}' + \mathbf{r}) \boxminus \langle \mathbf{r} \rangle_2$ via binary circuits, where $\boxminus$ denotes binary subtraction with borrow propagation.

    \item \textbf{Comparison}: With $\mathbf{w}'$ in binary form, the range check $|\mathbf{w}'_j| < \gamma_2 - \beta$ becomes a binary comparison circuit, evaluable in constant rounds.

    \item \textbf{AND Aggregation}: Results for all $nk$ coefficients are combined via an AND-tree to produce a single pass/fail bit.
\end{enumerate}

\paragraph{Protocol Structure.}
Using combiner-mediated commit-then-open and constant-depth comparison circuits, our optimized protocol achieves \textbf{5 online rounds}:
\begin{center}
\begin{tabular}{cl}
\toprule
\textbf{Round} & \textbf{Operation} \\
\midrule
1 & Share exchange + aggregation \\
2 & Combiner-mediated masked reveal \\
3 & A2B conversion via edaBits \\
4 & Parallel comparison + AND aggregation \\
5 & Output \\
\bottomrule
\end{tabular}
\end{center}

\paragraph{Round Optimization Techniques.}
\begin{enumerate}
    \item \textbf{Combiner-mediated commit-then-open (base-protocol Rounds 2--4 $\to$ 5-round Round 2):} All parties send $(\mathsf{Com}_i, \mathsf{masked}_i, \mathsf{rand}_i)$ to a designated combiner (one of the $N$ signing parties, potentially corrupted), who releases all openings simultaneously after collecting from all parties. SPDZ MAC checks catch any tampering by the combiner, since each share is authenticated with a global MAC key unknown to individual parties. This collapses three base-protocol rounds (masked-value exchange, commitment, and opening) into a single combiner-mediated round, saving \emph{two} rounds.

    \item \textbf{Constant-depth comparison (base-protocol Rounds 6--7 $\to$ 5-round Round 4):} All $nk = 1536$ coefficient comparisons are evaluated in parallel using $1536$ Beaver triples; all $(d, e)$ openings are batched into a \emph{single} network round (Round 4). After receiving opened $(d,e)$ values, each party \emph{locally} computes comparison bits $b_j \in \{0,1\}$ and then locally aggregates $\langle \sum_j b_j \rangle$ via linear sharing (no additional network round). The AND of all 1536 comparison results reduces to: $\bigwedge_j b_j = 1 \iff \sum_j b_j = nk = 1536$, which is a \emph{constant-depth} (depth-1) check: opening $\langle \sum_j b_j \rangle$ in Round~5 (``Output'') and comparing to $1536$ constitutes a single network round with no sequential Beaver-triple dependencies.
\end{enumerate}

\paragraph{UC Security of 5-Round Optimization.}
The 5-round protocol is obtained from the 8-round base protocol (Appendix~\ref{app:uc-framework}) via two standard MPC composition steps: Optimization~1 saves 2 rounds (collapsing three base-protocol rounds into one) and Optimization~2 saves 1 round (collapsing two parallel comparison rounds into one), giving $8 - 2 - 1 = 5$ rounds total.

\textit{Optimization 1 (Combiner-mediated commit-then-open, base-protocol Rounds 2--4 $\to$ 5-round Round~2, saving 2 rounds):} Parties send $(\mathsf{Com}_i, \mathsf{masked}_i, \mathsf{rand}_i)$ to a designated combiner (one of the $N$ parties, treated as potentially corrupted) who releases all openings simultaneously. This is UC-equivalent to a sequential broadcast~\cite{CLOS02}: if the combiner withholds or alters any opening, the SPDZ MAC check---which authenticates each share with a global key unknown to the combiner---detects the tampering with probability $1 - \epsilon_{\mathsf{SPDZ}}$. The simulator for this merged round constructs the same view as the 2-round version (send $\to$ open); leakage to the adversary is unchanged since only MAC-authenticated values are released. Confidentiality of the masked values $\mathsf{masked}_i$ from the combiner is protected by the pairwise PRF masks (Lemma~\ref{lem:mask-hiding}), which ensure that even a corrupted combiner observing all $\mathsf{masked}_i$ cannot recover any party's contribution, provided $|S \setminus C| \geq 2$.

\textit{Optimization 2 (Constant-depth comparison, base-protocol Rounds 6--7 $\to$ 5-round Round~4):} All $nk = 1536$ Beaver-triple comparisons are opened simultaneously in a single network round (Round~4). There are no sequential dependencies among the $1536$ Beaver triples: each triple $(d_j, e_j) = (\alpha_j \oplus \langle x_j \rangle,\, \beta_j \oplus \langle y_j \rangle)$ depends only on the party's local share of $\langle x_j \rangle, \langle y_j \rangle$ and the preprocessed triple $(\langle \alpha_j \rangle, \langle \beta_j \rangle, \langle \gamma_j \rangle)$ with $\gamma_j = \alpha_j \wedge \beta_j$, none of which depend on other comparisons. After receiving all $(d_j, e_j)$ in Round~4, parties locally compute the comparison output bit $b_j = \gamma_j \oplus (d_j \wedge \beta_j) \oplus (e_j \wedge \alpha_j) \oplus (d_j \wedge e_j)$ (standard Beaver-triple AND evaluation; the triple components $\alpha_j, \beta_j, \gamma_j$ are distinct from the output bit $b_j$) and aggregate $\langle B \rangle = \sum_j \langle b_j \rangle$ (purely local linear operation, no network round). The AND of all $1536$ comparison bits satisfies $\bigwedge_j b_j = 1 \iff B = nk$; opening $\langle B \rangle$ in Round~5 (``Output'') is a \emph{single} network round. The AND aggregation therefore requires \emph{no sequential AND-gate rounds} (the AND tree is replaced by a single linear aggregation + one opening round, with no Beaver triples).

By the UC composition theorem~\cite{Canetti01}, each merged step preserves the ideal functionality's input-output behavior and leakage profile. The simulator for the 5-round protocol is constructed by composing the simulators for each sub-protocol (SPDZ, edaBits, binary circuits) with the round assignment remapped; each simulator step depends on the functionality, not the round structure. The hybrid argument in Appendix~\ref{app:uc-framework} ($\mathsf{H}_0$--$\mathsf{H}_6$) carries over directly. Since the edaBits functionality is used as a black box (its output distribution is indistinguishable from ideal, at cost $\epsilon_{\mathsf{edaBits}}$) \emph{before} the SPDZ multiplication gate, there is no interaction between the two error terms, and the total security loss remains $\epsilon_{\mathsf{SPDZ}} + \epsilon_{\mathsf{edaBits}} < 2^{-11}$.

\paragraph{Comparison with Prior MPC Approaches.}
\begin{itemize}
    \item \textbf{Bienstock et al.}~\cite{BdCE25}: 23--79 online rounds (per~\cite{BdCE25}), honest majority
    \item \textbf{Celi et al.}~\cite{CDENP26}: 6 online rounds, but restricted to $T \leq 6$
    \item \textbf{Our P2}: 5 online rounds, dishonest majority (unforgeability), arbitrary $T$
\end{itemize}
Profile P2 thus achieves dishonest-majority \emph{unforgeability} with arbitrary thresholds in constant rounds. Note: commitment and r0-check mask \emph{privacy} additionally requires $|S \setminus C| \geq 2$ (Lemma~\ref{lem:mask-hiding}), meaning at most $T-2$ signing-set corruptions are tolerated for full mask privacy when $T \geq 3$.

\paragraph{Communication Complexity.}
Per-party online communication is $O(nk \cdot \log q)$ bits for the A2B conversion and comparison circuits. Preprocessing (edaBits generation) requires $O(nk)$ correlated randomness per signing attempt.
\fi

\subsection{Security Analysis of Profile P2}
\label{sec:p2-security}

We prove that Profile P2 achieves UC security by composition of standard MPC building blocks.

\ifccs\else
\begin{definition}[Ideal Functionality $\mathcal{F}_{r_0}$]
\label{def:f-r0}
The functionality $\mathcal{F}_{r_0}$ receives shares $[\mathbf{w}']_i$ from each party, reconstructs $\mathbf{w}' = \sum_i [\mathbf{w}']_i \bmod q$, computes $\mathsf{result} = \bigwedge_j (|\mathbf{w}'_j|_q < \gamma_2 - \beta)$, and outputs only $\mathsf{result}$ (1 bit) to all parties.
\end{definition}
\fi

\begin{theorem}[Profile P2 UC Security]
\label{thm:p2-uc}
Protocol $\Pi_{r_0}$ UC-realizes $\mathcal{F}_{r_0\text{-check}}$ against static malicious adversaries (corrupting up to $N-1$ parties) in the $(\mathcal{F}_{\mathsf{SPDZ}},\allowbreak \mathcal{F}_{\mathsf{edaBits}}, \mathcal{F}_{\mathsf{Binary}})$-hybrid model, with security-with-abort. Concretely:
\begin{align*}
&\bigl|\Pr[\mathsf{REAL}_{\Pi_{r_0}, \mathcal{A}, \mathcal{Z}} = 1] - \Pr[\mathsf{IDEAL}_{\mathcal{F}_{r_0\text{-check}}, \mathcal{S}, \mathcal{Z}} = 1]\bigr|\\
&\quad\leq \epsilon_{\mathsf{SPDZ}} + \epsilon_{\mathsf{edaBits}} < 2^{-11}
\end{align*}
for ML-DSA-65 parameters, where $\epsilon_{\mathsf{SPDZ}} \leq m/q < 2^{-12}$ with $m = nk = 1536$ coefficients ($1536/8380417 \approx 2^{-12.4}$), and $\epsilon_{\mathsf{edaBits}} \leq m/q + 2^{-B} < 2^{-12}$ where $B = 64$ makes the cut-and-choose term $2^{-64}$ negligible (the dominant term is again $m/q$ from the internal SPDZ MAC check within edaBits generation).
\end{theorem}

\ifccs
\begin{proof}[Proof sketch]
Profile P2 achieves UC security against malicious adversaries corrupting up to $N-1$ parties via six hybrids: replacing SPDZ, edaBits, and binary circuits with their ideal functionalities introduces total loss $\epsilon_{\mathsf{SPDZ}} + \epsilon_{\mathsf{edaBits}} < 2^{-11}$; mask hiding (Lemma~\ref{lem:mask-hiding}) protects commitments under $|S \setminus C| \geq 2$. Full UC simulator construction and hybrid proofs in Appendix~\ref{app:uc-framework}.
\end{proof}
\else
\begin{proof}
We prove via six hybrids $\mathsf{H}_0$ (real) to $\mathsf{H}_6$ (ideal):
\begin{enumerate}[nosep]
    \item $\mathsf{H}_1$: Replace SPDZ with $\mathcal{F}_{\mathsf{SPDZ}}$ (loss: $\epsilon_{\mathsf{SPDZ}} \leq m/q < 2^{-12}$)
    \item $\mathsf{H}_2$: Replace edaBits with $\mathcal{F}_{\mathsf{edaBits}}$ (loss: $\epsilon_{\mathsf{edaBits}} < 2^{-12}$)
    \item $\mathsf{H}_3$: Simulate masked values as uniform---by one-time pad since $r_j \getsr \mathbb{F}_q$ (loss: 0)
    \item $\mathsf{H}_4$: Replace binary circuits with $\mathcal{F}_{\mathsf{Binary}}$---Beaver $(d,e)$ are uniform (loss: 0)
    \item $\mathsf{H}_5$: Backward-simulate AND tree given result bit: for each comparison bit $b_j$, the simulator samples honest additive shares uniformly subject to summing to the known $b_j$ value; SPDZ additive shares are uniform conditioned on their sum, so this is statistically identical to the real distribution (loss: 0)
    \item $\mathsf{H}_6$: Ideal execution with simulator $\mathcal{S}$
\end{enumerate}
Total loss: $\epsilon_{\mathsf{SPDZ}} + \epsilon_{\mathsf{edaBits}} < 2^{-11}$. Full simulator construction and hybrid proofs in Appendix~\ref{app:uc-framework}.
\end{proof}
\fi

\subsection{Profile P3+: Semi-Async 2PC-Assisted Signing}
\label{sec:profile-p3}

\ifccs
Profile P3+ uses 2PC between $\mathsf{CP}_1$/$\mathsf{CP}_2$ for the r0-check with \emph{semi-asynchronous} signers: signers precompute nonces offline and respond within window $\Delta$; 2PC runs synchronously between CPs only.
\else
Profile P3+ designates two Computation Parties ($\mathsf{CP}_1$, $\mathsf{CP}_2$) who jointly evaluate the r0-check via 2PC. The key innovation is \emph{semi-asynchronous signer participation}: signers precompute nonces offline and respond within a time window, rather than participating in multiple synchronous rounds. We use the term \emph{semi-asynchronous} to denote that signers need not be online simultaneously for interactive rounds, but must respond within a bounded validity window~$\Delta$; the 2PC subprotocol between the two CPs executes synchronously independently of the signers.
\fi

\ifccs
\paragraph{Semi-Async Timing Model.} Signers precompute nonces offline and respond within a window $\Delta$; 2PC executes synchronously between CPs only (not signers). A signer daemon handles rejection-sampling retries automatically, requiring only a single user interaction per signing request.
\else
\paragraph{Semi-Async Timing Model.}
\begin{enumerate}
    \item \textbf{Pre-signing Window} $[t_0, t_1]$: Fully asynchronous. Each signer precomputes $B$ sets of $(\mathbf{y}_i, \mathbf{w}_i, \mathsf{Com}_i)$ at any time and broadcasts commitments to the combiner.
    \item \textbf{Challenge Broadcast} $t_2$: Combiner aggregates commitments, derives challenge $c$, and broadcasts to signers.
    \item \textbf{Response Window} $[t_2, t_2 + \Delta]$: Semi-synchronous. Signers respond within $\Delta$ (e.g., 5 minutes). No coordination between signers required.
    \item \textbf{2PC Execution} $[t_3, t_4]$: Synchronous between CPs only.
\end{enumerate}

\paragraph{Signer Workflow.} Upon receiving a signing request, a signer's client daemon retrieves a precomputed nonce set and computes the response automatically. Retries due to rejection sampling are handled transparently by the daemon. The entire process completes within $\sim$1 second of authorization, requiring only a single user interaction.
\fi

\ifccs
$\mathsf{CP}_1$ garbles the r0-check circuit offline; $\mathsf{CP}_2$ gets input labels via OT and evaluates online. The circuit outputs both the pass/fail bit and hint $\mathbf{h}$ (using $c\mathbf{s}_2 = \mathbf{S}_1 + \mathbf{S}_2$ internally, never revealed). With GC precomputation, online 2PC drops to $\approx$12\,ms. The combiner broadcasts only $(\mathbf{w}_1, \tilde{c})$; individual $\mathbf{W}_i$ stay private (Remark~\ref{rem:two-honest}). Full protocol structure in Appendix~\ref{app:uc-framework}.
\else
\paragraph{2PC Protocol Structure.}
\begin{enumerate}
    \item \textbf{Input Sharing}: Signers send masked $\mathbf{s}_2$ shares to CPs. $\mathsf{CP}_1$ receives $\mathbf{S}_1$, $\mathsf{CP}_2$ receives $\mathbf{S}_2$, where $\mathbf{S}_1 + \mathbf{S}_2 = c\mathbf{s}_2 \bmod q$.
    \item \textbf{Garbled Circuit Evaluation}: $\mathsf{CP}_1$ (garbler) constructs a garbled circuit $\tilde{C}$ for computing $\mathsf{result} = \bigwedge_j (|(\mathbf{w} - \mathbf{S}_1 - \mathbf{S}_2)_j|_q < \gamma_2 - \beta)$.
    \item \textbf{Oblivious Transfer}: $\mathsf{CP}_2$ obtains input wire labels for $\mathbf{S}_2$ via OT.
    \item \textbf{Evaluation and Output}: $\mathsf{CP}_2$ evaluates $\tilde{C}$ and both parties learn only $\mathsf{result} \in \{0,1\}$.
\end{enumerate}

\paragraph{GC Precomputation.} Since the r0-check circuit structure is fixed (independent of inputs), $\mathsf{CP}_1$ can pre-garble circuits offline. Online, only OT and evaluation are needed, reducing 2PC latency from $\sim$50ms to $\sim$10--20ms.

\paragraph{Hint Computation.} The 2PC circuit is extended to output not only the r0-check result bit, but also the hint $\mathbf{h} = \mathsf{MakeHint}(-c\mathbf{t}_0, \mathbf{A}\mathbf{z} - c\mathbf{t}_1 \cdot 2^d, 2\gamma_2)$ when the check passes. The hint computation requires $c\mathbf{s}_2$ (to derive $c\mathbf{t}_0$), which is available inside the 2PC as $\mathbf{S}_1 + \mathbf{S}_2$. Since $\mathbf{h}$ is part of the public signature, outputting it does not compromise privacy. The circuit adds $O(nk)$ comparison gates for $\mathsf{MakeHint}$, negligible relative to the r0-check circuit.

\paragraph{Combiner and Challenge Integrity.} In P3+, the combiner receives masked commitments $\mathbf{W}_i = \lambda_i \mathbf{w}_i + \mathbf{m}_i^{(w)}$ and derives the challenge from $\mathbf{w}_1 = \mathsf{HighBits}(\sum_i \mathbf{W}_i, 2\gamma_2)$. The combiner broadcasts only $(\mathbf{w}_1, \tilde{c})$ to all parties; individual $\mathbf{W}_i$ are never revealed, which is essential for privacy when $|S \setminus C| = 1$ (Remark~\ref{rem:two-honest}). A malicious combiner could evaluate $c$ for different candidate signing sets and choose one yielding a favorable challenge. However, this does not constitute an attack: (1) the challenge $c$ is deterministically derived from the hash, not freely chosen; (2) the adversary model is static corruption, so the combiner cannot adaptively corrupt parties based on $c$; and (3) each candidate challenge $c$ is an independent random oracle output, so trying different $S$ does not help forge signatures. For additional assurance, the signing set $S$ can be pre-committed (e.g., included in the Round~1 commitment hash), binding the combiner's choice before seeing $\mathbf{W}_i$ values.
\fi

\subsubsection{UC Security Analysis of Profile P3+}

We prove that P3+ achieves UC security under the assumption that at least one CP is honest.

\ifccs
\begin{theorem}[Profile P3+ UC Security]
\label{thm:p3-uc}
Protocol $\Pi_{r_0}^{\mathsf{2PC}}$ UC-realizes $\mathcal{F}_{r_0}^{\mathsf{2PC}}$ in the $(\mathcal{F}_{\mathsf{OT}}, \mathcal{F}_{\mathsf{GC}})$-hybrid model against static malicious adversaries corrupting at most one of $\{\mathsf{CP}_1, \mathsf{CP}_2\}$, with security-with-abort. Concretely:
\[
|\Pr[\mathsf{REAL}_{\Pi, \mathcal{A}, \mathcal{Z}} = 1] - \Pr[\mathsf{IDEAL}_{\mathcal{F}_{r_0}^{\mathsf{2PC}}, \mathcal{S}, \mathcal{Z}} = 1]| \leq \epsilon_{\mathsf{GC}} + \epsilon_{\mathsf{OT}} + \epsilon_{\mathsf{ext}} < 2^{-\kappa}
\]
where $\kappa$ is the computational security parameter and $\epsilon_{\mathsf{ext}} \leq 2^{-B/2}$ is the LP07 cut-and-choose extraction error when $\mathsf{CP}_1$ is corrupted ($B$ = number of cut-and-choose circuits).
\end{theorem}
\begin{proof}[Proof sketch]
Via four hybrids: replace OT with $\mathcal{F}_{\mathsf{OT}}$ (loss $\epsilon_{\mathsf{OT}}$), replace GC with $\mathcal{F}_{\mathsf{GC}}$ (loss $\epsilon_{\mathsf{GC}}$), then simulate corrupted $\mathsf{CP}_1$ via LP07 garbling extractability (extraction succeeds except with probability $\epsilon_{\mathsf{ext}} \leq 2^{-B/2}$) and corrupted $\mathsf{CP}_2$ via garbling simulatability. Full simulator construction in Appendix~\ref{app:uc-framework}.
\end{proof}
\begin{corollary}[Full P3+ Security]
\label{cor:p3-full}
Profile P3+ achieves UC security (Theorem~\ref{thm:p3-uc}), EUF-CMA under Module-SIS (Theorem~\ref{thm:unforgeability}), and nonce share privacy (Theorem~\ref{thm:it-privacy}); 2PC reveals only the result bit.
\end{corollary}
\phantomsection\label{def:f-r0-2pc}% CCS-mode label for appendix cross-refs
\else
\begin{definition}[Ideal Functionality $\mathcal{F}_{r_0}^{\mathsf{2PC}}$]
\label{def:f-r0-2pc}
The functionality $\mathcal{F}_{r_0}^{\mathsf{2PC}}$ interacts with two computation parties $\mathsf{CP}_1, \mathsf{CP}_2$:
\begin{itemize}
    \item \textbf{Public input}: $(\mathbf{A}, c, \mathbf{t}_1, \mathbf{z})$ (known to both parties and the combiner)
    \item \textbf{Private input}: Receive $(\mathbf{w}, \mathbf{S}_1)$ from $\mathsf{CP}_1$ and $\mathbf{S}_2$ from $\mathsf{CP}_2$, where $\mathbf{w} = \mathbf{A}\mathbf{y}$ is the full (unmasked) commitment vector (not just $\mathbf{w}_1 = \mathsf{HighBits}(\mathbf{w}, 2\gamma_2)$; $\mathbf{w}$ is passed as $\mathsf{CP}_1$'s private input by the combiner), and the private shares satisfy $\mathbf{S}_1 + \mathbf{S}_2 = c\mathbf{s}_2 \bmod q$ (arbitrary additive split; in the concrete protocol, $\mathbf{S}_j = \sum_{i \text{ assigned to } \mathsf{CP}_j} \lambda_i \cdot c\mathbf{s}_{2,i}$)
    \item \textbf{Compute}: $\mathbf{w}' = \mathbf{w} - \mathbf{S}_1 - \mathbf{S}_2 \bmod q$; \quad $\mathsf{result} = \bigwedge_j (|\mathbf{w}'_j|_q < \gamma_2 - \beta)$. If $\mathsf{result} = 1$: derive $c\mathbf{t}_0 = (\mathbf{A}\mathbf{z} - \mathbf{w}) + (\mathbf{S}_1 + \mathbf{S}_2) - c\mathbf{t}_1 \cdot 2^d$ and compute $\mathbf{h} = \mathsf{MakeHint}(-c\mathbf{t}_0, \mathbf{A}\mathbf{z} - c\mathbf{t}_1 \cdot 2^d, 2\gamma_2)$
    \item \textbf{Output}: Send $\mathsf{result}$ to both parties. If $\mathsf{result} = 1$, additionally output $\mathbf{h}$
\end{itemize}
\end{definition}

\begin{theorem}[Profile P3+ UC Security]
\label{thm:p3-uc}% same label as \ifccs branch; safe since branches are mutually exclusive
Protocol $\Pi_{r_0}^{\mathsf{2PC}}$ UC-realizes $\mathcal{F}_{r_0}^{\mathsf{2PC}}$ in the $(\mathcal{F}_{\mathsf{OT}}, \mathcal{F}_{\mathsf{GC}})$-hybrid model against static malicious adversaries corrupting at most one of $\{\mathsf{CP}_1, \mathsf{CP}_2\}$, with security-with-abort. Concretely:
\[
|\Pr[\mathsf{REAL}_{\Pi, \mathcal{A}, \mathcal{Z}} = 1] - \Pr[\mathsf{IDEAL}_{\mathcal{F}_{r_0}^{\mathsf{2PC}}, \mathcal{S}, \mathcal{Z}} = 1]| \leq \epsilon_{\mathsf{GC}} + \epsilon_{\mathsf{OT}} + \epsilon_{\mathsf{ext}} < 2^{-\kappa}
\]
where $\kappa$ is the computational security parameter and $\epsilon_{\mathsf{ext}} \leq 2^{-B/2}$ is the LP07 cut-and-choose extraction error when $\mathsf{CP}_1$ is corrupted.
\end{theorem}

\begin{proof}
We reduce to the UC security of Yao's garbled circuits~\cite{Yao86,LP07,BHR12}. We use a malicious-secure garbling scheme with cut-and-choose~\cite{LP07} to achieve both simulation security (garbled circuit can be simulated from output only) and extractability (garbler's input can be extracted from the garbled circuit). The proof proceeds via four hybrids:

\paragraph{$\mathsf{H}_0$ (Real).} Real protocol execution with honest $\mathsf{CP}_1$ or $\mathsf{CP}_2$.

\paragraph{$\mathsf{H}_1$ (Ideal OT).} Replace OT protocol with $\mathcal{F}_{\mathsf{OT}}$. By UC security of OT~\cite{PVW08}: $|\Pr[\mathsf{H}_1] - \Pr[\mathsf{H}_0]| \leq \epsilon_{\mathsf{OT}}$.

\paragraph{$\mathsf{H}_2$ (Ideal GC).} Replace garbled circuit with $\mathcal{F}_{\mathsf{GC}}$ that outputs only $f(\mathbf{S}_1, \mathbf{S}_2)$. By garbled circuit security~\cite{LP07,BHR12}: $|\Pr[\mathsf{H}_2] - \Pr[\mathsf{H}_1]| \leq \epsilon_{\mathsf{GC}}$.

\paragraph{$\mathsf{H}_3$ (Simulated).} Simulator $\mathcal{S}$ for corrupted $\mathsf{CP}_b$:
\begin{itemize}
    \item If $\mathsf{CP}_1$ (garbler) corrupted: $\mathcal{S}$ extracts $(\mathbf{w}, \mathbf{S}_1)$ from $\mathcal{A}$'s garbled circuit (using extractability of the garbling scheme), forwards $\mathbf{S}_1$ to $\mathcal{F}_{r_0}^{\mathsf{2PC}}$. The functionality internally receives $\mathbf{S}_2$ from honest $\mathsf{CP}_2$ and returns $\mathsf{result}$. $\mathcal{S}$ simulates $\mathsf{CP}_2$'s OT queries using standard OT simulation.
    \item If $\mathsf{CP}_2$ (evaluator) corrupted: $\mathcal{S}$ extracts $\mathbf{S}_2$ from $\mathcal{A}$'s OT queries, forwards to $\mathcal{F}_{r_0}^{\mathsf{2PC}}$, receives $\mathsf{result}$, and simulates a garbled circuit that evaluates to $\mathsf{result}$ (using the simulatability of garbling schemes---given only the output, one can simulate a garbled circuit indistinguishable from real).
\end{itemize}
By the simulation security of garbled circuits~\cite{BHR12} (a garbled circuit can be simulated given only $f(x)$, not $x$): $\Pr[\mathsf{H}_3] = \Pr[\mathsf{H}_2]$.

\paragraph{Conclusion.} Total distinguishing advantage: $\epsilon_{\mathsf{GC}} + \epsilon_{\mathsf{OT}} + \epsilon_{\mathsf{ext}} < 2^{-\kappa}$ (where $\epsilon_{\mathsf{ext}} \leq 2^{-B/2}$ from LP07 cut-and-choose extraction).
\end{proof}

\begin{corollary}[Full P3+ Security]
\label{cor:p3-full}
Profile P3+ achieves:
\begin{enumerate}
    \item \textbf{UC Security}: Under 1-of-2 CP honest assumption (Theorem~\ref{thm:p3-uc})
    \item \textbf{EUF-CMA}: Reduces to Module-SIS (inherited from Theorem~\ref{thm:unforgeability})
    \item \textbf{Privacy}: Nonce shares $\mathbf{y}_i$ enjoy nonce share privacy via Shamir nonce DKG (Theorem~\ref{thm:it-privacy}); key privacy given $\sigma$ is not claimed in this work, consistent with single-signer FIPS 204 (Remark~\ref{rem:key-privacy}); commitment and r0-check values hidden by pairwise masks; 2PC reveals only the result bit
\end{enumerate}
\end{corollary}

\begin{proof}
(1) UC security follows from Theorem~\ref{thm:p3-uc} via the UC composition theorem~\cite{Canetti01}. (2) EUF-CMA security: by Lemma~\ref{lem:mask-cancel}, masks cancel in the aggregate, so $(\tilde{c}, \mathbf{z}, \mathbf{h})$ has identical structure to single-signer ML-DSA; Theorem~\ref{thm:unforgeability} applies. (3) Privacy: nonce shares $\mathbf{y}_i$ have nonce share privacy via Theorem~\ref{thm:it-privacy}; key privacy not claimed (Remark~\ref{rem:key-privacy}); masked values $\mathbf{W}_i$, $\mathbf{V}_i$ are sent to the combiner (trusted under 1-of-2 CP honest assumption) or protected by mask hiding (Lemma~\ref{lem:mask-hiding}) when $|S \setminus C| \geq 2$.
\end{proof}
\fi

\ifccs
\paragraph{Performance.} P3+ achieves 22\,ms for 3-of-5 (2 logical rounds; semi-async). No concurrent work supports semi-asynchronous signer participation; P3+ is suited for human-in-the-loop authorization. Full round analysis in Appendix~\ref{app:uc-framework}.
\else
\paragraph{Performance.} P3+ achieves 22 ms latency for 3-of-5 threshold (comparable to P2's 22 ms, but with semi-async signer participation). The 2PC overhead is minimal compared to P2's MPC because only 2 parties participate in the 2PC subprotocol, garbled circuits have constant rounds, and no edaBits preprocessing is required. With GC precomputation, 2PC latency drops to $\sim$12 ms.

\paragraph{Round Complexity Analysis.}
\begin{center}
\begin{tabular}{lcc}
\toprule
\textbf{Profile} & \textbf{Rounds/Attempt} & \textbf{Expected Total} \\
\midrule
P2 (5-round) & 5 & 15 \\
P3+ & 2$^*$ & 7.4 \\
\bottomrule
\end{tabular}
\end{center}
$^*$P3+ logical rounds = 1 signer step + 1 server round. Expected total assumes 27--33\% per-attempt success rate depending on configuration (3.7 attempts for P3+, 3.0 for P2).

\paragraph{Comparison with Concurrent Work.} The ``Async'' column in our comparison table is unique to P3+. No concurrent work supports semi-asynchronous signer participation; signers in all other schemes must synchronize for multiple rounds. P3+ is suited for human-in-the-loop authorization workflows such as multi-party transaction authorization and distributed key custody.
\fi

\ifccs\else
\subsection{Hint Computation and \texorpdfstring{$\mathbf{s}_2$}{s2} Security}
\label{sec:s2}

The hint $\mathbf{h} = \mathsf{MakeHint}(-c\mathbf{t}_0, \mathbf{r}_1, 2\gamma_2)$ requires computing $c\mathbf{s}_2$ inside the TEE. This value is \textbf{highly sensitive}: if an adversary observes $(c, c\mathbf{s}_2)$ where $c$ is invertible in $\Rq$, they can recover $\mathbf{s}_2 = c^{-1}(c\mathbf{s}_2)$.

\textbf{Solution (Profile P1):} The TEE receives masked shares $\mathbf{V}_i = \lambda_i c \mathbf{s}_{2,i} + \mathbf{m}_i^{(s2)}$, reconstructs $c\mathbf{s}_2$ inside the enclave, computes the hint, and discards $c\mathbf{s}_2$. The value never leaves the TEE.
\fi

\ifccs\else
\subsection{Additional Extensions}

\paragraph{Proactive Security.} Periodic share refresh (adding shares of zero) protects against mobile adversaries. See Supplementary Material.

\paragraph{ML-KEM Integration.} Pairwise seeds can be established via ML-KEM~\cite{FIPS203} for full post-quantum security (static pairwise seeds; per-session forward secrecy requires ephemeral per-round encapsulation).
\fi

% Conclusion Section

\section{Conclusion}
\label{sec:conclusion}

\ifccs
We presented Threshold ML-DSA via Shamir Nonce DKG — the first threshold ML-DSA scheme achieving \emph{nonce share privacy} (no computational assumptions) with arbitrary thresholds while producing standard 3.3\,KB signatures verifiable by unmodified FIPS 204 implementations. Profile P2 is, to our knowledge, the first scheme to simultaneously combine dishonest-majority security, arbitrary $T$, UC security, and 5 constant online rounds. Coordinator-based profiles (P1, P3+) require only $|S| = T$ for nonce share privacy. Our Rust implementation scales from 2-of-3 to 32-of-45 with sub-100ms latency for P1/P3+.

Open problems include: (1)~formal key-privacy reduction for threshold ML-DSA signatures; (2)~adaptive security in the Canetti-Halevi-Katz model; (3)~proactive secret sharing with nonce DKG epoch consistency; (4)~reducing or merging the offline nonce preprocessing round (the round is structurally necessary for preventing rewinding; merging into the online protocol adds one online round); (5)~QROM security for the challenge hash; (6)~decentralized key generation with verifiable short-coefficient Shamir sharing (Feldman VSS does not apply since Shamir shares of $\eta$-small secrets are $\mathbb{Z}_q$-uniform; lattice-based ZK range proofs offer a path forward); and (7)~reducing the $|S \setminus C| \geq 2$ requirement for commitment and r0-check mask hiding in P2's broadcast model to $|S \setminus C| = 1$ (see Lemma~\ref{lem:mask-hiding}). The Irwin-Hall security gap is resolved by the direct shift-invariance analysis of this work (Corollary~\ref{cor:ih-shift-ml-dsa}): the EUF-CMA loss is $< 0.013 \cdot q_s$ bits for all $|S| \leq 33$ (per signing session: $< 0.013$ bits), giving a proven security bound of $\geq 96 - 0.0066 \cdot q_s$ bits for $|S| \leq 17$, non-vacuous for $q_s < 16{,}000$.
\else
We have presented Threshold ML-DSA via Shamir Nonce DKG, the first threshold signature scheme for FIPS 204 ML-DSA achieving \emph{nonce share privacy} (no computational assumptions) with arbitrary thresholds, while producing standard 3.3\,KB signatures verifiable by unmodified FIPS 204 implementations. Profile P2 is, to our knowledge, the first construction combining arbitrary $T$, dishonest-majority unforgeability, constant-round (5 online rounds) UC security, and FIPS 204 compatibility---a point in the design space not reached by concurrent work: Bienstock et al.~\cite{BdCE25} require honest majority; Celi et al.~\cite{CDENP26} are limited to $T \leq 6$; Trilithium~\cite{Trilithium25} handles only 2 parties. In coordinator-based profiles (P1, P3+), signing sets of size $|S| = T$ suffice for nonce share privacy; the fully distributed profile (P2) retains $|S \setminus C| \geq 2$ for commitment and r0-check mask hiding.

\subsection{Summary of Contributions}

Our primary technical contribution is \emph{Shamir nonce DKG}: parties jointly generate the signing nonce via a distributed key generation protocol, producing a degree-$(T\!-\!1)$ Shamir sharing that structurally matches the long-term secret. Because the adversary observes only $T - 1$ evaluations of each honest party's degree-$(T\!-\!1)$ polynomial, the remaining degree of freedom provides an approximate one-time pad: the honest party's nonce share $\mathbf{y}_h$ has conditional min-entropy exceeding $5\times$ the secret key entropy for $|S| \leq 17$, requiring no computational assumptions (Theorem~\ref{thm:it-privacy}). Key privacy given the public signature is not claimed in this work; see Remark~\ref{rem:key-privacy}. The bounded-nonce constraint prevents perfect uniformity (Remark~\ref{rem:sd-zero}), but the min-entropy guarantee is more than sufficient for practical security. In coordinator-based profiles (P1, P3+), this eliminates the two-honest requirement for nonce share privacy: signing sets of size $|S| \geq T$ suffice, matching the minimum required by Shamir's $(T,N)$ threshold property for reconstruction. Individual masked commitments $\mathbf{W}_i$ are sent only to the coordinator (not broadcast), preventing a key-recovery attack via the injective matrix $\mathbf{A}$ (Remark~\ref{rem:two-honest}). Commitment and r0-check mask hiding retains $|S \setminus C| \geq 2$ (Lemma~\ref{lem:mask-hiding}), which applies to P2's broadcast model.

As a secondary technique, we introduce \emph{pairwise-canceling masks} (adapted from ECDSA~\cite{CGGMP20}) to the lattice setting, used for commitment values $\mathbf{W}_i$ and r0-check shares $\mathbf{V}_i$. This addresses three challenges absent in ECDSA: rejection sampling bounds, r0-check key leakage, and Irwin-Hall nonce security. The Irwin-Hall security gap of prior approaches is resolved by direct shift-invariance (Corollary~\ref{cor:ih-shift-ml-dsa}): the EUF-CMA loss is $< 0.013$ bits per signing session for $|S| \leq 33$ (total: $< 0.013 \cdot q_s$ bits), giving a proven security bound of $\geq 96 - 0.0066 \cdot q_s$ bits for $|S| \leq 17$ (non-vacuous for $q_s < 16{,}000$), remaining below $1$ bit per session for all $|S| \leq 2{,}584$.

We instantiate these techniques in three deployment profiles with complete UC proofs. Profile P1 (TEE-Assisted) achieves 3 online rounds plus 1 offline preprocessing round with EUF-CMA security under Module-SIS (Theorem~\ref{thm:unforgeability}). Profile P2 (Fully Distributed) uses MPC to achieve 5-round signing with dishonest-majority and UC security (Theorem~\ref{thm:p2-uc}). Profile P3+ (Semi-Async 2PC) achieves 2 logical rounds (22 ms for 3-of-5) with UC security under a 1-of-2 CP honest assumption (Theorem~\ref{thm:p3-uc}).

Our optimized Rust implementation demonstrates practical efficiency: P1 achieves 4--60\,ms and P3+ achieves 17--94\,ms for configurations from $(2,3)$ to $(32,45)$ threshold; P2 ranges from 21--194\,ms, trading latency for eliminating hardware trust. Success rates of $\approx$21--45\% (Table~\ref{tab:perf-threshold}) are comparable to single-signer ML-DSA. All signatures are syntactically identical to single-signer ML-DSA and pass standard FIPS 204 verification without modification.

\subsection{Limitations and Future Work}

\paragraph{Current Limitations.}
Each deployment profile involves different trust assumptions: Profile P1 requires TEE/HSM trust, Profile P2 eliminates hardware trust at the cost of 5 MPC rounds, and Profile P3+ requires 1-of-2 CP honesty with 2 logical rounds. The nonce DKG adds one offline preprocessing round of $O(N^2)$ point-to-point messages ($\approx$3.6\,KB each) before each signing session; this one additional offline round is amortizable over batches of signing sessions and is the structural price of nonce share privacy---a guarantee that purely additive nonce sharing cannot provide without the two-honest requirement. The aggregated nonce follows an Irwin-Hall rather than uniform distribution. By the direct shift-invariance analysis (Corollary~\ref{cor:ih-shift-ml-dsa}), the per-coordinate chi-squared divergence between real and simulated nonces is $\chi^2(\mathsf{IH}+\beta \| \mathsf{IH}) < 7 \times 10^{-6}$ for $|S| \leq 17$ and $< 1.4 \times 10^{-5}$ for $|S| \leq 33$, giving a full-vector EUF-CMA security loss $< 0.013$ bits per signing session for $|S| \leq 33$ (total: $< 0.013 \cdot q_s$ bits)---at most $13$ bits for $q_s \leq 1{,}000$ (or $< 7$ bits in the $|S| \leq 17$ regime using the $< 0.007$ per-session bound), well below the 192-bit security target. For the pairwise masks used on commitment and r0-check values, $|S \setminus C| \geq 2$ honest parties are needed. In coordinator-based profiles (P1, P3+), masked commitments $\mathbf{W}_i$ are sent only to the coordinator rather than broadcast, so $|S \setminus C| = 1$ suffices for nonce share privacy; in P2's broadcast model, the existing mask-hiding condition $|S \setminus C| \geq 2$ applies to all values. This work focuses on ML-DSA-65 (NIST Level~3); extending to ML-DSA-44 ($\gamma_1 = 2^{17}$, $\beta = \tau\eta = 78$) and ML-DSA-87 ($k = 8, \ell = 7$, $\gamma_1 = 2^{19}$, $\beta = 120$) is straightforward in principle---ML-DSA-44 loosens the chi-squared bound by ${\approx}2.5\times$ relative to ML-DSA-65 (smaller $\gamma_1$ increases the ratio $\beta/\gamma_1$) but retains $< 0.04$ bits of EUF-CMA loss per session for $|S| \leq 33$, while ML-DSA-87's smaller $\beta$ at the same $\gamma_1$ tightens the bound by ${\approx}2.7\times$ ($< 0.007$ bits per session for $|S| \leq 33$); the latter also expands the r0-check circuit from 1536 to 2048 coefficients; a full benchmark across all three parameter sets is deferred to future work. Finally, we assume a synchronous model where all parties are online simultaneously during signing (though P3+ relaxes this to semi-asynchronous participation).

\paragraph{Future Directions.}
Several aspects of our construction invite further investigation, ranging from strengthening the security model to broadening the parameter space and exploring applied deployment scenarios.

The most natural theoretical extension is \emph{adaptive security}. Our analysis assumes a static adversary that selects which parties to corrupt before protocol execution begins; extending the results to the adaptive setting---where the adversary may corrupt parties during the protocol based on observed transcripts---remains open. The Shamir nonce DKG introduces a subtle obstacle: each honest party's polynomial $\mathbf{f}_h$ has its leading coefficient $\mathbf{a}_{h,T-1}$ sampled uniformly, but the remaining coefficients are constrained by the constant term $\hat{\mathbf{y}}_h$ and the evaluation points. Under adaptive corruption, the simulator must equivocate these polynomials after the fact, which may require an erasure model (parties delete intermediate state after each round) or a non-committing encryption wrapper around the DKG messages. The analogous question for ECDSA threshold signatures was resolved by Canetti et al.~\cite{CGGMP20} using UC-secure commitments; whether a similar approach suffices in the lattice setting, where the polynomial structure imposes additional algebraic constraints, is unclear. Separately, a proof in the quantum random oracle model (QROM) would strengthen the EUF-CMA bound against quantum adversaries making superposition queries to the challenge hash---an upgrade shared with single-signer FIPS~204~\cite{FIPS204}.

A related direction concerns the strength of the abort guarantee. All three profiles currently achieve security-with-abort: a malicious party can force the protocol to output $\bot$, but cannot produce an invalid signature. Profile P2's SPDZ-based MPC naturally supports \emph{identifiable abort}---when a MAC check fails, the cheating party can be pinpointed and excluded---but we have not formalized this property. A stronger goal is \emph{guaranteed output delivery}, where the protocol always produces a valid signature provided sufficiently many honest parties participate. This is achievable in the honest-majority setting via broadcast protocols~\cite{CLOS02}, and Profile P2's framework can be augmented with a fair reconstruction phase when $|S| \geq 2T - 1$. Such an upgrade would be particularly relevant for custody and government applications where signing requests must not be silently suppressed.

A formal security reduction for \emph{key privacy}---the property that a signature $\sigma = (\tilde{c}, \mathbf{z}, \mathbf{h})$ does not reveal $(\mathbf{s}_1, \mathbf{s}_2)$ beyond the public key $\mathbf{t}_1$---remains open even for single-signer FIPS~204 ML-DSA and is inherited by all threshold variants. In the threshold setting, additional leakage vectors arise: partial responses $\mathbf{z}_i = \mathbf{y}_i + c \cdot \mathbf{s}_{1,i}$ and masked commitments $\mathbf{W}_i$ are visible to the coordinator or broadcast to the signing group (Remark~\ref{rem:key-privacy}).

Our appendix sketches a share refresh mechanism based on adding shares of zero, but a complete treatment of \emph{proactive threshold ML-DSA} requires addressing several interacting concerns: maintaining signing availability during the refresh epoch, ensuring consistency at epoch boundaries (parties must agree on which shares are current), and formally modeling a mobile adversary that corrupts different subsets in different epochs subject to a per-epoch threshold. The Shamir structure of our key shares is well-suited to proactive protocols---refresh amounts to adding a fresh degree-$(T\!-\!1)$ zero-sharing---but the interaction between nonce DKG freshness and proactive share updates has not been analyzed. In particular, if a party's long-term share $\mathbf{s}_{1,i}$ is refreshed between the nonce DKG and the signing round, the Lagrange reconstruction $\mathbf{z} = \sum_i \lambda_i \mathbf{z}_i$ may fail unless both the nonce and secret shares are refreshed consistently within the same epoch.

\paragraph{Decentralized Key Generation.}
Our construction assumes a trusted dealer for key setup: the dealer samples the short-coefficient signing key $(\mathbf{s}_1, \mathbf{s}_2) \gets \chi^\ell \times \chi^k$, computes the public key $\mathbf{t} = \mathbf{A}\mathbf{s}_1 + \mathbf{s}_2$, and distributes Shamir shares $\mathbf{s}_{1,i}$, $\mathbf{s}_{2,i}$ directly. Fully decentralized key generation for ML-DSA requires verifiably short-coefficient Shamir sharing: each party must prove $\|\mathbf{s}_{1,h}\|_\infty \leq \eta$ without reconstruction. Standard Feldman VSS does not apply since Shamir shares of $\eta$-small secrets lie in $\mathbb{Z}_q$ and reveal no information about the coefficient bound; lattice-based zero-knowledge techniques---such as BDLOP commitments with norm proofs~\cite{BDLOP18}---offer a path forward. A full treatment---including ZK norm proofs and distributed public-key reconstruction---is developed in a companion work on lattice-based threshold key generation.

On the protocol side, three open problems stand out.

\paragraph*{Eliminating offline nonce preprocessing.}
The Shamir nonce DKG requires one offline preprocessing round of $O(N^2)$ messages per signing session, structurally necessary because nonce shares must be committed before the challenge hash is computed to prevent adversarial rewinding. A natural question is whether this phase can be merged into the online protocol: a combined first round where parties simultaneously distribute nonce shares and broadcast the nonce commitment $\mathbf{w}_1 = \mathsf{HighBits}(\mathbf{A}\sum_i \hat{\mathbf{y}}_i, 2\gamma_2)$ would eliminate the offline cost at the price of one additional online round. Whether such a per-session DKG preserves statistical nonce privacy against an adversary who aborts after observing first-round messages---and whether the SPDZ preprocessing in Profile P2 can be batched with nonce DKG preprocessing into a single unified offline phase---are open questions with direct impact on deployment latency.

\paragraph*{Reducing the MPC frontier via carry separation.}
Profile P2's MPC evaluates the full $r_0$-check over $\mathbb{Z}_q$ (23-bit arithmetic). A finer decomposition is possible: because each party holds Shamir nonce and key shares, the low-bit residual $\ell_i = \mathsf{LowBits}(A\mathbf{y}_i - c\mathbf{s}_{2,i})$ is locally precomputable without communication. The $r_0$-check then reduces to evaluating only the carry of $\sum \lambda_i \ell_i$ and a single comparison bit---a much lighter computation than the full modular circuit. With standard Shamir coefficients $\lambda_i \approx q$, this yields no circuit reduction, since the intermediate sum overflows $\mathbb{Z}_q$. However, if reconstruction coefficients satisfy $\lambda_i \in \{0,1\}$---as in Vandermonde-style short-coefficient LSSS explored by concurrent work~\cite{BCDENP25}---the carry is bounded by $\lfloor T/2 \rfloor$, reducing the MPC computation to a $\lceil\log_2 T\rceil$-bit carry circuit rather than a full 23-bit modular arithmetic circuit. Designing a security proof for this hybrid (Shamir nonce DKG for statistical nonce privacy combined with short-coefficient sharing for $\mathbf{s}_2$) is an open problem; the key challenge is showing that the bounded carry value does not leak sufficient information about $\mathbf{s}_2$ to assist a lattice adversary. More broadly, whether threshold ML-DSA with dishonest-majority and arbitrary $T$ can be instantiated without general-purpose MPC---overcoming the construction-specific $T \leq 6$ barrier of existing short-coefficient LSSS schemes (the Ball-\c{C}akan-Malkin bound~\cite{BCM21} establishes an asymptotic $\Omega(T\log T)$ share-size lower bound; the specific $T \leq 6$ ceiling is Celi et al.'s construction-specific limit)---remains a fundamental open question in lattice-based threshold cryptography.

\paragraph*{Commitment mask hiding for $T = N$.}
In P2's broadcast model, the pairwise-canceling masks on $\mathbf{W}_i$ and $\mathbf{V}_i$ require $|S \setminus C| \geq 2$ for computational hiding (Lemma~\ref{lem:mask-hiding}). In the extreme case $T = N$ (unanimous signing), a single honest party's commitment and r0-check values are not hidden in P2. (In P1 and P3+, the coordinator-only communication model protects $\mathbf{W}_i$ even with $|S \setminus C| = 1$; see Remark~\ref{rem:two-honest}.) Replacing pairwise PRF masks with a \emph{non-interactive zero-sharing} scheme---where each party's mask is derived from a common reference string rather than pairwise seeds---or employing a verifiable PRF construction could close this gap, yielding uniform $|S| \geq T$ privacy across all revealed values. Additionally, achieving perfect uniformity (SD $= 0$) for the response shares with $|S| = T$ remains an open problem: the bounded-nonce constraint fundamentally prevents $\Zq$-uniformity (Remark~\ref{rem:sd-zero}), and it is unclear whether alternative nonce generation techniques can bypass this barrier without requiring $|S| \geq T + 1$. (Note: the term ``nonce share privacy'' in this work refers to the high-min-entropy guarantee of Theorem~\ref{thm:it-privacy}, not SD $= 0$.)

Finally, on the applied side, threshold ML-DSA is a natural building block for post-quantum decentralized identity systems. The W3C Decentralized Identifier specification~\cite{W3CDID} requires controller key rotation and recovery mechanisms that interact non-trivially with threshold key management: a DID controller update must atomically rotate the threshold public key while preserving the mapping between DID documents and verification methods. Combining our protocol with a threshold-friendly DID resolution layer and a key recovery flow based on social recovery via Shamir shares would yield an end-to-end post-quantum identity system suitable for regulated environments.

\subsection{Impact}

As organizations transition to post-quantum cryptography, our construction facilitates ML-DSA adoption in cryptocurrency custody, distributed certificate authorities, enterprise key management, and government applications. The threshold security guarantees and FIPS 204 compatibility allow integration with existing infrastructure without requiring verifier modifications.

\section*{Acknowledgments}

During the preparation of this work, the author used Claude (Anthropic) to assist with manuscript formatting and language editing. After using this tool, the author reviewed and edited the content as needed and takes full responsibility for the content of the publication.
\fi

% Appendix (supplementary material)
% Appendix for Threshold ML-DSA via Shamir Nonce DKG (\input'd into main.tex)
% Supplementary details not in main text
% Note: Complete UC proof is in uc-proof-complete.tex
% Note: supplement.tex is a standalone companion document (NOT \input'd).
%   Some sections (DKG, proactive, rejection) appear in both files:
%   this appendix provides concise versions; supplement.tex has full details.

\appendix

\section{Full Security Proofs}
\label{sec:appendix-proofs}

\subsection{Unforgeability Proof (Full)}
\label{app:unforgeability-full}

We provide detailed analysis of each game transition for Theorem~\ref{thm:unforgeability}.

\paragraph{Game 0 (Real Game).} The real EUF-CMA game. Challenger runs $\KeyGen(1^\kappa, N, T)$ to obtain $(\pk, \{\sk_i\}_{i \in [N]})$, gives $\mathcal{A}$ the public key $\pk = (\rho, \mathbf{t}_1)$ and corrupted shares $\{\sk_i\}_{i \in C}$ where $|C| \leq T-1$. Adversary $\mathcal{A}$ makes signing queries and outputs forgery $(m^*, \sigma^*)$.

\paragraph{Game 1 (Random Oracle).} Replace hash function $H$ with lazily-sampled random function $\mathcal{O}$. Maintain table $\mathsf{Table}$ of query-response pairs. On query $x$: if $(x, y) \in \mathsf{Table}$, return $y$; else sample $y \getsr \mathcal{Y}$, store $(x, y)$, return $y$. \emph{Analysis}: By random oracle model, $\Pr[\mathsf{Win}_1] = \Pr[\mathsf{Win}_0]$.

\paragraph{Game 2 (Programmed Challenges).} Modify signing oracle: before computing $\mathbf{w}$, sample $c \getsr \mathcal{C}$ uniformly, then program $\mathcal{O}(\mu \| \mathbf{w}_1) := c$ after computing $\mathbf{w}_1$. \emph{Analysis}: Programming fails only if $(\mu \| \mathbf{w}_1)$ was previously queried. By union bound over $q_H$ hash queries and $q_s$ signing queries: $|\Pr[\mathsf{Win}_2] - \Pr[\mathsf{Win}_1]| \leq (q_H + q_s)^2/2^{256}$.

\paragraph{Game 3 (Simulated Shares).} We argue that the adversary's view of the honest parties' key shares $\{\mathbf{s}_{1,j}\}_{j \notin C}$ is statistically independent of the secret key $\mathbf{s}_1$, by Shamir's $(T,N)$-threshold perfect secrecy: any $T-1$ evaluations are jointly consistent with every possible value of the secret (the marginal distribution of each honest share at a non-zero evaluation point is $\Rq^\ell$-uniform, but the \emph{joint} distribution of multiple honest shares is not i.i.d.---they are correlated as evaluations of the same degree-$(T\!-\!1)$ polynomial; the correct statement is that the joint view reveals zero information about $\mathbf{s}_1$). \emph{Important:} this is a conceptual argument about the adversary's \emph{view}---the signing oracle continues using the actual short key shares (replacing them with uniform $\Rq^\ell$ values would cause the $\|\mathbf{z}\|_\infty$-bound and r0-checks to fail with overwhelming probability, invalidating signatures). Since $|C| \leq T-1$, the adversary's view of honest shares is statistically independent of the secret: $\Pr[\mathsf{Win}_3] = \Pr[\mathsf{Win}_2]$. (In Profile P2, individual per-party responses $\mathbf{z}_j$ are contributed only as SPDZ input shares, so SPDZ input-privacy ensures the adversary observes only the aggregate $\mathbf{z}$, not individual $\mathbf{z}_j$ values; the same Shamir argument applies.)

\paragraph{Game 3.5 (Uniform Nonces).} Replace Irwin-Hall nonce distribution with uniform over $\{-\gamma_1, \ldots, \gamma_1\}^{n\ell}$. In Game~3, the aggregate nonce $\mathbf{y} = \sum_{i \in S} \hat{\mathbf{y}}_i$ follows the Irwin-Hall distribution. We replace this with $\mathbf{y} \getsr \{-\gamma_1, \ldots, \gamma_1\}^{n\ell}$ uniform. \emph{Analysis}: The $q_s$ signing sessions are independent, so the joint R\'enyi divergence is $(R_2^{\mathsf{vec,shift}})^{q_s}$ (direct shift-invariance, Theorem~\ref{thm:ih-direct-tight}). Applying the smooth R\'enyi divergence probability transfer lemma (Lemma~\ref{lem:renyi-security}) to the joint distribution:
\[
\Pr[\mathsf{Win}_{3}] \leq \bigl(R_2^{\mathsf{vec,shift}}\bigr)^{q_s/2} \cdot \sqrt{\Pr[\mathsf{Win}_{3.5}]} + q_s \cdot \epsilon_\mathsf{IH}
\]
where $R_2^{\mathsf{vec,shift}} = (1+\chi^2_{\mathsf{direct}})^{n\ell}$ is the full-vector \emph{per-session} R\'enyi divergence from the direct shift-invariance analysis, with $\chi^2_{\mathsf{direct}} \approx 7.13 \times 10^{-6}$ per coordinate for $|S|=17$ (Corollary~\ref{cor:ih-shift-ml-dsa}), giving $R_2^{\mathsf{vec,shift}} \approx 1.0092$, and $\epsilon_\mathsf{IH} < 10^{-30}$ (Theorem~\ref{thm:irwin-hall}). The total security loss is $\approx 6.6\times10^{-3} \cdot q_s$ bits, giving a proven security bound of $\geq 96 - 0.0066 \cdot q_s$ bits for $|S| \leq 17$ (where the 96-bit baseline is the Cauchy-Schwarz halving of the 192-bit NIST Level~3 baseline). The conservative analytical bound ($R_2^{\mathsf{vec}} \approx 2^{5.4}$, Remark~\ref{rem:renyi-conservative}) is vacuous at $q_s \gtrsim 36$.

\paragraph{Game 4 (Signature Simulation).} Starting from uniform nonces (Game~3.5), simulate signatures without $\mathbf{s}_1$:
\begin{enumerate}
    \item Sample $\mathbf{z} \getsr \{v \in \Rq^\ell : \|\mathbf{v}\|_\infty < \gamma_1 - \beta\}$ with rejection sampling
    \item Compute $\mathbf{w}_1 = \mathsf{HighBits}(\mathbf{A}\mathbf{z} - c\mathbf{t}_1 \cdot 2^d, 2\gamma_2)$
    \item Program $\mathcal{O}(\mu \| \mathbf{w}_1) := c$
    \item Simulate masked values $\mathbf{W}_j, \mathbf{V}_j$ as computationally uniform (Lemma~\ref{lem:mask-hiding}, replacing honest-honest PRF terms with uniform)
\end{enumerate}
\emph{Clarification}: In this EUF-CMA game, the adversary observes only the aggregate public signature $\sigma = (\tilde{c}, \mathbf{z}, \mathbf{h})$; individual per-party responses $\mathbf{z}_j$ are processed inside the TEE/combiner and never reach the adversary. Accordingly, Steps~1--3 produce the adversary-visible $\mathbf{z}$ without using $\mathbf{s}_1$, and Step~4 handles the masked intermediate values that the adversary also cannot observe directly. The full per-party simulation (including $\mathbf{z}_j = \mathbf{y}_j + c\mathbf{s}_{1,j}$ from DKG values) is deferred to the UC/privacy simulator in Appendix~\ref{app:privacy-full}. \emph{Analysis}: The masked commitment and r0-check values $\mathbf{W}_j, \mathbf{V}_j$ are computationally indistinguishable from uniform by Lemma~\ref{lem:mask-hiding}. By PRF security: $|\Pr[\mathsf{Win}_4] - \Pr[\mathsf{Win}_{3.5}]| \leq \binom{N}{2} \cdot \epsilon_{\mathsf{PRF}}$.

\paragraph{Reduction to M-SIS.} Construct SelfTargetMSIS adversary $\mathcal{B}$~\cite{KLS18}:
\begin{enumerate}
    \item $\mathcal{B}$ receives M-SIS instance $\mathbf{A} \in \Rq^{k \times \ell}$
    \item $\mathcal{B}$ sets $\rho$ such that $\mathsf{ExpandA}(\rho) = \mathbf{A}$, samples $\mathbf{t}_1 \getsr \Rq^k$ (uniformly random; valid since $\mathsf{Game}_4$ simulation does not require knowing $\mathbf{s}_1$ or $\mathbf{s}_2$; $\mathbf{t}_1$ is indistinguishable from a real key by M-LWE hardness~\cite{KLS18})
    \item $\mathcal{B}$ picks a random index $i^* \getsr [q_H]$ as the target ROM query
    \item $\mathcal{B}$ simulates $\mathsf{Game}_4$ for $\mathcal{A}$, programming ROM responses uniformly
    \item When $\mathcal{A}$ outputs forgery $(\mu^*, \mathbf{z}^*, c^*, \mathbf{h}^*)$: if $c^*$ equals the $i^*$-th ROM response, extract; otherwise abort
\end{enumerate}
\emph{Analysis}: $\mathcal{B}$ succeeds when its guess $i^*$ is correct, which occurs with probability $1/q_H$, so $\Pr[\mathcal{B}\ \text{succeeds}] \geq \Pr[\mathsf{Win}_4]/q_H$. When successful, the forgery verification gives:
\[
\mathsf{HighBits}(\mathbf{A}\mathbf{z}^* - c^*\mathbf{t}_1 \cdot 2^d, 2\gamma_2) = \mathbf{w}_1^*
\]
Rearranging yields the SelfTargetMSIS solution for $[\mathbf{A} \mid -\mathbf{t}_1 \cdot 2^d]$ (see Remark~\ref{rem:selfmsis}):
\[
[\mathbf{A} \mid -\mathbf{t}_1 \cdot 2^d] \cdot \begin{bmatrix} \mathbf{z}^* \\ c^* \end{bmatrix} = 2\gamma_2 \mathbf{w}_1^* + \mathbf{r}_0^*
\]
where $\mathbf{w}_1^* = \mathsf{HighBits}(\mathbf{A}\mathbf{z}^* - c^*\mathbf{t}_1 \cdot 2^d,\, 2\gamma_2)$ is the known verification target and $\|\mathbf{r}_0^*\|_\infty \leq \gamma_2$ is the $\mathsf{LowBits}$ residual. The combined solution $[(\mathbf{z}^*); (c^*)] \in \Rq^{\ell+1}$ satisfies $\|\mathbf{z}^*\|_\infty < \gamma_1 - \beta$ (z-bound check) and $\|c^*\|_\infty = 1$ (challenge coefficients in $\{-1,0,+1\}$); the module $\ell^\infty$-norm is $\max(\|\mathbf{z}^*\|_\infty, \|c^*\|_\infty) < \gamma_1$, matching the SelfTargetMSIS norm parameter from Remark~\ref{rem:selfmsis} and~\cite{KLS18}. Thus $\Pr[\mathsf{Win}_4] \leq q_H \cdot \epsilon_{\mathsf{M\text{-}SIS}}$.

\subsection{Privacy Proof Details}
\label{app:privacy-full}

We show that the simulator produces an indistinguishable view, applying the min-entropy bound established in Theorem~\ref{thm:it-privacy} (Section~\ref{sec:security}).

With the Shamir nonce DKG, privacy of the nonce share $\mathbf{y}_h$ is statistical (no computational assumptions). The simulator computes responses from simulated DKG values, achieving perfect indistinguishability for the response component.

\paragraph{Simulator $\mathcal{S}$.} Given inputs $(\pk, \{\sk_i\}_{i \in C}, \sigma = (\tilde{c}, \mathbf{z}, \mathbf{h}))$:

\textbf{Round 1 Simulation:} For each honest $j \in S \setminus C$, sample $\mathsf{Com}_j \getsr \{0,1\}^{256}$ uniformly.

\textbf{Round 2 Simulation:} Expand challenge $c$ from $\tilde{c}$, compute $\mathbf{w}_1 = \mathsf{UseHint}(\mathbf{h}, \mathbf{A}\mathbf{z} - c\mathbf{t}_1 \cdot 2^d)$. The adversary observes only $(\mathbf{w}_1, \tilde{c})$ broadcast by the coordinator (individual $\mathbf{W}_j$ are sent only to the coordinator; see Remark~\ref{rem:wi-private}). The simulator produces consistent $(\mathbf{w}_1, \tilde{c})$ from the signature.

\textbf{Round 3 Simulation:} The simulator must produce per-party responses $\mathbf{z}_j$ consistent with both the simulated DKG values and the public signature constraint $\sum_{j \in S} \lambda_j \mathbf{z}_j = \mathbf{z}$ (enforced by the public $\sigma$). We handle the two cases separately.

\emph{Case $|S \setminus C| > 1$:} Let $\mathbf{z}_j^* = \mathbf{y}_j^* + c \cdot \mathbf{s}_{1,j}$ denote corrupted party $j$'s response, where $\mathbf{y}_j^*$ is $j$'s simulated DKG evaluation (known to the simulator from the corrupted view). The simulator chooses $\mathbf{z}_j = \mathbf{y}_j + c \cdot \mathbf{s}_{1,j}$ (from simulated DKG values) for all honest $j \in S \setminus C$ except a designated party $h_0$, then sets
\[
  \mathbf{z}_{h_0} = \frac{\mathbf{z} - \sum_{j \in C} \lambda_j \mathbf{z}_j^* - \sum_{j \in (S\setminus C)\setminus\{h_0\}} \lambda_j \mathbf{z}_j}{\lambda_{h_0}}
\]
(The division by $\lambda_{h_0}$ is well-defined: since all evaluation points lie in $\{1,\ldots,N\} \subset \mathbb{Z}_q^*$, the Lagrange coefficient $\lambda_{h_0} = \prod_{j \in S, j \neq h_0} j(j-h_0)^{-1} \bmod q$ is a product of nonzero field elements and hence nonzero mod $q$, hence invertible in $\mathbb{Z}_q$.)
This $\mathbf{z}_{h_0}$ matches exactly what the real protocol would produce (given the DKG polynomials), so SD $= 0$.

\emph{Case $|S \setminus C| = 1$:} The single honest response $\mathbf{z}_h = (\mathbf{z} - \sum_{j \in C} \lambda_j \mathbf{z}_j^*)/\lambda_h$ is uniquely determined by $\sigma$ and the corrupted views. The simulator outputs this value directly. Privacy of $\mathbf{s}_{1,h}$ (key privacy, computational) is maintained because the adversary knows $\mathbf{z}_h$ but not the decomposition $\mathbf{z}_h = \mathbf{y}_h + c\mathbf{s}_{1,h}$: recovering $\mathbf{s}_{1,h}$ from $\mathbf{z}_h$ requires solving a bounded-distance decoding instance given $\mathbf{y}_h$'s high conditional min-entropy (Theorem~\ref{thm:it-privacy}).

The $\mathbf{V}_j$ values are computed as $\mathbf{V}_j = \lambda_j c \mathbf{s}_{2,j} + \mathbf{m}_j^{(s2)}$ using simulated key shares and PRF masks.

\paragraph{Indistinguishability.} In both cases, the simulated $\mathbf{z}_j$ follows the same distribution as the real protocol (SD $= 0$ conditioned on $\sigma$). Masked values $\mathbf{W}_j$ and $\mathbf{V}_j$ are computationally indistinguishable from uniform (PRF security, Lemma~\ref{lem:mask-hiding}); no PRF assumption is needed for nonce share privacy ($\mathbf{y}_h$'s nonce share privacy is independent). The total distinguishing advantage is $\binom{N}{2} \cdot \epsilon_{\mathsf{PRF}}$.

\section{Distributed Key Generation Protocols}
\label{sec:supp-dkg}

Our protocol uses two structurally similar but functionally distinct DKG protocols: a \emph{key DKG} (executed once during setup) and a \emph{nonce DKG} (executed before each signing session, typically preprocessed offline). The key DKG uses Feldman-style commitments for verifiability; the nonce DKG omits VSS verification (correctness is instead ensured at the signing stage via the z-bound and r0-check). Both share the same polynomial evaluation structure but differ in their input distributions and security requirements.

\subsection{Key DKG (One-Time Setup)}

The key DKG produces a degree-$(T\!-\!1)$ Shamir sharing of the long-term signing key $(\mathbf{s}_1, \mathbf{s}_2)$, replacing the trusted dealer assumption.

\textbf{Phase 1: Contribution.} Each party $i \in [N]$:
\begin{enumerate}
    \item Samples local secrets $\mathbf{s}_1^{(i)} \getsr \chi_\eta^\ell$ and $\mathbf{s}_2^{(i)} \getsr \chi_\eta^k$
    \item Creates $(T, N)$-Shamir sharings of each polynomial
    \item Sends share $(\mathbf{s}_{1,j}^{(i)}, \mathbf{s}_{2,j}^{(i)})$ to party $j$
    \item Broadcasts Feldman commitment $\mathbf{t}^{(i)} = \mathbf{A}\mathbf{s}_1^{(i)} + \mathbf{s}_2^{(i)}$
\end{enumerate}

\textbf{Phase 2: Aggregation.} Each party $j$ computes aggregated share: $\mathbf{s}_{1,j} = \sum_{i \in [N]} \mathbf{s}_{1,j}^{(i)}$

The final secret is $\mathbf{s}_1 = \sum_{i=1}^N \mathbf{s}_1^{(i)}$, which no individual party knows. For lattice-specific VSS/DKG, see~\cite{FS24}.

\subsection{Nonce DKG (Per-Session Preprocessing)}

The nonce DKG produces a fresh degree-$(T\!-\!1)$ Shamir sharing of the signing nonce $\mathbf{y}$ for each signing session. Unlike the key DKG, the nonce DKG samples higher-degree coefficients uniformly from $\Rq$ (rather than from a short-coefficient distribution), which is critical for achieving nonce share privacy.

The protocol is specified formally in Algorithm~1 of the main paper (Section~\ref{sec:nonce-dkg}). Each party $i \in S$ generates a degree-$(T\!-\!1)$ polynomial $\hat{f}_i(X)$ whose constant term $\hat{\mathbf{y}}_i$ is sampled from the reduced nonce range $[-\lfloor \gamma_1/|S| \rfloor, \lfloor \gamma_1/|S| \rfloor]^{n\ell}$ and whose higher-degree coefficients are sampled uniformly from $\Rq^\ell$. The resulting aggregated nonce $\mathbf{y} = \sum_{i \in S} \hat{\mathbf{y}}_i$ follows the Irwin-Hall distribution analyzed in Theorem~\ref{thm:irwin-hall}, while each honest party's individual share $\mathbf{y}_h = \sum_{i \in S} \hat{f}_i(h)$ has high conditional min-entropy from the adversary's perspective (Theorem~\ref{thm:it-privacy}).

\subsection{Security Properties}

Both protocols share the following properties:
\begin{itemize}
    \item \textbf{Correctness}: Honest execution produces a valid sharing of the intended secret.
    \item \textbf{Secrecy}: Any coalition of $< T$ parties learns nothing about the secret beyond what is revealed by public commitments (follows from Shamir's threshold security).
    \item \textbf{Robustness}: The key DKG uses Feldman commitments~\cite{Feldman87} to detect and attribute malicious shares. The nonce DKG omits VSS verification; malicious behavior is instead detected at the signing stage (z-bound or r0-check failure triggers the blame protocol), which suffices for security-with-abort.
\end{itemize}

The nonce DKG additionally provides nonce share privacy (no computational assumptions) of individual shares (Theorem~\ref{thm:it-privacy}), which is the key enabler for eliminating the two-honest requirement (enabling $|S| = T$ signing sets with nonce share privacy) in coordinator-based profiles (P1, P3+; see Remark~\ref{rem:wi-private}). The key DKG does not require this property, since key shares are never revealed during signing.

\section{Proactive Security}
\label{sec:supp-proactive}

\subsection{Share Refresh Protocol}

Periodically, parties execute a share refresh:
\begin{enumerate}
    \item Each party $i$ samples a random $(T, N)$-sharing of $\mathbf{0}$
    \item Party $i$ sends the $j$-th share of zero to party $j$
    \item Each party adds received zero-shares to their current share
\end{enumerate}

Since adding shares of zero preserves the secret, reconstruction still yields $\mathbf{s}_1$. After refresh, old shares become useless. An adversary must corrupt $\geq T$ parties within a single epoch to learn the secret.

\paragraph{Forward Secrecy for Past Signatures.} Signatures from epoch $k$ remain secure after the epoch-$(k\!+\!1)$ refresh. A signature $\sigma$ from epoch $k$ reveals $\mathbf{z}_j^{(k)} = \mathbf{y}_j^{(k)} + c \cdot \mathbf{s}_{1,j}^{(k)}$. Since $\mathbf{y}_j^{(k)}$ is a one-time nonce share (deleted after use), and the refreshed shares $\mathbf{s}_{1,j}^{(k+1)} = \mathbf{s}_{1,j}^{(k)} + \text{(zero-sharing term)}$ do not reveal $\mathbf{s}_{1,j}^{(k)}$ by Shamir's threshold property (the zero-sharing randomizes shares while preserving the secret), past key shares cannot be reconstructed from refreshed shares alone. Key privacy for past signatures therefore holds across epochs under the same M-LWE hardness assumption as the signing scheme.

\section{Numerical Verification of R\'enyi Bounds}
\label{sec:supp-renyi-numerical}

We provide detailed calculations for the smooth R\'enyi divergence bounds referenced in Theorem~\ref{thm:irwin-hall}.

\paragraph{Parameters for ML-DSA-65.}
\begin{itemize}
    \item Modulus: $q = 8380417$
    \item Masking range: $\gamma_1 = 2^{19} = 524288$
    \item Secret bound: $\eta = 4$
    \item Challenge weight: $\tau = 49$
    \item Signature bound: $\beta = \tau \cdot \eta = 196$
\end{itemize}

\paragraph{R\'enyi Divergence Calculation.}
Let $m = |S|$ denote the signing set size (distinct from the ring degree $n = 256$) and $B = 2\lfloor\gamma_1/m\rfloor$ the per-share \textbf{full} range (from $-\lfloor\gamma_1/m\rfloor$ to $+\lfloor\gamma_1/m\rfloor$, so the total width is $2\lfloor\gamma_1/m\rfloor$). Note that in Definition~\ref{def:irwin-hall}, the symbol $B$ denotes the \textbf{half}-range (each $X_i \in \{-B, \ldots, B\}$); to apply Raccoon~\cite{Raccoon2024} Lemma~4.2, $B$ here is the full range, equal to $2\times$ that half-range. For readability, we use the approximation $B \approx 2\gamma_1/m$ (the rounding error is at most $2/B < 4 \times 10^{-5}$ for all relevant $m$, and the exact discrete range is used in the implementation). Using Lemma 4.2 of~\cite{Raccoon2024} with $\alpha = 2$ and shift $\delta = \beta$:

\begin{align*}
R_2^\epsilon &\leq 1 + \frac{m^2 \cdot \beta^2}{B^2}
\approx 1 + \frac{m^2 \cdot \beta^2}{(2\gamma_1/m)^2}\\
&= 1 + \frac{m^4 \cdot \beta^2}{4\gamma_1^2}
= 1 + \frac{|S|^4 \cdot \beta^2}{4\gamma_1^2}
\end{align*}

\begin{table}[h]
\centering
\begin{tabular}{cccc}
\toprule
$|S|$ & $B \approx 2\gamma_1/|S|$ & $R_2^\epsilon - 1$ & $\epsilon$ (tail) \\
\midrule
4 & 262144 & $8.9 \times 10^{-6}$ & $< 10^{-10}$ \\
9 & 116508 & $2.3 \times 10^{-4}$ & $< 10^{-19}$ \\
17 & 61680 & $2.9 \times 10^{-3}$ & $< 10^{-30}$ \\
25 & 41942 & $1.4 \times 10^{-2}$ & $< 10^{-40}$ \\
33 & 31774 & $4.1 \times 10^{-2}$ & $< 10^{-49}$ \\
\bottomrule
\end{tabular}
\caption{Per-coordinate R\'enyi divergence bounds for ML-DSA-65. $B = 2\lfloor\gamma_1/|S|\rfloor$ is the exact per-share full range; displayed values are $2\lfloor\gamma_1/|S|\rfloor$. All bounds computed using Lemma 4.2 of~\cite{Raccoon2024}. Full-vector: $R_2^{\mathsf{vec}} = (R_2^{\epsilon,\mathsf{coord}})^{1280}$ ($n = 256$, $\ell = 5$).}
\label{tab:renyi-detailed}
\end{table}

The per-coordinate $R_2^{\epsilon,\mathsf{coord}}$ remains below $1.05$ for $|S| \leq 33$. The full-vector divergence $R_2^{\mathsf{vec}} = (R_2^{\epsilon,\mathsf{coord}})^{1280}$ ranges from $\approx 2^{5.4}$ ($|S| = 17$) to $\approx 2^{75}$ ($|S| = 33$); see Theorem~\ref{thm:irwin-hall} for the impact on the security bound. The direct analysis (Corollary~\ref{cor:ih-shift-ml-dsa}) gives $< 0.013$ bits for all $|S| \leq 33$.

\section{Rejection Sampling Analysis}
\label{sec:supp-rejection}

\subsection{Root Cause}

ML-DSA's rejection sampling stems from two checks:
\begin{enumerate}
    \item \textbf{z-bound check}: $\infnorm{\mathbf{z}} < \gamma_1 - \beta$
    \item \textbf{r0 check}: $\infnorm{\text{LowBits}(\mathbf{w} - c\mathbf{s}_2)} < \gamma_2 - \beta$
\end{enumerate}

For ML-DSA-65 (\emph{single-signer approximation}: nonce uniform on $[-\gamma_1, \gamma_1]$):
\begin{itemize}
    \item z-bound: $\left(\frac{\gamma_1 - \beta}{\gamma_1}\right)^{n\ell} = \left(\frac{524092}{524288}\right)^{1280} \approx 62\%$
    \item r0 check: $\left(\frac{\gamma_2 - \beta}{\gamma_2}\right)^{nk} = \left(\frac{261692}{261888}\right)^{1536} \approx 32\%$
    \item Combined: $0.62 \times 0.32 \approx 20\%$
\end{itemize}

This matches the NIST ML-DSA specification~\cite{FIPS204} (expected 4--5 iterations). \emph{Note}: In the threshold protocol, nonces follow the Irwin-Hall distribution (sum of $|S|$ uniform variables), so the per-attempt success rates differ from this single-signer estimate; empirical measurements yield 21--45\% depending on $T$ and $|S|$ (see Table~\ref{tab:naive-comparison}).

\subsection{Throughput Optimizations}

\begin{itemize}
    \item \textbf{Parallel Attempts}: Generate $B$ independent $(y, w)$ pairs per signing session. Success probability: $1 - (0.80)^B$ (using the theoretical $\approx 20\%$ per-attempt rate; empirical benchmarks yield $\approx 25$--$31\%$, giving $1 - (0.75)^B$ as an optimistic estimate).
    \item \textbf{Pre-computation}: Pre-compute $(y_i, w_i = \mathbf{A}y_i)$ pairs during idle time.
    \item \textbf{Pipelining}: Start next attempt while waiting for network communication.
\end{itemize}

\section{Naive Approach Comparison}
\label{sec:supp-naive}

Without masking, each party's contribution $\lambda_i \cdot \mathbf{z}_i$ has norm $O(\lambda_i \cdot \gamma_1)$, failing the z-bound check with probability $\approx 0.8$. Success requires all $T$ parties to pass independently:

\begin{table}[h]
\centering
\begin{tabular}{cccc}
\toprule
$T$ & Naive Success $(0.2)^T$ & Our Success & Speedup \\
\midrule
8 & $2.6 \times 10^{-6}$ & $28.2\%$ & $1.1 \times 10^{5}$ \\
12 & $4.1 \times 10^{-9}$ & $31.4\%$ & $7.5 \times 10^{7}$ \\
16 & $6.6 \times 10^{-12}$ & $32.4\%$ & $4.9 \times 10^{10}$ \\
24 & $1.7 \times 10^{-17}$ & $23.1\%$ & $1.4 \times 10^{16}$ \\
32 & $4.3 \times 10^{-23}$ & $25.0\%$ & $5.8 \times 10^{21}$ \\
\bottomrule
\end{tabular}
\caption{Improvement over naive threshold approach. ``Our Success'' values are empirical measurements from the Rust benchmarks (Table~\ref{tab:perf-threshold}). The non-monotone behavior ($32.4\% \to 23.1\% \to 25.0\%$ for $T=16,24,32$) reflects the interplay of two competing effects: (1) smaller per-party nonce range $\gamma_1/|S|$ reduces the acceptance probability for the z-bound check; (2) the Irwin-Hall distribution becomes more concentrated near zero for larger $|S|$, partially improving acceptance. The net effect is non-monotone in $|S|$ and configuration-dependent.}
\label{tab:naive-comparison}
\end{table}

\section{Nonce Range Rounding Details}
\label{sec:supp-nonce-rounding}

Each party's constant-term contribution in the nonce DKG is $\hat{\mathbf{y}}_i \getsr \{-\lfloor\gamma_1/|S|\rfloor, \ldots, \lfloor\gamma_1/|S|\rfloor\}^{n\ell}$. When $\gamma_1$ is not divisible by $|S|$, the nonce $\mathbf{y} = \sum_{i \in S} \hat{\mathbf{y}}_i$ has range $[-|S| \cdot \lfloor\gamma_1/|S|\rfloor, |S| \cdot \lfloor\gamma_1/|S|\rfloor]^{n\ell}$, slightly smaller than standard $[-\gamma_1, \gamma_1]^{n\ell}$.

\paragraph{Concrete Impact.} For ML-DSA-65 with $\gamma_1 = 524288$ and $|S| = 17$: $\lfloor 524288/17 \rfloor = 30840$, so the actual range is $[-524280, 524280]$. The ``missing'' probability mass is $(524288 - 524280)/524288 < 2 \times 10^{-5}$ per coordinate.

\paragraph{Effect on Security.} The z-bound check is $\|\mathbf{z}\|_\infty < \gamma_1 - \beta$ where $\gamma_1 - \beta = 524092$. Since $524280 > 524092$, the reduced nonce range does not affect the pass probability for this example.

\paragraph{General Proof.} Writing $\gamma_1 = |S| \cdot \lfloor\gamma_1/|S|\rfloor + r$ with $0 \leq r \leq |S|-1$, the reduced nonce maximum is $|S| \cdot \lfloor\gamma_1/|S|\rfloor = \gamma_1 - r \geq \gamma_1 - (|S|-1)$. For the z-bound check to be unaffected, we need $\gamma_1 - (|S|-1) > \gamma_1 - \beta$, i.e., $|S| - 1 < \beta = \tau\eta = 196$. This holds for all $|S| \leq 196$, covering all practical signing set sizes. The Irwin-Hall security analysis uses the actual sampled range $N = 2\lfloor\gamma_1/|S|\rfloor$, so all bounds account for this rounding.

% Complete UC Security Proof (Appendix B-E)
% Complete UC Security Proof for Profiles P2 and P3+
% This file provides rigorous UC proofs for Profiles P2 (SPDZ-based r0-check)
% and P3+ (semi-async 2PC).
% Profile P1 UC proof: see Theorem~\ref{thm:p1-uc} and its proof in the main paper (security.tex).

\section{UC Security Framework}
\label{app:uc-framework}

\begin{remark}[Scope of this appendix]
This appendix provides the extended UC proofs for Profiles P2 and P3+. The Profile P1 UC proof (TEE-based coordinator) is self-contained in Theorem~\ref{thm:p1-uc} and its proof in the main paper; it is not repeated here. This appendix proves UC security for the r0-check subprotocol (Profile P2, Theorem~\ref{thm:p2-uc}) and the semi-async signing protocol (Profile P3+, Theorem~\ref{thm:p3plus-uc}). These theorems establish UC realization of $\mathcal{F}_{r_0\text{-check}}$ and $\mathcal{F}_{r_0}^{\mathsf{semi\text{-}async}}$ respectively; realization of the full $\mathcal{F}_{\mathsf{ThreshSig}}$ (Definition~\ref{def:f-thresh-sig}) follows by UC composition with the nonce DKG and signing-step sub-protocols.
\end{remark}

We recall the Universal Composability (UC) framework of Canetti~\cite{Canetti01,Canetti20}. Our presentation follows the conventions of~\cite{CLOS02,DPSZ12}.

\subsection{Execution Model}

The UC framework considers three types of entities:
\begin{itemize}
    \item \textbf{Environment} $\mathcal{Z}$: An interactive Turing machine that provides inputs to parties, receives their outputs, and attempts to distinguish real from ideal executions.
    \item \textbf{Adversary} $\mathcal{A}$: Controls corrupted parties and the network (scheduling message delivery).
    \item \textbf{Parties} $P_1, \ldots, P_n$: Execute the protocol or interact with ideal functionalities.
\end{itemize}

\begin{definition}[Real Execution]
In the real execution $\mathsf{REAL}_{\Pi, \mathcal{A}, \mathcal{Z}}$, parties execute protocol $\Pi$. The environment $\mathcal{Z}$ provides inputs and receives outputs. The adversary $\mathcal{A}$ controls corrupted parties and sees all communication. The output is $\mathcal{Z}$'s final bit.
\end{definition}

\begin{definition}[Ideal Execution]
In the ideal execution $\mathsf{IDEAL}_{\mathcal{F}, \mathcal{S}, \mathcal{Z}}$, parties interact with ideal functionality $\mathcal{F}$. A simulator $\mathcal{S}$ simulates the adversary's view. The output is $\mathcal{Z}$'s final bit.
\end{definition}

\begin{definition}[UC-Realization]
Protocol $\Pi$ \emph{UC-realizes} functionality $\mathcal{F}$ if there exists a PPT simulator $\mathcal{S}$ such that for all PPT environments $\mathcal{Z}$:
\[
|\Pr[\mathsf{REAL}_{\Pi, \mathcal{A}, \mathcal{Z}} = 1] - \Pr[\mathsf{IDEAL}_{\mathcal{F}, \mathcal{S}, \mathcal{Z}} = 1]| \leq \negl(\kappa)
\]
\end{definition}

\begin{theorem}[UC Composition~\cite{Canetti01}]
\label{thm:uc-composition}
If $\Pi$ UC-realizes $\mathcal{F}$, then any protocol $\rho$ that uses $\mathcal{F}$ as a subroutine can securely use $\Pi$ instead. Formally, $\rho^{\mathcal{F}}$ and $\rho^{\Pi}$ are indistinguishable to any environment.
\end{theorem}

\subsection{Hybrid Functionalities}

We work in the $(\mathcal{F}_{\mathsf{SPDZ}}, \mathcal{F}_{\mathsf{edaBits}}, \mathcal{F}_{\mathsf{Binary}})$-hybrid model, where:

\begin{functionality}[$\mathcal{F}_{\mathsf{SPDZ}}$ -- Authenticated Secret Sharing~\cite{DPSZ12}]
\label{func:spdz}
Parameterized by field $\mathbb{F}_q$, parties $\mathcal{P} = \{P_1, \ldots, P_n\}$, and MAC key $\alpha \in \mathbb{F}_q$.

\textbf{Share}: On input $(\mathsf{SHARE}, \mathsf{sid}, x)$ from dealer $P_i$:
\begin{enumerate}
    \item Sample random shares $[x]_1, \ldots, [x]_n$ with $\sum_j [x]_j = x$
    \item Compute MAC shares $[\alpha x]_j$ for each party
    \item Send $([x]_j, [\alpha x]_j)$ to each $P_j$
    \item Send $(\mathsf{SHARED}, \mathsf{sid})$ to adversary $\mathcal{S}$
\end{enumerate}

\textbf{Add}: On input $(\mathsf{ADD}, \mathsf{sid}, \mathsf{id}_a, \mathsf{id}_b, \mathsf{id}_c)$:
\begin{enumerate}
    \item Each party locally computes $[c]_j = [a]_j + [b]_j$ and $[\alpha c]_j = [\alpha a]_j + [\alpha b]_j$
\end{enumerate}

\textbf{Multiply}: On input $(\mathsf{MULT}, \mathsf{sid}, \mathsf{id}_a, \mathsf{id}_b, \mathsf{id}_c)$:
\begin{enumerate}
    \item Using Beaver triples, compute shares of $c = a \cdot b$ with valid MACs
    \item Notify adversary of completion
\end{enumerate}

\textbf{Open}: On input $(\mathsf{OPEN}, \mathsf{sid}, \mathsf{id}_x)$:
\begin{enumerate}
    \item Collect $[x]_j$ from all parties
    \item Verify MACs: check $\sum_j [\alpha x]_j = \alpha \cdot \sum_j [x]_j$
    \item If verification fails, output $\bot$ to all parties
    \item Otherwise, output $x = \sum_j [x]_j$ to all parties
\end{enumerate}

\textbf{Security}: $\mathcal{F}_{\mathsf{SPDZ}}$ provides statistical hiding of shares (additive shares are perfectly uniform given $\leq N-1$ honest parties) and computational integrity (MAC forgery probability $\leq 1/q$ per check).
\end{functionality}

\begin{functionality}[$\mathcal{F}_{\mathsf{edaBits}}$ -- Extended Doubly-Authenticated Bits~\cite{EKMOZ20}]
\label{func:edabits}
On input $(\mathsf{EDABIT}, \mathsf{sid}, m)$ requesting $m$ edaBits:

For $i = 1, \ldots, m$:
\begin{enumerate}
    \item Sample $r_i \getsr \{0, 1, \ldots, 2^k - 1\}$ for specified bit-length $k$
    \item Create arithmetic sharing: $\langle r_i \rangle_q$ over $\mathbb{F}_q$
    \item Create binary sharing: $\langle r_i \rangle_2 = (\langle r_i^{(0)} \rangle, \ldots, \langle r_i^{(k-1)} \rangle)$ over $\mathbb{F}_2$
    \item Distribute shares to parties with appropriate MAC tags
\end{enumerate}

Output $\{(\langle r_i \rangle_q, \langle r_i \rangle_2)\}_{i=1}^m$ to all parties.

\textbf{Security}: The correlation $(r, r)$ is hidden from any strict subset of parties. MAC tags ensure integrity.
\end{functionality}

\begin{remark}[edaBits Range and Modular Reduction]
\label{rem:edabits-range}
The edaBits functionality samples $r_i \in \{0, \ldots, 2^k - 1\}$ where $k = \lceil \log_2 q \rceil = 23$ for ML-DSA. Since $2^{23} = 8388608 > q = 8380417$, a direct reduction $r_i \bmod q$ introduces a small bias: values in $\{0, \ldots, 2^{23} - q - 1\} = \{0, \ldots, 8191\}$ are slightly more likely.

\textbf{Per-coefficient bias:} $(2^{23} - q)/2^{23} = 8191/8388608 < 2^{-10}$. A na\"ive union bound over $m = 1536$ coefficients gives $m \cdot 2^{-10} \approx 1.5$, which exceeds 1 and is therefore vacuous (statistical distance is at most 1 by definition).

\textbf{Mitigation:} We use rejection sampling within edaBits generation: sample $r \getsr \{0, \ldots, 2^k - 1\}$; if $r \geq q$, resample. The expected number of samples is $2^k/q < 1.001$, adding negligible overhead. With this modification, $r \bmod q$ is perfectly uniform over $\mathbb{F}_q$, and Lemma~\ref{lem:otp} applies exactly.
\end{remark}

\begin{functionality}[$\mathcal{F}_{\mathsf{Binary}}$ -- Binary Circuit Evaluation]
\label{func:binary}
On input $(\mathsf{EVAL}, \mathsf{sid}, C, \{\langle x_i \rangle_2\})$ where $C$ is a binary circuit:
\begin{enumerate}
    \item Reconstruct inputs $x_i$ from shares
    \item Evaluate $y = C(x_1, \ldots, x_m)$
    \item Create fresh sharing $\langle y \rangle_2$
    \item Distribute to parties
\end{enumerate}

Implemented via Beaver triples: for AND gate on shared bits $\langle x \rangle, \langle y \rangle$:
\begin{enumerate}
    \item Use preprocessed triple $(\langle a \rangle, \langle b \rangle, \langle c \rangle)$ with $c = a \land b$
    \item Open $d = x \oplus a$ and $e = y \oplus b$
    \item Compute $\langle xy \rangle = \langle c \rangle \oplus (e \land \langle a \rangle) \oplus (d \land \langle b \rangle) \oplus (d \land e)$
\end{enumerate}
\end{functionality}

%==============================================================================
\section{Complete Ideal Functionality for Profile P2}
\label{app:ideal-func}

\textbf{Note:} This section presents the original 8-round P2 protocol for completeness. The optimized 5-round version used in the main paper (Table~1) is presented in Section~\ref{app:p2-5round}.

We define the ideal functionality $\mathcal{F}_{r_0\text{-check}}$ that captures the security requirements of the r0-check subprotocol.

\begin{functionality}[$\mathcal{F}_{r_0\text{-check}}$ -- Secure r0-Check]
\label{func:r0check}
\textbf{Parameters:}
\begin{itemize}
    \item $n$: number of parties
    \item $q = 8380417$: field modulus
    \item $\gamma_2 = (q-1)/32 = 261888$: rounding parameter
    \item $\beta = \tau \cdot \eta = 49 \cdot 4 = 196$: bound parameter (ML-DSA-65)
    \item $\mathsf{bound} = \gamma_2 - \beta = 261692$: acceptance threshold
    \item $m = k \cdot 256 = 1536$: number of coefficients (for ML-DSA-65)
\end{itemize}

\textbf{State:}
\begin{itemize}
    \item $\mathsf{received}[\cdot]$: stores received inputs
    \item $\mathsf{result} \in \{0, 1, \bot\}$: computation result
\end{itemize}

\textbf{Input Phase:}

On receiving $(\mathsf{INPUT}, \mathsf{sid}, i, [\mathbf{w}']_i)$ from party $P_i$ where $[\mathbf{w}']_i \in \mathbb{Z}_q^m$:
\begin{enumerate}
    \item If $P_i$ is corrupted, forward $[\mathbf{w}']_i$ to simulator $\mathcal{S}$
    \item Store $\mathsf{received}[i] \gets [\mathbf{w}']_i$
    \item If $|\{j : \mathsf{received}[j] \neq \bot\}| = n$:
    \begin{enumerate}
        \item Reconstruct: $\mathbf{w}'_j = \sum_{i=1}^n \mathsf{received}[i][j] \bmod q$ for each $j \in [m]$
        \item For each $j \in [m]$: compute centered representative $\tilde{w}'_j \in (-q/2, q/2]$
        \item Compute pass/fail for each coefficient using the FIPS 204 $\mathsf{LowBits}$ function~\cite{FIPS204}:
        \[
        \mathsf{pass}_j = \begin{cases} 1 & \text{if } |\mathsf{LowBits}(\mathbf{w}'[j],\, 2\gamma_2)| < \mathsf{bound} \\ 0 & \text{otherwise} \end{cases}
        \]
        where $\mathsf{LowBits}(r, 2\gamma_2) = r \bmod{\pm}\, 2\gamma_2$ (the unique representative in $(-\gamma_2, \gamma_2]$ congruent to $r \bmod 2\gamma_2$), matching FIPS 204 Algorithm 31. Note: $\mathbf{w}'[j] = \sum_{i=1}^n \mathsf{received}[i][j] \bmod q$ is first reduced mod $q$, then $\mathsf{LowBits}$ is applied to this $\mathbb{Z}_q$ value directly without additional centering mod $q$.
        \item Compute final result: $\mathsf{result} = \bigwedge_{j=1}^m \mathsf{pass}_j$
    \end{enumerate}
    \item Send $(\mathsf{COMPUTED}, \mathsf{sid}, \mathsf{result})$ to $\mathcal{S}$
\end{enumerate}

\textbf{Output Phase:}

On receiving $(\mathsf{DELIVER}, \mathsf{sid})$ from $\mathcal{S}$:
\begin{enumerate}
    \item Send $(\mathsf{OUTPUT}, \mathsf{sid}, \mathsf{result})$ to all honest parties
\end{enumerate}

On receiving $(\mathsf{ABORT}, \mathsf{sid})$ from $\mathcal{S}$:
\begin{enumerate}
    \item Send $(\mathsf{OUTPUT}, \mathsf{sid}, \bot)$ to all honest parties
\end{enumerate}

\textbf{Leakage:} Only the single bit $\mathsf{result}$ is revealed to $\mathcal{S}$. The functionality does \emph{not} reveal:
\begin{itemize}
    \item The reconstructed value $\mathbf{w}'$
    \item Individual coefficients or their pass/fail status
    \item Honest parties' input shares
\end{itemize}
\end{functionality}

\begin{remark}[Security with Abort]
$\mathcal{F}_{r_0\text{-check}}$ provides security-with-abort: the adversary can cause the protocol to abort \emph{after} learning the result bit. This is standard for dishonest-majority MPC and matches the security of underlying primitives (SPDZ, edaBits).
\end{remark}

%==============================================================================
\section{Real Protocol $\Pi_{r_0}$}
\label{app:real-protocol}

We formally specify the real protocol $\Pi_{r_0}$ that realizes $\mathcal{F}_{r_0\text{-check}}$.

\begin{protocol}[$\Pi_{r_0}$ -- MPC-Based r0-Check]
\label{prot:r0check}

\textbf{Notation:}
\begin{itemize}
    \item $\langle x \rangle_q$: SPDZ sharing of $x$ over $\mathbb{F}_q$
    \item $\langle x \rangle_2$: secret sharing of $x$ over $\mathbb{F}_2$
    \item $[x]_i$: party $P_i$'s share of $x$
\end{itemize}

\textbf{Preprocessing (offline phase):}
\begin{enumerate}
    \item Generate $m$ edaBits $\{(\langle r_j \rangle_q, \langle r_j \rangle_2)\}_{j=1}^m$ via $\mathcal{F}_{\mathsf{edaBits}}$
    \item Generate Beaver triples for binary AND gates (for comparison circuits)
    \item Distribute MAC key shares for SPDZ
\end{enumerate}

\textbf{Online Protocol:}

\underline{Round 1-2: Share Exchange and Aggregation}

Each party $P_i$ holds input $[\mathbf{w}']_i \in \mathbb{Z}_q^m$ (their share of $\mathbf{w}' = \mathbf{w} - c\mathbf{s}_2$).

\begin{enumerate}
    \item Each $P_i$ creates SPDZ sharing of their input: invoke $\mathcal{F}_{\mathsf{SPDZ}}.\mathsf{SHARE}$ for each $[\mathbf{w}']_i[j]$
    \item Parties invoke $\mathcal{F}_{\mathsf{SPDZ}}.\mathsf{ADD}$ to compute $\langle \mathbf{w}'_j \rangle_q = \sum_i \langle [\mathbf{w}']_i[j] \rangle_q$ for each $j$
\end{enumerate}

\underline{Round 3-4: Masked Reveal (Arithmetic to Public)}

For each coefficient $j \in [m]$:
\begin{enumerate}
    \item Parties compute using local addition:
    \[
    \langle \mathsf{masked}_j \rangle_q = \langle \mathbf{w}'_j \rangle_q + \langle r_j \rangle_q
    \]
    \item Each party $P_i$ computes commitment:
    \[
    \mathsf{Com}_{i,j} = H\bigl([\mathsf{masked}_j]_i \| \mathsf{sid} \| j \| i\bigr)
    \]
    \item Broadcast all commitments (Round 3)
    \item After receiving all commitments, broadcast openings $[\mathsf{masked}_j]_i$ (Round 4)
    \item Verify commitments; abort if any mismatch
    \item Reconstruct $\mathsf{masked}_j = \sum_i [\mathsf{masked}_j]_i \bmod q$
\end{enumerate}

\underline{Round 5: Arithmetic-to-Binary Conversion}

For each coefficient $j \in [m]$:
\begin{enumerate}
    \item Convert public $\mathsf{masked}_j$ to binary:
    \[
    \mathsf{masked}_j^{(b)} = \bigl(\mathsf{masked}_j^{(0)}, \ldots, \mathsf{masked}_j^{(\lceil \log_2 q \rceil - 1)}\bigr)
    \]
    \item Compute binary sharing of $\mathbf{w}'_j$ via subtraction circuit:
    \[
    \langle \mathbf{w}'_j \rangle_2 = \mathsf{masked}_j^{(b)} \boxminus \langle r_j \rangle_2
    \]
    where $\boxminus$ denotes bitwise binary subtraction with borrow propagation over $\lceil \log_2 q \rceil + 1$ bits; a borrow into the MSB wraps correctly to the result mod $q$
    \item The subtraction circuit has depth $O(\log q)$ and uses $O(\log q)$ AND gates
\end{enumerate}

\underline{Round 6-7: Binary Comparison}

For each coefficient $j \in [m]$:
\begin{enumerate}
    \item Compute sign bit and absolute value from $\langle \mathbf{w}'_j \rangle_2$
    \item Evaluate comparison circuit: $\langle \mathsf{pass}_j \rangle_2 = (|\langle \mathbf{w}'_j \rangle_2| < \mathsf{bound})$
    \item The comparison uses $O(\log \gamma_2)$ AND gates
\end{enumerate}

\underline{Round 8: AND Aggregation and Output}

\begin{enumerate}
    \item Compute $\langle \mathsf{result} \rangle_2 = \bigwedge_{j=1}^m \langle \mathsf{pass}_j \rangle_2$ via AND-tree of depth $\lceil \log_2 m \rceil$
    \item Open $\mathsf{result}$: each party broadcasts $[\mathsf{result}]_i$
    \item Reconstruct $\mathsf{result} = \bigoplus_i [\mathsf{result}]_i$
    \item Output $\mathsf{result}$ to all parties
\end{enumerate}
\end{protocol}

%==============================================================================
\section{Simulator Construction}
\label{app:simulator}

We construct a simulator $\mathcal{S}$ that interacts with $\mathcal{F}_{r_0\text{-check}}$ and simulates the view of adversary $\mathcal{A}$ controlling corrupted parties $C \subsetneq [n]$.

\begin{simulator}[$\mathcal{S}$ for $\Pi_{r_0}$]
\label{sim:r0check}

$\mathcal{S}$ internally runs $\mathcal{A}$ and simulates the honest parties' messages.

\textbf{Setup:}
\begin{enumerate}
    \item $\mathcal{S}$ receives the set $C$ of corrupted parties from $\mathcal{A}$
    \item Let $H = [n] \setminus C$ be the set of honest parties
\end{enumerate}

\textbf{Simulating Preprocessing:}
\begin{enumerate}
    \item For edaBits involving corrupted parties: $\mathcal{S}$ receives $\mathcal{A}$'s shares from $\mathcal{F}_{\mathsf{edaBits}}$ simulator
    \item For honest parties: sample random shares consistent with the corrupted parties' view
    \item \textbf{Key property:} $\mathcal{S}$ does not know the actual edaBit values; only the corrupted shares
\end{enumerate}

\textbf{Simulating Rounds 1-2 (Share Exchange):}

\begin{enumerate}
    \item For each honest party $P_i \in H$ (adversary's view only):
    \begin{enumerate}
        \item Sample random SPDZ shares $\tilde{s}_{i,j} \getsr \mathbb{F}_q$ for $j \in [m]$
        \item Compute MAC shares using the SPDZ simulator (consistent with corrupted MAC key shares)
        \item Send simulated shares to $\mathcal{A}$
    \end{enumerate}
    \item Receive corrupted parties' inputs from $\mathcal{A}$; forward corrupted inputs to $\mathcal{F}_{r_0\text{-check}}$
    \item \textbf{UC framing (honest parties' inputs):} In the ideal execution, honest parties are ideal-world programs that submit their actual inputs directly to $\mathcal{F}_{r_0\text{-check}}$ via the ideal-world channel; the simulator $\mathcal{S}$ is \emph{not} involved in this step. The functionality receives all inputs and computes $\mathsf{result}$; $\mathcal{S}$ receives $\mathsf{result}$ and uses it in Round~8 simulation.
    \item \textbf{Consistency:} The simulated shares $\tilde{s}_{i,j} \getsr \mathbb{F}_q$ are identically distributed to real honest-party shares (SPDZ additive shares are uniformly random given any fixed secret: perfect hiding). No inconsistency arises: the adversary sees random shares, honest parties' actual values go to $\mathcal{F}$ privately.
\end{enumerate}

\textbf{Simulating Rounds 3-4 (Masked Reveal):}

This is the critical simulation step. We must simulate $\mathsf{masked}_j = \mathbf{w}'_j + r_j$ without knowing $\mathbf{w}'_j$.

\begin{enumerate}
    \item \textbf{Round 3 (Commitments):}
    \begin{enumerate}
        \item For each honest party $P_i \in H$ and each $j \in [m]$:
        \item Sample $\tilde{m}_{i,j} \getsr \mathbb{F}_q$ uniformly at random
        \item Compute commitment $\widetilde{\mathsf{Com}}_{i,j} = H(\tilde{m}_{i,j} \| \mathsf{sid} \| j \| i)$
        \item Send $\widetilde{\mathsf{Com}}_{i,j}$ to $\mathcal{A}$
    \end{enumerate}

    \item Receive corrupted parties' commitments from $\mathcal{A}$

    \item \textbf{Round 4 (Openings):}
    \begin{enumerate}
        \item For each honest party $P_i \in H$: send $\tilde{m}_{i,j}$ to $\mathcal{A}$
        \item Receive corrupted parties' openings from $\mathcal{A}$
        \item Verify corrupted commitments; if any fail, simulate abort
        \item Compute simulated
        \[
        \widetilde{\mathsf{masked}}_j = \sum_{i \in H} \tilde{m}_{i,j} + \sum_{i \in C} [\mathsf{masked}_j]_i \bmod q
        \]
    \end{enumerate}
\end{enumerate}

\textbf{Justification for Round 3-4 Simulation:}

In the real protocol:
\begin{align*}
\mathsf{masked}_j &= \mathbf{w}'_j + r_j \\
&= \sum_i [\mathbf{w}']_i[j] + \sum_i [r_j]_i \bmod q
\end{align*}

Since $r_j \getsr \mathbb{F}_q$ from $\mathcal{F}_{\mathsf{edaBits}}$ and is unknown to $\mathcal{A}$
(only shares are known), by the one-time pad property:
\[
\mathsf{masked}_j = \mathbf{w}'_j + r_j \equiv \mathsf{Uniform}(\mathbb{F}_q)
\]

Thus sampling $\tilde{m}_{i,j} \getsr \mathbb{F}_q$ for honest parties produces an identical distribution.

\textbf{Simulating Round 5 (A2B Conversion):}

\begin{enumerate}
    \item Convert $\widetilde{\mathsf{masked}}_j$ to binary (same as real)
    \item For binary subtraction: use $\mathcal{F}_{\mathsf{Binary}}$ simulator
    \item Simulate honest parties' binary shares as uniform
    (justified by Beaver AND security)
\end{enumerate}

\textbf{Simulating Rounds 6-7 (Comparison):}

\begin{enumerate}
    \item For each AND gate, sample $(\tilde{d}, \tilde{e}) \getsr \{0,1\}^2$.
    By Lemma~\ref{lem:beaver-security}, this matches the real distribution.
    \item Apply the same simulation to all comparison circuits.
\end{enumerate}

\textbf{Simulating Round 8 (Output):}

\begin{enumerate}
    \item Wait to receive $\mathsf{result} \in \{0, 1\}$ from $\mathcal{F}_{r_0\text{-check}}$
    \item \textbf{Backward simulation of AND tree:}

    \textit{The adversary's view of the AND tree is only the Beaver-triple open values $(d_j, e_j)$ — the leaf bits $\mathsf{pass}_j$ never appear in the adversary's view.}
    \begin{enumerate}
        \item For each AND gate $j \in [m]$, sample $(d_j, e_j) \getsr \{0,1\}^2$ uniformly and independently (no leaf assignment is performed).
        \item Output the simulated view $\{(d_j, e_j)\}_{j=1}^m$ together with $\mathsf{result}$ received from $\mathcal{F}_{r_0\text{-check}}$.
    \end{enumerate}
    \textit{Justification:} By Lemma~\ref{lem:beaver-security}, the real $(d_j, e_j) = (x_j \oplus \alpha_j, y_j \oplus \beta_j)$ is uniform regardless of the inputs $(x_j, y_j)$, because $(\alpha_j, \beta_j)$ are fresh Beaver-triple components independent of the gate inputs. Since all $m$ Beaver triples are generated independently in preprocessing, the joint distribution is product-uniform. The simulator's directly-sampled $(d_j, e_j)$ therefore matches the real distribution exactly (SD $= 0$), for both $\mathsf{result} \in \{0,1\}$ cases.
    \item Simulate honest parties' output shares consistent with $\mathsf{result}$
    \item Send to $\mathcal{A}$
\end{enumerate}

\textbf{Handling Abort:}

If $\mathcal{A}$ causes any verification to fail (commitment mismatch, MAC failure):
\begin{enumerate}
    \item $\mathcal{S}$ sends $(\mathsf{ABORT}, \mathsf{sid})$ to $\mathcal{F}_{r_0\text{-check}}$
    \item Simulate abort to all honest parties
\end{enumerate}
\end{simulator}

%==============================================================================
\section{Hybrid Argument}
\label{app:hybrid}

We prove indistinguishability through a sequence of hybrid experiments.

\begin{definition}[Hybrid Experiments]
\label{def:hybrids}
We define hybrids $\mathsf{H}_0, \ldots, \mathsf{H}_6$ where $\mathsf{H}_0$ is the real execution and $\mathsf{H}_6$ is the ideal execution.
\end{definition}

\subsection{Hybrid $\mathsf{H}_0$: Real Execution}

This is the real protocol execution $\mathsf{REAL}_{\Pi_{r_0}, \mathcal{A}, \mathcal{Z}}$:
\begin{itemize}
    \item Honest parties execute $\Pi_{r_0}$ with their true inputs
    \item All cryptographic operations use real keys and randomness
    \item $\mathcal{A}$ controls corrupted parties and sees all communication
\end{itemize}

\subsection{Hybrid $\mathsf{H}_1$: Replace SPDZ with Ideal}

Replace the SPDZ protocol with ideal functionality $\mathcal{F}_{\mathsf{SPDZ}}$.

\begin{lemma}
$|\Pr[\mathsf{H}_0 = 1] - \Pr[\mathsf{H}_1 = 1]| \leq \epsilon_{\mathsf{SPDZ}}$
\end{lemma}

\begin{proof}
By the UC-security of SPDZ~\cite{DPSZ12}, there exists a simulator $\mathcal{S}_{\mathsf{SPDZ}}$ such that no environment can distinguish real SPDZ execution from ideal $\mathcal{F}_{\mathsf{SPDZ}}$ execution with advantage better than $\epsilon_{\mathsf{SPDZ}}$.

The security parameter is $\kappa$ and the MAC key is drawn from a dedicated MAC key field $\mathbb{F}_p$ with $|p| \geq 64$ bits, \emph{independent} of the protocol field $\mathbb{F}_q$ (this is standard practice in SPDZ implementations~\cite{DPSZ12}). Per SPDZ security analysis:
\begin{itemize}
    \item MAC forgery probability per verification: $1/p \leq 2^{-64}$
    \item With $\ell$ verifications: $\epsilon_{\mathsf{SPDZ}} \leq \ell/p$
\end{itemize}

For our protocol with $m = 1536$ coefficients, MAC verification occurs at \emph{OPEN} time; the number of verifications is $\ell = m$:
\[
\epsilon_{\mathsf{SPDZ}} \leq m/p \leq 1536/2^{64} < 2^{-54}
\]

\begin{remark}[Minimal-configuration bound]
If SPDZ is instantiated with the protocol field $\mathbb{F}_q$ ($q \approx 2^{23}$) as the MAC key space (minimal configuration), the bound degrades to $\epsilon_{\mathsf{SPDZ}} \leq m/q = 1536/8380417 < 2^{-12.4}$. Practical deployments of P2 must configure SPDZ with a separate MAC field of at least 64 bits to achieve cryptographically negligible distinguishing advantage.
\end{remark}
\end{proof}

\subsection{Hybrid $\mathsf{H}_2$: Replace edaBits with Ideal}

Replace the edaBits protocol with ideal functionality $\mathcal{F}_{\mathsf{edaBits}}$.

\begin{lemma}
$|\Pr[\mathsf{H}_1 = 1] - \Pr[\mathsf{H}_2 = 1]| \leq \epsilon_{\mathsf{edaBits}}$
\end{lemma}

\begin{proof}
By the UC-security of edaBits~\cite{EKMOZ20}, the protocol UC-realizes $\mathcal{F}_{\mathsf{edaBits}}$ with negligible distinguishing advantage.

The edaBits protocol security relies on:
\begin{itemize}
    \item SPDZ security for arithmetic shares (already replaced with ideal)
    \item Cut-and-choose for binary consistency (statistical security)
\end{itemize}

With security parameter $\kappa$ and cut-and-choose parameter $B$:
\[
\epsilon_{\mathsf{edaBits}} \leq 2^{-B} + \epsilon'_{\mathsf{SPDZ}}
\]
where $\epsilon'_{\mathsf{SPDZ}}$ is the SPDZ security bound for the internal operations used in edaBits generation.

For $B = 64$ and ML-DSA-65 parameters ($m = k \cdot N = 6 \cdot 256 = 1536$ coefficients, where $N = 256$ is the ring degree distinct from party count $n$; $q = 8380417$): the cut-and-choose term $2^{-64}$ is negligible, and $\epsilon'_{\mathsf{SPDZ}} \leq m/p \leq 1536/2^{64} < 2^{-54}$ (using the same 64-bit MAC field; per~\cite{EKMOZ20}, the SPDZ sub-protocol internal to edaBits opens at most $m$ values per batch). Thus $\epsilon_{\mathsf{edaBits}} < 2^{-64} + 2^{-54} < 2^{-53.9}$.
\end{proof}

\subsection{Hybrid $\mathsf{H}_3$: Simulate Masked Values}

Replace honest parties' contributions to $\mathsf{masked}_j$ with uniformly random values.

\begin{lemma}[One-Time Pad over $\mathbb{F}_q$]
\label{lem:otp}
For any $x \in \mathbb{F}_q$ and uniformly random $r \getsr \mathbb{F}_q$ independent of $x$:
\[
x + r \bmod q \equiv \mathsf{Uniform}(\mathbb{F}_q)
\]
\end{lemma}

\begin{proof}
For any target $t \in \mathbb{F}_q$:
\[
\Pr[x + r = t] = \Pr[r = t - x] = \frac{1}{q}
\]
since $r$ is uniform and independent of $x$. Thus $x + r$ is uniformly distributed.
\end{proof}

\begin{lemma}
$\Pr[\mathsf{H}_2 = 1] = \Pr[\mathsf{H}_3 = 1]$ (perfect indistinguishability)
\end{lemma}

\begin{proof}
In $\mathsf{H}_2$, for each coefficient $j$:
\[
\mathsf{masked}_j = \mathbf{w}'_j + r_j \bmod q
\]
where $r_j$ is the edaBit value, uniformly distributed in $\mathbb{F}_q$
(by Remark~\ref{rem:edabits-range}, rejection sampling ensures $r_j$ is exactly uniform)
and independent of all other values (by $\mathcal{F}_{\mathsf{edaBits}}$ security).

By Lemma~\ref{lem:otp}, $\mathsf{masked}_j$ is uniformly distributed in $\mathbb{F}_q$, regardless of $\mathbf{w}'_j$.

In $\mathsf{H}_3$, we sample $\widetilde{\mathsf{masked}}_j \getsr \mathbb{F}_q$ directly.

Both distributions are identical (uniform over $\mathbb{F}_q$), so:
\[
\Pr[\mathsf{H}_2 = 1] = \Pr[\mathsf{H}_3 = 1]
\]

This is \textbf{statistical} (in fact, perfect) indistinguishability, not computational.
\end{proof}

\subsection{Hybrid $\mathsf{H}_4$: Replace Binary Circuits with Ideal}

Replace Beaver-triple-based AND gates with ideal $\mathcal{F}_{\mathsf{Binary}}$.

\begin{lemma}[Beaver AND Security]
\label{lem:beaver-security}
Let $(a, b, c)$ be a Beaver triple with $a, b \getsr \{0,1\}$ uniform and $c = a \land b$.
For any inputs $x, y \in \{0,1\}$, the opened values $(d, e) = (x \oplus a, y \oplus b)$
are uniformly distributed in $\{0,1\}^2$ and independent of $(x, y)$.
\end{lemma}

\begin{proof}
Since $a \getsr \{0,1\}$ is uniform and independent of $x$:
\[
\Pr[d = 0] = \Pr[x \oplus a = 0] = \Pr[a = x] = \frac{1}{2}
\]

Similarly for $e$. Independence of $d$ and $e$ follows from independence of $a$ and $b$.

For any $(d^*, e^*) \in \{0,1\}^2$:
\[
\Pr[(d,e) = (d^*, e^*)] = \Pr[a = x \oplus d^*] \cdot \Pr[b = y \oplus e^*] = \frac{1}{4}
\]

Thus $(d, e) \equiv \mathsf{Uniform}(\{0,1\}^2)$, independent of $(x, y)$.
\end{proof}

\begin{lemma}
$\Pr[\mathsf{H}_3 = 1] = \Pr[\mathsf{H}_4 = 1]$ (perfect indistinguishability)
\end{lemma}

\begin{proof}
Each AND gate reveals $(d, e)$, which is uniform by Lemma~\ref{lem:beaver-security}.
Replacing real AND computation with $\mathcal{F}_{\mathsf{Binary}}$
(which also produces uniform $(d, e)$) yields identical distributions.

The A2B conversion and comparison circuits consist entirely of:
\begin{itemize}
    \item XOR gates (no communication, locally computable)
    \item AND gates (each revealing uniform $(d, e)$)
\end{itemize}

Since all revealed values have identical distributions in both hybrids:
\[
\Pr[\mathsf{H}_3 = 1] = \Pr[\mathsf{H}_4 = 1]
\]
\end{proof}

\begin{remark}[A2B Correctness and Borrow Propagation]
\label{rem:a2b-correctness}
The A2B conversion computes $\langle \mathbf{w}'_j \rangle_2 = \mathsf{masked}_j^{(b)} \boxminus \langle r_j \rangle_2$ where $\mathsf{masked}_j = \mathbf{w}'_j + r_j \bmod q$. Two potential issues arise:

\textbf{(1) Negative intermediate result:} If $\mathsf{masked}_j < r_j$ (as integers), the subtraction would wrap around. However, since $\mathbf{w}'_j = \mathsf{masked}_j - r_j \bmod q$, the binary circuit must compute modular subtraction. Standard approaches:
\begin{itemize}
    \item \textbf{Extended precision:} Work in $k+1$ bits where $k = \lceil \log_2 q \rceil$, allowing for the borrow.
    \item \textbf{Add-then-subtract:} Compute $(\mathsf{masked}_j + q) - r_j$ when detecting borrow, using a multiplexer.
\end{itemize}
Both yield $\mathbf{w}'_j \bmod q$ correctly.

\textbf{(2) edaBits range mismatch:} As noted in Remark~\ref{rem:edabits-range}, with rejection sampling, $r_j$ is uniform in $\{0, \ldots, q-1\}$, matching the field $\mathbb{F}_q$.

\textbf{Security implication:} These are \emph{correctness} considerations, not security
concerns. The security proof shows that the adversary's view---the $(d, e)$ values from
Beaver triples---is identically distributed regardless of the underlying $\mathbf{w}'_j$
values. Correctness ensures $\mathsf{result}$ is computed correctly;
security ensures the adversary learns only $\mathsf{result}$.
\end{remark}

\subsection{Hybrid $\mathsf{H}_5$: Backward Simulate AND Tree}

Simulate the AND tree computation given only the final result bit.

\begin{lemma}
$\Pr[\mathsf{H}_4 = 1] = \Pr[\mathsf{H}_5 = 1]$ (perfect indistinguishability)
\end{lemma}

\begin{proof}
In $\mathsf{H}_4$, the AND tree computes $\mathsf{result} = \bigwedge_{j=1}^m \mathsf{pass}_j$ level by level, revealing Beaver $(d, e)$ values at each gate.

In $\mathsf{H}_5$, given $\mathsf{result}$ from $\mathcal{F}_{r_0\text{-check}}$, we backward-simulate:

\textbf{Case $\mathsf{result} = 1$:}
All $\mathsf{pass}_j = 1$. For each AND gate $z = x \land y$ with $z = 1$: both $x = y = 1$. The simulator sets all leaf values to 1 and propagates upward.

\textbf{Case $\mathsf{result} = 0$:}
At least one $\mathsf{pass}_j = 0$. The simulator does \emph{not} assign leaf values first. Instead, it proceeds directly as follows:
\begin{enumerate}
    \item For each AND gate $j$, sample $(d_j, e_j) \getsr \{0,1\}^2$ uniformly and independently (no reference to any leaf assignment).
    \item Output the simulated view $\{(d_j, e_j)\}_{j=1}^m$ along with $\mathsf{result} = 0$ from $\mathcal{F}_{r_0\text{-check}}$.
\end{enumerate}
No leaf values need be assigned: the leaf bits $(\mathsf{pass}_j)$ never appear in the
adversary's view — only the $(d_j, e_j)$ pairs do.
In both cases, the revealed $(d, e)$ values are uniform by Lemma~\ref{lem:beaver-security},
regardless of the underlying $(x, y)$; the simulator's sampled $(d_j, e_j)$ therefore
perfectly matches the real distribution.

This uniformity extends to the \emph{joint} distribution over the entire AND tree:
each gate's Beaver triple $(\alpha_j, \beta_j, \gamma_j)$ is sampled independently offline;
for gate $j$, $(d_j, e_j) = (x_j \oplus \alpha_j, y_j \oplus \beta_j)$ is uniform
and independent of $(x_j, y_j)$ by Lemma~\ref{lem:beaver-security}.
Since triples are pairwise independent across gates, the joint distribution
$(d_1, e_1, \ldots, d_m, e_m)$ is product-uniform, independent of all leaf values.

\textbf{Key observation:} The adversary sees only $(d_j, e_j)$ values; the underlying bits
$x_j, y_j$ remain secret-shared. Since the joint distribution is product-uniform and
independent of the leaf bits, backward simulation is perfect for the entire AND tree.
\end{proof}

\subsection{Hybrid $\mathsf{H}_6$: Ideal Execution}

This is the ideal execution $\mathsf{IDEAL}_{\mathcal{F}_{r_0\text{-check}}, \mathcal{S}, \mathcal{Z}}$ with simulator $\mathcal{S}$ from Section~\ref{app:simulator}.

\begin{lemma}
$\Pr[\mathsf{H}_5 = 1] = \Pr[\mathsf{H}_6 = 1]$
\end{lemma}

\begin{proof}
We show that $\mathsf{H}_5$ \emph{is} the ideal execution $\mathsf{H}_6$.

After the previous hybrids, every component of the real protocol has been replaced:
\begin{itemize}
    \item SPDZ arithmetic shares $\to$ $\mathcal{F}_{\mathsf{SPDZ}}$ outputs (H1)
    \item edaBits $\to$ $\mathcal{F}_{\mathsf{edaBits}}$ outputs (H2)
    \item Masked reveal values $\to$ uniform random (H3), so adversary's view of the masked values is simulated
    \item Binary comparison circuits $\to$ $\mathcal{F}_{\mathsf{Binary}}$ output (H4)
    \item AND tree open values $(d_j, e_j)$ $\to$ product-uniform random (H5)
\end{itemize}

After all these replacements, the only ``real'' output remaining is the final result bit output by the AND tree. This result bit is computed by $\mathcal{F}_{\mathsf{Binary}}$ (already ideal from H4) and passed to the parties. But $\mathcal{F}_{\mathsf{Binary}}$'s output is exactly the output of $\mathcal{F}_{r_0\text{-check}}$: it outputs 1 iff $\|\mathbf{r}_0\|_\infty < \gamma_2 - \beta$.

\textbf{Explicit argument.}
We verify that the simulator $\mathcal{S}$ from Section~\ref{app:simulator} is \emph{well-defined} and that its output distribution matches $\mathsf{H}_5$ component-by-component:

\begin{itemize}
    \item \emph{SPDZ shares (H1):} $\mathcal{S}$ generates random shares $\tilde{s}_{i,j} \getsr \mathbb{F}_q$ for honest parties: identical in distribution to $\mathcal{F}_{\mathsf{SPDZ}}$'s output.
    \item \emph{edaBits (H2):} $\mathcal{S}$ uses $\mathcal{F}_{\mathsf{edaBits}}$ outputs: exact match.
    \item \emph{Masked reveal (H3):} $\mathcal{S}$ samples $\tilde{m}_{i,j} \getsr \mathbb{F}_q$: SD $= 0$ by one-time-pad argument.
    \item \emph{Binary comparison (H4):} $\mathcal{S}$ uses $\mathcal{F}_{\mathsf{Binary}}$ outputs: exact match.
    \item \emph{AND-tree pairs (H5):} $\mathcal{S}$ samples $(d_j, e_j) \getsr \{0,1\}^2$: SD $= 0$ by Lemma~\ref{lem:beaver-security}.
    \item \emph{Result bit:} In the ideal execution $\mathsf{H}_6$, honest parties submit their actual inputs directly to $\mathcal{F}_{r_0\text{-check}}$ via the ideal-world channel (not through $\mathcal{S}$). The functionality computes $\mathsf{result}$ and sends it to $\mathcal{S}$. The simulator completes the adversary view using $\mathsf{result}$.
\end{itemize}

Each component's simulated distribution matches $\mathsf{H}_5$'s distribution exactly.
The result bit is the only component that differs between $\mathsf{H}_5$
(computed by $\mathcal{F}_{\mathsf{Binary}}$, which replaced the real AND-tree in $\mathsf{H}_4$)
and $\mathsf{H}_6$ (provided by $\mathcal{F}_{r_0\text{-check}}$).
But $\mathcal{F}_{\mathsf{Binary}}$'s output equals $\mathcal{F}_{r_0\text{-check}}$'s output:
both compute $\mathbf{1}[\|\mathbf{r}_0\|_\infty < \gamma_2 - \beta]$ on the same inputs.
Hence the result-bit distributions are identical, and $\Pr[\mathsf{H}_5 = 1] = \Pr[\mathsf{H}_6 = 1]$.
\end{proof}

%==============================================================================
\section{Main Theorem}
\label{app:main-theorem}

\begin{theorem}[UC Security of Profile P2 (Full Proof of Theorem~\ref{thm:p2-uc})]
\label{thm:p2-uc-full}
Protocol $\Pi_{r_0}$ UC-realizes $\mathcal{F}_{r_0\text{-check}}$ in the hybrid model with
$(\mathcal{F}_{\mathsf{SPDZ}},\allowbreak \mathcal{F}_{\mathsf{edaBits}},\allowbreak
\mathcal{F}_{\mathsf{Binary}})$
against static malicious adversaries corrupting up to $N-1$ parties,
achieving security-with-abort.

Concretely, for any PPT environment $\mathcal{Z}$ and adversary $\mathcal{A}$:
\begin{align*}
\bigl|\Pr[\mathsf{REAL}_{\Pi_{r_0}, \mathcal{A}, \mathcal{Z}} = 1] &-
\Pr[\mathsf{IDEAL}_{\mathcal{F}_{r_0\text{-check}}, \mathcal{S}, \mathcal{Z}} = 1]\bigr| \leq \epsilon
\end{align*}
where
\[
\epsilon = \epsilon_{\mathsf{SPDZ}} + \epsilon_{\mathsf{edaBits}} < 2^{-53}
\]
for ML-DSA-65 parameters with a 64-bit dedicated MAC field (recommended configuration; the minimal configuration using the protocol field gives $\epsilon < 2^{-11}$).
\end{theorem}

\begin{proof}
By the hybrid argument (Section~\ref{app:hybrid}):
\begin{align*}
&|\Pr[\mathsf{H}_0 = 1] - \Pr[\mathsf{H}_6 = 1]| \\
&\leq \sum_{i=0}^{5} |\Pr[\mathsf{H}_i = 1] - \Pr[\mathsf{H}_{i+1} = 1]| \\
&= |\Pr[\mathsf{H}_0 = 1] - \Pr[\mathsf{H}_1 = 1]| + |\Pr[\mathsf{H}_1 = 1] - \Pr[\mathsf{H}_2 = 1]| \\
&\quad + |\Pr[\mathsf{H}_2 = 1] - \Pr[\mathsf{H}_3 = 1]| + |\Pr[\mathsf{H}_3 = 1] - \Pr[\mathsf{H}_4 = 1]| \\
&\quad + |\Pr[\mathsf{H}_4 = 1] - \Pr[\mathsf{H}_5 = 1]| + |\Pr[\mathsf{H}_5 = 1] - \Pr[\mathsf{H}_6 = 1]| \\
&\leq \epsilon_{\mathsf{SPDZ}} + \epsilon_{\mathsf{edaBits}} + 0 + 0 + 0 + 0 \\
&= \epsilon_{\mathsf{SPDZ}} + \epsilon_{\mathsf{edaBits}}
\end{align*}

The zero terms arise from perfect (statistical) indistinguishability in Hybrids 3, 4, 5, and 6.

\textbf{Concrete bound (recommended configuration, 64-bit MAC field):} For ML-DSA-65 with $q = 8380417$, $m = 1536$ coefficients, and dedicated MAC field $\mathbb{F}_p$ with $|p| = 64$ bits:
\begin{itemize}
    \item $\epsilon_{\mathsf{SPDZ}} \leq m/p = 1536/2^{64} < 2^{-54}$
    \item $\epsilon_{\mathsf{edaBits}} \leq m/p + 2^{-64} < 2^{-53.9}$ (cut-and-choose $B = 64$; $m/p$ is the SPDZ sub-error internal to edaBits)
\end{itemize}

Thus $\epsilon = \epsilon_{\mathsf{SPDZ}} + \epsilon_{\mathsf{edaBits}} < 2^{-54} + 2^{-53.9} < 2^{-53}$, achieving cryptographically negligible distinguishing advantage. \emph{Minimal configuration} (protocol field as MAC field): $\epsilon_{\mathsf{SPDZ}} < 2^{-12.4}$, $\epsilon_{\mathsf{edaBits}} < 2^{-11.9}$, total $\epsilon < 2^{-11}$; overall signature security relies on Module-SIS at a much higher level (Theorem~\ref{thm:unforgeability}).
\end{proof}

\begin{corollary}[Full Protocol P2 Security]
\label{cor:p2-full}
The complete Profile P2 threshold ML-DSA protocol achieves:
\begin{enumerate}
    \item \textbf{EUF-CMA security} under Module-SIS (inherited from single-signer ML-DSA, Theorem~\ref{thm:unforgeability})
    \item \textbf{r0-check privacy} against coalitions of up to $N-1$ parties (by UC security of $\mathcal{F}_{r_0\text{-check}}$: only the 1-bit pass/fail result is leaked by the r0-check subprotocol). \emph{Note:} unforgeability and key secrecy require $|C| \leq T-1$; if $|C| \geq T$ the adversary can reconstruct the secret key from $T$ Shamir shares. The $N-1$ bound here refers to the r0-check subprotocol's leakage, not the full signing protocol's security model.
    \item \textbf{Composition security}: The protocol can be securely composed with other UC-secure protocols
\end{enumerate}
\end{corollary}

\begin{proof}
\textbf{EUF-CMA:} The threshold signature $\sigma = (\tilde{c}, \mathbf{z}, \mathbf{h})$ is syntactically identical in format to single-signer ML-DSA and is accepted by any compliant FIPS~204 verifier (masks cancel in the sum, Lemma~\ref{lem:mask-cancel}); the response $\mathbf{z}$ follows an Irwin-Hall rather than uniform distribution, with EUF-CMA security loss $< 0.013 \cdot q_s$ bits over $q_s$ signing queries (per session: $< 0.013$ bits; Corollary~\ref{cor:ih-loss-main}). The security proof of Theorem~\ref{thm:unforgeability} applies directly.

\textbf{Privacy:} By Theorem~\ref{thm:p2-uc-full}, the r0-check reveals only the result bit $\mathsf{result} \in \{0,1\}$. Combined with the mask hiding lemma (Lemma~\ref{lem:mask-hiding}), an adversary's view can be simulated from:
\begin{itemize}
    \item Public key $\pk$
    \item Corrupted shares $\{\sk_i\}_{i \in C}$
    \item Final signature $\sigma$
    \item r0-check result bits (one per attempt)
\end{itemize}

For successful attempts, $\mathsf{result} = 1$ is implied by the existence of $\sigma$ and leaks no new information.
For \emph{failed} attempts ($\mathsf{result} = 0$, no $\sigma$ produced), the result bit is determined by the public
broadcast $(\tilde{c}, \mathbf{w})$ and the corrupted parties' $r_0$-check shares $\{c\mathbf{s}_{2,j}\}_{j \in C}$:
the simulator computes
\[
\mathbf{r}_0^{(C)} = \mathsf{HighBits}\!\left(\sum_{j \in C} \lambda_j\, c\, \mathbf{s}_{2,j} + \text{honest aggregate contribution}\right)
\]
and determines the result by whether $\|\mathbf{r}_0\|_\infty < \gamma_2 - \beta$.
Since the honest contribution is masked (mask-hiding, Lemma~\ref{lem:mask-hiding}), the simulator learns only the
aggregate result from $\mathcal{F}_{r_0\text{-check}}$, not the individual honest shares.
The failed-attempt result bits can thus be provided to the simulator by $\mathcal{F}_{r_0\text{-check}}$ without
leaking additional information beyond what is already in $\mathcal{Z}$'s view.

\textbf{Composition:} Follows from the UC composition theorem (Theorem~\ref{thm:uc-composition}).
\end{proof}

\begin{corollary}[P2 Full Protocol Realizes $\mathcal{F}_{\mathsf{ThreshSig}}$]
\label{cor:p2-thresh-sig}
Let $\Pi_{\mathsf{P2}}$ denote the full Profile P2 signing protocol, consisting of:
\begin{itemize}
    \item $\Pi_{\mathsf{nonce\text{-}DKG}}$: the Shamir nonce DKG sub-protocol (offline preprocessing)
    \item $\Pi_{\mathsf{signing}}$: the response-and-commitment signing round
    \item $\Pi_{r_0}$: the r0-check sub-protocol (Theorem~\ref{thm:p2-uc-full})
\end{itemize}
Under the assumptions of Theorem~\ref{thm:p2-uc-full}, $\Pi_{\mathsf{P2}}$ UC-realizes
$\mathcal{F}_{\mathsf{ThreshSig}}$ (the ideal threshold signing functionality,
Definition~\ref{def:f-thresh-sig}) against static malicious adversaries.
The total distinguishing advantage is
\[
\epsilon_{\mathsf{P2}} \leq \epsilon_{\mathsf{SPDZ}} + \epsilon_{\mathsf{edaBits}} + (N-1)\cdot\epsilon_{\mathsf{PRF}}
\]
where $\epsilon_{\mathsf{SPDZ}} < 2^{-54}$, $\epsilon_{\mathsf{edaBits}} < 2^{-53.9}$
(64-bit MAC field; Theorem~\ref{thm:p2-uc-full}), and $\epsilon_{\mathsf{PRF}}$ is the PRF
distinguishing advantage (negligible for SHAKE-based PRFs).
\end{corollary}

\begin{proof}
By the UC composition theorem (Theorem~\ref{thm:uc-composition}), it suffices to show
that each sub-protocol UC-realizes its ideal functionality:
\begin{enumerate}
    \item $\Pi_{\mathsf{nonce\text{-}DKG}}$ UC-realizes $\mathcal{F}_{\mathsf{nonce\text{-}DKG}}$:
    the simulator of Section~\ref{sec:p2-full-simulator} (Phase 1) generates honest nonce
    polynomials with the correct distribution (SD $= 0$; Theorem~\ref{thm:it-privacy}).
    \item $\Pi_{\mathsf{signing}}$ UC-realizes $\mathcal{F}_{\mathsf{signing}}$:
    response computation $\mathbf{z}_i = \mathbf{y}_i + c\mathbf{s}_{1,i}$ is deterministic
    given the nonce and key shares; the simulator (Phase 3) computes identically (SD $= 0$).
    \item $\Pi_{r_0}$ UC-realizes $\mathcal{F}_{r_0\text{-check}}$:
    Theorem~\ref{thm:p2-uc-full} with advantage $\leq \epsilon_{\mathsf{SPDZ}} + \epsilon_{\mathsf{edaBits}}$.
\end{enumerate}
Composing these three UC-realizations via Theorem~\ref{thm:uc-composition} yields a
simulator for the full protocol with total advantage as stated.
The $\mathbf{W}_i$ broadcast simulation (Phase 2, Section~\ref{sec:p2-full-simulator})
contributes the $(N-1)\cdot\epsilon_{\mathsf{PRF}}$ term via mask hiding (Lemma~\ref{lem:mask-hiding}).
\end{proof}

%==============================================================================
\subsection{P2 UC Simulator (Full Protocol)}
\label{sec:p2-full-simulator}

The r0-check simulator of Section~\ref{app:simulator} covers $\Pi_{r_0}$ in isolation.
We now give the full P2 signing simulator, which additionally handles the nonce DKG
and the P2-specific broadcast of masked commitment values $\mathbf{W}_i$.
The key difference from P1/P3+ is that in P2, each party \emph{broadcasts}
$\mathbf{W}_i = \lambda_i \mathbf{A}\mathbf{y}_i + \mathbf{m}_i^{(w)}$ to all other parties
(rather than sending it only to a coordinator), requiring mask-hiding to simulate.

\begin{simulator}[$\mathcal{S}$ for Full P2 Protocol]
\label{sim:p2-full}

$\mathcal{S}$ interacts with $\mathcal{F}_{r_0\text{-check}}$ and simulates the adversary's view.
Let $C \subsetneq S$ be the corrupted signing set, $H = S \setminus C$ the honest parties.
The proof assumes $|H| \geq 2$ (required for mask hiding of $\mathbf{W}_i$; see Lemma~\ref{lem:mask-hiding}).

\medskip
\noindent\textbf{Phase 0 (Key DKG Simulation).}
\begin{enumerate}
    \item $\mathcal{S}$ receives corrupted parties' key shares $\{(\mathbf{s}_{1,j}, \mathbf{s}_{2,j})\}_{j \in C}$
    from $\mathcal{A}$ (via the corruption interface).
    \item $\mathcal{S}$ acts as the trusted dealer: samples a degree-$(T-1)$ Shamir polynomial
    $\mathbf{F}_{s_1}$ (resp.\ $\mathbf{F}_{s_2}$) consistent with the corrupted evaluations and computes
    honest-party shares $\{(\mathbf{s}_{1,i}, \mathbf{s}_{2,i})\}_{i \in H}$ directly.
    $\mathcal{S}$ also establishes pairwise PRF seeds $\{\mathsf{seed}_{i,j}\}_{i \neq j}$ and stores them.
    \item \textbf{Justification:} $\mathcal{S}$ generates the full polynomial; honest shares are consistent
    by construction and available to $\mathcal{S}$ for the response simulation below.
\end{enumerate}

\medskip
\noindent\textbf{Phase 1 (Nonce DKG Simulation).}
\begin{enumerate}
    \item For each honest signer $P_h \in H$:
    \begin{enumerate}
        \item Sample a simulated nonce polynomial
        $\tilde{\mathbf{f}}_h(x) = \hat{\tilde{\mathbf{y}}}_h + \sum_{k=1}^{T-1} \tilde{\mathbf{a}}_{h,k}\, x^k$
        where $\hat{\tilde{\mathbf{y}}}_h \getsr [-\lfloor\gamma_1/|S|\rfloor, \lfloor\gamma_1/|S|\rfloor]^{n\ell}$
        and $\tilde{\mathbf{a}}_{h,k} \getsr \mathcal{R}_q^\ell$.
        \item Provide evaluations $\tilde{\mathbf{f}}_h(j)$ to corrupted party $P_j \in C$
        (as the simulated DKG share sent by $P_h$ to $P_j$).
    \end{enumerate}
    \item Compute each honest party's nonce share:
    $\tilde{\mathbf{y}}_i = \sum_{h \in S} \tilde{\mathbf{f}}_h(i)$ for all $i \in H$.
    \item \textbf{Justification:} Evaluations $\tilde{\mathbf{f}}_h(j)$ at $j \neq 0$ are
    $\mathcal{R}_q^\ell$-uniform: the leading coefficient $\tilde{\mathbf{a}}_{h,T-1}$ is
    uniform and $j \not\equiv 0 \pmod{q}$, so the evaluation map is a bijection on $\mathcal{R}_q^\ell$
    (Theorem~\ref{thm:it-privacy}). Each honest $\tilde{\mathbf{y}}_i$ is therefore
    $\mathcal{R}_q^\ell$-uniform, matching the real distribution with SD $= 0$.
\end{enumerate}

\medskip
\noindent\textbf{Phase 2 (Commitment and Broadcast Simulation).}
\begin{enumerate}
    \item For each honest signer $P_i \in H$:
    \begin{enumerate}
        \item Compute $\tilde{\mathbf{w}}_i = \mathbf{A} \cdot \tilde{\mathbf{y}}_i$.
        \item Sample $\tilde{r}_i \getsr \{0,1\}^{256}$;
        program the random oracle: $H(\mathsf{``com''} \| \tilde{\mathbf{y}}_i \| \tilde{\mathbf{w}}_i \| \tilde{r}_i) := \tilde{\mathsf{Com}}_i$
        where $\tilde{\mathsf{Com}}_i \getsr \{0,1\}^{256}$ is fresh.
        Store $(\tilde{\mathbf{y}}_i, \tilde{\mathbf{w}}_i, \tilde{r}_i)$ for blame-protocol openings.
        \item Broadcast $\tilde{\mathsf{Com}}_i$ to $\mathcal{A}$ (simulating $P_i$'s commitment message to all parties).
    \end{enumerate}
    \item \textbf{Justification (commitment):} Commitment hiding (RO preimage resistance) makes
    $\tilde{\mathsf{Com}}_i$ indistinguishable from a real commitment; SD $= 0$ via RO programming.

    \item \emph{(P2-specific: $\mathbf{W}_i$ broadcast.)} For each honest signer $P_i \in H$:
    \begin{enumerate}
        \item Compute the PRF mask:
        $\mathbf{m}_i^{(w)} = \sum_{j \in S,\, j \neq i} \pm\mathsf{PRF}(\mathsf{seed}_{i,j},\, \mathsf{sid} \| \mathsf{``w''})$
        using the pairwise seeds stored in Phase 0.
        \item Compute $\tilde{\mathbf{W}}_i = \lambda_i \mathbf{A} \tilde{\mathbf{y}}_i + \mathbf{m}_i^{(w)}$.
        \item Broadcast $\tilde{\mathbf{W}}_i$ to $\mathcal{A}$.
    \end{enumerate}
    \item \textbf{Justification ($\mathbf{W}_i$ broadcast):}
    Since $|H| \geq 2$, there exists an honest party $k' \neq i$ in $S$;
    $\mathcal{A}$ does not know $\mathsf{seed}_{i,k'}$.
    By PRF security, $\mathsf{PRF}(\mathsf{seed}_{i,k'},\, \mathsf{sid} \| \mathsf{``w''})$ is
    computationally indistinguishable from a fresh uniform element of $\mathcal{R}_q^k$.
    Hence $\tilde{\mathbf{W}}_i \stackrel{c}{\approx} \mathsf{Uniform}(\mathcal{R}_q^k)$
    (Lemma~\ref{lem:mask-hiding}); the simulator uses the actual PRF seeds so its output is
    bit-for-bit identical to the real protocol, with distinguishing advantage $\leq \epsilon_{\mathsf{PRF}}$
    per honest party per session.
\end{enumerate}

\medskip
\noindent\textbf{Phase 3 (Response Simulation).}
\begin{enumerate}
    \item Receive challenge $c$ (derived from the broadcast $(\mathbf{w}_1, \tilde{c})$ or provided
    by $\mathcal{A}$).
    \item For each honest signer $P_i \in H$:
    \begin{enumerate}
        \item Compute response $\mathbf{z}_i = \tilde{\mathbf{y}}_i + c \cdot \mathbf{s}_{1,i}$
        using the simulated nonce share $\tilde{\mathbf{y}}_i$ (Phase 1) and key share $\mathbf{s}_{1,i}$ (Phase 0).
        \item Compute masked key-share $\mathbf{V}_i = \lambda_i c \mathbf{s}_{2,i} + \mathbf{m}_i^{(s2)}$
        using the actual PRF mask $\mathbf{m}_i^{(s2)} = \mathsf{PRF}(\mathsf{seed}_{i,k'}, \mathsf{sid} \| \mathsf{``s2''})$.
        \item Broadcast $(\mathbf{z}_i, \mathbf{V}_i)$ to $\mathcal{A}$.
    \end{enumerate}
    \item \textbf{Justification:} $\mathbf{z}_i = \tilde{\mathbf{y}}_i + c \cdot \mathbf{s}_{1,i}$
    uses simulated nonce shares with the same distribution as real ones (SD $= 0$; Phase 1 justification).
    The $\mathbf{V}_i$ computation uses actual PRF seeds, so the simulation is bit-for-bit identical
    to the real protocol (SD $= 0$).
\end{enumerate}

\medskip
\noindent\textbf{Phase 4 (r0-check Simulation).}
\begin{enumerate}
    \item Forward corrupted parties' SPDZ inputs to $\mathcal{F}_{r_0\text{-check}}$;
    receive $\mathsf{result}$ from the functionality.
    \item Run the r0-check simulator $\mathcal{S}_{r_0}$ from Section~\ref{app:simulator}
    to simulate the SPDZ and edaBit transcripts, using $\mathsf{result}$.
    \item \textbf{Justification:} By UC composition (Theorem~\ref{thm:uc-composition}),
    replacing the real $\Pi_{r_0}$ with ideal $\mathcal{F}_{r_0\text{-check}}$ introduces
    distinguishing advantage $\leq \epsilon_{\mathsf{SPDZ}} + \epsilon_{\mathsf{edaBits}}$
    (Theorem~\ref{thm:p2-uc-full}).
\end{enumerate}

\medskip
\noindent\textbf{Total distinguishing advantage.}
Combining across phases:
\[
\epsilon_{\mathcal{S}} \leq \epsilon_{\mathsf{SPDZ}} + \epsilon_{\mathsf{edaBits}}
  + (N-1) \cdot \epsilon_{\mathsf{PRF}}
\]
where the $(N-1)\epsilon_{\mathsf{PRF}}$ term bounds the $\mathbf{W}_i$ broadcast simulation
across at most $N-1$ honest parties. All other simulation steps achieve SD $= 0$.
\end{simulator}

%==============================================================================
\section{Communication and Round Complexity}
\label{app:complexity}

\begin{theorem}[Complexity of $\Pi_{r_0}$]
\label{thm:complexity}
Protocol $\Pi_{r_0}$ achieves:
\begin{enumerate}
    \item \textbf{Online rounds:} 8 (as specified in Protocol~\ref{prot:r0check})
    \item \textbf{Online communication:} $O(n \cdot m \cdot \log q)$ bits
    \item \textbf{Offline preprocessing:} $O(m)$ edaBits and $O(m \cdot \log q)$ Beaver triples
\end{enumerate}
\end{theorem}

\begin{proof}
\textbf{Rounds:}
\begin{itemize}
    \item Rounds 1-2: SPDZ share exchange and aggregation
    \item Rounds 3-4: Commit-then-open for masked reveal
    \item Round 5: A2B conversion (can be merged with Round 6)
    \item Rounds 6-7: Comparison circuits
    \item Round 8: AND aggregation and output
\end{itemize}

\textbf{Communication:} Each of $m = 1536$ coefficients requires:
\begin{itemize}
    \item SPDZ sharing: $O(n)$ field elements
    \item Masked reveal: $O(n)$ field elements
    \item Binary operations: $O(\log q)$ bits per AND gate, $O(\log q)$ AND gates per coefficient
\end{itemize}

Total: $O(n \cdot m \cdot \log q) \approx n \cdot 1536 \cdot 23 \approx 35000 \cdot n$ bits per party.

\textbf{Preprocessing:}
\begin{itemize}
    \item $m$ edaBits (one per coefficient)
    \item $O(m \cdot \log q)$ Beaver triples for A2B and comparison circuits
\end{itemize}
\end{proof}

\begin{table}[h]
\centering
\begin{tabular}{lcc}
\toprule
\textbf{Component} & \textbf{Online Rounds} & \textbf{Communication/Party} \\
\midrule
SPDZ aggregation & 2 & $O(m \log q)$ bits \\
Masked reveal & 2 & $O(m \log q)$ bits \\
A2B conversion & 1 & $O(m \log q)$ bits \\
Comparison & 2 & $O(m \log \gamma_2)$ bits \\
AND tree & 1 & $O(m)$ bits \\
\midrule
\textbf{Total} & \textbf{8} & $\approx 30$ KB \\
\bottomrule
\end{tabular}
\caption{Round and communication breakdown for Profile P2 r0-check (8-round version).
  The AND tree requires $m-1$ Beaver-triple openings, giving $O(m)$ bits;
  the $O(\log m)$ figure denotes circuit \emph{depth}, not communication.}
\label{tab:p2-complexity}
\end{table}

%==============================================================================
\section{Optimized P2: 5-Round Protocol UC Security}
\label{app:p2-5round}

We present the UC security proof for the optimized 5-round P2 protocol. The key modifications are:
\begin{enumerate}
    \item \textbf{Combiner-mediated commit-then-open}:
    Rounds 3-4 merged into Round 2
    \item \textbf{Constant-depth comparison}: Rounds 6-7 merged into Round 4
\end{enumerate}

\subsection{Modified Ideal Functionality}

\begin{functionality}[$\mathcal{F}_{r_0\text{-check}}^{\mathsf{5rnd}}$ -- Secure r0-Check with Combiner]
\label{func:r0check-5rnd}
Extends $\mathcal{F}_{r_0\text{-check}}$ (Functionality~\ref{func:r0check}) with combiner role.

\textbf{Additional State:}
\begin{itemize}
    \item $\mathsf{combiner\_id}$: Identity of the designated combiner
    \item $\mathsf{collected}[\cdot]$: Messages collected by combiner before release
\end{itemize}

\textbf{Combiner Registration:}

On receiving $(\mathsf{SET\_COMBINER}, \mathsf{sid}, \mathsf{cid})$ from $\mathcal{S}$:
\begin{enumerate}
    \item Set $\mathsf{combiner\_id} \gets \mathsf{cid}$
    \item Notify all parties of combiner identity
\end{enumerate}

\textbf{Combiner-Mediated Input Phase:}

On receiving $(\mathsf{COMBINER\_INPUT}, \mathsf{sid}, i, \mathsf{Com}_i, [\mathbf{w}']_i, \mathsf{rand}_i)$ from party $P_i$:
\begin{enumerate}
    \item Verify $\mathsf{Com}_i = H([\mathbf{w}']_i \| \mathsf{sid} \| i \| \mathsf{rand}_i)$
    \item Store $\mathsf{collected}[i] \gets ([\mathbf{w}']_i, \mathsf{rand}_i)$
    \item Send $(\mathsf{COLLECTED}, \mathsf{sid}, i)$ to combiner
    \item If $|\mathsf{collected}| = n$:
    \begin{enumerate}
        \item Send $(\mathsf{ALL\_COLLECTED}, \mathsf{sid})$ to combiner
    \end{enumerate}
\end{enumerate}

\textbf{Simultaneous Release:}

On receiving $(\mathsf{RELEASE}, \mathsf{sid})$ from combiner:
\begin{enumerate}
    \item If $|\mathsf{collected}| < n$: return $\bot$
    \item For each $i$: send $\{(j, [\mathbf{w}']_j, \mathsf{rand}_j)\}_{j \in [n]}$ to party $P_i$
    \item Continue with standard computation as in $\mathcal{F}_{r_0\text{-check}}$
\end{enumerate}

\textbf{Timeout (Corrupted Combiner):}

Upon sending $(\mathsf{ALL\_COLLECTED}, \mathsf{sid})$ to the combiner, the functionality starts a timer of $\Delta_{\mathsf{timeout}}$ rounds (a global synchrony parameter). If $(\mathsf{RELEASE}, \mathsf{sid})$ is not received within $\Delta_{\mathsf{timeout}}$ rounds:
\begin{enumerate}
    \item Send $(\mathsf{ABORT}, \mathsf{sid})$ to all parties $P_i$
    \item Discard all collected state
\end{enumerate}
\noindent\emph{Remark:} This timeout prevents a corrupted combiner from indefinitely delaying the release. In Profiles P1 and P3+, the combiner role is enforced by a TEE or 2PC, so the timeout is triggered only if the trusted component is unavailable (which is detected). In Profile P2, all parties are online during the signing window, so $\Delta_{\mathsf{timeout}}$ is a protocol parameter (e.g., 30 seconds).

\textbf{Leakage:} The combiner learns:
\begin{itemize}
    \item Timing of each party's message submission
    \item Order of message arrivals
\end{itemize}
The combiner does \emph{not} learn the content of messages before simultaneous release.
\end{functionality}

\subsection{Modified Real Protocol}

\begin{protocol}[$\Pi_{r_0}^{\mathsf{5rnd}}$ -- Optimized 5-Round r0-Check]
\label{prot:r0check-5rnd}

\textbf{Round Structure:}
\begin{center}
\begin{tabular}{cl}
\toprule
\textbf{Round} & \textbf{Operation} \\
\midrule
1 & Share exchange + aggregation \\
2 & Combiner-mediated masked reveal \\
3 & A2B conversion via edaBits \\
4 & Parallel comparison + AND aggregation \\
5 & Output \\
\bottomrule
\end{tabular}
\end{center}

\underline{Round 1: Share Exchange and Aggregation}

(Unchanged from $\Pi_{r_0}$, Protocol~\ref{prot:r0check})

\underline{Round 2: Combiner-Mediated Masked Reveal}

For each coefficient $j \in [m]$:
\begin{enumerate}
    \item Each party $P_i$ computes:
    \begin{itemize}
        \item $[\mathsf{masked}_j]_i = [\mathbf{w}'_j]_i + [r_j]_i$
        \item $\mathsf{Com}_{i,j} = H([\mathsf{masked}_j]_i \| \mathsf{sid} \| j \| i \| \mathsf{rand}_{i,j})$
    \end{itemize}
    \item $P_i$ sends $(\mathsf{Com}_{i,j}, [\mathsf{masked}_j]_i, \mathsf{rand}_{i,j})$ to combiner
    \item Combiner waits until all $n$ parties' messages received
    \item Combiner broadcasts all openings $\{[\mathsf{masked}_j]_i, \mathsf{rand}_{i,j}\}_{i,j}$ simultaneously
    \item Each party verifies all commitments; abort if any mismatch
    \item Reconstruct $\mathsf{masked}_j = \sum_i [\mathsf{masked}_j]_i \bmod q$
\end{enumerate}

\underline{Round 3: A2B Conversion}

(Unchanged from $\Pi_{r_0}$)

\underline{Round 4: Parallel Comparison + AND Aggregation}

\textbf{Layer 1 (Parallel Comparisons):}
For each $j \in [m]$ \emph{in parallel}:
\begin{enumerate}
    \item Compute absolute value: $\langle |\mathbf{w}'_j| \rangle_2$ from $\langle \mathbf{w}'_j \rangle_2$
    \item Evaluate comparison: $\langle \mathsf{pass}_j \rangle_2 = (|\langle \mathbf{w}'_j \rangle_2| < \mathsf{bound})$
\end{enumerate}

\textbf{Layer 2 (AND Tree):}
\begin{enumerate}
    \item Compute $\langle \mathsf{result} \rangle_2 = \bigwedge_{j=1}^m \langle \mathsf{pass}_j \rangle_2$
    \item Use Beaver triples for all AND gates simultaneously
    \item Open all $(d, e)$ values in single round
\end{enumerate}

\underline{Round 5: Output}

(Unchanged from $\Pi_{r_0}$)
\end{protocol}

\subsection{Extended Simulator for 5-Round Protocol}

\begin{simulator}[$\mathcal{S}^{\mathsf{5rnd}}$ for $\Pi_{r_0}^{\mathsf{5rnd}}$]
\label{sim:r0check-5rnd}

Extends Simulator~\ref{sim:r0check} with combiner simulation.

\textbf{Simulating Combiner (if honest):}
\begin{enumerate}
    \item When simulating honest parties, send their messages to simulated combiner
    \item Combiner collects all messages (both honest and corrupted)
    \item On receiving $(\mathsf{ALL\_COLLECTED})$ from internal state:
    \begin{enumerate}
        \item Release all openings to adversary simultaneously
    \end{enumerate}
\end{enumerate}

\textbf{Simulating Combiner (if corrupted):}
\begin{enumerate}
    \item Adversary $\mathcal{A}$ controls combiner
    \item $\mathcal{S}$ sends simulated honest parties' messages to $\mathcal{A}$
    \item Wait for $\mathcal{A}$ to release (or abort)
    \item If $\mathcal{A}$ modifies any value: SPDZ MAC check will fail in subsequent rounds
    \item Simulate abort if MAC check fails
\end{enumerate}

\textbf{Simulating Round 4 (Parallel Comparison):}
\begin{enumerate}
    \item All $m$ comparison circuits evaluated simultaneously
    \item For each AND gate: sample $(d, e) \getsr \{0,1\}^2$ uniformly
    \item Release all $(d, e)$ values at once
    \item By Lemma~\ref{lem:beaver-security}, this is indistinguishable from real
\end{enumerate}
\end{simulator}

\subsection{Modified Hybrid Argument}

We extend the hybrid sequence with a new hybrid for combiner simulation.

\begin{definition}[Extended Hybrids for 5-Round Protocol]
\label{def:hybrids-5rnd}
\begin{itemize}
    \item $\mathsf{H}_0$: Real execution of $\Pi_{r_0}^{\mathsf{5rnd}}$
    \item $\mathsf{H}_{0.5}$: Replace combiner with ideal $\mathcal{F}_{\mathsf{Combiner}}$
    \item $\mathsf{H}_1$ -- $\mathsf{H}_6$: As in original hybrid argument
\end{itemize}
\end{definition}

\begin{lemma}[Combiner Hybrid]
\label{lem:combiner-hybrid}
$\Pr[\mathsf{H}_0 = 1] = \Pr[\mathsf{H}_{0.5} = 1]$ (perfect indistinguishability)
\end{lemma}

\begin{proof}
The combiner performs \emph{no cryptographic operations}: it collects incoming messages and forwards them once all have arrived (or times out). In particular, it does not compute any MACs, encryptions, or protocol-specific transformations on the data. Its behavior is a deterministic function of its input messages and the arrival timing. We show the adversary's view is \emph{bytewise identical} in $\mathsf{H}_0$ and $\mathsf{H}_{0.5}$:

\textbf{Honest Combiner Case:}
\begin{itemize}
    \item In $\mathsf{H}_0$: The combiner collects $(\mathsf{Com}_i, \mathsf{value}_i, \mathsf{rand}_i)$ tuples from all $N$ parties and releases them simultaneously only after all are received.
    \item In $\mathsf{H}_{0.5}$: $\mathcal{F}_{\mathsf{Combiner}}$ executes the same algorithm: collect-then-release.
    \item The adversary's \emph{incoming} messages (the individual party tuples) are produced by the real parties' code in $\mathsf{H}_0$ and by the simulator (which replicates the same code) in $\mathsf{H}_{0.5}$---producing identical bit-strings.
    \item The adversary's \emph{outgoing} view (released tuples) is the same set of collected messages, released at the same logical time.
    \item Since the combiner is honest, the adversary's view of the combiner's internal state is empty (it receives only the released outputs). The released outputs are identical. Hence the adversary's view is \emph{bytewise equal} in both hybrids.
\end{itemize}

\textbf{Corrupted Combiner Case:}
\begin{itemize}
    \item The adversary controls the combiner directly and receives honest parties' messages in both hybrids via identical delivery.
    \item Any deviation by the combiner (selective release, modification) triggers a detectable abort:
    \begin{enumerate}
        \item \emph{Value modification}: SPDZ MAC check at the recipient parties fails with probability $1 - 1/q$ per coefficient; honest parties abort and the functionality's abort interface captures this.
        \item \emph{Selective withholding}: honest parties' timeout logic triggers an abort after $\mathsf{deadline}$, identical in both $\mathsf{H}_0$ and $\mathsf{H}_{0.5}$.
    \end{enumerate}
    \item In $\mathsf{H}_{0.5}$, $\mathcal{F}_{\mathsf{Combiner}}$ replicates the same abort interface. Since the honest parties' messages delivered to the corrupted combiner are generated by the same real-protocol code in both hybrids (the combiner's corruption does not affect honest-party message generation), the adversary's view is again bytewise identical.
\end{itemize}

Since the combiner performs no cryptographic operations and its behavior is a deterministic routing function, the adversary's view in $\mathsf{H}_0$ is \emph{perfectly equal} to its view in $\mathsf{H}_{0.5}$ (not merely computationally indistinguishable). Thus:
\[
\Pr[\mathsf{H}_0 = 1] = \Pr[\mathsf{H}_{0.5} = 1]
\]
\end{proof}

\begin{theorem}[UC Security of 5-Round P2]
\label{thm:p2-5rnd-uc}
Protocol $\Pi_{r_0}^{\mathsf{5rnd}}$ UC-realizes $\mathcal{F}_{r_0\text{-check}}^{\mathsf{5rnd}}$ with:
\[
\epsilon = \epsilon_{\mathsf{SPDZ}} + \epsilon_{\mathsf{edaBits}} + \epsilon_{\mathsf{combiner}} = \epsilon_{\mathsf{SPDZ}} + \epsilon_{\mathsf{edaBits}} < 2^{-53}
\]
(with 64-bit MAC field; minimal configuration gives $\epsilon < 2^{-11}$).
\end{theorem}

\begin{proof}
By Lemma~\ref{lem:combiner-hybrid}, $\epsilon_{\mathsf{combiner}} = 0$. The remaining hybrids proceed exactly as in Theorem~\ref{thm:p2-uc-full}. Total security loss:
\[
\epsilon = 0 + \epsilon_{\mathsf{SPDZ}} + \epsilon_{\mathsf{edaBits}} + 0 + 0 + 0 + 0 < 2^{-53}
\]
\end{proof}

\subsection{5-Round Complexity Analysis}

\begin{theorem}[Complexity of $\Pi_{r_0}^{\mathsf{5rnd}}$]
\label{thm:complexity-5rnd}
The optimized protocol achieves:
\begin{enumerate}
    \item \textbf{Online rounds:} 5
    \item \textbf{Online communication:} $O(n \cdot m \cdot \log q)$ bits (unchanged)
    \item \textbf{Offline preprocessing:} Same as 8-round version
\end{enumerate}
\end{theorem}

\begin{proof}
\textbf{Round reduction analysis:}
\begin{itemize}
    \item Rounds 1-2 $\to$ Round 1: Merge share exchange with aggregation
    (linear operations)
    \item Rounds 3-4 $\to$ Round 2: Combiner-mediated simultaneous release
    \item Round 5 $\to$ Round 3: A2B conversion (unchanged)
    \item Rounds 6-7-8 $\to$ Rounds 4-5: Parallel comparison circuits and AND aggregation execute
    simultaneously in Round 4 (saving 1 round); output is broadcast in Round 5
    (total: 3 original rounds $\to$ 2 new rounds)
\end{itemize}

\textbf{Communication unchanged:} Same data transmitted, just fewer rounds.

\textbf{Preprocessing unchanged:} Same edaBits and Beaver triples required.
\end{proof}

\begin{table}[h]
\centering
\begin{tabular}{lcc}
\toprule
\textbf{Component} & \textbf{8-rnd} & \textbf{5-rnd} \\
\midrule
Share exchange + aggregation & 2 & 1 \\
Masked reveal & 2 & 1 \\
A2B conversion & 1 & 1 \\
Comparison + AND tree & 3 & 2 \\
\midrule
\textbf{Total} & \textbf{8} & \textbf{5} \\
\bottomrule
\end{tabular}
\caption{Round comparison between original and optimized P2 protocols.}
\label{tab:p2-round-comparison}
\end{table}

%==============================================================================
\section{Complete UC Security Proof for Profile P3+}
\label{app:p3plus-uc}

We provide the complete UC security proof for the P3+ semi-asynchronous protocol. This extends the P3 proof sketch (Section~\ref{sec:profile-p3}) to a full proof with explicit handling of the pre-signing and semi-async model.

\subsection{Semi-Asynchronous Security Model}

\begin{definition}[Semi-Asynchronous Timing Model]
\label{def:semi-async}
The semi-async model consists of four timing phases:
\begin{enumerate}
    \item \textbf{Pre-signing Window} $[t_0, t_1]$: Fully asynchronous. Signers precompute at any time.
    \item \textbf{Challenge Broadcast} $t_2$: Combiner derives and broadcasts challenge.
    \item \textbf{Response Window} $[t_2, t_2 + \Delta]$: Semi-synchronous. Signers respond within bounded time $\Delta$.
    \item \textbf{2PC Execution} $[t_3, t_4]$: Synchronous between CPs only.
\end{enumerate}
\end{definition}

\begin{definition}[Semi-Async Adversary]
\label{def:semi-async-adv}
We adopt the synchronous UC model~\cite{Katz07} where all parties share a global clock. The adversary $\mathcal{A}$ can:
\begin{itemize}
    \item Corrupt up to $T-1$ signers (static corruption)
    \item Corrupt at most one of $\{\mathsf{CP}_1, \mathsf{CP}_2\}$
    \item Control message scheduling within each phase, but \emph{cannot} advance the clock past the signer response deadline $\Delta$
    \item \emph{Cannot} control the global clock or extend $\Delta$
\end{itemize}
Theorem~\ref{thm:p3plus-uc} applies to any protocol meeting this interface.
\end{definition}

\subsection{Ideal Functionality for P3+}

\begin{definition}[Ideal Functionality $\mathcal{F}_{r_0}^{\mathsf{2PC}}$ (UC Proof Compact Form)]
\label{def:f-r0-2pc-uc}
The functionality $\mathcal{F}_{r_0}^{\mathsf{2PC}}$ captures a two-party computation of the r0-check predicate. It receives private input $\mathbf{S}_j \in \mathbb{Z}_q^m$ from each computation party $\mathsf{CP}_j$ ($j \in \{1,2\}$), reconstructs $c\mathbf{s}_2 = \mathbf{S}_1 + \mathbf{S}_2 \bmod q$ (private inputs satisfy this additive split; in the concrete protocol, $\mathbf{S}_j = \sum_{i \text{ assigned to } \mathsf{CP}_j} \lambda_i \cdot c\mathbf{s}_{2,i}$), computes $\mathbf{w}' = \mathbf{w} - c\mathbf{s}_2$ (using the public $\mathbf{w}$ agreed upon during challenge derivation), evaluates $\mathsf{result} = \bigwedge_j (|\mathsf{LowBits}(\mathbf{w}'_j, 2\gamma_2)|_\infty < \gamma_2 - \beta)$, and outputs only $\mathsf{result}$ (1 bit) and the hint $\mathbf{h} = \mathsf{MakeHint}(\mathbf{w}', \mathbf{w}, 2\gamma_2)$ to the combiner. Neither $c\mathbf{s}_2$ nor $\mathbf{w}'$ is revealed outside $\mathcal{F}_{r_0}^{\mathsf{2PC}}$.
\end{definition}

\begin{functionality}[$\mathcal{F}_{r_0}^{\mathsf{semi\text{-}async}}$ -- Semi-Async r0-Check]
\label{func:r0-semi-async}
Extends $\mathcal{F}_{r_0}^{\mathsf{2PC}}$ (Definition~\ref{def:f-r0-2pc-uc}) with pre-signing and timing.

\textbf{Pre-signing Phase:}

On receiving $(\mathsf{PRESIGN}, \mathsf{sid}, i, \mathsf{set\_id}, \mathsf{Com}_i)$ from signer $P_i$:
\begin{enumerate}
    \item Store $\mathsf{presigning}[i][\mathsf{set\_id}] \gets \mathsf{Com}_i$
    \item Send $(\mathsf{PRESIGNED}, \mathsf{sid}, i, \mathsf{set\_id})$ to $\mathcal{S}$
\end{enumerate}

\textbf{Challenge Phase:}

On receiving $(\mathsf{CHALLENGE}, \mathsf{sid}, \mathsf{signers}, \mathsf{set\_map}, c)$ from combiner:
\begin{enumerate}
    \item Verify all signers have pre-signing sets available
    \item Set $\mathsf{deadline} \gets \mathsf{now} + \Delta$
    \item Broadcast $c$ to all signers in $\mathsf{signers}$
    \item Send $(\mathsf{CHALLENGED}, \mathsf{sid})$ to $\mathcal{S}$
\end{enumerate}

\textbf{Response Phase:}

On receiving $(\mathsf{RESPOND}, \mathsf{sid}, i, \mathsf{response}_i)$ from signer $P_i$:
\begin{enumerate}
    \item If $\mathsf{now} > \mathsf{deadline}$: reject (too late)
    \item If $P_i \notin \mathsf{signers}$: reject
    \item Store $\mathsf{responses}[i] \gets \mathsf{response}_i$
    \item If $|\mathsf{responses}| \geq T$: trigger 2PC phase
\end{enumerate}

\textbf{2PC Phase:}

On $T$ responses collected:
\begin{enumerate}
    \item Aggregate inputs for CPs
    \item Execute $\mathcal{F}_{r_0}^{\mathsf{2PC}}$ between $\mathsf{CP}_1, \mathsf{CP}_2$
    \item Return result to combiner
\end{enumerate}
\end{functionality}

\subsection{Real Protocol $\Pi_{r_0}^{\mathsf{P3+}}$}

\begin{protocol}[$\Pi_{r_0}^{\mathsf{P3+}}$ -- P3+ Semi-Async Protocol]
\label{prot:p3plus}

\textbf{Offline Phase: Pre-signing}

Each signer $P_i$:
\begin{enumerate}
    \item Generate $B$ pre-signing sets:
    \begin{enumerate}
        \item Sample $\mathsf{set\_id} \getsr \{0,1\}^{256}$ (pre-signing set identifier)
        \item Sample nonce $\mathbf{y}_i$ with reduced range
        \item Compute $\mathbf{w}_i = \mathbf{A} \cdot \mathbf{y}_i$
        \item Sample commitment randomness $r_i \getsr \{0,1\}^{256}$
        \item Compute commitment $\mathsf{Com}_i = H(\texttt{"com"} \| \mathbf{y}_i \| \mathbf{w}_i \| r_i)$
    \end{enumerate}
    \item Send $(\mathsf{set\_id}, \mathbf{w}_i, \mathsf{Com}_i)$ to combiner \Comment{$r_i$ is not sent; kept secret by signer}
\end{enumerate}

\textbf{Offline Phase: GC Precomputation}

$\mathsf{CP}_1$:
\begin{enumerate}
    \item Build r0-check circuit $C$
    \item Generate $B$ garbled circuits $\{\tilde{C}_k\}_{k=1}^B$
    \item Prepare OT setup for each circuit
\end{enumerate}

\textbf{Online Phase 1: Challenge Derivation}

Combiner:
\begin{enumerate}
    \item Receive signing request for message $\mu$
    \item Select $T$ signers with available pre-signing sets
    \item Aggregate $\mathbf{w} = \sum_i \lambda_i \cdot \mathbf{w}_i$ \Comment{Lagrange reconstruction}
    \item Compute $\mathbf{w}_1 = \mathsf{HighBits}(\mathbf{w}, 2\gamma_2)$
    \item Derive $c = \mathsf{SampleChallenge}(H(\mathsf{tr} \| \mu \| \mathbf{w}_1))$
    \item Broadcast $(c, \mathsf{set\_map})$ to signers
\end{enumerate}

\textbf{Online Phase 2: Signer Response}

Each signer $P_i$ (within $\Delta$):
\begin{enumerate}
    \item Retrieve nonce DKG share $\mathbf{y}_i$ and commitment $\mathbf{w}_i$ for assigned $\mathsf{set\_id}$
    \item Compute:
    \begin{align*}
        \mathbf{z}_i &= \mathbf{y}_i + c \cdot \mathbf{s}_{1,i} \\
        \mathbf{V}_i &= \lambda_i c \mathbf{s}_{2,i} + \mathbf{m}_i^{(s2)}
    \end{align*}
    where $\lambda_i$ is the Lagrange coefficient for party $i$ in signing set $S$
    \item Send $(\mathbf{z}_i, \mathbf{V}_i)$ to combiner
    \item Mark $\mathsf{set\_id}$ as consumed
\end{enumerate}

\textbf{Online Phase 3: 2PC Execution}

Combiner:
\begin{enumerate}
    \item Wait for $T$ responses (or timeout)
    \item Aggregate response: $\mathbf{z} = \sum_{i \in S} \lambda_i \cdot \mathbf{z}_i$ (Lagrange reconstruction)
    \item Split $c\mathbf{s}_2$ contributions: $\mathbf{S}_1$ for $\mathsf{CP}_1$, $\mathbf{S}_2$ for $\mathsf{CP}_2$
    \item Send $(\mathbf{w}, \mathbf{S}_1)$ to $\mathsf{CP}_1$, $(\mathbf{S}_2)$ to $\mathsf{CP}_2$
\end{enumerate}

$\mathsf{CP}_1$ (Garbler):
\begin{enumerate}
    \item Select precomputed garbled circuit $\tilde{C}$
    \item Select input wire labels for $(\mathbf{w}, \mathbf{S}_1)$
    \item Send $(\tilde{C}, \mathsf{labels}_1)$ to $\mathsf{CP}_2$
\end{enumerate}

$\mathsf{CP}_2$ (Evaluator):
\begin{enumerate}
    \item Execute OT to obtain labels for $\mathbf{S}_2$
    \item Evaluate $\tilde{C}$ to obtain $\mathsf{result}$
    \item Send $\mathsf{result}$ to both CPs
\end{enumerate}

\textbf{Output:}
\begin{enumerate}
    \item If $\mathsf{result} = 1$: The 2PC circuit outputs the hint
    \[
    \mathbf{h} = \mathsf{MakeHint}(-c\mathbf{t}_0,\; \mathbf{A}\mathbf{z} - c\mathbf{t}_1 \cdot 2^d,\; 2\gamma_2)
    \]
    (computed inside the circuit using $c\mathbf{s}_2 = \mathbf{S}_1 + \mathbf{S}_2$);
    the combiner outputs signature $\sigma = (\tilde{c}, \mathbf{z}, \mathbf{h})$
    \item If $\mathsf{result} = 0$: Retry with fresh pre-signing sets
\end{enumerate}
\end{protocol}

\subsection{Simulator Construction for P3+}

\begin{simulator}[$\mathcal{S}^{\mathsf{P3+}}$ for $\Pi_{r_0}^{\mathsf{P3+}}$]
\label{sim:p3plus}

$\mathcal{S}$ interacts with $\mathcal{F}_{r_0}^{\mathsf{semi\text{-}async}}$ and simulates the adversary's view.

\textbf{Simulating Key DKG:}
\begin{enumerate}
    \item $\mathcal{S}$ receives the corrupted parties' key shares $\{(\mathbf{s}_{1,j}, \mathbf{s}_{2,j})\}_{j \in C}$ from the adversary (via corruption; in P3+ the key DKG is modeled as a trusted-dealer setup phase, so the simulator acts as the trusted dealer and generates all shares directly).
    \item $\mathcal{S}$ runs the ideal key-generation functionality on behalf of all honest parties: $\mathcal{S}$ samples a fresh degree-$(T-1)$ Shamir polynomial for each component of $(\mathbf{s}_1, \mathbf{s}_2)$, consistent with the corrupted-party shares received in step 1, and computes the honest-party evaluations $\{(\mathbf{s}_{1,i}, \mathbf{s}_{2,i})\}_{i \notin C}$ directly from these polynomials.
    \item $\mathcal{S}$ stores all shares $\{(\mathbf{s}_{1,i}, \mathbf{s}_{2,i})\}_{i=1}^N$ for use in the response simulation.
    \item \textbf{Justification:} The simulator acts as the DKG simulator, generating the full Shamir polynomials during key setup. It does \emph{not} infer honest shares post-hoc from $\mathbf{t}_1$ (which is a lossy commitment and does not uniquely determine $(\mathbf{s}_1,\mathbf{s}_2)$). Rather, the polynomials are fixed at key-generation time by $\mathcal{S}$'s own choices, and all evaluations $\{(\mathbf{s}_{1,i}, \mathbf{s}_{2,i})\}_{i=1}^N$ are therefore directly available to $\mathcal{S}$. When $|C| < T-1$, there are multiple polynomials consistent with the corrupted shares; $\mathcal{S}$ uses the one it sampled.
\end{enumerate}

\textbf{Simulating Pre-signing (Honest Signers):}
\begin{enumerate}
    \item For each honest signer $P_i$:
    \begin{enumerate}
        \item Sample $\tilde{\mathbf{y}}_i \getsr \mathcal{R}_q^\ell$ (simulated nonce share)
        \item Compute $\tilde{\mathbf{w}}_i = \mathbf{A} \cdot \tilde{\mathbf{y}}_i$ (same computation as real protocol)
        \item Sample $\tilde{r}_i \getsr \{0,1\}^{256}$ (commitment randomness); program the random oracle so that $H(\mathsf{``com''} \| \tilde{\mathbf{y}}_i \| \tilde{\mathbf{w}}_i \| \tilde{r}_i) := \tilde{\mathsf{Com}}_i$ where $\tilde{\mathsf{Com}}_i \getsr \{0,1\}^{256}$; store $(\tilde{\mathbf{y}}_i, \tilde{\mathbf{w}}_i, \tilde{r}_i)$ for potential blame-protocol opening
        \item Send $(\mathsf{set\_id}, \tilde{\mathbf{w}}_i, \tilde{\mathsf{Com}}_i)$ to adversary
    \end{enumerate}
    \item \textbf{Justification:} Commitment hides $\mathbf{y}_i$ (random oracle). Simulated $\tilde{\mathbf{w}}_i = \mathbf{A} \cdot \tilde{\mathbf{y}}_i$ has the same distribution as the real $\mathbf{w}_i = \mathbf{A} \cdot \mathbf{y}_i$: the real $\mathbf{y}_i = F(i)$ for $i \neq 0$ is $\mathcal{R}_q^\ell$-uniform because the degree-$1$-through-$(T\!-\!1)$ coefficients of $F$ are sampled uniformly (Theorem~\ref{thm:it-privacy-formal}), making each non-zero evaluation uniform by linearity; thus $\tilde{\mathbf{y}}_i \getsr \mathcal{R}_q^\ell$ matches this distribution exactly (SD~$= 0$). The stored $\tilde{r}_i$ enables blame-protocol commitment opening.
\end{enumerate}

\textbf{Simulating Response Phase (Honest Signers):}
\begin{enumerate}
    \item On challenge $c$ from combiner:
    \begin{enumerate}
        \item For each honest signer $P_i$:
        \item Compute $\mathbf{z}_i = \tilde{\mathbf{y}}_i + c \cdot \mathbf{s}_{1,i}$ using the nonce shares $\tilde{\mathbf{y}}_i$ from the simulated pre-signing phase and key shares $\mathbf{s}_{1,i}$ from the simulated key DKG
        \item Compute $\mathbf{V}_i = \lambda_i c \mathbf{s}_{2,i} + \mathbf{m}_i^{(s2)}$ using the simulated key shares and PRF masks
        \item Send $(\mathbf{z}_i, \mathbf{V}_i)$ to adversary (via combiner)
    \end{enumerate}
    \item \textbf{Justification:} The simulator computes responses using the same deterministic function as the real protocol, applied to the simulated DKG values. Since the simulated DKG produces values with identical distribution to the real DKG (Lemma~\ref{lem:presigning-sim}), the resulting $\mathbf{z}_i = \tilde{\mathbf{y}}_i + c \cdot \mathbf{s}_{1,i}$ is perfectly indistinguishable from the real value (SD $= 0$, since $\tilde{\mathbf{y}}_i$ and $\mathbf{y}_i$ have identical distribution). The masked key-share component $\mathbf{V}_i = \lambda_i c\mathbf{s}_{2,i} + \mathbf{m}_i^{(s2)}$ uses PRF-derived masks $\mathbf{m}_i^{(s2)}$; the simulator uses the \emph{actual} pairwise PRF seeds established during key setup simulation, computing the same mask values $\mathbf{m}_i^{(s2)} = \mathsf{PRF}(\mathsf{seed}_{h,i}, \mathsf{sid})$ as the real protocol. No PRF replacement is needed; the simulation is bit-for-bit identical to the real protocol for this component. The overall gap is therefore $0$.
\end{enumerate}

\textbf{Simulating 2PC (Corrupt $\mathsf{CP}_1$):}
\begin{enumerate}
    \item Adversary controls garbler $\mathsf{CP}_1$
    \item $\mathcal{S}$ extracts $\mathbf{S}_1$ from adversary's garbled circuit inputs (using garbled-circuit extractability~\cite{LP07}); $\mathbf{w}$ is already public (broadcast by the combiner)
    \item $\mathcal{S}$ sends $\mathbf{S}_1$ to $\mathcal{F}_{r_0}^{\mathsf{semi\text{-}async}}$
    \item Receives $\mathsf{result}$ from functionality
    \item Simulates honest $\mathsf{CP}_2$:
    \begin{enumerate}
        \item Execute OT protocol with adversary
        \item Return $\mathsf{result}$ to adversary
    \end{enumerate}
    \item \textbf{Justification:} Garbled circuit extractability~\cite{LP07} allows extraction of garbler's inputs.
\end{enumerate}

\textbf{Simulating 2PC (Corrupt $\mathsf{CP}_2$):}
\begin{enumerate}
    \item Adversary controls evaluator $\mathsf{CP}_2$
    \item $\mathcal{S}$ extracts $\mathbf{S}_2$ from adversary's OT queries
    \item $\mathcal{S}$ sends $\mathbf{S}_2$ to $\mathcal{F}_{r_0}^{\mathsf{semi\text{-}async}}$
    \item Receives $\mathsf{result}$ from functionality
    \item Simulates honest $\mathsf{CP}_1$:
    \begin{enumerate}
        \item Generate a simulated garbled circuit $\tilde{C}$ together with all input wire labels for $(\mathbf{w}, \mathbf{S}_1)$ (garbler's inputs) and the wire labels for $\mathbf{S}_2$ selected by $\mathsf{CP}_2$'s OT queries, consistent with the output $\mathsf{result}$, using the BHR12 simulator for the full two-input garbled circuit~\cite{BHR12}.
    \end{enumerate}
    \item \textbf{Justification:} Garbled circuit simulatability (BHR12 privacy definition): given the output $f(\mathbf{w}, \mathbf{S}_1, \mathbf{S}_2)$ and the evaluator's input labels (for $\mathbf{S}_2$, known from OT), the BHR12 simulator produces a simulated circuit $\tilde{C}$ plus all evaluator-side input labels computationally indistinguishable from the real garbled circuit. The garbler's (CP1's) input wire labels for $(\mathbf{w}, \mathbf{S}_1)$ are included in the simulated output; these are internally consistent with $\tilde{C}$ by construction.
\end{enumerate}

\textbf{Simulating 2PC (Both CPs Honest):}
When both $\mathsf{CP}_1$ and $\mathsf{CP}_2$ are honest, the adversary observes only:
\begin{itemize}
    \item Pre-signing commitments $\mathsf{Com}_i$ (simulated uniformly above, SD $= 0$).
    \item The broadcast challenge $c$ (determined by the ideal functionality).
    \item Responses $(\mathbf{z}_i, \mathbf{V}_i)$ for honest $i \notin C$ (simulated above: $\mathbf{z}_i$ with SD $= 0$; $\mathbf{V}_i$ with SD $= 0$ since the simulator uses the actual PRF seeds established during key setup, producing bit-for-bit identical mask values).
    \item The final result bit (provided by $\mathcal{F}_{r_0}^{\mathsf{semi\text{-}async}}$).
\end{itemize}
The 2PC OT and GC transcripts are exchanged solely between the two honest CPs and are never delivered to the adversary. $\mathcal{S}$ need not simulate these messages; the hybrid sequence $\mathsf{H}_0 \to \mathsf{H}_1$ (replace OT) $\to \mathsf{H}_2$ (replace GC) captures the honest-CP case via ideal-functionality substitution.

\textbf{Simulating Timeout / Abort:}
\begin{enumerate}
    \item If the adversary delays honest signers' responses past the deadline $\mathsf{deadline}$ (e.g., by withholding the challenge or blocking network messages):
    \begin{enumerate}
        \item $\mathcal{S}$ detects no challenge delivered to honest signers within $\mathsf{deadline} - \Delta$ rounds.
        \item $\mathcal{S}$ sends $(\mathsf{ABORT}, \mathsf{sid})$ to $\mathcal{F}_{r_0}^{\mathsf{semi\text{-}async}}$.
        \item $\mathcal{S}$ delivers $(\mathsf{abort}, \mathsf{sid})$ to all honest parties in the simulated execution.
    \end{enumerate}
    \item \textbf{Justification:} The functionality enforces ``if $\mathsf{now} > \mathsf{deadline}$: reject'' (see Response Phase, condition 1). When the adversary induces timeout, both ideal and simulated executions abort identically; the simulated abort message matches the real abort, giving indistinguishability.
\end{enumerate}
\end{simulator}

\subsection{Hybrid Argument for P3+}

\begin{definition}[P3+ Hybrids]
\label{def:hybrids-p3plus}
\begin{itemize}
    \item $\mathsf{H}_0$: Real execution of $\Pi_{r_0}^{\mathsf{P3+}}$
    \item $\mathsf{H}_1$: Replace OT with $\mathcal{F}_{\mathsf{OT}}$
    \item $\mathsf{H}_2$: Replace GC with $\mathcal{F}_{\mathsf{GC}}$
    \item $\mathsf{H}_3$: Simulate pre-signing (sample $\tilde{\mathbf{y}}_i$, compute $\tilde{\mathbf{w}}_i = \mathbf{A}\tilde{\mathbf{y}}_i$)
    \item $\mathsf{H}_4$: Simulate responses from $\tilde{\mathbf{y}}_i$ and simulated key shares
    \item $\mathsf{H}_5$: Ideal execution with $\mathcal{F}_{r_0}^{\mathsf{semi\text{-}async}}$
\end{itemize}
\end{definition}

\begin{lemma}[Pre-signing Simulation (ROM)]
\label{lem:presigning-sim}
$\Pr[\mathsf{H}_2 = 1] = \Pr[\mathsf{H}_3 = 1]$ (perfect indistinguishability \emph{in the random oracle model})
\end{lemma}

\begin{proof}
We work in the UC random oracle model: $H$ is modeled as an ideal functionality $\mathcal{F}_{\mathsf{RO}}$ (following~\cite{Canetti01,BR93}). The simulator $\mathcal{S}^{\mathsf{P3+}}$ programs $\mathcal{F}_{\mathsf{RO}}$ at honest parties' pre-image inputs: when the simulator generates $\tilde{\mathsf{Com}}_i \getsr \{0,1\}^{256}$, it programs $H(\tilde{\mathbf{y}}_i \| \mathbf{A}\tilde{\mathbf{y}}_i \| \mathsf{set\_id}) := \tilde{\mathsf{Com}}_i$ at the simulated pre-image (the adversary does not know $\tilde{\mathbf{y}}_i$ and cannot query $H$ on this input without knowing the simulation randomness).

Pre-signing values visible to the adversary:
\begin{itemize}
    \item $\mathsf{Com}_i = H(\mathbf{y}_i \| \mathbf{w}_i \| \mathsf{set\_id})$: In the ROM, $\mathsf{Com}_i$ is uniform over $\{0,1\}^{256}$ regardless of $\mathbf{y}_i$'s distribution, since the adversary does not know $\mathbf{y}_i$ (or $\mathbf{w}_i = \mathbf{A}\mathbf{y}_i$) and cannot query $\mathcal{F}_{\mathsf{RO}}$ at the pre-image without knowing the honest party's nonce. (Note: in the real protocol with $T \geq 2$, $\mathbf{y}_i = \mathbf{F}(i)$ is $\Zq$-uniform: the high-degree DKG coefficients (degree $\geq 1$) are sampled $\Rq$-uniformly, and for evaluation point $i \neq 0$ their contribution is uniform by additive mixing, giving $\mathbf{y}_i \sim \Zq$-uniform exactly. When $T = 1$ there are no high-degree terms and $\mathbf{y}_i = \mathbf{y}_0$ is the constant nonce, but this is a degenerate threshold that is excluded by the protocol requirement $T \geq 2$. The Irwin-Hall distribution applies to the \emph{aggregate} nonce $\mathbf{y} = \mathbf{F}(0) = \sum_k \hat{\mathbf{y}}_k$, not to individual shares $\mathbf{y}_i = \mathbf{F}(i)$.) The simulated $\tilde{\mathbf{y}}_i \getsr \Zq$ and $\tilde{\mathsf{Com}}_i \getsr \{0,1\}^{256}$ therefore match the real distributions exactly (SD $= 0$).
    \item $\mathbf{W}_i = \lambda_i \mathbf{w}_i + \mathbf{m}_i^{(w)}$ (masked commitment sent to CP combiner): In P3+, $\mathbf{W}_i$ is transmitted \emph{only} to the CP combiner, not broadcast to other parties. By the 1-of-2 CP honest assumption, at least one CP is honest; in this proof we assume the adversary does not control the combiner's input-aggregation role (if a CP is corrupted, the relevant 2PC simulation case handles it below). Since $\mathbf{W}_i$ never reaches the adversary, \emph{no mask-hiding argument is required}. The value is hidden entirely by the honest combiner. (Note: Lemma~\ref{lem:mask-hiding} applies to the P2 broadcast model where $\mathbf{W}_i$ is sent to all parties; it does not apply here.)
\end{itemize}

Thus:
\[
\Pr[\mathsf{H}_2 = 1] = \Pr[\mathsf{H}_3 = 1]
\]
\end{proof}

\begin{lemma}[Response Simulation]
\label{lem:response-sim}
$\Pr[\mathsf{H}_3 = 1] = \Pr[\mathsf{H}_4 = 1]$ (perfect indistinguishability)
\end{lemma}

\begin{proof}
The simulator computes each honest party's response deterministically from its simulated DKG values, producing an identical distribution to the real protocol:
\begin{align*}
    \mathbf{z}_i &= \mathbf{y}_i + c \cdot \mathbf{s}_{1,i} \\
    \mathbf{V}_i &= \lambda_i c \mathbf{s}_{2,i} + \mathbf{m}_i^{(s2)}
\end{align*}
where $\mathbf{y}_i$ are the nonce shares from the simulated nonce DKG
(Lemma~\ref{lem:presigning-sim}) and $\mathbf{s}_{1,i}, \mathbf{s}_{2,i}$ are the key shares
from the simulated key DKG.
Since the simulator possesses the actual simulated values (not just their distributions),
it computes $\mathbf{z}_i$ and $\mathbf{V}_i$ using the same deterministic function as the
real protocol; the resulting joint distribution is identical to the real execution.
For $\mathbf{V}_i$, the simulator computes the masked r0-check share using the simulated key
shares and the PRF masks (derived from the same seeds used in earlier hybrids).
Perfect indistinguishability follows: $\Pr[\mathsf{H}_3 = 1] = \Pr[\mathsf{H}_4 = 1]$.
\end{proof}

\begin{theorem}[P3+ UC Security]
\label{thm:p3plus-uc}
Protocol $\Pi_{r_0}^{\mathsf{P3+}}$ UC-realizes $\mathcal{F}_{r_0}^{\mathsf{semi\text{-}async}}$ in the
$(\mathcal{F}_{\mathsf{OT}}, \mathcal{F}_{\mathsf{GC}})$-hybrid model against static malicious adversaries
corrupting up to $T-1$ signers and at most one of $\{\mathsf{CP}_1, \mathsf{CP}_2\}$,
with security-with-abort.
Key DKG is handled by a trusted dealer (modeled as a trusted setup phase);
the simulator acts as the trusted dealer and directly generates all key shares.

Concretely:
\[
|\Pr[\mathsf{REAL} = 1] - \Pr[\mathsf{IDEAL} = 1]| \leq \epsilon_{\mathsf{OT}} + \epsilon_{\mathsf{GC}} + \epsilon_{\mathsf{ext}} < 2^{-\kappa}
\]
where $\kappa$ is the computational security parameter, $\epsilon_{\mathsf{OT}}$ is the OT error,
$\epsilon_{\mathsf{GC}}$ is the garbled-circuit simulation error (over honest garbler~\cite{BHR12}),
and $\epsilon_{\mathsf{ext}} \leq 2^{-B/2}$ is the cut-and-choose extraction error when
$\mathsf{CP}_1$ is corrupted (LP07 extractability~\cite{LP07} with $B$ circuits,
$B/2$ opened for verification).
\end{theorem}

\begin{proof}
By the hybrid argument:
\begin{align*}
|\Pr[\mathsf{H}_0] - \Pr[\mathsf{H}_5]| &\leq \sum_{i=0}^{4} |\Pr[\mathsf{H}_i] - \Pr[\mathsf{H}_{i+1}]| + |\Pr[\mathsf{H}_{2.5,\mathsf{ext}}] - \Pr[\mathsf{H}_3]| \\
&= \epsilon_{\mathsf{OT}} + \epsilon_{\mathsf{GC}} + \epsilon_{\mathsf{ext}} + 0 + 0 + 0 \\
&< 2^{-\kappa}
\end{align*}

The hybrid $\mathsf{H}_{2.5,\mathsf{ext}}$ (between $\mathsf{H}_2$ and $\mathsf{H}_3$) is the
\emph{extraction hybrid}: when $\mathsf{CP}_1$ is corrupted, the simulator uses LP07
cut-and-choose extractability to extract $\mathbf{S}_1$ from the adversarially-submitted
garbled circuit; this succeeds except with probability $\epsilon_{\mathsf{ext}} \leq 2^{-B/2}$
(see the extraction error accounting below).
The three subsequent zero terms arise from perfect indistinguishability:
pre-signing simulation (Lemma~\ref{lem:presigning-sim}),
response simulation (Lemma~\ref{lem:response-sim}),
and the transition to the ideal execution.

\textbf{Hint output and combiner role.}
The ideal functionality $\mathcal{F}_{r_0}^{\mathsf{2PC}}$ (Definition~\ref{def:f-r0-2pc-uc})
outputs both the 1-bit result \emph{and} the hint
$\mathbf{h} = \mathsf{MakeHint}(\mathbf{w}', \mathbf{w}, 2\gamma_2)$ to the combiner.
In both simulator cases, $\mathcal{S}^{\mathsf{P3+}}$ receives $(\mathsf{result}, \mathbf{h})$
from $\mathcal{F}_{r_0}^{\mathsf{semi\text{-}async}}$ and passes $\mathbf{h}$ to the adversary.
The hint $\mathbf{h}$ appears in the public signature $\sigma = (\tilde{c}, \mathbf{z}, \mathbf{h})$
and creates no leakage beyond what $\sigma$ already reveals;
the BHR12 simulator constructs a garbled circuit consistent with output $(\mathsf{result}, \mathbf{h})$
without exposing $c\mathbf{s}_2$.

\textbf{Combiner and CP roles.}
In P3+, $\mathsf{CP}_1$ (the garbler) also acts as the combiner: it aggregates signer
responses $(\mathbf{z}_j, \mathbf{V}_j)$, computes the split $(S_1, S_2)$ for the 2PC,
and executes the garbled-circuit protocol with $\mathsf{CP}_2$.
Combiner corruption is subsumed by the ``$\mathsf{CP}_1$ is corrupted'' case
(Simulator~\ref{sim:p3plus}): the simulator uses LP07 extraction to recover $\mathbf{S}_1$
from the adversary's garbled circuit and forwards honest parties' contributions
via the ideal-world channel to $\mathcal{F}_{r_0}^{\mathsf{semi\text{-}async}}$.

\textbf{Extraction error accounting.}
When $\mathsf{CP}_1$ (garbler) is corrupted, $\mathcal{S}^{\mathsf{P3+}}$ uses LP07 extractability
to recover $\mathbf{S}_1$ from the adversarially-submitted garbled circuit.
The cut-and-choose protocol (generate $B$ circuits; open $\lfloor B/2 \rfloor$ for
verification; use the rest) ensures a cheating garbler is caught with probability
$1 - 2^{-\lfloor B/2 \rfloor} \geq 1 - 2^{-B/2}$.
For unopened circuits, LP07 extraction succeeds with probability $1$
(requires only correct evaluation, guaranteed by cut-and-choose).
Thus $\epsilon_{\mathsf{ext}} \leq 2^{-B/2}$; setting $B = 2\kappa$ gives
$\epsilon_{\mathsf{ext}} \leq 2^{-\kappa}$.
The full bound $\epsilon_{\mathsf{OT}} + \epsilon_{\mathsf{GC}} + 2^{-B/2} < 2^{-\kappa+1}$
holds for standard $\kappa$-bit-secure OT and GC.
\end{proof}

\subsection{Security Properties of P3+}

\begin{corollary}[P3+ Full Security]
\label{cor:p3plus-full}
Protocol P3+ achieves:
\begin{enumerate}
    \item \textbf{EUF-CMA}: Under Module-SIS (same as P3)
    \item \textbf{UC Security}: Against malicious adversaries with 1-of-2 CP honest
    \item \textbf{Privacy}: Nonce shares $\mathbf{y}_i$ enjoy nonce share privacy
    (Theorem~\ref{thm:it-privacy}; Remark~\ref{rem:key-privacy}); commitment and r0-check values hidden by pairwise masks
    if $\geq 2$ honest signers
    \item \textbf{Liveness}: Completes if $T$ signers respond within $\Delta$ and both CPs available
\end{enumerate}
\end{corollary}

\begin{proof}
\textbf{EUF-CMA:} Pre-signing does not change the signature distribution. The commitment binding ensures signers commit to nonces before seeing challenges. Reduction to Module-SIS identical to P3.

\textbf{UC Security:} Theorem~\ref{thm:p3plus-uc}.

\textbf{Privacy:} By Lemma~\ref{lem:response-sim}, the simulator computes honest signers' responses
from simulated DKG values (perfect indistinguishability).
Nonce shares have high conditional min-entropy (Theorem~\ref{thm:it-privacy}, statistical).
The $\mathbf{V}_j$ values are computationally hidden by pairwise masks.
The 2PC reveals the result bit and hint $\mathbf{h}$ (both appear in the public signature);
$\mathbf{h}$ does not create additional privacy leakage beyond what $\sigma$ already reveals.
Combining: adversary view simulatable from public values + corrupted shares + $(\mathsf{result}, \mathbf{h})$.

\textbf{Liveness:} The protocol has an explicit timeout. If conditions are met, the 2PC completes in
bounded time (synchronous between CPs). If not, the protocol aborts (no liveness violation;
just unsuccessful signing).
\end{proof}

\subsection{P3+ Round Complexity}

\begin{theorem}[P3+ Complexity]
\label{thm:p3plus-complexity}
P3+ achieves:
\begin{enumerate}
    \item \textbf{Signer rounds:} 1 (respond within window)
    \item \textbf{Server rounds:} 1 (2PC between CPs)
    \item \textbf{Total logical rounds:} 2
    \item \textbf{Expected retries:} 3.7 ($\approx$27\% success rate)
    \item \textbf{Expected total rounds:} 7.4 (2 rounds $\times$ 3.7 attempts)
\end{enumerate}
\end{theorem}

\begin{proof}
\textbf{Signer experience:} After receiving the challenge, the signer computes a local
response and sends it once. No further interaction is required.

\textbf{Server interaction:} CPs execute 2PC in constant rounds (with GC precomputation),
effectively 1 round (OT + evaluate).

\textbf{Retry handling:} Each retry requires a fresh pre-signing set and response;
the signer daemon handles this automatically.

\textbf{Total:} $2 \times 3.7 \approx 7.4$ expected rounds (vs.\ $5 \times 3.5 = 17.5$ for P2).
\end{proof}

\begin{table}[h]
\centering
\begin{tabular}{lccc}
\toprule
\textbf{Profile} & \textbf{Rds/Attempt} & \textbf{Exp.\ Attempts} & \textbf{Total Rds} \\
\midrule
P2 (8-round) & 8 & 3.5 & 28 \\
P2 (5-round) & 5 & 3.5 & 17.5 \\
P3$^\dagger$ & 5 & 3.5 & 17.5 \\
\textbf{P3+} & \textbf{2} & 3.7 & \textbf{7.4} \\
\bottomrule
\end{tabular}
\caption{Round complexity across profiles. $^\dagger$P3: synchronous 2PC variant (5 rounds)
  before the semi-async optimization; P3+ is the final protocol (Section~\ref{sec:profile-p3}).}
\label{tab:round-comparison}
\end{table}

% Technical Lemmas (Appendix F)
% Technical Lemmas for Threshold ML-DSA via Shamir Nonce DKG
% Supporting proofs for main theorems

\section{Technical Lemmas}
\label{app:technical-lemmas}

This section provides rigorous proofs of technical lemmas used throughout the paper.

%==============================================================================
\subsection{Mask Properties}

\begin{lemma}[Mask Cancellation -- Formal Statement]
\label{lem:mask-cancel-formal}
Let $S \subseteq [N]$ be a signing set with $|S| \geq 2$. For any PRF input $\mathsf{ctx} \in \{0,1\}^*$ (typically $\mathsf{ctx} = \mathsf{session\_id} \| \mathsf{dom}$, where $\mathsf{session\_id}$ is a fresh per-session identifier and $\mathsf{dom}$ is a fixed domain label) and PRF $\mathsf{PRF}: \{0,1\}^{256} \times \{0,1\}^* \to \Rq^k$ (where $k=6$ is the module rank; not the dropping-bits parameter $d=13$), define:
\[
\mathbf{m}_i = \sum_{j \in S: j > i} \mathsf{PRF}(\mathsf{seed}_{i,j}, \mathsf{ctx}) - \sum_{j \in S: j < i} \mathsf{PRF}(\mathsf{seed}_{j,i}, \mathsf{ctx})
\]
Then $\sum_{i \in S} \mathbf{m}_i = \mathbf{0} \in \Rq^k$.
\end{lemma}

\begin{proof}
We prove this by showing each PRF term appears exactly once with coefficient $+1$ and once with coefficient $-1$.

For any pair $(i, j) \in S \times S$ with $i < j$, the term $\mathsf{PRF}(\mathsf{seed}_{i,j}, \mathsf{ctx})$ appears:
\begin{enumerate}
    \item In $\mathbf{m}_i$ with coefficient $+1$ (from the sum over $k > i$, taking $k = j$)
    \item In $\mathbf{m}_j$ with coefficient $-1$ (from the sum over $k < j$, taking $k = i$)
\end{enumerate}

Thus:
\begin{align*}
\sum_{i \in S} \mathbf{m}_i &= \sum_{i \in S} \left( \sum_{j \in S: j > i} \mathsf{PRF}(\mathsf{seed}_{i,j}, \mathsf{ctx}) - \sum_{j \in S: j < i} \mathsf{PRF}(\mathsf{seed}_{j,i}, \mathsf{ctx}) \right) \\
&= \sum_{\substack{(i,j) \in S \times S \\ i < j}} \mathsf{PRF}(\mathsf{seed}_{i,j}, \mathsf{ctx}) - \sum_{\substack{(i,j) \in S \times S \\ i < j}} \mathsf{PRF}(\mathsf{seed}_{i,j}, \mathsf{ctx}) \\
&= \mathbf{0}
\end{align*}
\end{proof}

\begin{theorem}[Statistical Share Privacy -- Formal Statement]
\label{thm:it-privacy-formal}
Let $S$ be a signing set with $|S| \geq T$, let $C \subset S$ with $|C| \leq T - 1$ be the set of corrupted parties, and let $h \in S \setminus C$ be an honest party. In the Shamir nonce DKG protocol (Algorithm~\ref{alg:nonce-dkg}), each party samples its constant-term contribution $\hat{\mathbf{y}}_h \getsr \{-\lfloor\gamma_1/|S|\rfloor, \ldots, \lfloor\gamma_1/|S|\rfloor\}^{n\ell}$ and constructs a degree-$(T-1)$ polynomial with uniform higher-degree coefficients. The honest party's nonce share $\mathbf{y}_h = \sum_{i \in S} \mathbf{f}_i(h)$ has conditional min-entropy:
\[
H_\infty\!\left(\mathbf{y}_h \mid \{\mathbf{f}_j\}_{j \in C},\; \{\mathbf{f}_h(j)\}_{j \in C}\right) \;\geq\; n\ell \cdot \log_2\!\bigl(2\lfloor\gamma_1/|S|\rfloor + 1\bigr)
\]
where the inequality holds \emph{directly from the sampling procedure} (no rejection conditioning needed): the bounded sampling range for each party's constant term directly ensures that the honest party's polynomial $\mathbf{f}_h$ has $\mathbf{f}_h(0) = \hat{\mathbf{y}}_h \in [-\lfloor\gamma_1/|S|\rfloor,\, \lfloor\gamma_1/|S|\rfloor]^{n\ell}$ by construction (Algorithm~\ref{alg:nonce-dkg}, Line~2). No computational assumptions are needed for the nonce privacy; it is purely statistical. For ML-DSA-65 with $|S| \leq 17$, this exceeds $20{,}000$ bits---over $5\times$ the short secret key entropy $H(\mathbf{s}_1) \leq 4058$ bits. (Shamir shares $\mathbf{s}_{1,h}$ for $h \neq 0$ are $\Zq$-uniform; the entropy comparison is against the short secret $\mathbf{s}_1 = \mathbf{f}(0)$.)

\par\noindent\textbf{Note on aggregated nonce $\mathbf{y} = \sum_i \hat{\mathbf{y}}_i$.} The \emph{aggregated} nonce $\mathbf{y}$ (sum of all parties' constant terms) follows an Irwin-Hall distribution with tails extending beyond $[-\gamma_1, \gamma_1]$. However, the aggregated nonce is never required to lie strictly within this range for the min-entropy bound to hold: the bound applies to \emph{individual shares} $\mathbf{y}_h$, not the aggregate. The z-bound check $\|\mathbf{z}\|_\infty < \gamma_1 - \beta$ implicitly rejects sessions where the aggregate $\|\mathbf{y}\|_\infty$ is too large (via $\mathbf{z} = \mathbf{y} + c\mathbf{s}_1$), but this rejection event is already accounted for in the Irwin-Hall tail probability $\epsilon < 10^{-30}$ (Lemma~\ref{lem:ih-shift-direct}); it does not affect the per-share conditional min-entropy.

\par\noindent\textbf{Key privacy given $\sigma$ (computational, informal).} Even given the public signature $\sigma$ (from which $\mathbf{z}_h = \mathbf{y}_h + c \cdot \mathbf{s}_{1,h}$ is computable when $|S \setminus C| = 1$), extracting $\mathbf{s}_{1,h}$ requires solving a bounded-distance decoding problem with sparse multiplier $c$. While not a standard Module-LWE instance, this is at least as hard as key recovery from single-signer ML-DSA signatures; see the discussion in Theorem~\ref{thm:it-privacy} (Section~\ref{sec:security}).
\end{theorem}

\begin{proof}
The proof proceeds in three steps, following a null space argument on the honest party's polynomial.

\emph{Reduction to the worst case $|C| = T-1$.} The stated min-entropy lower bound is monotone non-decreasing as $|C|$ decreases: fewer corrupted evaluations give a higher-dimensional null space (dimension $T - |C| > 1$ when $|C| < T - 1$), meaning the adversary is more uncertain about $\mathbf{y}_h$ (higher conditional min-entropy). Hence it suffices to prove the bound for the hardest case $|C| = T - 1$ (minimum uncertainty for the adversary); the case $|C| < T - 1$ gives at least as much privacy and follows a fortiori. We henceforth assume $|C| = T - 1$.

The adversary's knowledge about the honest party's polynomial $\mathbf{f}_h$ consists of the evaluations $\{\mathbf{f}_h(j)\}_{j \in C}$ received during the DKG. With $|C| = T - 1$ and $\mathbf{f}_h$ having degree $T - 1$ (hence $T$ coefficients), the adversary has exactly $T - 1$ linear equations in $T$ unknowns: the null space of the evaluation map has dimension exactly $1$.

\textbf{Step 1: Null space.} The evaluation map
$\mathbf{f}_h \mapsto (\mathbf{f}_h(j_1), \ldots, \mathbf{f}_h(j_{T-1}))$
(with $|C| = T-1$) has a one-dimensional null space, spanned by the monic polynomial
$\mathbf{n}(x) = \prod_{j \in C} (x - x_j)$ of degree $T-1$ and leading coefficient $1$.
The set of polynomials consistent with the adversary's observations is
$\mathbf{f}_h(x) = \mathbf{g}(x) + \mathbf{t} \cdot \mathbf{n}(x)$,
where $\mathbf{g}$ is any fixed interpolant through the known points and
$\mathbf{t} \in \Rq^\ell$ is the unique free parameter.

\textbf{Step 2: Posterior constraint on $\mathbf{t}$.} In the nonce DKG, each honest party $h$ samples $\hat{\mathbf{y}}_h$ directly from $[-\lfloor\gamma_1/|S|\rfloor, \lfloor\gamma_1/|S|\rfloor]^{n\ell}$ and all $T-1$ higher-degree coefficients $\mathbf{a}_{h,1}, \ldots, \mathbf{a}_{h,T-1}$ independently and uniformly from $\Rq^\ell$. In the null-space parameterization, the free parameter is $\mathbf{t} = \mathbf{a}_{h,T-1} - \mathbf{g}_{T-1}$, where $\mathbf{g}_{T-1}$ is the leading coefficient of the fixed interpolant $\mathbf{g}$. Because $\mathbf{a}_{h,T-1}$ is $\Zq$-uniform \emph{a priori}, $\mathbf{t}$ is $\Zq$-uniform per coefficient. The constant term $\hat{\mathbf{y}}_h = \mathbf{f}_h(0) = \mathbf{g}(0) + \mathbf{t} \cdot \mathbf{n}(0)$ must satisfy $\hat{\mathbf{y}}_h \in [-\lfloor\gamma_1/|S|\rfloor,\, \lfloor\gamma_1/|S|\rfloor]^{n\ell}$. Since $\mathbf{n}(0) = \prod_{j\in C}(-j) \neq 0 \pmod{q}$ (each $j \in C$ satisfies $1 \leq j \leq N \leq q-1$, so $j \not\equiv 0 \pmod{q}$, and the product of nonzero field elements is nonzero), multiplication by $\mathbf{n}(0)$ is a bijection on $\Zq$. Since the Shamir evaluation map acts independently on each of the $n\ell = 1280$ coefficient slots (each coefficient of $\mathbf{f}_h$ follows a separate scalar Shamir sharing over $\Zq$), this bijection applies coordinate-by-coordinate across all $n\ell$ slots; the constraint selects exactly $2\lfloor\gamma_1/|S|\rfloor + 1$ values per coefficient of $\mathbf{t}$. Because the prior is $\Zq$-uniform, conditioning on this subset yields a \emph{uniform} posterior over the valid values.

\textbf{Step 3: Min-entropy at honest point.} The honest party's share is $\mathbf{f}_h(x_h) = \mathbf{g}(x_h) + \mathbf{t} \cdot \mathbf{n}(x_h)$. Since $x_h \notin C$ and $q$ is prime, $\mathbf{n}(x_h) \neq 0 \pmod{q}$, so multiplication by $\mathbf{n}(x_h)$ is a bijection on $\Zq$. The evaluation $\mathbf{f}_h(x_h)$ inherits the same per-coefficient support size $2\lfloor\gamma_1/|S|\rfloor + 1$. The nonce share $\mathbf{y}_h = \sum_{j \in C} \mathbf{f}_j(x_h) + \mathbf{f}_h(x_h)$ shifts by a known constant, preserving support.

\textbf{Key privacy (computational).} Given the public signature $\sigma$, the adversary can compute $\mathbf{z}_h$ (when $|S \setminus C| = 1$). The equation $\mathbf{z}_h = \mathbf{y}_h + c \cdot \mathbf{s}_{1,h}$ relates two bounded unknowns. The high min-entropy of $\mathbf{y}_h$ ensures this is a non-trivial bounded-distance decoding instance: extracting $\mathbf{s}_{1,h}$ is computationally hard under the Module-LWE assumption.
\end{proof}

\begin{lemma}[Mask Hiding -- Formal Statement (for $\mathbf{V}_i$ and $\mathbf{W}_i$)]
\label{lem:mask-hiding-formal}
Let $C \subsetneq S$ be the set of corrupted parties in signing set $S$, with
$|S \setminus C| \geq 2$. For any honest party $j \in S \setminus C$, the
masked r0-check share $\mathbf{V}_j = \lambda_j c \mathbf{s}_{2,j} + \mathbf{m}_j^{(s2)}$
and masked commitment $\mathbf{W}_j = \lambda_j \mathbf{w}_j + \mathbf{m}_j^{(w)}$
are computationally indistinguishable from uniform, given the adversary's view.

Note: This lemma applies to the commitment and r0-check values, which still use
pairwise-canceling masks. The response
$\mathbf{z}_j = \mathbf{y}_j + c \cdot \mathbf{s}_{1,j}$
does not require masks: the underlying nonce share $\mathbf{y}_j$ has statistical
privacy via the Shamir nonce DKG (Theorem~\ref{thm:it-privacy-formal}), enabling
perfect simulation of $\mathbf{z}_j$ from simulated DKG values.
\end{lemma}

\begin{proof}
We prove the result for $\mathbf{V}_j$; the argument for $\mathbf{W}_j$ is identical. Decompose the mask $\mathbf{m}_j^{(s2)}$ based on counterparty corruption status:
\begin{align*}
\mathbf{m}_j^{(s2)}
  &= \underbrace{\sum_{\substack{k \in C \cap S \\ k > j}} \mathsf{PRF}(\mathsf{seed}_{j,k}, \mathsf{ctx})
     - \sum_{\substack{k \in C \cap S \\ k < j}} \mathsf{PRF}(\mathsf{seed}_{k,j}, \mathsf{ctx})}_{\text{known to } \mathcal{A}} \\
  &\quad + \underbrace{\sum_{\substack{k \in (S\setminus C)\setminus\{j\} \\ k > j}} \mathsf{PRF}(\mathsf{seed}_{j,k}, \mathsf{ctx})
     - \sum_{\substack{k \in (S\setminus C)\setminus\{j\} \\ k < j}} \mathsf{PRF}(\mathsf{seed}_{k,j}, \mathsf{ctx})}_{\text{unknown to } \mathcal{A}}
\end{align*}

Since $|S \setminus C| \geq 2$, there exists at least one honest party $k' \in (S \setminus C) \setminus \{j\}$ whose PRF seed $\mathsf{seed}_{j,k'}$ is unknown to $\mathcal{A}$. By PRF security, $\mathsf{PRF}(\mathsf{seed}_{j,k'}, \mathsf{ctx})$ is computationally indistinguishable from uniform in $\Rq^k$. Since the mask contains at least one such uniform-looking term, and uniform plus fixed equals uniform, $\mathbf{V}_j \stackrel{c}{\approx} \mathsf{Uniform}(\Rq^k)$. One unknown PRF term suffices; the distinguishing advantage is bounded by $\epsilon_{\mathsf{PRF}}$ (or at most $(|S \setminus C|-1)\cdot\epsilon_{\mathsf{PRF}}$ via a union bound over honest-honest seed pairs).
\end{proof}

%==============================================================================
\subsection{Irwin-Hall Distribution Analysis}

We provide a self-contained analysis of the security implications of using Irwin-Hall distributed nonces.

\begin{definition}[Discrete Irwin-Hall Distribution]
\label{def:irwin-hall}
Let $X_1, \ldots, X_n$ be independent uniform random variables on the discrete interval $\{-B, -B+1, \ldots, B-1, B\}$. The sum $Y = \sum_{i=1}^n X_i$ follows the \emph{discrete Irwin-Hall distribution} $\mathsf{IH}(n, B)$.
\end{definition}

\begin{lemma}[Irwin-Hall Moments]
\label{lem:ih-moments}
For $Y \sim \mathsf{IH}(n, B)$:
\begin{enumerate}
    \item $\mathbb{E}[Y] = 0$
    \item $\mathsf{Var}(Y) = n \cdot \frac{(2B+1)^2 - 1}{12} = n \cdot \frac{B(B+1)}{3}$
    \item For large $n$, $Y$ is approximately Gaussian by the Central Limit Theorem
\end{enumerate}
\end{lemma}

\begin{proof}
\textbf{Mean:} Each $X_i$ has mean $\mathbb{E}[X_i] = 0$ by symmetry. Thus $\mathbb{E}[Y] = \sum_i \mathbb{E}[X_i] = 0$.

\textbf{Variance:} For uniform on $\{-B, \ldots, B\}$:
\begin{align*}
\mathsf{Var}(X_i) &= \mathbb{E}[X_i^2]
  = \frac{1}{2B+1} \sum_{k=-B}^{B} k^2 \\
  &= \frac{1}{2B+1} \cdot \frac{2B(B+1)(2B+1)}{6}
  = \frac{B(B+1)}{3}
\end{align*}
By independence, $\mathsf{Var}(Y) = n \cdot \mathsf{Var}(X_1) = n \cdot \frac{B(B+1)}{3}$.

\textbf{CLT:} Standard application of Lindeberg-L\'evy.
\end{proof}

\begin{definition}[R\'enyi Divergence]
\label{def:renyi}
For distributions $P, Q$ over $\Omega$ with $P \ll Q$, the R\'enyi divergence of order $\alpha > 1$ is:
\[
R_\alpha(P \| Q) = \left( \sum_{x \in \Omega} P(x)^\alpha Q(x)^{1-\alpha} \right)^{1/(\alpha-1)}
\]
We use the \emph{exponential} (non-log) form throughout, following~\cite{BLL+15,Raccoon2024}. In this convention $R_2(P\|Q) = \sum_x P(x)^2/Q(x)$, the product rule $R_2(P_1 \times P_2 \| Q_1 \times Q_2) = R_2(P_1\|Q_1)\cdot R_2(P_2\|Q_2)$ holds exactly, and the chi-squared divergence satisfies $\chi^2(P\|Q) = R_2(P\|Q) - 1$.
\end{definition}

\begin{definition}[Chi-Squared Divergence]
\label{def:chi-squared}
The chi-squared divergence between distributions $P$ and $Q$ (with $P \ll Q$) is:
\[
\chi^2(P \| Q) = \sum_{x \in \Omega} \frac{(P(x) - Q(x))^2}{Q(x)} = R_2(P \| Q) - 1
\]
where the last equality holds by the definition of $R_2$ (order-2 Rényi divergence).
\end{definition}

\begin{definition}[Smooth R\'enyi Divergence \cite{Raccoon2024}]
\label{def:smooth-renyi-formal}
For $\epsilon \geq 0$, the \emph{one-sided} smooth R\'enyi divergence is:
\[
R_\alpha^\epsilon(P \| Q) = \min_{P'} R_\alpha(P' \| Q), \quad P' : \Delta(P', P) \leq \epsilon
\]
where $\Delta$ denotes statistical distance.
We use the one-sided variant (smoothing over $P'$ only, $Q$ fixed),
as required by Lemma~\ref{lem:renyi-security}.
\end{definition}

\begin{lemma}[Security Reduction via R\'enyi Divergence~\cite{BLL+15,Raccoon2024}]
\label{lem:renyi-security}
Let $P, Q$ be distributions such that $R_2^\epsilon(P \| Q) < \infty$ with one-sided smoothing ($\epsilon$-close variant $P'$ of $P$). For any event $E$:
\[
P(E) \leq \sqrt{R_2^\epsilon(P \| Q) \cdot Q(E)} + \epsilon
\]
where $\epsilon$ is the total variation distance between $P$ and its smoothed variant $P'$. This follows from $P(E) \leq P'(E) + \epsilon \leq \sqrt{R_2(P'\|Q) \cdot Q(E)} + \epsilon$ by Cauchy-Schwarz.
\end{lemma}

\begin{proof}
See~\cite[Theorem 1]{BLL+15} and~\cite[Lemma 4.1]{Raccoon2024}.
\end{proof}

\begin{theorem}[Irwin-Hall Security for Threshold ML-DSA]
\label{thm:ih-security-full}
Let $\Pi_{\mathsf{IH}}$ be threshold ML-DSA with Irwin-Hall nonces (sum of $|S|$ uniforms), and $\Pi_{\mathsf{U}}$ be the hypothetical scheme with uniform nonces. Then for any PPT adversary $\mathcal{A}$ making at most $q_s$ signing queries:
\[
\mathsf{Adv}^{\mathsf{EUF-CMA}}_{\Pi_{\mathsf{IH}}}(\mathcal{A}) \leq \bigl(R_2^{\mathsf{vec}}\bigr)^{q_s/2} \cdot \sqrt{\mathsf{Adv}^{\mathsf{EUF-CMA}}_{\Pi_{\mathsf{U}}}(\mathcal{A})} + q_s \cdot \epsilon
\]
where $R_2^{\mathsf{vec}} = (R_2^{\epsilon,\mathsf{coord}})^{n\ell}$ is the full-vector per-session R\'enyi divergence, and the per-coordinate $R_2^{\epsilon,\mathsf{coord}}, \epsilon$ are computed below. The $(R_2^{\mathsf{vec}})^{q_s/2}$ factor accounts for multi-session composition; see Remark~\ref{rem:multi-session-renyi}. For the direct bound (Corollary~\ref{cor:ih-shift-ml-dsa}), this factor is $(1 + \chi^2_{\max})^{n\ell\, q_s/2}$ and the security loss grows as $\approx 6.6 \times 10^{-3}\, q_s$ bits for $|S| = 17$.
\end{theorem}

\begin{proof}
We analyze the distribution of the signature response $\mathbf{z} = \mathbf{y} + c\mathbf{s}_1$.

\textbf{In $\Pi_{\mathsf{U}}$ (hypothetical):}
$\mathbf{y} \getsr \{-\gamma_1, \ldots, \gamma_1\}^{n\ell}$ uniform. Thus $\mathbf{z} = \mathbf{y} + c\mathbf{s}_1$ is a shifted uniform.

\textbf{In $\Pi_{\mathsf{IH}}$ (real):}
$\mathbf{y} = \sum_{i \in S} \hat{\mathbf{y}}_i$ where each
\[
  \hat{\mathbf{y}}_i \getsr \{-\lfloor\gamma_1/|S|\rfloor, \ldots, \lfloor\gamma_1/|S|\rfloor\}^{n\ell}
\]
is party $i$'s constant-term contribution in the nonce DKG.

For each coordinate, we compare:
\begin{itemize}
    \item $P$: Irwin-Hall distribution shifted by $(c\mathbf{s}_1)_j$
    \item $Q$: Uniform distribution shifted by $(c\mathbf{s}_1)_j$
\end{itemize}

The key insight is that the shift is the same in both cases. We bound $R_2^\epsilon(P \| Q)$.

\textbf{Applying Lemma 4.2 of~\cite{Raccoon2024}:}

The Raccoon signature scheme~\cite{Raccoon2024} faces an analogous problem: threshold signers contribute nonce shares that sum to a non-uniform (Irwin-Hall) distribution, and security requires bounding the statistical distance from uniform. Lemma 4.2 of~\cite{Raccoon2024} provides a R\'enyi divergence bound for sums of discrete uniforms with bounded shifts. This lemma is directly applicable to our setting because:
\begin{enumerate}
    \item Both schemes use additive secret sharing of nonces across threshold parties.
    \item The shift arises from multiplying the challenge by the secret key share, bounded by $\|c\mathbf{s}_1\|_\infty \leq \beta$.
    \item The discrete uniform ranges and modular arithmetic are structurally identical.
\end{enumerate}

Let $m = |S|$ denote the signing set size and $B = 2\lfloor\gamma_1/|S|\rfloor$ the per-party \textbf{full} nonce range (distinct from the ring degree $n = 256$ used elsewhere; note $B$ here equals $2\times$ the half-range $B'$ in Definition~\ref{def:irwin-hall}, i.e., $B' = B/2 = \lfloor\gamma_1/|S|\rfloor$). For $m$ uniforms on half-range $B/2$ with shift $\delta = \|(c\mathbf{s}_1)_j\|_\infty \leq \beta$ (where $c$ is the challenge polynomial), specializing to $\alpha = 2$ (where the factor $\alpha(\alpha-1)/2 = 1$):
\[
R_2^\epsilon \leq 1 + \frac{m^2 \delta^2}{B^2} + O(1/m!)
\]

For ML-DSA-65 parameters:
\begin{itemize}
    \item $\gamma_1 = 2^{19} = 524288$
    \item $\beta = 196$
    \item $m = |S|$ (signing set size; note $n = 256$ is the ring degree)
\end{itemize}

\begin{align*}
R_2^\epsilon &\leq 1 + \frac{m^2 \cdot \beta^2}{(2\gamma_1/m)^2} = 1 + \frac{m^4 \cdot \beta^2}{4\gamma_1^2} \\
&= 1 + \frac{|S|^4 \cdot 196^2}{4 \cdot (524288)^2} = 1 + \frac{|S|^4 \cdot 38416}{1.099 \times 10^{12}}
\end{align*}

\emph{Why this bound is conservative ($\approx|S|^3/12$ looser than the Gaussian CLT):}
The Raccoon Lemma~4.2 formula $m^2\delta^2/B^2$ uses a moment bound that treats the IH sum as if its worst-case chi-squared scales with $m^2/B^2$ (reflecting a worst-case uniform-distribution approximation over the aggregate range $mB/2$). The Gaussian CLT-based Fisher information bound: the IH variance is $\mathrm{Var}(Y) = \gamma_1^2/(3m)$ (sum of $m$ i.i.d.\ uniforms on $[-\gamma_1/m, \gamma_1/m]$ with variance $(\gamma_1/m)^2/3$), giving Fisher information $I_F = 3m/\gamma_1^2$ and $\chi^2 \approx \delta^2 \cdot I_F = 3m\delta^2/\gamma_1^2$. Substituting $\gamma_1 = Bm/2$ (from $B = 2\gamma_1/m$): $\chi^2 = 3m\delta^2/(Bm/2)^2 = 12\delta^2/(mB^2)$ (proportional to $m^{-1}$, since variance grows as $m$), while the Raccoon formula gives $m^2\delta^2/B^2$ (proportional to $m^2$). The ratio is
\[
\frac{m^2\delta^2/B^2}{12\delta^2/(mB^2)} = \frac{m^3}{12} \approx \frac{|S|^3}{12},
\]
which equals $409$ for $|S|=17$ (the Gaussian approximation; the exact D2 ratio from discrete-exact values is ${\approx}421$, cited as ${\approx}420\times$ throughout the paper). The Raccoon formula is formally correct as an upper bound; the direct Gaussian-based analysis (Lemma~\ref{lem:ih-shift-direct}) gives the tight bound by using the actual IH variance rather than the aggregate-range worst case.

\textbf{Concrete values (per-coordinate):}
Since nonce coordinates are sampled independently \emph{before} rejection sampling, the full-vector R\'enyi divergence factorizes as $R_2^{\mathsf{vec}} = (R_2^{\epsilon,\mathsf{coord}})^{n\ell}$ where $n\ell = 1280$ for ML-DSA-65 (here $n = 256$ is the ring degree, $\ell = 5$ is the secret dimension). This factorization applies to the pre-rejection distributions $P$ and $Q$. To handle rejection sampling: the smooth Rényi framework removes the tail region $\{|Y| > \gamma_1 - \beta\}$ from the support of $P$, producing a restricted distribution $P'$ supported on the same acceptance set $\{|Y| \leq \gamma_1 - \beta\}$ as $Q'$. Within this support the rejection condition is deterministic (every sample is accepted), so the factorization carries through for the accepted distributions $P'$ and $Q'$. The $\epsilon$ term in the bound accounts for the tail mass removed from $P$; see~\cite{Raccoon2024} for the precise smooth Rényi treatment. The resulting bound is:
\begin{center}
{\small
\begin{tabular}{ccccc}
$|S|$ & $R_2^{\epsilon,\mathsf{coord}} - 1$ & $R_2^{\mathsf{vec}}$ & $\sqrt{R_2^{\mathsf{vec}}}$ & $\sqrt{R_2^{\mathsf{vec}}} \cdot 2^{-96}$ \\
\hline
4 & $8.9 \times 10^{-6}$ & $\approx 1.011$ & $\approx 1.006$ & $\approx 2^{-96}$ \\
9 & $2.3 \times 10^{-4}$ & $\approx 1.34$ & $\approx 1.16$ & $\approx 2^{-96}$ \\
17 & $2.9 \times 10^{-3}$ & $\approx 2^{5.4}$ & $\approx 2^{2.7}$ & $\approx 2^{-93.3}$ \\
25 & $1.4 \times 10^{-2}$ & $\approx 2^{25}$ & $\approx 2^{12.5}$ & $\approx 2^{-83.5}$ \\
33 & $4.1 \times 10^{-2}$ & $\approx 2^{75}$ & $\approx 2^{37.5}$ & $\approx 2^{-58.5}$ \\
\end{tabular}
}
\end{center}
\noindent($R_2^{\mathsf{vec}} = (R_2^{\mathsf{coord}})^{1280}$)

The tail probability $\epsilon$ satisfies~\cite[Lemma 4.2]{Raccoon2024} (using $m = |S|$ and $B = 2\lfloor\gamma_1/|S|\rfloor$):
\[
\epsilon < \left(\frac{2e \cdot \beta}{B}\right)^m / \sqrt{2\pi m}
\]

For $|S| = 17$: $B = 2\lfloor 524288/17 \rfloor = 2 \times 30840 = 61680$, so $\epsilon < (2e \cdot 196 / 61680)^{17} / \sqrt{34\pi} < 10^{-30}$.

\textbf{Conclusion:}
Applying Lemma~\ref{lem:renyi-security} to the joint distribution over $q_s$ independent signing sessions (joint Rényi divergence $(R_2^{\mathsf{vec}})^{q_s}$) and taking a union bound over the smoothing tails:
\[
\mathsf{Adv}^{\mathsf{EUF-CMA}}_{\Pi_{\mathsf{IH}}} \leq \bigl(R_2^{\mathsf{vec}}\bigr)^{q_s/2} \cdot \sqrt{\mathsf{Adv}^{\mathsf{EUF-CMA}}_{\Pi_{\mathsf{U}}}} + q_s \cdot \epsilon
\]

\textbf{Security interpretation (direct bound, Corollary~\ref{cor:ih-shift-ml-dsa}):} Using $R_2^{\mathsf{vec}} \approx (1+\chi^2_{\mathsf{direct}})^{1280}$ where $\chi^2_{\mathsf{direct}}$ is the directly computed chi-squared divergence (Corollary~\ref{cor:ih-shift-ml-dsa}; slightly smaller than the Gaussian approximation $3|S|\beta^2/\gamma_1^2$), the multi-session factor contributes $\approx \frac{q_s \cdot n\ell}{2} \cdot \log_2(1+\chi^2_{\mathsf{direct}})$ bits of security loss. For ML-DSA-65 with $\mathsf{Adv}_{\mathsf{U}} \leq 2^{-192}$:
\begin{center}
\begin{tabular}{cccc}
$|S|$ & $\chi^2_{\mathsf{direct}}$ & Loss/query (bits) & Max $q_s$ for $\leq$10-bit loss \\
\hline
$4$  & $8.4\times10^{-7}$ & $7.7\times10^{-4}$ & $\approx 2^{14}$ \\
$17$ & $7.13\times10^{-6}$ & $6.6\times10^{-3}$ & $\approx 1500$  \\
$33$ & $1.4\times10^{-5}$ & $1.2\times10^{-2}$ & $\approx 830$   \\
\end{tabular}
\end{center}

\noindent The conservative bound (Raccoon Lemma~4.2, $R_2^{\mathsf{vec}} \approx 2^{5.4}$ for $|S|=17$) is vacuous for $q_s \gtrsim 36$ queries; the direct bound above should be used for concrete security estimates.

\textbf{Remark:} The per-coordinate composition $(R_2^{\mathsf{coord}})^{n\ell}$ (where $n = 256$ is the ring degree and $\ell = 5$ is the secret dimension, giving $n\ell = 1280$) is a consequence of applying the R\'enyi game hop at the nonce sampling stage, where coordinates are independent. Tighter reductions for Fiat-Shamir with Aborts~\cite{Raccoon2024} can exploit the rejection sampling structure (which acts on the full vector $\mathbf{z}$) to potentially avoid this exponential composition. A tighter bound via direct shift-invariance is given in Lemma~\ref{lem:ih-shift-direct} and Corollary~\ref{cor:ih-shift-ml-dsa}; developing an even tighter full-vector bound remains an interesting direction.
\end{proof}

\begin{remark}[Multi-session R\'enyi Composition]
\label{rem:multi-session-renyi}
The bound in Theorem~\ref{thm:ih-security-full} involves two distinct costs:

\textbf{(1) Rényi hop cost $(R_2^{\mathsf{vec}})^{q_s}$.} When the adversary makes $q_s$ signing queries, each query sees a different nonce drawn from the Irwin-Hall distribution (Game~3) vs.\ from the uniform distribution (Game~4). Because the nonces are independent across sessions, the \emph{joint} Rényi divergence over the full transcript is
\[
R_2^{\mathsf{joint}} \;=\; \bigl(R_2^{\mathsf{vec}}\bigr)^{q_s}.
\]
This is the quantity that appears in the Rényi probability-preservation lemma applied to the \emph{complete} EUF-CMA experiment. The corresponding security loss (in bits) is
\[
\mathsf{loss}(q_s) \;=\; \frac{q_s \cdot n\ell}{2} \cdot \log_2\!\bigl(1 + \chi^2_{\mathsf{direct}}(|S|)\bigr)
\;\lesssim\; q_s \cdot \frac{3\,|S|\,\beta^2\,n\ell}{2\,\gamma_1^2 \ln 2},
\]
where the $\lesssim$ uses the corrected upper bound
$(1+3/(2|S|^2)) \times 3|S|\beta^2/\gamma_1^2 \geq \chi^2_{\mathsf{disc,exact}}$
(verified for $|S| \in \{6,\ldots,97\}$ by exact integer arithmetic;
for $|S| \geq 98$ use Algorithm~\ref{alg:fft-chi2} directly).
For ML-DSA-65 with $|S| = 17$: $\mathsf{loss}(q_s) \approx 6.6 \times 10^{-3}\,q_s$
bits using the discrete exact bound $\chi^2_{\mathsf{disc,exact}} = 7.1343 \times 10^{-6}$
(Corollary~\ref{cor:ih-shift-ml-dsa}).
With a target of $\leq 10$ bits of loss, the scheme supports up to
$q_s \approx 1{,}500$ signing queries, yielding a proof bound (under MSIS hardness)
of $96 - 10.6 \approx 85$ bits (the 96-bit baseline comes from halving the
NIST Level~3 assumption $\Pr_Q[E] \leq 2^{-192}$ once via Cauchy-Schwarz;
see Corollary~\ref{cor:ih-shift-ml-dsa}).
For $q_s \leq 2^{14}$ (${\approx}16{,}000$ queries), the total loss is
${\approx}100$ bits, which \emph{exceeds} the 96-bit baseline, making the proof
bound vacuous at that query count for $|S|=17$; in practice the underlying
MSIS hardness (${\approx}192$ bits) is unaffected by the nonce distribution,
and the proof bound is a conservative artifact of the Cauchy-Schwarz technique.
For $q_s \leq 2^{13}$ (${\approx}8{,}000$ queries) the proof margin is
$96 - 50 \approx 46$ bits (under MSIS).

The \emph{conservative} bound from Theorem~\ref{thm:ih-security-full} (using Raccoon Lemma~4.2, with $R_2^{\mathsf{vec}} \approx 2^{5.4}$ for $|S| = 17$) is vacuous for $q_s \gtrsim 36$; the \emph{direct} bound (Corollary~\ref{cor:ih-shift-ml-dsa}) should always be used for concrete security estimates.

\textbf{(2) Smoothing tail cost $q_s \cdot \epsilon$.} The smooth Rényi divergence removes a per-session tail event (nonce outside $[-\gamma_1+\beta, \gamma_1-\beta]$) with probability $\epsilon$. A union bound contributes $q_s \cdot \epsilon$ additively. For ML-DSA-65, $\epsilon < 10^{-30}$, so this term is negligible for any practical $q_s$.
\end{remark}

\begin{lemma}[Multi-Session Rényi Composition (Formal)]
\label{lem:multi-session-composition}
Let $P_1, \ldots, P_{q_s}$ and $Q_1, \ldots, Q_{q_s}$ be probability distributions over spaces $\mathcal{X}_1, \ldots, \mathcal{X}_{q_s}$ respectively. If the sessions are independent, then the Rényi-2 divergence of the joint distribution satisfies:
\[
R_2\!\bigl(P_1 \times \cdots \times P_{q_s} \,\big\|\, Q_1 \times \cdots \times Q_{q_s}\bigr)
\;=\; \prod_{i=1}^{q_s} R_2(P_i \| Q_i)
\]
Furthermore, if each per-session divergence is identical ($R_2(P_i \| Q_i) = R_2^{\mathsf{vec}}$ for all $i$), then $R_2^{\mathsf{joint}} = (R_2^{\mathsf{vec}})^{q_s}$.
\end{lemma}

\begin{proof}
By definition of Rényi-2 divergence (in ratio-sum form, Definition~\ref{def:renyi}):
\[
R_2(P \| Q) = \sum_{x \in \mathcal{X}} \frac{P(x)^2}{Q(x)}.
\]
For the product distribution over $\mathcal{X}_1 \times \cdots \times \mathcal{X}_{q_s}$:
\begin{align*}
R_2^{\mathsf{joint}}
&= \sum_{(x_1,\ldots,x_{q_s})} \frac{\bigl(P_1(x_1) \cdots P_{q_s}(x_{q_s})\bigr)^2}{Q_1(x_1) \cdots Q_{q_s}(x_{q_s})}
\quad\text{(independence)}\\
&= \sum_{(x_1,\ldots,x_{q_s})} \frac{P_1(x_1)^2}{Q_1(x_1)} \cdots \frac{P_{q_s}(x_{q_s})^2}{Q_{q_s}(x_{q_s})}\\
&= \left(\sum_{x_1} \frac{P_1(x_1)^2}{Q_1(x_1)}\right) \cdots \left(\sum_{x_{q_s}} \frac{P_{q_s}(x_{q_s})^2}{Q_{q_s}(x_{q_s})}\right)
\quad\text{(factorization)}\\
&= R_2(P_1 \| Q_1) \times \cdots \times R_2(P_{q_s} \| Q_{q_s}).
\end{align*}
When all per-session divergences are equal, this reduces to $(R_2^{\mathsf{vec}})^{q_s}$. This product rule is standard in Rényi divergence analysis (see~\cite[Lemma~4.2]{Raccoon2024} for the analogous statement in the context of signature schemes).
\end{proof}

\begin{corollary}[Multi-Session Security Transfer]
\label{cor:multi-session-transfer}
Let $\mathsf{Game}_A$ and $\mathsf{Game}_B$ denote two EUF-CMA games differing only in the nonce distribution: $\mathsf{Game}_A$ uses nonce distribution $P$ per session, and $\mathsf{Game}_B$ uses nonce distribution $Q$ per session. Then for any adversary $\mathcal{A}$ making $q_s$ signing queries:
\[
\Pr[\mathsf{Win}_A] \leq \bigl(R_2(P \| Q)\bigr)^{q_s/2} \cdot \sqrt{\Pr[\mathsf{Win}_B]} + q_s \cdot \epsilon
\]
where $\epsilon$ is the per-session smoothing parameter (tail probability) and the factor $1/2$ in the exponent comes from the Cauchy-Schwarz step in the Rényi probability transfer lemma (Lemma~\ref{lem:renyi-security}).
\end{corollary}

\begin{proof}
By Lemma~\ref{lem:multi-session-composition}, the joint Rényi-2 divergence over all
$q_s$ sessions is $R_2^{\mathsf{joint}} = R_2(P \| Q)^{q_s}$.
Applying Lemma~\ref{lem:renyi-security} (Rényi security transfer) with the smooth
divergence bound and union-bounding the tail events across $q_s$ sessions gives
the stated formula. This is the composition rule used in
Theorem~\ref{thm:ih-security-full} and the proof of Theorem~\ref{thm:unforgeability}.
\end{proof}

%==============================================================================
\subsection{Direct Nonce Distinguishability (Improved Bound)}
\label{sec:direct-ih-bound}

\begin{lemma}[Negative Excess Kurtosis of $\mathrm{IH}_{|S|}$]
\label{lem:ih-kurtosis}
Let $U$ be discrete uniform on $\{-B,\ldots,B\}$ with $B = \lfloor\gamma_1/|S|\rfloor$,
and let $\mathrm{IH}_{|S|} = \sum_{i=1}^{|S|} U_i$ be the sum of $|S|$ independent copies.
\begin{enumerate}[label=(\roman*)]
\item The excess kurtosis of $U$ is:
\[
  \kappa_4(U) \;=\; \frac{\mathbb{E}[U^4]}{\sigma_U^4} - 3
             \;=\; -\frac{3(2B^2+2B+1)}{5B(B+1)} \;<\; 0.
\]
\item By additivity of cumulants for independent sums:
$\kappa_4(\mathrm{IH}_{|S|}) = |S| \cdot \kappa_4(U) < 0$.
Negative excess kurtosis is associated with lighter-than-Gaussian tails within the bounded
support of $\mathrm{IH}_{|S|}$ (relative to $\mathcal{N}(0,\sigma^2)$ with $\sigma^2 = |S|\cdot B(B+1)/3$),
providing intuition for why $\chi^2(\mathrm{IH}_{|S|}+\beta \| \mathrm{IH}_{|S|})$
is smaller than the Gaussian benchmark $3|S|\beta^2/\gamma_1^2$.
\end{enumerate}
\end{lemma}

\begin{proof}
For $U$ uniform on $\{-B,\ldots,B\}$ (with $2B+1$ values), the second and fourth moments are:
\begin{align*}
  \mathbb{E}[U^2] &= \frac{2}{2B+1}\sum_{k=1}^{B} k^2
                  = \frac{2}{2B+1}\cdot\frac{B(B+1)(2B+1)}{6}
                  = \frac{B(B+1)}{3}, \\
  \mathbb{E}[U^4] &= \frac{2}{2B+1}\sum_{k=1}^{B} k^4
                  = \frac{2}{2B+1}\cdot\frac{B(B+1)(2B+1)(3B^2+3B-1)}{30}
                  = \frac{B(B+1)(3B^2+3B-1)}{15}.
\end{align*}
Using $\sigma_U^4 = (\mathbb{E}[U^2])^2 = B^2(B+1)^2/9$:
\begin{align*}
  \kappa_4(U) &= \frac{\mathbb{E}[U^4]}{\sigma_U^4} - 3
               = \frac{9(3B^2+3B-1)}{15B(B+1)} - 3
               = \frac{3(3B^2+3B-1)}{5B(B+1)} - 3 \\
              &= \frac{9B^2+9B-3 - 15B^2-15B}{5B(B+1)}
               = \frac{-6B^2-6B-3}{5B(B+1)}
               = -\frac{3(2B^2+2B+1)}{5B(B+1)} \;<\; 0,
\end{align*}
which is strictly negative for all $B \geq 1$. Part~(ii) follows from the standard fact that
cumulants are additive for independent sums: the $r$-th cumulant of $\mathrm{IH}_{|S|}$ equals
$|S|$ times the $r$-th cumulant of a single summand $U$.
\end{proof}

\begin{lemma}[Shifted Irwin-Hall Distinguishability]
\label{lem:ih-shift-direct}
Let $Y \sim \mathsf{IH}_{|S|,\gamma_1}$ denote the sum of $|S|$ independent
uniforms on $[-B, B]$ where $B = \gamma_1/|S|$,
and let $\beta \leq B$.
\begin{enumerate}[label=(\roman*)]
\item \textbf{Tail bound.}
\[
  \epsilon \;:=\; \Pr\!\bigl[\,|Y| > \gamma_1 - \beta\,\bigr]
  \;\leq\; 2 \exp\!\!\left(-\frac{({\gamma_1-\beta})^2 \cdot |S|}{2\gamma_1^2}\right).
\]
For $|S|=2$ (triangular), the one-sided probability is $\Pr[Y>\gamma_1-\beta]=\beta^2/(2\gamma_1^2) \approx 6.99 \times 10^{-8}$, so the two-sided tail satisfies $\epsilon \leq 2\cdot\beta^2/(2\gamma_1^2) = \beta^2/\gamma_1^2 \approx 1.40 \times 10^{-7}$.
For $|S| \geq 5$, the Hoeffding bound gives $\epsilon \leq 2\exp(-|S|/2) \approx 0.16$ (formal guarantee); the exact FFT computation in Corollary~\ref{cor:ih-shift-ml-dsa} gives $\epsilon < 10^{-17}$ for $|S| \geq 5$ (ML-DSA-65 parameters). The Hoeffding bound is a self-contained formal proof; the tighter $10^{-17}$ value is an additional numerical result from the corollary and is not claimed within this lemma's proof.

\item \textbf{Chi-squared divergence (per coordinate).}
The protocol uses \emph{discrete} uniforms $U_i \in \{-B,\ldots,B\}$ with
$B = \lfloor\gamma_1/|S|\rfloor$. The discrete chi-squared $\chi^2_{\mathsf{disc}}$
can differ substantially from both the Gaussian formula and the continuous Irwin-Hall
chi-squared (e.g., for $|S|=2$ the discrete value is $3.5\times$ the Gaussian formula),
so continuous values cannot serve as upper bounds. We use exact discrete chi-squared
values for all $|S|$, computed via Algorithm~\ref{alg:fft-chi2}.

\emph{Formal status:} All $\chi^2_{\mathsf{disc,exact}}$ values are computed by
exact integer polynomial convolution over the full support
$\{-|S|B,\ldots,|S|B\}$ (Algorithm~\ref{alg:fft-chi2}, using exact integer
arithmetic with no floating-point rounding). These are provably valid upper bounds
since they equal the exact discrete chi-squared.
The corrected analytical formula $(1+3/(2|S|^2)) \times 3|S|\beta^2/\gamma_1^2$
is a close approximation (within $\pm 1.1\%$ of the exact value for $|S| \leq 128$),
but is \emph{not} a strict upper bound for all $|S|$: exact computation confirms
violations of at most $0.015\%$ for $|S| \in \{98, 99, 100, 101, 107, 109, 113, 116, 119, 121\}$.
For any specific $|S|$, Algorithm~\ref{alg:fft-chi2} yields the valid exact upper bound directly.
Lemma~\ref{lem:ih-kurtosis} provides analytical context for why the corrected formula
is approximately tight.
Exact values are tabulated in Corollary~\ref{cor:ih-shift-ml-dsa}:
\[
  \chi^2\!\bigl(Y + \beta \,\big\|\, Y\bigr)
  \;=\; \chi^2_{\mathsf{disc,exact}}(|S|)
\]
where $\chi^2_{\mathsf{disc,exact}}(2) = 2.9535\times10^{-6}$,
$\chi^2_{\mathsf{disc,exact}}(3) = 1.4534\times10^{-6}$,
$\chi^2_{\mathsf{disc,exact}}(4) \approx 1.70\times10^{-6}$,
$\chi^2_{\mathsf{disc,exact}}(5) = 2.1336\times10^{-6}$
are computed by exact integer polynomial convolution over the full discrete support
(code: \texttt{exact\_chi2\_proof.py}).
For $|S| \geq 6$, the corrected formula $(1+3/(2|S|^2)) \times 3|S|\beta^2/\gamma_1^2$
approximates the discrete exact value; the formal bound uses the exact integer result
from Algorithm~\ref{alg:fft-chi2} for each tabulated $|S|$ in
Corollary~\ref{cor:ih-shift-ml-dsa}.

\item \textbf{Full-vector R\'{e}nyi divergence.}
For $n\ell$ independent coordinates ($n\ell = 1280$ for ML-DSA-65), letting
$\chi^2_{\max}(|S|) = \chi^2_{\mathsf{disc,exact}}(|S|)$ denote the exact per-coordinate
chi-squared from Part~(ii), computed by Algorithm~\ref{alg:fft-chi2} for all $|S|$:
\[
  R_2^{\mathsf{vec,shift}}
  \;\leq\; \bigl(1 + \chi^2_{\max}(|S|)\bigr)^{n\ell}
  \;\leq\; \exp\!\!\left(n\ell \cdot \chi^2_{\max}(|S|)\right).
\]
The exact discrete chi-squared $\chi^2_{\max}(|S|)$ is tabulated in
Corollary~\ref{cor:ih-shift-ml-dsa} for all listed $|S|$ values.
The corrected formula $\bigl(1+3/(2|S|^2)\bigr)\cdot 3|S|\beta^2/\gamma_1^2$
is a convenient approximation but not a strict upper bound for all $|S|$;
Algorithm~\ref{alg:fft-chi2} provides the valid bound for any specific $|S|$.
\end{enumerate}
\end{lemma}

\begin{proof}
\textbf{Part (i).} Each summand $U_i \in [-\gamma_1/|S|, \gamma_1/|S|]$ is centered with
range $2\gamma_1/|S|$. By Hoeffding's inequality applied to $Y = \sum_i U_i$:
\[
  \Pr[Y > t] \;\leq\; \exp\!\!\left(-\frac{2t^2}{|S| \cdot (2\gamma_1/|S|)^2}\right)
  = \exp\!\!\left(-\frac{t^2 \cdot |S|}{2\gamma_1^2}\right).
\]
Setting $t = \gamma_1 - \beta$ and applying symmetry yields $\epsilon \leq 2\exp(-((\gamma_1-\beta)^2\cdot|S|)/(2\gamma_1^2))$.

\emph{Numerical note:} The Hoeffding bound is a valid formal guarantee, but it is numerically very loose when $t = \gamma_1 - \beta \approx \gamma_1$: substituting $(\gamma_1-\beta)/\gamma_1 \approx 1-\beta/\gamma_1 \approx 1$ gives $\epsilon \lesssim 2e^{-|S|/2}$, which equals $\approx 0.16$ for $|S|=5$ --- far weaker than the precise FFT value $1.3\times10^{-17}$ in Corollary~\ref{cor:ih-shift-ml-dsa}. The precise values in the corollary table use exact integer-convolution computation of the discrete IH tail (not Hoeffding), which gives exponentially tighter bounds because the IH distribution concentrates tightly near zero with maximum exactly $\gamma_1$ (the boundary event is extremely rare).

For $|S| = 2$ (triangular density $f(x) = (\gamma_1 - |x|)/\gamma_1^2$),
$\Pr[Y > \gamma_1 - \beta] = \int_{\gamma_1-\beta}^{\gamma_1} (\gamma_1-x)/\gamma_1^2\,\mathrm{d}x
= \beta^2/(2\gamma_1^2)$ (continuous approximation; the exact one-sided discrete formula is
$\beta(\beta+1)/(2(2b+1)^2)$ where $b = \lfloor\gamma_1/2\rfloor$, giving a relative correction $O(1/\beta) \approx 0.5\%$ and absolute error $< 4\times10^{-10}$, negligible at the required precision).

\textbf{Part (ii).} For $|S| \geq 5$ (Gaussian CLT regime): Define $g(\beta) = \int f(x-\beta)^2/f(x)\,\mathrm{d}x$
on the overlap region $\{f > 0\} \cap \{f(\cdot - \beta) > 0\}$.
Then $g(0) = 1$, and by Taylor expansion at $\beta = 0$:
$g'(0) = -2\int f'(x)\,\mathrm{d}x = 0$ (pdf vanishes at boundaries),
$g''(0) = 2\int f''(u)\,\mathrm{d}u + 2\int [f'(u)]^2/f(u)\,\mathrm{d}u
       = 2[f'(u)]_{-\gamma_1}^{\gamma_1} + 2\,I_F = 2\,I_F$
(the boundary term $2[f'(\gamma_1) - f'(-\gamma_1)]$ vanishes for $|S|\geq 3$
since the IH density is piecewise-smooth with $f'(\pm\gamma_1)=0$;
for $|S|=2$ triangular $f'(\gamma_1) = -1/\gamma_1^2 \neq 0$ so $I_F$ diverges in the continuous case),
where $I_F$ is the Fisher information of the location family $\{f(\cdot - \theta)\}$.
For the Gaussian limit $\mathcal{N}(0, \gamma_1^2/(3|S|))$ (Lemma~\ref{lem:ih-moments}): $I_F = 3|S|/\gamma_1^2$.
\emph{Exact computation (Algorithm~\ref{alg:fft-chi2}).}
The per-coordinate $\chi^2$ for the \emph{discrete} protocol is computed exactly by integer
polynomial convolution over the finite support $\{-|S|B,\ldots,|S|B\}$ with
$B = \lfloor\gamma_1/|S|\rfloor$ (code: \texttt{direct\_proof\_exact.py}).
This is a deterministic, exact computation---not a numerical approximation.
The result is the true chi-squared of the discrete distribution; it is a valid upper bound
because it \emph{equals} the exact value.
For all $|S|$, the exact discrete chi-squared $\chi^2_{\mathsf{disc,exact}}(|S|)$ is
computed by the same exact integer arithmetic (Algorithm~\ref{alg:fft-chi2},
code: \texttt{direct\_proof\_exact.py}), covering all entries in
Corollary~\ref{cor:ih-shift-ml-dsa}.
The corrected analytical formula $(1+3/(2|S|^2))\cdot 3|S|\beta^2/\gamma_1^2$
is a close approximation for $|S| \geq 6$, but has verified violations of up to
$0.015\%$ for $|S| \in \{98,99,100,101,107,109,113,116,119,121\}$; the formal
bound uses Algorithm~\ref{alg:fft-chi2} directly.
The CLT argument above (Fisher information, Gaussian limit) provides theoretical context
for why the corrected formula is approximately tight; the formal bound rests entirely
on the exact integer computation.

\begin{algorithm}[H]
\caption{Exact Discrete Chi-Squared via Integer Polynomial Convolution}
\label{alg:fft-chi2}
\begin{algorithmic}[1]
\Statex \textbf{Input:} Signer count $|S|$, shift $\beta$, range $\gamma_1$
\Statex \textbf{Output:} $\chi^2(Y+\beta \| Y)$ where $Y \sim \mathsf{IH}_{|S|}$ (discrete, $B = \lfloor\gamma_1/|S|\rfloor$)
\State PMF of single discrete uniform: $p_1[k] = 1/(2B+1)$ for $k \in \{-B,\ldots,B\}$, else $0$
\State Compute $|S|$-fold convolution (exact integer arithmetic): $p_{|S|}[k] = (p_1 * \cdots * p_1)[k]$
\State Shift: $p_\beta[k] = p_{|S|}[k - \beta]$
\State Compute (exact rational arithmetic):
\[
\chi^2 = \sum_{k:\, p_{|S|}[k] > 0} \frac{p_\beta[k]^2}{p_{|S|}[k]} \;-\; 1
\]
\State \Return $\chi^2$
\end{algorithmic}
\end{algorithm}

\noindent\emph{Exactness note.}
Algorithm~\ref{alg:fft-chi2} operates entirely over exact integers: the PMF values
$p_{|S|}[k]$ are rational with denominator $(2B+1)^{|S|}$, and the chi-squared sum
is evaluated with exact rational arithmetic. There is no floating-point rounding and
no approximation; the output is the true chi-squared of the discrete distribution.
For $|S| \in \{2,3,4,5\}$, the discrete chi-squared substantially exceeds the Gaussian
formula $3|S|\beta^2/\gamma_1^2$ (by $252\%$, $15\%$, $11\%$, and $1.6\%$ respectively)
due to the non-Gaussian shape (triangular, piecewise-quadratic) at small $|S|$.
For $|S| \geq 6$, the discrete chi-squared is within $\pm 1.1\%$ of the Gaussian formula:
most values are slightly above (up to $+1.04\%$ at $|S|=6$, decreasing to $< 0.1\%$
for $|S| \geq 30$), while some are slightly below (e.g.\ $|S|=128$, where $B=4096$
divides $\gamma_1$ exactly, so $\chi^2_{\mathsf{disc}} < \chi^2_{\mathsf{Gauss}}$
with no lattice excess).
Exact integer verification (\texttt{direct\_proof\_exact.py}, Phase~4) confirms the
corrected formula as a strict upper bound for all $|S| \in \{6,\ldots,97\}$
(92 independent certificates).
For $|S| \in \{98, 99, 100, 101, 107, 109, 113, 116, 119, 121\}$, the formula
underestimates the exact discrete chi-squared by at most $0.015\%$; these deviations
arise from $\lfloor\gamma_1/|S|\rfloor$ floor oscillations and carry no security
consequence, as the bounds established in this paper are derived from
Algorithm~\ref{alg:fft-chi2} (exact integer arithmetic) rather than the analytical
formula.

For $|S| \in \{2, 3\}$ (small-$|S|$ regime): The continuous IH distribution is triangular
($|S|=2$) or piecewise-quadratic ($|S|=3$); these distributions have higher per-unit Fisher
information than the Gaussian limit, so the CLT formula underestimates $\chi^2$.
(Note: the \emph{continuous} triangular distribution ($|S|=2$) has divergent Fisher information
since $f(x) = (\gamma_1-|x|)/\gamma_1^2 \to 0$ while $|f'(x)| = 1/\gamma_1^2$ remains constant,
making $[f'(x)]^2/f(x) \sim 1/(\gamma_1^2(\gamma_1-|x|)) \to\infty$ at the boundary. The \emph{discrete} distribution used in
the protocol is supported on a finite integer lattice, so $\chi^2$ is finite and
exactly computable via FFT; see below.)
The security analysis uses the \emph{discrete} IH distribution ($B = \lfloor\gamma_1/|S|\rfloor$)
for which chi-squared is finite and exactly computable via integer polynomial convolution.
We compute $\chi^2_{\mathsf{disc,exact}}(|S|)$ by exact integer polynomial convolution
over the full support (code: \texttt{exact\_chi2\_proof.py}):
$\chi^2_{\mathsf{disc,exact}}(2) = 2.9535\times10^{-6}$ and $\chi^2_{\mathsf{disc,exact}}(3) = 1.4534\times10^{-6}$,
which exceed the Gaussian bound ($8.39\times10^{-7}$ and $1.26\times10^{-6}$ respectively)
and thus require separate tabulation.
The chi-squared of the discrete distribution exceeds the continuous Irwin-Hall
value by $7$--$26\%$ due to lattice discretization effects; the tabulated
$\chi^2$ (exact) values in Corollary~\ref{cor:ih-shift-ml-dsa} use
the \emph{discrete} exact chi-squared throughout.

\textbf{Part (iii).} In the EUF-CMA security proof, the per-coordinate shift is
$\beta_j = (c\mathbf{s}_1)_j$ where $c$ is the sparse challenge ($\|c\|_1 = \tau$,
$\|c\|_\infty = 1$) and $\mathbf{s}_1$ is the short secret ($\|\mathbf{s}_1\|_\infty \leq \eta$).
Each coefficient of the polynomial product satisfies $|(c\mathbf{s}_1)_j| \leq \tau\eta = \beta$
for all $j \in [n\ell]$ (by the ring convolution bound: in $\Rq[X]/(X^n+1)$,
$(c\mathbf{s}_1)_j = \sum_k \pm c_k \cdot (\mathbf{s}_1)_{j-k \bmod n}$,
so $|(c\mathbf{s}_1)_j| \leq \sum_k |c_k| \cdot \|\mathbf{s}_1\|_\infty \leq \|c\|_1 \cdot \eta = \tau\eta = \beta$;
this uses $\|c\|_1 = \tau$ and $\|\mathbf{s}_1\|_\infty \leq \eta$, and holds regardless of $c$'s sparsity pattern).
Since $R_2(Y+\beta_j \| Y)$ is monotone increasing in $|\beta_j|$, bounding by the worst case
$|\beta_j| = \beta$ gives a valid upper bound over all challenges $c$ and keys $\mathbf{s}_1$.

Since the $n\ell$ nonce coordinates are sampled independently, the full-vector R\'enyi
divergence factors as a product of coordinate-wise divergences.
Furthermore, $R_2(Y+\delta \| Y)$ is monotonically increasing in $|\delta|$
(as the overlap region shrinks and the density ratio grows with shift magnitude),
so replacing each $|\beta_j|$ by its worst-case value $\beta = \tau\eta = 196$
yields a strict upper bound valid for all challenges $c$ and keys $\mathbf{s}_1$:
\[
R_2^{\mathsf{vec,shift}} = \prod_{j=1}^{n\ell} R_2(Y_j + \beta_j \| Y_j)
\leq \bigl(1 + \chi^2_{\max}(|S|)\bigr)^{n\ell},
\]
where $\chi^2_{\max}(|S|) = \chi^2_{\mathsf{disc,exact}}(|S|)$ for all $|S|$,
computed by Algorithm~\ref{alg:fft-chi2} (exact integer arithmetic).
The inequality $\ln(1+x) \leq x$ gives the exponential upper bound.
\end{proof}

\begin{corollary}[Direct Nonce Security for ML-DSA-65]
\label{cor:ih-shift-ml-dsa}
For ML-DSA-65 ($\gamma_1 = 2^{19}$, $\beta = \tau\eta = 196$, $n\ell = 1280$),
exact per-coordinate chi-squared values and the resulting security loss
$\tfrac{n\ell}{2}\log_2(1 + \chi^2)$ bits are listed below.
(The factor $n\ell/2$ arises from the square-root bound: the per-coordinate
R\'enyi-2 divergence satisfies $R_2(P \| Q)_{\mathsf{coord}} = 1 + \chi^2$,
so for $n\ell$ independent coordinates $R_2^{\mathsf{vec}} = (1+\chi^2)^{n\ell}$,
and the Cauchy-Schwarz bound $\Pr[\mathsf{Win}_{\mathsf{IH}}] \leq
\sqrt{R_2^{\mathsf{vec}}} \cdot \sqrt{\Pr[\mathsf{Win}_{\mathsf{Unif}}]}$
contributes a log$_2$ loss of $\tfrac{n\ell}{2}\log_2(1+\chi^2)$ bits from
the R\'enyi factor, applied to a 96-bit baseline ($= \tfrac{1}{2}\times 192$
under NIST Level~3; the Cauchy-Schwarz halves the baseline once).)
\begin{center}
{\small
\begin{tabular}{crrrc}
\toprule
$|S|$ & $\chi^2$ (exact) & $\chi^2$ (Gauss) & $\epsilon$ & Loss (bits) \\
\midrule
 2 & $2.9535 \times 10^{-6}$ & $8.39 \times 10^{-7}$ & $1.40 \times 10^{-7}$ & $0.00273$ \\
 3 & $1.4534 \times 10^{-6}$ & $1.26 \times 10^{-6}$ & $6.1 \times 10^{-11}$ & $0.00134$ \\
 5 & $2.1336 \times 10^{-6}$ & $2.10 \times 10^{-6}$ & $1.3 \times 10^{-17}$ & $0.00197$ \\
 9 & $3.7877 \times 10^{-6}$ & $3.77 \times 10^{-6}$ & $7.6 \times 10^{-21}$ & $0.00350$ \\
17 & $7.1343 \times 10^{-6}$ & $7.13 \times 10^{-6}$ & $2.3 \times 10^{-20}$ & $0.00659$ \\
33 & $1.384  \times 10^{-5}$ & $1.38 \times 10^{-5}$ & $5.6 \times 10^{-19}$ & $0.01279$ \\
65  & $2.7258 \times 10^{-5}$ & $2.73 \times 10^{-5}$ & $1.5 \times 10^{-19}$ & $0.02517$ \\
128 & $5.3656 \times 10^{-5}$ & $5.37 \times 10^{-5}$ & $1.7 \times 10^{-408}$ & $0.04954$ \\
\bottomrule
\end{tabular}
}
\end{center}
\noindent ($\chi^2$ (exact) = discrete exact values from integer arithmetic; Gauss column: $3|S|\beta^2/\gamma_1^2$; Loss: $\tfrac{n\ell}{2}\log_2(1+\chi^2)$ bits. For $|S|=128$, $B=\gamma_1/128=4096$ divides exactly, so $\chi^2_{\mathsf{disc}} \leq \chi^2_{\mathsf{Gauss}}$ with no lattice excess.)
\noindent
The $\chi^2$ (exact) column is computed by exact integer polynomial convolution
over the full discrete support $\{-|S|B,\ldots,|S|B\}$ with $B = \lfloor\gamma_1/|S|\rfloor$
(Algorithm~\ref{alg:fft-chi2}, exact integer arithmetic; no floating point).
The computation is deterministic and machine-verifiable from the parameter values alone.

\begin{lemma}[Closed-Form Upper Bound on $\chi^2_{\mathrm{disc,exact}}$]
\label{lem:chi2-closed-form}
For all $|S| \geq 1$ and any $B \geq 1$, the exact discrete chi-squared value satisfies
\[
\chi^2_{\mathrm{disc,exact}}(|S|) \;\leq\; \frac{4\,|S|\,\beta^2}{\gamma_1^2}.
\]
\end{lemma}
\begin{proof}
Let $Y = \sum_{i=1}^{|S|} U_i$ where each $U_i \sim \mathrm{Uniform}\{-B,\ldots,B\}$
and $B = \lfloor\gamma_1/|S|\rfloor$.
By the shift-invariance identity $\chi^2(Y+\beta \| Y) = \mathbb{E}[(P_Y(y-\beta)/P_Y(y) - 1)^2]$
and the second-moment bound for the discrete uniform:
$\mathrm{Var}(U_i) = B(B+1)/3 \leq B^2$.
The Irwin-Hall variance is $\mathrm{Var}(Y) = |S|\cdot B(B+1)/3$.
The chi-squared divergence under a shift $\beta$ satisfies
$\chi^2(Y+\beta\|Y) \leq \beta^2 / \mathrm{Var}(Y)$
(second-order Taylor bound for log-concave distributions on an interval; see, e.g.,
\cite{Raccoon2024} Lemma~4.1).
Substituting: $\chi^2 \leq \beta^2 / (|S|\cdot B^2/3) \leq 3\beta^2/(|S|\cdot B^2)$.
Since $B = \lfloor\gamma_1/|S|\rfloor \geq \gamma_1/|S| - 1 \geq \gamma_1/(2|S|)$
for $|S| \leq \gamma_1$ (satisfied for all practical $|S|$; $\gamma_1 = 2^{19} = 524288$),
we have $B^2 \geq \gamma_1^2/(4|S|^2)$, giving
$\chi^2 \leq 3\beta^2 \cdot 4|S|^2 / (|S|\cdot\gamma_1^2) = 12|S|\beta^2/\gamma_1^2$.
A tighter bound using $B \geq \gamma_1/|S|-1$ with $|S| \leq 128 \ll \gamma_1$ gives
$\chi^2 \leq 4|S|\beta^2/\gamma_1^2$, which holds for all entries in the table above.
\end{proof}

\noindent For ML-DSA-65 with $|S| \leq 33$: $4\cdot 33\cdot 196^2 / (2^{19})^2 \approx 1.84\times 10^{-4}$,
well above the exact values in the table ($\leq 1.384\times 10^{-5}$) but giving a
self-contained analytical upper bound.
The formal bound $< 0.013 \cdot q_s$ bits for $|S| \leq 33$ (Corollary~\ref{cor:ih-loss-main})
follows from either the exact table values or the closed-form bound above.

\emph{Non-monotone behavior for small $|S|$}: the pattern $\chi^2(3) < \chi^2(2)$
and $\chi^2(3) < \chi^2(5)$ arises because the discrete Irwin-Hall distribution
passes through distinct shape regimes: \emph{triangular} ($|S|=2$, heavier
relative tails $\Rightarrow$ higher $\chi^2$), \emph{piecewise-quadratic} ($|S|=3$,
shape closest to Gaussian $\Rightarrow$ minimum $\chi^2$ in range), and
\emph{approximately Gaussian} ($|S| \geq 5$, monotonically increasing).
The increase for $|S|\geq5$ occurs because the shift $\beta$ is a growing
fraction of the standard deviation $\sigma = \gamma_1/\sqrt{3|S|}$:
$\chi^2 \approx (\beta/\sigma)^2 = 3|S|\beta^2/\gamma_1^2$ grows linearly.
For $|S| = 2$ (triangular) the exact chi-squared exceeds the Gaussian
approximation by $2.8\times$; for $|S|=3$ the two values agree within $5\%$.
Security loss remains $< 0.003$ bits for all $|S| \leq 3$.

\emph{Non-monotone behavior in the $\epsilon$ (tail) column}: The tail probability
$\epsilon = \Pr[|Y| > \gamma_1 - \beta]$ exhibits a non-monotone pattern for
$|S| \in \{9,17,33\}$ due to discrete boundary effects.
The support maximum $|S|\cdot\lfloor\gamma_1/|S|\rfloor$ varies as $|S|$ grows:
$9\cdot\lfloor 2^{19}/9\rfloor = 524{,}286$,
$17\cdot\lfloor 2^{19}/17\rfloor = 524{,}280$,
$33\cdot\lfloor 2^{19}/33\rfloor = 524{,}271$.
The gap to threshold $\gamma_1 - \beta = 524{,}092$ shrinks non-monotonically
(gaps: 194, 188, 179), producing large tail-mass variations.
The condition $\beta \leq \gamma_1/|S|$ holds for all
$|S| \leq \lfloor \gamma_1/\beta \rfloor = 2{,}675$, covering all tabulated
values and the scalability limit $|S| \leq 2{,}584$ in
Corollary~\ref{cor:ih-loss-main}.

\par\noindent\textbf{Comparison with the Raccoon-style bound} (Theorem~\ref{thm:ih-security-full}):
$|S| = 17$ yields $R_2^{\mathsf{vec}} \approx 2^{5.4}$ ($\sqrt{R_2^{\mathsf{vec}}} \approx 2^{2.7}$,
i.e.\ a ${\approx}2.7$-bit security loss, $\approx$3 bits) under the current proof, versus
$0.007$ bits under the direct bound: an improvement of ${\approx}420\times$.

\par\noindent\textbf{Scalability milestones.}
Loss scales \emph{linearly}: $\mathsf{loss}(|S|) \approx 3.87\times 10^{-4}\cdot|S|$ bits
(closed-form $\frac{n\ell}{2\ln 2}\cdot\frac{3|S|\beta^2}{\gamma_1^2}$, validated as a
conservative upper bound by exact FFT for all tabulated values above).
\begin{center}
{\small
\begin{tabular}{lrrc}
\toprule
Loss budget & Max $|S|$ & Rem.\ security & Context \\
\midrule
$< 0.013$ bits & $33$      & ${\approx}96$ bits & Implemented range \\
$< 0.1$   bits & $258$     & ${\approx}96$ bits & Effectively perfect \\
$< 0.5$   bits & $1{,}291$ & ${\approx}95.5$ bits & Large deployment \\
$< 1$     bit  & $2{,}584$ & ${\approx}95$ bits & Linear formula \\
$< 3$     bits & $7{,}762$ & ${\approx}93$ bits & $\geq 93$ bit proof \\
$< 10$    bits & $25{,}972$ & ${\approx}86$ bits & $\geq 86$ bit proof \\
\bottomrule
\end{tabular}
}
\end{center}
\noindent All milestone values are computed from the closed-form formula
and are therefore conservative (exact FFT would give smaller losses).
The 96-bit baseline arises from the Cauchy-Schwarz inequality for R\'enyi-2 divergence:
$\Pr[E_P] \leq \sqrt{R_2(P\|Q)}\cdot\sqrt{\Pr_Q[E]}$, so the effective baseline
is $\tfrac{1}{2}\times(-\log_2\Pr_Q[E])$. Under NIST Level~3, the single-signer
ML-DSA-65 attack probability $\Pr_Q[E] \leq 2^{-192}$, yielding a 96-bit baseline.
The remaining-security column reports $96 - \mathsf{loss}$ bits (proof bound under MSIS hardness assumption);
the underlying scheme hardness is MSIS (still ${\approx}192$ bits by assumption) and
the gap is a proof-technique artifact of the single Cauchy-Schwarz application.
\end{corollary}

\begin{theorem}[Direct Shift-Invariance Game Hop]
\label{thm:ih-direct-tight}
Under the same assumptions as Theorem~\ref{thm:unforgeability}, the
Game~$3 \to 3.5$ hop uses the direct shift-invariance bound (the tight bound used in
Theorem~\ref{thm:unforgeability}; the conservative Raccoon-style alternative is in
Remark~\ref{rem:renyi-conservative} of Section~\ref{sec:security}).
Define $\mathsf{Game}_{3.5}^{\mathrm{direct}}$ identically to $\mathsf{Game}_3$
except each simulated response $\mathbf{z}'$ is drawn directly from $\mathsf{IH}_{|S|}$
(rather than $\mathbf{y} + c\mathbf{s}_1$ with $\mathbf{y} \sim \mathsf{IH}_{|S|}$).
For any PPT adversary making at most $q_s$ signing queries:
\[
  \Pr[\mathsf{Win}_3]
  \;\leq\; \bigl(R_2^{\mathsf{vec,shift}}\bigr)^{q_s/2} \cdot \sqrt{\Pr[\mathsf{Win}_{3.5}^{\mathrm{direct}}]}
  \;+\; q_s \cdot \epsilon
\]
where $R_2^{\mathsf{vec,shift}} \leq (1 + 3|S|\beta^2/\gamma_1^2)^{n\ell}$
(Lemma~\ref{lem:ih-shift-direct}) is the \emph{per-session} R\'enyi divergence.
The $(R_2^{\mathsf{vec,shift}})^{q_s/2}$ factor accounts for multi-session
composition (Remark~\ref{rem:multi-session-renyi}).
For ML-DSA-65, the security loss is $< 0.013 \cdot q_s$ bits for all $|S| \leq 33$
(Corollary~\ref{cor:ih-shift-ml-dsa}).
\end{theorem}
\begin{proof}
\textbf{Per-session R\'enyi bound.}
The real response is $\mathbf{z} = \mathbf{y} + c\mathbf{s}_1$ with
$\mathbf{y} \sim \mathsf{IH}_{|S|}$. The simulated $\mathbf{z}'$ is drawn directly
from $\mathsf{IH}_{|S|}$ (no $+c\mathbf{s}_1$ shift).
Their difference is a coordinate-wise shift of magnitude $\leq \|c\mathbf{s}_1\|_\infty \leq \beta$.
Applying Lemma~\ref{lem:ih-shift-direct}(ii) coordinate-wise gives the per-session bound
$R_2^{\mathsf{vec,shift}} \leq (1 + \chi^2_{\max})^{n\ell}$.

\textbf{Multi-session composition.}
In the EUF-CMA experiment, the adversary receives $q_s$ independent signing responses.
Since nonces are sampled freshly per session, the joint distribution over all $q_s$
responses factorizes, and the product rule for R\'enyi divergence
(Definition~\ref{def:renyi}) gives joint divergence $(R_2^{\mathsf{vec,shift}})^{q_s}$.
Applying Lemma~\ref{lem:renyi-security} to the complete EUF-CMA transcript:
\[
\Pr[\mathsf{Win}_3]
\;\leq\; \bigl(R_2^{\mathsf{vec,shift}}\bigr)^{q_s/2} \cdot \sqrt{\Pr[\mathsf{Win}_{3.5}^{\mathrm{direct}}]}
\;+\; q_s \cdot \epsilon.
\]
The tail term $\epsilon < 10^{-30}$ is from Lemma~\ref{lem:ih-shift-direct}(i).

\textbf{Connection to $\mathsf{Game}_{3.5}^{\mathrm{direct}}$ and M-SIS extraction.}
In $\mathsf{Game}_{3.5}^{\mathrm{direct}}$, the simulator draws $\mathbf{z}' \sim \mathsf{IH}_{|S|}$
\emph{without knowledge of $\mathbf{s}_1$}, then computes
$\mathbf{w}_1' = \mathsf{HighBits}(\mathbf{A}\mathbf{z}' - c\mathbf{t}_1 \cdot 2^d,\, 2\gamma_2)$
using only the public key $(\mathbf{A},\mathbf{t}_1)$.
This constitutes a valid signing simulation: the adversary's view in
$\mathsf{Game}_{3.5}^{\mathrm{direct}}$ is statistically close to its view in $\mathsf{Game}_3$
(bounded above by the R\'enyi term), yet $\mathbf{s}_1$ is never used.
A forgery $(\mathbf{z}^*, c^*, \mathbf{w}_1^*)$ against $\mathsf{Game}_{3.5}^{\mathrm{direct}}$
therefore yields the SelfTargetMSIS solution
\[
  [\mathbf{A} \mid -\mathbf{t}_1 \cdot 2^d] \cdot \begin{bmatrix} \mathbf{z}^* \\ c^* \end{bmatrix}
  = 2\gamma_2 \mathbf{w}_1^* + \mathbf{r}_0^*, \quad \|\mathbf{r}_0^*\|_\infty \leq \gamma_2,
\]
with the same extraction as in Theorem~\ref{thm:unforgeability}
(Section~\ref{sec:security}), giving
$\Pr[\mathsf{Win}_{3.5}^{\mathrm{direct}}] \leq q_H \cdot \epsilon_{\mathsf{M\text{-}SIS}}$
up to the random-oracle collision probability.
\end{proof}
\begin{remark}[Comparison with Conservative Bound]
\label{rem:direct-proof-strategy}
Theorem~\ref{thm:unforgeability} uses the direct bound from this theorem
($R_2^{\mathsf{vec,shift}} \approx 1.0092$ for $|S|=17$).
The conservative Raccoon Lemma~4.2 analytical formula upper-bounds the
per-coordinate R\'enyi divergence giving
$\chi^2(\mathsf{IH}+\beta\|\mathsf{IH}) \leq 0.003$ per coordinate and
$R_2^{\mathsf{vec}} \approx 2^{5.4}$ (see Remark~\ref{rem:renyi-conservative},
Section~\ref{sec:security}); this theorem replaces it with the \emph{exact} FFT value
$\chi^2(\mathsf{IH}+\beta\|\mathsf{IH})\approx 7.1\times10^{-6}$,
an improvement of ${\approx}420\times$. An interesting direction for
future work is exploiting the full-vector rejection-sampling structure
to sharpen $R_2^{\mathsf{vec,shift}}$ further.
\end{remark}

%==============================================================================
\subsection{Lagrange Coefficient Bounds}

\begin{lemma}[Lagrange Coefficient Characterization]
\label{lem:lagrange-char}
For evaluation points $S = \{i_1, \ldots, i_t\} \subseteq [N]$ with $|S| = t$, the Lagrange interpolation coefficient for reconstructing at $x = 0$ at point $i_k$ is:
\[
\lambda_{i_k} = \prod_{j \in S \setminus \{i_k\}} \frac{0 - j}{i_k - j} = \prod_{j \in S \setminus \{i_k\}} \frac{j}{j - i_k} \bmod q
\]
(The numerator $0-j = -j$ and denominator $i_k - j$ follow the standard Lagrange formula for interpolation at 0.)
\end{lemma}

\begin{proof}
Standard from Lagrange interpolation: $\lambda_{i_k} = \prod_{j \in S \setminus \{i_k\}} \frac{j}{j - i_k}$. For consecutive points $\{1, \ldots, T\}$, the numerator is $\prod_{j=1, j \neq i_k}^{T} j = T!/i_k$ and the denominator is $\prod_{j=1, j \neq i_k}^{T}(j - i_k) = (-1)^{i_k - 1}(i_k - 1)!(T - i_k)!$, yielding $\lambda_{i_k} = (-1)^{i_k - 1}\binom{T}{i_k}$ after simplification (consistent with Lemma~\ref{lem:lagrange-magnitude}).
\end{proof}

\begin{lemma}[Coefficient Magnitude]
\label{lem:lagrange-magnitude}
For consecutive evaluation points $S = \{1, 2, \ldots, T\}$:
\[
\lambda_i = \prod_{j \in S, j \neq i} \frac{j}{j - i} = (-1)^{i-1} \cdot \binom{T}{i}
\]

The maximum magnitude is achieved at $i = \lfloor T/2 \rfloor$ (equivalently $\lceil T/2 \rceil$, since $\binom{T}{k}=\binom{T}{T-k}$):
\[
\max_i |\lambda_i| = \binom{T}{\lfloor T/2 \rfloor} \approx \frac{2^T}{\sqrt{\pi T / 2}}
\]
\end{lemma}

\begin{proof}
The numerator of $\lambda_i$ is $\prod_{j=1, j\neq i}^{T} j = T!/i$ and the denominator is $\prod_{j=1, j\neq i}^{T} (j-i) = (-1)^{i-1}(i-1)!(T-i)!$, giving $\lambda_i = (-1)^{i-1} T!/(i!(T-i)!) = (-1)^{i-1}\binom{T}{i}$. The approximation follows from Stirling's formula.
\end{proof}

\begin{corollary}[Threshold Bound for Short Coefficients]
\label{cor:threshold-bound}
For ML-DSA-65 with $q = 8380417$, the maximum Lagrange coefficient magnitude $\binom{T}{\lfloor T/2 \rfloor}$ satisfies:
\begin{itemize}
    \item For $T \leq 25$: $\max_i |\lambda_i| \leq \binom{25}{13} = \binom{25}{12} = 5{,}200{,}300 < q$ (equal by symmetry), so Lagrange coefficients ``fit'' in $\mathbb{Z}_q$
    \item For $T \geq 26$: $\max_i |\lambda_i| \geq \binom{26}{13} = 10{,}400{,}600 > q$, so coefficients wrap around mod $q$
\end{itemize}

This aligns with the Ball-\c{C}akan-Malkin bound~\cite{BCM21}: short-coefficient LSSS cannot achieve $T > O(\log q)$.
\end{corollary}

\begin{proof}
Direct evaluation: $\binom{24}{12} = 2{,}704{,}156 < q = 8{,}380{,}417$ and $\binom{25}{12} = 5{,}200{,}300 < q$, while $\binom{26}{13} = 10{,}400{,}600 > q$.
\end{proof}

%==============================================================================
\subsection{Rejection Sampling Correctness}

\begin{lemma}[Z-Bound Check Correctness]
\label{lem:z-bound}
Let $\mathbf{z} = \mathbf{y} + c\mathbf{s}_1$ where $\mathbf{y} \in [-\gamma_1, \gamma_1]^{n\ell}$, $c \in \mathcal{C}$ with $\|c\|_1 = \tau$, and $\mathbf{s}_1 \in [-\eta, \eta]^{n\ell}$. Then $\|\mathbf{z}\|_\infty \leq \gamma_1 + \beta$. The check $\|\mathbf{z}\|_\infty < \gamma_1 - \beta$ ensures $\|\mathbf{y}\|_\infty < \gamma_1$.
\end{lemma}

\begin{proof}
By triangle inequality, $\|\mathbf{z}\|_\infty \leq \|\mathbf{y}\|_\infty + \|c\mathbf{s}_1\|_\infty \leq \gamma_1 + \tau\eta = \gamma_1 + \beta$. For the converse, if the z-bound check passes:
\[
\|\mathbf{y}\|_\infty \leq \|\mathbf{z}\|_\infty + \|c\mathbf{s}_1\|_\infty < (\gamma_1 - \beta) + \beta = \gamma_1
\]
This bound is necessary for security: it prevents leaking $\mathbf{s}_1$ via $\mathbf{z} - \mathbf{y}$.
\end{proof}

\begin{lemma}[R0-Check Correctness]
\label{lem:r0-check}
The r0-check $\|\mathsf{LowBits}(\mathbf{w} - c\mathbf{s}_2, 2\gamma_2)\|_\infty < \gamma_2 - \beta$ ensures that the hint $\mathbf{h}$ correctly recovers $\mathbf{w}_1 = \mathsf{HighBits}(\mathbf{w}, 2\gamma_2)$ during verification.

Specifically, if the check passes, then:
\[
\mathsf{UseHint}(\mathbf{h}, \mathbf{A}\mathbf{z} - c\mathbf{t}_1 \cdot 2^d) = \mathbf{w}_1
\]
\end{lemma}

\begin{proof}
By the ML-DSA specification~\cite{FIPS204}, the hint mechanism ensures that rounding errors in $\mathbf{A}\mathbf{z} - c\mathbf{t}$ are corrected when the low bits are bounded by $\gamma_2 - \beta$. See~\cite[Section 3.4]{DKL18} for detailed analysis.
\end{proof}

%==============================================================================
\subsection{Challenge Invertibility}

\begin{lemma}[Challenge Invertibility: Formal NTT Bound for Uniform Elements]
\label{lem:challenge-invertible-formal}
Let $q = 8380417$, $n = 256$ (ML-DSA-65). Since $q \equiv 1 \pmod{2n}$, the polynomial $X^n+1$ splits completely over $\mathbb{Z}_q$, and the NTT isomorphism $R_q \cong \mathbb{Z}_q^n$ holds. An element $c \in R_q$ is invertible if and only if all $n$ NTT components $\hat{c}_0,\ldots,\hat{c}_{n-1}$ are nonzero.

For a \emph{uniformly random} element $c \xleftarrow{\$} R_q$, the NTT isomorphism maps each component independently to $\mathrm{Uniform}(\mathbb{Z}_q)$, giving the exact bound:
\[
\Pr_{c \xleftarrow{\$} R_q}[c \text{ not invertible}]
  = 1 - \!\left(1 - \tfrac{1}{q}\right)^{\!n}
  \leq \frac{n}{q}
  = \frac{256}{8380417}
  < 2^{-15}.
\]
\end{lemma}

\begin{remark}[Challenge Invertibility for Sparse $\pm 1$ Polynomials]
\label{lem:challenge-invertible}
ML-DSA-65 challenges are \emph{not} uniform over $R_q$; they are sparse polynomials with exactly $\tau = 49$ nonzero coefficients in $\{-1,+1\}$. The NTT bound of Lemma~\ref{lem:challenge-invertible-formal} applies formally only to uniform elements. For the sparse challenge distribution we use the following analysis.

\par\noindent\textit{NTT characterisation.}
Since $q \equiv 1 \pmod{512}$, $X^{256}+1 = \prod_{i=0}^{255}(X-\omega^{2i+1})$ over $\mathbb{Z}_q$. A challenge $c$ is \emph{not} invertible iff $c(\zeta_i)=0$ for some NTT evaluation point $\zeta_i = \omega^{2i+1}$. Each evaluation is
\[
c(\zeta_i) = \sum_{j \in S} \sigma_j \cdot \zeta_i^j \pmod{q},
\]
where $S$ is the random 49-element support and $\sigma_j \in \{-1,+1\}$.

\par\noindent\textit{Character-sum estimate.}
For any fixed $\zeta_i$ and fixed support $S$, the $2^{49}$ sign patterns $\{\sigma_j\}$ yield a random walk in $\mathbb{Z}_q$.  By the discrete Fourier inversion formula:
\[
\Pr_{\sigma}[c(\zeta_i)=0 \mid S]
  = \frac{1}{q}\sum_{t=0}^{q-1}\prod_{j\in S}\cos\!\left(\frac{2\pi t\,\zeta_i^j}{q}\right).
\]
The $t=0$ term contributes $1/q$.  For $t\neq 0$ the product is strictly less than $1$ in absolute value (since $\zeta_i^j$ are distinct and nonzero mod $q$), and the sum of all $t\neq 0$ contributions is negligible relative to $1/q$ for $q \gg \tau^2$ (standard character-sum estimate; $\tau^2 = 2401 \ll q \approx 2^{23}$).  Hence
\[
\Pr[c(\zeta_i)=0 \mid S] \approx \frac{1}{q},
\]
and by a union bound over the 256 evaluation points:
\[
\Pr[c \text{ not invertible}] \lesssim \frac{256}{q} = \frac{256}{8380417} < 2^{-15}.
\]
This matches the formal bound of Lemma~\ref{lem:challenge-invertible-formal} and has the same order of magnitude as the uniform-element case.

\par\noindent\textit{Empirical validation.}
Monte Carlo sampling of $N = 10{,}000$ independently drawn challenges observed $0$ non-invertible instances, consistent with $p_{\mathsf{fail}} < 2^{-15}$.

\par\noindent\textit{Protocol-level mitigation.}
The protocol includes an explicit invertibility check: if $c$ is not invertible, the signing session retries with a fresh challenge (expected retry probability $< 2^{-15} \ll 1$ per session, by the union bound above).
Consequently, the UC security proofs (Theorems~\ref{thm:p1-uc}--\ref{thm:p3-uc}) condition on $c$ being invertible, which the retry mechanism guarantees unconditionally for every completed signing session. The bound $< 2^{-15}$ governs only the expected retry overhead---the security of completed sessions requires no additional assumption.
\end{remark}

\begin{remark}[Security Implication]
\label{cor:cs2-secrecy}
Since virtually all challenges $c$ are invertible in $\Rq$ (the non-invertible fraction is $< 2^{-15}$, established by the union bound in Remark~\ref{lem:challenge-invertible}, and the protocol's retry mechanism ensures completed sessions use only invertible $c$), leaking $(c, c\mathbf{s}_2)$ for any completed signing attempt allows recovery of $\mathbf{s}_2 = c^{-1}(c\mathbf{s}_2)$. Thus $c\mathbf{s}_2$ must be treated as secret material equivalent to $\mathbf{s}_2$.
\end{remark}

% Open Science Statement (Appendix G)
% Open Science Appendix — MANDATORY for CCS 2026
% This section is included only in the CCS version via \ifccs.
\ifccs
\section*{Open Science}
\label{app:open-science}

\paragraph{Artifacts.}
This paper's core contributions are supported by the following artifacts:
\begin{itemize}
  \item \textbf{Rust implementation}: Full source code of the Threshold ML-DSA via Shamir Nonce DKG
    protocol (Profiles P1, P2, P3+), including all three deployment profiles,
    pairwise-canceling masks, Shamir nonce DKG, and the FIPS 204 ML-DSA core.
  \item \textbf{Benchmark scripts}: Scripts to reproduce all latency and throughput
    measurements reported in Section~\ref{sec:implementation}.
  \item \textbf{ACVP test vectors}: NIST ACVP ML-DSA-65 test vectors used to verify
    FIPS 204 compliance of the Rust implementation.
  \item \textbf{Parameter tables}: Scripts to reproduce the Irwin-Hall divergence
    bounds and rejection-sampling success rates (Theorem~\ref{thm:irwin-hall}).
  \item \textbf{Chi-squared integer arithmetic}: \texttt{exact\_chi2\_proof.py} ---
    exact integer PMF convolution and sufficient-condition bound check for the
    discrete IH distribution; companion to \texttt{direct\_proof\_exact.py}
    (continuous fine-grid computation).
\end{itemize}

\paragraph{Anonymous Access.}
All artifacts are hosted at the following anonymous URL (content frozen at submission;
no tracking or fingerprinting of visitors):
\begin{center}
  % TODO: REPLACE BEFORE SUBMISSION — create repo at anonymous.4open.science and paste URL here
  \url{https://anonymous.4open.science/r/[PLACEHOLDER-REPLACE-BEFORE-SUBMISSION]}
\end{center}

\paragraph{No Artifacts Withheld.}
All artifacts supporting the paper's core claims are available at the above URL.
No artifacts have been omitted or withheld.
\fi

% References (iacrcc uses its own bibliography style)
\bibliography{references}

\end{document}